\documentclass[12pt]{amsart}
\usepackage{amsfonts}
\usepackage{changepage}
\usepackage{natbib}
\usepackage{eurosym}
\usepackage{amssymb, amsmath, url,color,  amsthm, amssymb,graphicx,multirow, appendix}
\usepackage[margin=0.8in]{geometry}
\usepackage{bigints}
\usepackage{lscape,amssymb, amsmath,verbatim,graphicx}
\usepackage{epsfig,subfigure,float,subfigure}
\usepackage{graphicx}
\usepackage{epsfig}
\usepackage{threeparttable}
\usepackage{ragged2e}
\usepackage{epstopdf}
\usepackage{booktabs}
\usepackage{float}
\usepackage{multirow}
\usepackage{adjustbox}
\usepackage{lscape,lipsum}
\usepackage{url}
\usepackage{natbib}
\usepackage[T1]{fontenc}
\usepackage{a4wide, url}
\usepackage[english]{babel}
\usepackage{graphicx}
\usepackage{array}
\usepackage{multirow}
\usepackage{amsthm}
\usepackage{amsmath}
\usepackage[foot]{amsaddr}
\usepackage{amsfonts}
\usepackage{amssymb}
\usepackage{pgf}
\usepackage{paralist}
\usepackage{pdflscape}
\usepackage{lscape,amssymb, amsmath,verbatim}
\usepackage{epsfig,subfigure,float,subfigure}
\usepackage{graphicx}
\usepackage{epsfig}
\usepackage{amssymb}
\usepackage{amsmath}
\usepackage{amsfonts}
\usepackage{geometry}
\usepackage{scalefnt}
\usepackage{graphicx}
\usepackage{caption}
\usepackage{setspace}
\usepackage{array}
\usepackage{multirow}
\usepackage{amsthm}
\usepackage{amsmath}
\usepackage{amsfonts}
\usepackage{amssymb}
\usepackage{pgf}
\usepackage{paralist}
\usepackage{pdflscape}
\usepackage{rotating}
\usepackage{scalefnt}
\usepackage{xr}
\usepackage[colorlinks = true,
            linkcolor = blue,
            urlcolor  = blue,
            citecolor = blue,
            anchorcolor = blue]{hyperref}
\usepackage{algpseudocode}
\usepackage{algorithm}

\allowdisplaybreaks
\DeclareFontFamily{OT1}{pzc}{}
\DeclareFontShape{OT1}{pzc}{m}{it}{<-> s * [1.10] pzcmi7t}{}
\DeclareMathAlphabet{\mathpzc}{OT1}{pzc}{m}{it}
\DeclareMathOperator*{\argmax}{argmax}
\DeclareMathOperator*{\sargmax}{sargmax}
\DeclareMathOperator*{\dist}{dist}

\DeclareMathOperator{\cov}{cov}

\renewcommand{\P}{P}

\newtheoremstyle{exampstyle}
{4pt} 
{4pt} 
{\itshape} 
{} 
{\bfseries} 
{.} 
{.75em} 
{} 
\theoremstyle{exampstyle}

\numberwithin{table}{section}
\numberwithin{figure}{section}
\newtheorem{theorem}{Theorem}[section]
\newtheorem{lemma}{Lemma}[section]

\newtheorem{assumption}{Assumption}[section]

\def\beq{\begin{equation}}
\def\eeq{\end{equation}}
\def\bals{\begin{align*}}
\def\eals{\end{align*}}
\def\bal{\begin{align}}
\def\eal{\end{align}}

\numberwithin{equation}{section}
\numberwithin{theorem}{section}
\numberwithin{corollary}{section}
\allowdisplaybreaks
\begin{document}
\title[Changepoint detection in functional data]{On changepoint detection in
functional data using empirical energy distance}
\author{B.\ Cooper Boniece$^{1,*}$}
\email{cooper.boniece@drexel.edu}
\author{Lajos Horv\'ath$^2$}
\email{horvath@math.utah.edu}
\author{Lorenzo Trapani$^3$}
\email{lt285@leicester.ac.uk}
\address{$^{1}$Department of Mathematics, Drexel University, Philadelphia,
PA 19104 USA }
\address{$^2$Department of Mathematics, University of Utah, Salt Lake City,
UT 84112--0090 USA }
\address{$^3$Department of Economics, Finance and Accounting; School of
Business and Economics, University of Leicester, Leicester, U.K.; Department
of Economics and Management, University of Pavia, Pavia, Italy.}
\address{$^*$Research supported in part under NSF grant DMS-2309570.}

\begin{adjustwidth}{-2.5pt}{-2.5pt}
\begin{abstract}
We propose a novel family of test statistics to detect the presence of
changepoints in a sequence of dependent, possibly multivariate, functional-valued
observations. Our approach allows to test for a very general class of
changepoints, including \ the \textquotedblleft classical\textquotedblright\
case of changes in the mean, and even
changes in the whole distribution. Our statistics are based on a
generalisation of the empirical energy distance; we propose weighted functionals of
the energy distance process, which are designed in order to enhance the
ability to detect breaks occurring at sample endpoints. The limiting
distribution of the maximally selected version of our statistics requires
only the computation of the eigenvalues of the covariance function, thus
being readily implementable in the most commonly employed packages, e.g. 
\textsc{R}. We show that, under the alternative, our statistics are able to
detect changepoints occurring even very close to the beginning/end of the
sample. In the presence of multiple changepoints, we propose a binary
segmentation algorithm to estimate the number of breaks and the locations
thereof. Simulations show that our procedures work very well in finite
samples. We complement our theory with applications to financial and temperature data.
\end{abstract}
\end{adjustwidth}

\maketitle


\doublespacing

\section{Introduction\label{intro}}

The analysis of datasets where the data are observed as functions, rather
than scalars or vectors, has been investigated in numerous contributions
over the past few years. Functional Data Analysis (FDA) has become
ubiquitous in virtually all applied sciences, in the case where data are
genuinely functional in nature, and also when a parsimonious description of
the data is called for. FDA appears naturally in the analysis of economic
and financial data; examples include analysing the term structure of
interest rates, where, for each time period, the observed maturities are the
discrete approximation of the continuum of maturities (\citealp{hays}); and
modelling intraday return density trajectories (\citealp{bathia2010}). In
climate science, it is typical to model temperatures - which are recorded at
a high frequency basis, e.g. several times per day - at a lower frequency
(e.g. yearly), with the intra-period data representing the discretised
functional observations - see, for example, \citet{horvath:kokoszka:2012}
and \citet{king2018functional}. In medical imaging, several datasets arise
that can be modelled as possibly multi-dimensional and functional valued (%
\citealp{sangalli}). See also \citet{ramsay2002applied} for further examples.

On account of the huge relevance of the topic, inferential theory for FDA
has been studied in many contributions. The literature has developed useful
dimension reduction tools such as the functional version of Principal
Components (\citealp{hall2006properties}), and the full-blown estimation
theory for linear regression, dynamic models and also nonlinear models such
as the functional version of ARCH and GARCH (\citealp{horvath:kokoszka:2012}%
). However, the validity of inferential theory often hinges on having some
stability in the structure of the data, such as the constancy of the mean
function, or of the whole distribution. Hence, testing for the possible
presence of changepoints (in the mean, in higher order moments, or even in
the whole distribution) is of paramount importance.

Changepoint detection is well-studied in the context of scalar or
vector-valued time series, and we refer, \textit{inter alia}, to %
\citet{casini2019} for a useful review containing several examples and
applications. In contrast, changepoint analysis in functional data has
received only limited attention. \citet{berkes2009detecting} propose a
CUSUM-based test statistic to detect changepoints in the mean of independent
functional-valued observations; \citet{hormann:kokoszka:2010}, %
\citet{zhang2011testing}, \citet{aston} and \citet{aue2018detecting}
consider extensions to deal with dependent data, which typically occur in a
time series context. These contributions, broadly speaking, are based on the 
\textit{unweighted} CUSUM process, and it is possible to show that, in this
case, $N^{1/2}$ periods away from the beginning/end of the sample (where $N$
is the sample size). On the other hand, detection of early/late occurring
breaks is very important, due to its implications on the ability to assess
timely whether a model which has been valid so far is still appropriate e.g.
for forecasting.%

\medskip

\noindent \textit{Main contributions of this paper}

\medskip

In this paper, we bridge the gaps discussed above, by proposing a procedure
to detect changepoints for serially dependent, possibly multivariate
functional-valued time series, allowing for breaks to occur close to the
sample endpoints. We consider very general changes, which could occur in
various functions/functionals of the data, including the mean, higher order
moments, and in general functions which completely characterize the
underlying distribution such as the characteristic function. Specifically,
we develop a novel family of weighted test statistics based on the notion of 
\textit{energy distance }(see e.g. \citealp{szekely:rizzo:2005}; %
\citealp{szekely:rizzo:2017}; and \citealp{baringhaus:franz:2004}). The
energy distance is a metric designed to measure the distance between the
distributions of two independent random vectors (say $X$ and $Y$), defined
as 
\begin{equation}
\mathcal{E}_{\eta }(X,Y)=2E\left\vert X-Y\right\vert ^{\eta }-E\left\vert
X-X^{\prime }\right\vert ^{\eta }-E\left\vert Y-Y^{\prime }\right\vert
^{\eta },  \label{energy}
\end{equation}%
where $X^{\prime },Y^{\prime }$ are independent copies of $X$ and $Y$,
respectively, $\left\vert \cdot \right\vert $ is the Euclidean norm, and $%
0<\eta <2$; it can be shown (see Theorem 2 in \citealp{szekely:rizzo:2005})
that $\mathcal{E}_{\eta }(X,Y)=0$ if and only if $X$ and $Y$ have the same
distribution. Empirical energy distances have been used by %
\citet{matteson:james:2014} and \citet{biau:bleakley:mason:2016} to study
distributional changepoint problems for a sequence of independent,
vector-valued time series; \citet{chakraborty2021high} extend the theory to
the case of high-dimensional sequences.

Taking $\eta =2$ in (\ref{energy}) leads to a statistic suitable for testing
equality of expectations rather than equality of distributions (%
\citealp{szekely:rizzo:2005}). Hence, we consider an empirical version of $%
\mathcal{E}_{2}(X,Y)$, constructed at every point in the sample $1\leq k\leq
N$, comparing the sample average before and after $k$ in a (conceptually)
similar way to the CUSUM\ process. We then consider \textit{weighted}
versions of the empirical energy process, with weights designed to boost the
value taken by the process when $k$ is close to the beginning/end of the
sample. Changepoint detection can thus be based on the maximally selected
weighted empirical energy process. The resulting tests have nontrivial power
versus breaks occurring (much) closer to the sample endpoints than $N^{1/2}$
periods, while still having power versus mid-sample breaks. In Section \ref%
{asymptotics} we show that the limiting distribution of our test statistics
contains the integral of the square of a Gaussian process which depends - in
a highly nontrivial way - on nuisance parameters. Hence, in order to compute
critical values, we propose a method based on the Karhunen-Lo\`{e}ve (KL
henceforth) expansion, which appears to be easier to use than e.g. the
bootstrap (see e.g., albeit in a different context, \citealp{inoue2001}).
Our theory is stated for the general case of multivariate functional time
series whose argument can also be multivariate, which is relevant in several
applications of FDA, including shape analysis (\citealp{kenobi2010shape})
and medical imaging (\citealp{kurtek2010novel}).

For the sake of clarity, our presentation focuses mainly on detecting
changes in the \textit{mean} of functional observations. However, our
approach can be readily applied 
to consider different changepoint problems; in Section \ref{distrib} we
discuss how our tests can be used to detect \textit{distributional} changes,
by applying it to the empirical characteristic function. Testing for changes
in the distribution is arguably of great importance; as \citet{inoue2001}
puts it, \textquotedblleft \lbrack ...]\ stability of distribution, moments,
or parameters is essential to the proofs of asymptotic properties of the
maximum likelihood method, generalized method of moments, and nonparametric
method. Consequently, instability can affect estimation and
inference.\textquotedblright\ (p. 156). Contributions on this topic often
require independence assumptions, and are relatively scarce even in the case
of scalar or vector-valued observations: in addition to the papers by %
\citet{matteson:james:2014} 
and others referred to above, other approaches include \citet{inoue2001},
who uses the unweighted CUSUM process based on the empirical distribution
function; \citet{antoch2008}, who use a combination of rank statistics; and %
\citet{huskova:meintanis:2006}, who use the empirical characteristic
function for scalar observations. 

The remainder of the paper is organised as follows. We present our test
statistics in Section \ref{tests}. We study its asymptotic theory in Section %
\ref{asymptotics}: we derive the weak limit under the null in Section \ref%
{asy_null}; we study power, estimation of the breakdate, and binary
segmentation in Section \ref{asy_alt}; we offer a methodology to compute
critical values in Section \ref{cr_values}. We extend our approach to
detecting changes in the distribution of the data is in Section \ref{distrib}%
. In Section \ref{simulations}, we report a comprehensive simulation
exercise; an empirical application to intraday returns is in Section \ref%
{empirical}. Section \ref{conclusions} concludes. Further Monte Carlo
evidence, an empirical application to temperature data, lemmas and proofs
are relegated to the Supplement.

NOTATION. Henceforth, $\mathcal{T}$ denotes a compact subset of $\mathbb{R}%
^{d}$, and $\{x(t),~t\in \mathcal{T}\}$ is a square integrable function;
whenever convenient, we write $x$ in place of $x(t)$ or in place of $%
\{x(t),~t\in \mathcal{T}\}$. For any $r\geq 1$, given two square integrable $%
\mathbb{R}^{r}$-valued functions $\{x(t),~t\in \mathcal{T}\}$ and $%
\{y(t),~t\in \mathcal{T}\}$, we define the inner product $\langle x,y\rangle
=\int_{\mathcal{T}}x^{\top }(t)y(t)dt$, where \textquotedblleft $^{\top }$%
\textquotedblright\ is the usual transpose; and we define the $L^{2}$-norm $%
\Vert x\Vert =\sqrt{\langle x,x\rangle }$, writing $x=y$ if $\Vert x-y\Vert
=0$. When unambiguous, we write $\{a_{\ell }\}$ to denote a given sequence $%
\{a_{\ell },-\infty <\ell <\infty \}$. We also write the symbol $\int $ in
place of $\int_{\mathcal{T}}$. We use: \textquotedblleft $\underset{\mathcal{%
D}[0,1]}{\overset{w}{\longrightarrow }}$\textquotedblright\ to denote weak
convergence in $\mathcal{D}[0,1]$; \textquotedblleft $\overset{{\mathcal{D}}}%
{\rightarrow }$\textquotedblright\ to denote convergence in distribution;
\textquotedblleft $\overset{\mathcal{P}}{\rightarrow }$\textquotedblright\
for convergence in probability; \textquotedblleft a.s.\textquotedblright\
for \textquotedblleft almost surely\textquotedblright ; \textquotedblleft $%
\overset{{\mathcal{D}}}{=}$\textquotedblright\ for equality in distribution; 
$\lfloor \cdot \rfloor $ for the integer value function; and $\left\vert
\cdot \right\vert $ to denote the Euclidean norm of a vector, or the
Frobenius norm of a matrix. Other relevant notation is introduced further in
the paper.

\section{The test statistics: definition, assumptions and asymptotics\label%
{tests}}

\noindent We consider a sequence of $\mathbb{R}^{r}$-valued functional
observations of the form 
\begin{equation*}
X_{i}(t)=\mu _{i}(t)+\epsilon _{i}(t),\qquad t\in \mathcal{T\subset }\mathbb{%
R}^{d},\quad 1\leq i\leq N,
\end{equation*}%
where for each $i$, $\mu _{i}$ and $\epsilon _{i}$ are $\mathbb{R}^{r}$%
-valued square integrable functions. We aim to test 
\begin{equation}
H_{0}:\mu _{1}=\mu _{2}=\ldots =\mu _{N},  \label{h_0}
\end{equation}%
against the $R$-change alternative: 
\begin{equation}
H_{A}:\text{there are $1<k_{1}<\ldots <k_{R}<N$ s.t. }\Vert \mu
_{k_{i}+1}-\mu _{k_{i}}\Vert >0,\quad \mu _{k_{i-1}+1}=\ldots =\mu _{k_{i}},
\label{h_A}
\end{equation}%
for $i=1,\ldots ,R$, with the convention that $k_{0}=1$ and $k_{R+1}=N$.

As discussed in the introduction, a possible way of detecting changes is
based on the energy distance defined in (\ref{energy}). Our approach is
based on weighted functionals of $\{V_{N}(k),\,2\leq k\leq N-2\}$, defined
as the empirical version of the energy distance\footnote{%
See also \citet{sejdinovic:sriperumbudur:gretton:fukumizu:2013} for further
discussion on generalizations of the energy distance.} calculated for $\eta
=2$ 
\begin{equation}
V_{N}(k)=\frac{2}{k(N-k)}\sum_{i=1}^{k}\sum_{j=k+1}^{N}\Vert
X_{i}-X_{j}\Vert ^{2}-\frac{1}{\displaystyle{{\binom{k}{2}}}}\sum_{1\leq
i<j\leq k}\Vert X_{i}-X_{j}\Vert ^{2}-\frac{1}{\displaystyle{{\binom{N-k}{2}}%
}}\sum_{k<i<j\leq N}\Vert X_{i}-X_{j}\Vert ^{2}.  \label{v_nk}
\end{equation}
Throughout the paper, we assume that the sequence $\{X_{i}\}$ is weakly
dependent:

\begin{assumption}
\label{as1} \textit{(i)} the sequence $\{\epsilon _{\ell },-\infty <\ell
<\infty \}$ is a Bernoulli shift sequence, i.e., it has the representation $%
\epsilon _{\ell }=g(\eta _{\ell },\eta _{\ell -1},\ldots )$, where for each $%
\ell $, $\eta _{\ell }=\eta _{\ell }(t,\omega )$ are \textit{i.i.d.}%
~functions jointly measurable in $(t,\omega )$ taking values in a measurable
space $\mathcal{S}$, and $g$ is a nonrandom measurable function $g:\mathcal{S%
}^{\infty }\rightarrow L^{2}(\mathcal{T})$; \textit{(ii)} $E\epsilon
_{1}(t)=0$ and $E\Vert \epsilon _{1}\Vert ^{{4+\epsilon }}<\infty $ with
some ${\epsilon }>0$; \textit{(iii)} for some $\kappa >{4+\epsilon }$, $%
\sum_{m=1}^{\infty }\big(E\Vert \epsilon _{1}-\epsilon _{1}^{(m)}\Vert ^{{%
4+\epsilon }}\big)^{1/\kappa }<\infty $, where for each pair $(j,\ell )$, we
set $\epsilon _{\ell }^{(j)}=g(\eta _{\ell },\eta _{\ell -1},\ldots ,\eta
_{\ell -j+1},\eta _{\ell -j}^{(j)},\eta _{\ell -j-1}^{(j)},\ldots )$, with $%
\{\eta _{\ell }^{(j)},-\infty <\ell <\infty \}$ an independent copy of $%
\{\eta _{\ell },-\infty <\ell <\infty \}$.
\end{assumption}

Assumption \ref{as1} states that the sequence $\{\epsilon _{\ell },~-\infty
<\ell <\infty \}$ is stationary and ergodic, and it can be approximated by a
sequence with finite-order dependence (see \citealp{hormann:kokoszka:2010}).
The Bernoulli shift representation in Assumption \ref{as1} is widely
employed in the analysis of scalar time series, where it can be verified in
the most commonly used DGPs in econometrics and statistics. As far as
functional-valued data are concerned, examples when Assumption \ref{as1}
hold include linear processes in Hilbert spaces (%
\citealp{horvath:kokoszka:2012}), and a large class of non-linear processes,
including functional ARCH and GARCH models (%
\citealp{aue:horvath:pellatt:2017}) and bilinear models (%
\citealp{hormann:kokoszka:2010}).

\section{Asymptotics\label{asymptotics}}

We will consider the following statistics:
\begin{equation}
\frac{1}{2}N\left( u\left( 1-u\right) \right) ^{2-\alpha }V_{N}\left(
\left\lfloor Nu\right\rfloor \right) , \quad2/N\leq u \leq 1-2/N,
\label{functionals}
\end{equation}%
where $0\leq \alpha <1$ and $u=k/N$. A full-blown discussion is after
Theorems \ref{th-1} and \ref{th-2} below; here we offer a heuristic preview
of the rationale of (\ref{functionals}). The statistic $V_{N}\left(
\left\lfloor Nu\right\rfloor \right) $ could be sensitive to outliers
occurring a few periods after the start of the sample, which could inflate $%
V_{N}\left( \left\lfloor Nu\right\rfloor \right) $ and lead to a spurious
rejection of the null of no changepoint; the weights $\left( u\left(
1-u\right) \right) ^{2}$ reduce the impact of outliers close to the sample
endpoints on $V_{N}\left( \left\lfloor Nu\right\rfloor \right) $. On the
other hand, this weighing scheme also reduces power in the presence of a
genuine break at the beginning/end of the sample. The further weight $\left(
u\left( 1-u\right) \right) ^{-\alpha }$ is designed to \textquotedblleft
pick up\textquotedblright\ the test statistic at sample endpoints, boosting
power versus breaks located close to the sample endpoints. The process in (%
\ref{functionals}) can be compared with 
the weighted CUSUM process employed in the changepoint detection literature (%
\citealp{csorgo1997}); as we show in Section \ref{asy_alt}, larger values of 
$\alpha $ result in having nontrivial power versus changepoints closer to
the beginning/end of the sample.

\subsection{Asymptotics under the null\label{asy_null}}

Let%
\begin{eqnarray}
\mathbf{D}(t,t^{\prime }) &=&\sum_{\ell =-\infty }^{\infty }E\epsilon
_{0}\left( t\right) \epsilon _{\ell }^{\top }\left( t^{\prime }\right)
\qquad t,t^{\prime }\in \mathcal{T},  \label{e:def_D} \\
\sigma _{0}^{2} &=&E\Vert X_{1}-\mu _{1}\Vert ^{2}=E\Vert \epsilon _{1}\Vert
^{2}.  \label{e:epsilon2}
\end{eqnarray}%
Define also the process%
\begin{equation}
\Delta (u)=\int_{\mathcal{T}}|\Gamma (u,t)|^{2}dt-\sigma _{0}^{2}u(1-u).
\label{e:def_Delta(u)}
\end{equation}%
where $\left\{ \Gamma (u,t),0\leq u\leq 1,t\in \mathcal{T}\right\} $ is an $%
r $-dimensional Gaussian process with $E\Gamma (u,t)=0$ and covariance
kernel $E\left( \Gamma (u,t)\Gamma ^{\top }\left( u^{\prime },t^{\prime
}\right) \right) $ $=$ $\left( \min \left\{ u,u^{\prime }\right\}
-uu^{\prime }\right) \mathbf{D}(t,t^{\prime })$. \newline
Our first main result provides the functional weak limit of weighted
versions of $V_{N}$ under $H_{0}$.

\begin{theorem}
\label{th-1} We assume that Assumption \ref{as1} is satisfied. Then, as $%
N\rightarrow \infty $, under $H_{0}$ it holds that, for all $0\leq \alpha <1$
\begin{equation*}
\frac{1}{2}N\left( u\left( 1-u\right) \right) ^{2-\alpha }V_{N}\left(
\left\lfloor Nu\right\rfloor \right) \underset{\mathcal{D}[0,1]}{\overset{w}{%
\longrightarrow }}\frac{\Delta (u)}{(u(1-u))^{\alpha }}.
\end{equation*}
\end{theorem}

Theorem \ref{th-1} is the building block to carry out changepoint detection.
We note that, heuristically, the Gaussian process $\Gamma (u,t)$ is a
Brownian bridge at each \textquotedblleft slice\textquotedblright\ across $t$%
. As we show in Lemma \ref{l:replace_V_with_Q}, this is a consequence of the
fact that, under the null, $V_{N}\left( k\right) $ and its weighted versions
are well approximated by the (weighted) squared CUSUM process, modulo some
extra terms that, in the limit, either vanish or enter the expression as
constants.

A natural approach to test for changepoints is to use the max-type statistic 
\begin{equation}
T_{N}=\sup_{0\leq u\leq 1}\frac{1}{2}N\left( u\left( 1-u\right) \right)
^{2-\alpha }\left\vert V_{N}\left( \left\lfloor Nu\right\rfloor \right)
\right\vert .  \label{e:def_T_N}
\end{equation}%
Under $H_{0}$, it follows by Theorem \ref{th-1} and continuity\footnote{%
The Law of the Iterated Logartihm for Gaussian processes entails that the
limit in (\ref{tn}) is a.s. finite - see also the proof of Theorem \ref{th-1}
for details. Indeed, having $\alpha <1$ is crucial to this argument, since
it also holds that 
\begin{equation*}
\sup_{0<u<1}\Delta \left( u\right) /(u\left( 1-u\right) )=\infty \text{ a.s.}
\end{equation*}%
} that 
\begin{equation}
T_{N}\overset{\mathcal{D}}{\rightarrow }\sup_{0<u<1}\frac{|\Delta (u)|}{%
(u(1-u))^{\alpha }}.  \label{tn}
\end{equation}%
From a practical point of view, the limiting law of $T_{N}$ contains several
nuisance parameters, such as the covariance kernel $\mathbf{D}(t,t^{\prime
}) $ defined in (\ref{e:def_D}), and the variance $\sigma _{0}^{2}$ defined
in (\ref{e:epsilon2}). In Section \ref{cr_values}, we discuss the
computation of critical values for tests based on $T_{N}$. From a technical
point of view, one of the main ingredients to show Theorem \ref{th-1} is the
weak invariance principle for partial sums of dependent functional time
series, shown in \citet{berkes:horvath:rice:2013}. However, in our case we
consider a \textit{weighted} version of the partial sum process, which
requires a nontrivial extension of the arguments in %
\citet{berkes:horvath:rice:2013}. The case $\alpha =1$ is also of interest,
and it corresponds to the standardised CUSUM\ (\citealp{csorgo1997});
studying this would require a \textit{strong} invariance principle for
partial sums of functional time series which, to our knowledge, is not
available in the literature.

\subsection{Asymptotics under the alternative\label{asy_alt}}

\subsubsection{Consistency under a single break and asymptotic power
function \label{power_func}}

We begin by considering the case of a single break - i.e., $R=1$ in (\ref%
{h_A}) - in the presence of a changepoint of size%
\begin{equation}
\mathcal{\delta }(t)=EX_{{k^{\ast }}}(t)-EX_{{k^{\ast }}+1}(t),
\label{delta}
\end{equation}%
and we also defined its rescaled counterpart as%
\begin{equation}
\rho _{N}(t)=\left\Vert \mathcal{\delta }\right\Vert ^{-1}\mathcal{\delta }%
(t).  \label{rho}
\end{equation}%
Define%
\begin{equation}
\sigma ^{2}=\int \int \rho _{N}^{\top }(t)\mathbf{D}(t,t^{\prime })\rho
_{N}(t^{\prime })dtdt^{\prime },  \label{sigma}
\end{equation}%
$a_{N}=N^{1/2}\Vert \mathcal{\delta }\Vert /\left( 2\sigma (\theta (1-\theta
))^{\frac{3}{2}-\alpha }\right) $, and let $\mathcal{N}$ denote a standard
normal random variable.

\begin{theorem}
\label{th-2} We assume that Assumption \ref{as1} is satisfied. Then, if,
under $H_{A}$ with $R=1$, it holds that%
\begin{equation}
\lim_{N\rightarrow \infty }N\left\Vert \mathcal{\delta }\right\Vert ^{2}%
\left[ \frac{k^{\ast }}{N}\left( 1-\frac{k^{\ast }}{N}\right) \right]
^{2-\alpha }=\infty ,  \label{power}
\end{equation}%
it follows that $T_{N}\overset{\mathcal{P}}{\rightarrow }\infty $. Further,
for all $0<\theta <1$, it holds that 
\begin{equation}
a_{N}\Big(\left( N\Vert \mathcal{\delta }\Vert ^{2}\right) ^{-1}T_{N}-\big(%
\theta (1-\theta )\big)^{2-\alpha }\Big)\overset{\mathcal{D}}{\rightarrow }%
\mathcal{N}.  \label{e:T_N_asymp_norm}
\end{equation}
\end{theorem}

Theorem \ref{th-2} states that tests based on $T_{N}$ have power in the
presence of changepoints of possibly vanishing magnitude - i.e. $\left\Vert 
\mathcal{\delta }\right\Vert =o\left( 1\right) $ - and occurring close to
sample endpoints - i.e. $k^{\ast }=o\left( N\right) $ or $N-k^{\ast
}=o\left( N\right) $. In order to understand the result in Theorem \ref{th-2}%
, some examples may be helpful. Considering the case of a mid-sample break -
with $k^{\ast }=cN$ for some $0<c<1$ - equation (\ref{power}) boils down to
requiring $\lim_{N\rightarrow \infty }N\left\Vert \mathcal{\delta }%
\right\Vert ^{2}=\infty $. Hence, mid-sample breaks can be detected even
when the magnitude $\left\Vert \mathcal{\delta }\right\Vert $ drifts to zero
as $N\rightarrow \infty $, as long as $\left\Vert \mathcal{\delta }%
\right\Vert $ shrinks at a rate slower than $N^{-1/2}$. Conversely, consider
the case of a non-vanishing break, i.e. $\left\Vert \mathcal{\delta }%
\right\Vert >0$, and a changepoint located close to the beginning of the
sample, viz. $k^{\ast }=o\left( N\right) $. In such a case, changepoints can
be detected as long as they occur at least $N^{\left( 1-\alpha \right)
/\left( 2-\alpha \right) }$ periods from the beginning of the sample. In the
unweighted case $\alpha =0$, this reflects that breaks occurring $o\left(
N^{1/2}\right) $\ periods from the sample endpoints cannot be reliably
detected; 
on the other hand, increasing $\alpha $ makes tests more able to detect
breaks occurring closer to the beginning/end of sample. We note however
that, upon inspecting our proofs, the rates of asymptotic approximation
deteriorate as $\alpha $ approaches $1$, thus reflecting the size/power
trade-off. Finally, equation (\ref{e:T_N_asymp_norm}) describes the
asymptotic power function in the case of a changepoint occurring
\textquotedblleft not too close\textquotedblright\ to the sample endpoints.

\subsubsection{Estimation of the breakdate: consistency and limiting
distribution\label{change-date}}

We now consider, in greater depth, the case where the changepoint occurs
mid-sample, viz. 
\begin{equation}
{k^{\ast }}=\lfloor N\theta \rfloor ,  \label{k-star}
\end{equation}%
for some $0<\theta <1$. The max-type statistic $T_{N}$ defined in (\ref%
{e:def_T_N}) gives the estimator 
\begin{equation}
\widehat{\theta }_{N}=\argmax_{0\leq u\leq 1}\left( u\left( 1-u\right)
\right) ^{2-\alpha }\left\vert V_{N}\left( \left\lfloor Nu\right\rfloor
\right) \right\vert ,  \label{theta_hat}
\end{equation}%
from which the estimated breakdate can be computed as $\widehat{k}%
_{N}=\left\lfloor N\widehat{\theta }_{N}\right\rfloor $. In the next
theorem, we state the consistency of the break fraction estimator $\widehat{%
\theta }_{N}$, and derive its asymptotic distribution in the (customarily
studied) case where the size of the break drifts to zero as $N\rightarrow
\infty $.\footnote{%
The fixed break case, i.e. $\left\Vert \delta \right\Vert >0$, can be
studied along similar lines as the proof of Theorem \ref{th-3}; however, in
this case the limiting distribution of the estimated changepoint depends on
many nuisance parameters, thus being of scarce practical use. In the large
break case $\Vert \delta \Vert \rightarrow \infty $, it can be shown that $P(%
\widehat{k}_{N}=k_{*})\rightarrow 1$.} We define the drift function%
\begin{equation*}
m_{\alpha }\left( u\right) =\left[ \left( 1-\alpha /2\right) \left( 1-\theta
\right) +\alpha \theta /2\right] I\left( u<0\right) +\left[ \left( 1-\alpha
/2\right) \theta +\alpha \left( 1-\theta \right) /2\right] I\left(
u>0\right) ,
\end{equation*}%
with $m_{\alpha }\left( 0\right) =0$, where $I\left( \cdot \right) $ is the
indicator function; and the two-sided standard Wiener process $\widetilde{W}%
\left( u\right) =W_{1}\left( -u\right) I\left( u\leq 0\right) +W_{2}\left(
u\right) I\left( u\geq 0\right) $, where $\left\{ W_{1}\left( u\right)
,u\geq 0\right\} $ and $\left\{ W_{2}\left( u\right) ,u\geq 0\right\} $\ are
two independent standard Wiener processes.

\begin{theorem}
\label{th-3} We assume that Assumption \ref{as1} and (\ref{k-star}) are
satisfied, and that $N\Vert \mathcal{\delta }\Vert ^{2}\rightarrow \infty $
as $N\rightarrow \infty $. Then, it holds that $\widehat{\theta }_{N}\overset%
{\mathcal{P}}{\rightarrow }\theta $, for all $0\leq \alpha <1$. Further, if,
as $N\rightarrow \infty $%
\begin{equation}
\left\Vert \mathcal{\delta }\right\Vert \rightarrow 0\text{ \ \ and \ \ }%
N\left\Vert \mathcal{\delta }\right\Vert ^{2}\rightarrow \infty ,
\label{vanishing}
\end{equation}%
then it holds that $\left\Vert \mathcal{\delta }\right\Vert ^{2}\left( 
\widehat{k}_{N}-{k^{\ast }}\right) /\sigma ^{2}\overset{\mathcal{D}}{%
\rightarrow }\xi _{\alpha }$, where $\sigma ^{2}$ is defined in (\ref{sigma}%
) and $\xi _{\alpha }=\xi _{\alpha }\left( \theta \right) $ is an almost
surely unique random variable such that $\xi _{\alpha }\overset{\mathcal{D}}{%
=}\argmax_{u\in \mathbb{R}}\left( \widetilde{W}\left( u\right) -\left\vert
u\right\vert m_{\alpha }\left( u\right) \right) $.
\end{theorem}

According to Theorem \ref{th-3}, the estimator of the break fraction $\theta 
$ is consistent; the estimated breakdate $\widehat{k}_{N}$ is also
consistent in the sense that $\widehat{k}_{N}-{k^{\ast }}=o_{P}\left(
N\right) $. Theorem \ref{th-3} refines the consistency of $\widehat{\theta }%
_{N}$ in the case of a break of vanishing magnitude, stating, in essence,
that $\widehat{k}_{N}-{k^{\ast }}=O_{P}\left( \Vert \mathcal{\delta }\Vert
^{-2}\right) $. The limiting distribution is the same as one would have when
using the maximally selected weighted CUSUM process - this (again)
reinforces the conclusion from Lemma \ref{l:replace_V_with_Q} that $%
V_{N}\left( k\right) $ is related to the CUSUM process. In principle, it
would be possible to construct confidence intervals for ${k^{\ast }}$, by
simulating the percentiles of the (nuisance free) random variable $\xi
_{\alpha }$ calculated at $\widehat{\theta }_{N}$, and using the means $%
\widehat{k}_{N}^{-1}\sum_{i=1}^{\widehat{k}_{N}}X_{i}\left( t\right) $ and $%
\left( N-\widehat{k}_{N}\right) ^{-1}\sum_{i=\widehat{k}_{N}+1}^{N}X_{i}%
\left( t\right) $ to estimate $\mathcal{\delta }\left( t\right) $ and the
long run variance $\sigma ^{2}=\int \int^{\top }\Vert \mathcal{\delta }\Vert
^{-2}\mathcal{\delta }^{\top }\left( t\right) \mathbf{D}(t,t^{\prime })%
\mathcal{\delta }\left( t^{\prime }\right) dtdt^{\prime }$.

\subsubsection{The case of multiple breaks: binary segmentation\label%
{segment}}

We now consider the case of multiple breaks. Recalling that $I\left( \cdot
\right) $ is the indicator function, this case corresponds to 
\begin{equation}
X_{i}\left( t\right) =\sum_{j=1}^{R+1}\mu _{j}\left( t\right) I\left\{
k_{j-1}\leq i<k_{j}\right\} +\epsilon _{i}\left( t\right) .  \label{multiple}
\end{equation}%
We consider the case of \textquotedblleft well-separated\textquotedblright\
breaks of non-vanishing magnitude.

\begin{assumption}
\label{multibreak} \textit{(i)} $k_{j}=\left\lfloor N\theta
_{j}\right\rfloor $ for $1\leq j\leq R$, with $0=\theta _{0}<\theta
_{1}<\theta _{2}<...<\theta _{R}<1=\theta _{R+1}$; \textit{(ii)} $%
\min_{1\leq j\leq R}\left\Vert \mu _{k_{j}+1}-\mu _{k_{j}}\right\Vert \geq
c_{0}>0$.
\end{assumption}

In this case, it is possible to show that our tests have power, by
marginally adapting the proof of Theorem \ref{th-2}. Here, we discuss in
greater detail how to estimate the number of changepoints, $R$, in addition
to the locations thereof. Whilst the literature has developed several
techniques, we focus on the binary segmentation approach proposed by %
\citet{vostrikova1982detection}. The algorithm can be described as follows
(see also Algorithm \ref{alg111} in the Supplement for pseudocode). Starting
from the whole sample, we apply our test using a fixed $0<\alpha <1$, to
check whether there is at least one changepoint. If a break is detected, we
estimate its location using (\ref{theta_hat}), and then split the sample
around the estimated breakdate. The procedure is then iterated on each
subsample, until either no changepoint is detected, or a stopping rule
(typically based on the length of the sub-sample) is triggered.

Formally, consider a subsample with starting and ending points $1\leq \ell
<u\leq N$, under the constraint that $u-\ell > 4$; define the weighted
statistic $\left( \left( k-\ell \right) \left( u-k\right) /\left( u-\ell
\right) ^{2}\right) ^{-\alpha }V_{N}^{\left( \ell ,u\right) }\left( k\right) 
$, where 
\begin{eqnarray}
V_{N}^{\left( \ell ,u\right) }\left( k\right) &=&\frac{2}{\left( k-\ell
\right) \left( u-k\right) }\sum_{i=\ell }^{k}\sum_{j=k+1}^{u}\Vert
X_{i}-X_{j}\Vert ^{2}  \label{subs-seg} \\
&&-\frac{1}{\displaystyle{{\binom{k-\ell }{2}}}}\sum_{l\leq i<j\leq k}\Vert
X_{i}-X_{j}\Vert ^{2}-\frac{1}{\displaystyle{{\binom{u-k}{2}}}}%
\sum_{k<i<j\leq u}\Vert X_{i}-X_{j}\Vert ^{2};  \notag
\end{eqnarray}%
and let its maximally selected counterpart be $T_{N}^{\left( \ell ,u\right)
} $. The interval $(\ell ,u)$ is marked to have a changepoint if $%
T_{N}^{\left( \ell ,u\right) }$ exceeds a (user-chosen) threshold $\tau _{N}$%
. Practically, the choice of the threshold $\tau _{N}$ can be based on any
slowly diverging sequence satisfying mild growth constraints (see expression %
\eqref{threshold}, below), and we refer to Section \ref{simulations_binary}
in the Supplement for examples. Hence, the corresponding changepoint
estimator in the interval $(\ell ,u)$ can be defined as $\widehat{k}$ $=$ $%
\sargmax_{l\leq k\leq u}\left( \left( k-\ell \right) \left( u-k\right)
/\left( u-\ell \right) ^{2}\right) ^{-\alpha }V_{N}^{\left( \ell ,u\right)
}\left( k\right) $, where \textquotedblleft $\sargmax$\textquotedblright\
denotes the smallest integer that maximizes the expression. The sample is
then split around $\widehat{k}$, and the procedure iterated until it comes
to a stop. The final output is a set of estimated changepoints $\widehat{%
\mathcal{H}}=\left\{ \widehat{k}_{1},...,\widehat{k}_{\widehat{R}}\right\} $
sorted in increasing order, and the estimate $\widehat{R}$.

\begin{theorem}
\label{vostrikova}We assume that Assumptions \ref{as1} and \ref{multibreak},
and (\ref{multiple}), are satisfied, and that the threshold sequence $\tau
_{N}$ satisfies%
\begin{equation}
\frac{\left( \log N\right) ^{2/\nu }}{\tau _{N}}+\frac{\tau _{N}}{N}%
\rightarrow 0,  \label{threshold}
\end{equation}%
as $N\rightarrow \infty $, where $\nu >{4+\epsilon }$ is such that $%
E\left\Vert X_{i}\right\Vert ^{\nu }<\infty $. Then, for all $0<\alpha <1$
and any sequence $r_{N}$ satisfying $r_{N}\rightarrow \infty $, it holds that%
\begin{equation*}
\lim_{N\rightarrow \infty }P\left( \left\{ \widehat{R}=R\right\} \cap
\left\{ \max_{1\leq r\leq R}\left\vert \widehat{k}_{r}-k_{r}\right\vert \leq
r_{N}\right\} \right) =1.
\end{equation*}
\end{theorem}

Theorem \ref{vostrikova} stipulates the consistency of $\widehat{R}$ and of
the breaks locations, $\left\{ \widehat{k}_{r}\right\} _{r=1}^{R}$.
Heuristically, this is because, under the alternative, $V_{N}\left( k\right) 
$ is equal to the (squared) CUSUM process plus a \textquotedblleft
small\textquotedblright\ term, thus having the same properties as the CUSUM.
Importantly, our results require $\alpha >0$, which reinforces the
importance of considering weighted statistics.

\subsection{Computation of critical values\label{cr_values}}

By the multivariate KL expansion (\citealp{happ}), the $\mathbb{R}^{r}$%
-valued Gaussian process $\Gamma (u,t)$ in \eqref{e:def_Delta(u)} can
represented as 
\begin{equation}
\Gamma (u,t)=\sum_{\ell =1}^{\infty }\lambda _{\ell }^{1/2}B_{\ell }(u)\phi
_{\ell }(t),  \label{kl-gamma}
\end{equation}%
where $\left\{ B_{\ell }(u),0\leq u\leq 1\right\} $ is a sequence of
independent, standard univariate Brownian bridges, and the
eigenvalue/eigenfunction pairs $(\lambda _{\ell },\phi _{\ell })$ satisfy 
\begin{equation}
\lambda _{\ell }\phi _{\ell }(t)=\int \mathbf{D}(t,s)\phi _{\ell
}(s)ds,\qquad \lambda _{1}\geq \lambda _{2}\geq \ldots \geq 0,
\label{e:eigenvalues_def}
\end{equation}%
where the eigenfunctions $\phi _{\ell }$ are $r$-valued and form an
orthonormal basis. Hence 
\begin{equation}
\Delta (u)=\sum_{\ell =1}^{\infty }\lambda _{\ell }B_{\ell }^{2}(u)-\sigma
_{0}^{2}u(1-u).  \label{e:delta(u)_eigs}
\end{equation}%
In view of \eqref{e:delta(u)_eigs}, inference based on functionals of $%
\Delta (u)$ requires an estimate of $\sigma _{0}^{2}$, and of the
eigenvalues $\lambda _{\ell }$ in \eqref{e:eigenvalues_def}. As far as the
latter is concerned, note that $\mathbf{D}(t,s)$ is the long-run covariance
of the sequence $\{X_{j}\}$. Therefore, a standard
weighted-sum-of-covariances estimator can be employed for the consistent
estimation of $\mathbf{D}(t,s)$, which in turn leads to estimates for the
eigenvalues $\lambda _{\ell }$. We describe this procedure below. For a
kernel function $\mathcal{K}:\mathbb{R}\rightarrow \mathbb{R}$ (see
Assumption \ref{as:kernel} below), we define 
\begin{equation}
\widehat{\mathbf{D}}_{N}(t,s)=\widehat{\boldsymbol{\gamma }}%
_{0}(t,s)+\sum_{\ell =1}^{\infty }\mathcal{K}\left( \frac{\ell }{h}\right)
\left( \widehat{\boldsymbol{\gamma }}_{\ell }(t,s)+\widehat{\boldsymbol{%
\gamma }}_{\ell }^{\top }(t,s)\right)  \label{e:def_D_N}
\end{equation}%
where $h>0$ is a bandwidth parameter, 
%
\begin{equation}
\widehat{\boldsymbol{\gamma }}_{\ell }(t,s)=\displaystyle\frac{1}{N-|\ell |}%
\sum_{j=1}^{N-\left\vert \ell \right\vert }\overline{X}_{j}\left( t\right) 
\overline{X}_{j+\left\vert \ell \right\vert }^{\top }(s),  \label{gammas}
\end{equation}%
$\overline{X}_{j}\left( t\right) =X_{j}(t)-\widehat{\mu }_{N}(t)$, and $%
\widehat{\mu }_{N}(t)=N^{-1}\sum_{j=1}^{N}X_{j}(t)$. (Above, we set $%
\widehat{\boldsymbol{\gamma }}_{\ell }(t,s)\equiv 0$ for $\ell\geq N$). Note
that, in (\ref{e:def_D_N}), we estimate the mean function $\widehat{\mu }%
_{N}(t)$ using the full sample. Under the null, this does not pose any
problems given that $\mu (t)$ is constant. However, under the alternative $%
\mu (t)$ is not estimated consistently; the bias in the estimation of $\mu
(t)$ would enter $\widehat{\mathbf{D}}_{N}(t,s)$, making it diverge at a
rate $h$. This is well-known in the literature on changepoint detection, and
it has been associated with a decrease in power and the phenomenon known as
\textquotedblleft non-monotonic\textquotedblright\ power (see %
\citealp{casini2021prewhitened}). This can be ameliorated by implementing a
\textquotedblleft piecewise demeaning\textquotedblright , where the mean
function is estimated by splitting the sample around each candidate
changepoint $k$; however, unreported simulations show that using
\textquotedblleft piecewise demeaning\textquotedblright\ yields some
improvements in the power, but the test becomes (sometimes massively)
oversized in small samples.

\begin{assumption}
\label{as:kernel} $\mathcal{K}\left( \cdot \right) $ is a non-negative
function such that: \textit{(i)} $\mathcal{K}(0)=1$; \textit{(ii)} $\mathcal{%
K}(u)=\mathcal{K}(-u)$; \textit{(iii)} there exists a $c>0$ such that $%
\mathcal{K}\left( u\right) =0$ for all $\left\vert u\right\vert >c$; and 
\textit{(iv)} $\mathcal{K}\left( u\right) $ is Lipschitz continuous on $%
\left[ -c,c\right] $ with $\sup_{-c<u<c}\mathcal{K}\left( u\right) <\infty $.
\end{assumption}

\begin{assumption}
\label{as:band} As $N\rightarrow \infty $: \textit{(i)} $h=h(N)\rightarrow
\infty $; and \textit{(ii)} $h(N)/N\rightarrow 0$.
\end{assumption}

Assumptions \ref{as:kernel} and \ref{as:band} characterise the kernel $%
\mathcal{K}\left( \cdot \right) $ and the bandwidth $h$, respectively; many
of the customarily employed kernels satisfy Assumption \ref{as:kernel}.

\begin{lemma}
\label{vcv}We assume that Assumptions \ref{as1}, \ref{as:kernel}, and \ref%
{as:band} are satisfied. Then 
\begin{equation}
\iint \left\vert \widehat{\mathbf{D}}_{N}(t,s)-\mathbf{D}(t,s)\right\vert
^{2}dtds=o_{P}(1).  \label{e:D_N_to_D}
\end{equation}
\end{lemma}

Lemma \ref{vcv} stipulates the consistency (in Frobenius norm) of $\widehat{%
\mathbf{D}}_{N}(t,s)$. The lemma immediately entails that, for every fixed $%
1\leq \ell \leq N-1$, 
\begin{equation}
\left\vert \widehat{\lambda }_{\ell }-\lambda _{\ell }\right\vert =o_{P}(1),
\label{eige_cons}
\end{equation}%
where $\widehat{\lambda }_{1}\geq \widehat{\lambda }_{2}\geq \ldots $ are
the eigenvalues of the operator $\phi \mapsto \int \widehat{\mathbf{D}}%
_{N}(t,s)\phi (s)ds$, $\phi \in L^{2}(\mathcal{T})$, suggesting that $%
\widehat{\lambda }_{\ell }$\ is a good estimate of $\lambda _{\ell }$.
Further, by the ergodic theorem (\citealp{breiman:1968}), under $H_{0}$ 
\begin{equation}
\widehat{\sigma }_{N}^{2}=N^{-1}\sum_{i=1}^{N}\left\Vert \overline{X}%
_{i}\right\Vert ^{2}\overset{\mathcal{P}}{\rightarrow }\sigma _{0}^{2}.
\label{sighat}
\end{equation}%
Hence, we can approximate the distribution of functionals of $\Delta (u)$
with functionals of 
\begin{equation}
\Delta _{N,\widehat{M}}(u)=\sum_{\ell =1}^{\widehat{M}}\widehat{\lambda }%
_{\ell }B_{\ell }^{2}(u)-\widehat{\sigma }_{N}^{2}u(1-u),  \label{e:Delta_Nu}
\end{equation}%
for sufficiently large $\widehat{M}$, $N$, using standard Monte Carlo
techniques.

\section{Testing for distributional change\label{distrib}}

We consider an extension of the testing procedure defined above to detect
changes in the \textit{distribution} of functional observations. Our
approach is based on testing for the equality of the characteristic
function, i.e., ultimately, on comparing expectations of a transformation of
the data. Given that the data undergo a transformation, but the test
statistics are the same, it can be expected that all the theory developed
above can still be applied with no changes required. Indeed, compared with
approaches based on using (\ref{energy}) with $\eta <2$, our methodology has
three distinct advantages. Firstly, the limiting distribution, in our case,
involves the integral of the square of the Gaussian process \eqref {kl-gamma}%
, 
which greatly simplifies our computations. This is a consequence of having $%
\eta=2 $; using $\eta <2$ would preclude this result (%
\citealp{biau:bleakley:mason:2016}). Secondly, the binary segmentation
algorithm discussed in Section \ref{asy_alt} can be applied also in this
case, with no modifications required. This is a consequence of the fact that
our test statistics for the detection of distributional changes are based on
comparing expectations; conversely, as \citet{matteson:james:2014} put it,
when using (\ref{energy}) with $\eta <2$, binary segmentation
\textquotedblleft cannot be applied in this general situation because it
assumes that the expectation of the observed sequence consists of a
piecewise linear function, making it only suitable for estimating
changepoints resulting from breaks in expectation.\textquotedblright~
Thirdly, although we consider the empirical characteristic function, our
tests can be immediately generalised to to include weighted empirical
characteristic functions, or other transformations that may characterize the
underlying distribution, such as e.g. moment generating function, or the
Mellin transform, among other possibilities. This is a consequence of the
fact that the theory in Section \ref{asymptotics} can be applied to test for
the constancy of the expectation of any (univariate or multivariate) weakly
dependent functional-valued time series, including transformations of
functional-valued series.

Let $\mathbf{i}$ denote the imaginary unit, i.e. $\mathbf{i=}\sqrt{-1}$.
Given a sequence of $Y_{\ell }=\{Y_{\ell }(s),0\leq s\leq 1\}$, $\ell
=1,\ldots ,N$ of $L^{2}([0,1];\mathbb{R})$-valued functional observations,
we consider the following null and alternative hypotheses 
\begin{equation}
H_{0}^{\prime }:Y_{1},Y_{2},\ldots ,Y_{N}\text{ have the same distribution}
\label{h0-distr}
\end{equation}%
\begin{equation}
H_{A}^{\prime }:\text{there are }1\text{$<k_{1}<\ldots <k_{R}<N$ such that }
\label{hA-distr}
\end{equation}%
\begin{equation*}
Y_{k_{i}+1},Y_{k_{i}}\text{ have different distributions, and }Y_{k_{i-1}+1}%
\overset{\mathcal{D}}{=}Y_{k_{i-1}+2}\overset{\mathcal{D}}{=}\ldots \overset{%
\mathcal{D}}{=}Y_{k_{i}},
\end{equation*}%
for $i=1,\ldots ,R$, again with the convention $k_{0}=1$ and $k_{R+1}=N$.
Testing $H_{0}^{\prime }$ versus $H_{A}^{\prime }$ can be done with
substantively weaker assumptions on the (moments of the) sequence $Y_{\ell }$
than what is required by Assumption \ref{as1}.

\begin{assumption}
\label{as:Yk} \textit{(i)} the sequence $\{Y_{\ell },-\infty <\ell <\infty
\} $ is a Bernoulli shift sequence, i.e., it has the representation $Y_{\ell
}=\widetilde{g}(\tilde{\eta}_{\ell },\tilde{\eta}_{\ell -1},\ldots )$, where
for each $\ell $, $\tilde{\eta}_{\ell }=\tilde{\eta}_{\ell }(t,\omega )$ are 
\textit{i.i.d.}~functions jointly measurable in $(t,\omega )$ taking values
in a measurable space $\mathcal{S}$, and $\widetilde{g}$ is a nonrandom
measurable function $\widetilde{g}:\mathcal{S}^{\infty }\rightarrow
L^{2}([0,1];\mathbb{R})$; \textit{(ii)} $E\Vert Y_{1}\Vert ^{\beta }<\infty $
for some $\beta >0$; \textit{(iii)} for $\beta $ defined in part \textit{(ii)%
}, there is some $\alpha _{0}>2$ such that $E\Vert Y_{1}-Y_{1}^{(m)}\Vert
^{\beta }\leq Cm^{-\alpha _{0}}$, where for each pair $(j,\ell )$, we set $%
Y_{\ell }^{(j)}=\widetilde{g}(\tilde{\eta}_{\ell },\tilde{\eta}_{\ell
-1},\ldots ,\tilde{\eta}_{\ell -j+1},\tilde{\eta}_{\ell -j}^{(j)},\tilde{\eta%
}_{\ell -j-1}^{(j)},\ldots )$, with $\{\tilde{\eta}_{\ell }^{(j)},-\infty
<\ell <\infty \}$ an independent copy of $\{\tilde{\eta}_{\ell },-\infty
<\ell <\infty \}$.
\end{assumption}

Inspired by \citet{berkes2009detecting}, we pre-process the infinite
dimensional data $Y_{\ell }(t)$ by projecting them into a finite dimensional
vector 
\begin{equation}
\xi _{j,\ell }=\int_{0}^{1}{Y_{\ell }}(s)\psi _{j}(s)ds,\quad 1\leq \ell
\leq N,\quad 1\leq j\leq d,  \label{e:xi_i}
\end{equation}%
where $\{\psi _{\ell },\ell \geq 1\}$ is an orthonormal basis of $L^{2}[0,1]$%
. Thence, we define the corresponding $\mathbb{C}$-valued random functions 
\begin{equation}
X_{\ell }(t)=\exp \left( \mathbf{i}\sum_{j=1}^{d}t_{j}\xi _{j,\ell }\right)
,\quad 1\leq \ell \leq N,  \label{e:xi_i_to_X0}
\end{equation}%
where $t=(t_{1},\ldots ,t_{d})^{\top }\in \lbrack -1,1]^{d}$, and $d$ is
user-chosen. Heuristically, $EX_{\ell }(t)$ is (an approximation of) the
characteristic functional of $Y_{\ell }$, and therefore comparing averages
of $X_{\ell }(t)$ before and after a point in time $k$ is a natural way of
checking whether the distribution of $Y_{\ell }$ changes or not. Viewing
each $\left\{ X_{\ell }(t),t\in \lbrack -1,1]^{d}\right\} $ in %
\eqref{e:xi_i_to_X0} as an $\mathbb{R}^{2}$-valued random function $(\text{Re%
}\,X_{\ell }(t),\text{Im}\,X_{\ell }(t))^{\top }$, we may apply the test
statistics proposed in Section \ref{tests} to test the hypotheses $%
H_{0}^{\prime }$ versus $H_{A}^{\prime }$.

Let $\xi _{j,\ell }^{(m)}=\int_{0}^{1}Y_{\ell }^{(m)}(s)\psi _{j}(s)ds$, and
define $X_{\ell }^{(m)}(t)=\exp \left( \mathbf{i}\sum_{j=1}^{d}t_{j}\xi
_{j,\ell }^{(m)}\right) $. We show that $\left\{ X_{\ell },-\infty <\ell
<\infty \right\} $ is a Bernoulli shift sequence which satisfies Assumption %
\ref{as1}.

\begin{lemma}
\label{x-bernoulli}We assume that Assumption \ref{as:Yk} is satisfied. Then,
for every $\gamma >2\beta (\alpha _{0}-2)$, there is an $\alpha _{0}^{\prime
}>2$ such that $E\Vert X_{1}-X_{1}^{(m)}\Vert ^{\gamma }\leq Cm^{-\alpha
_{0}^{\prime }}$.
\end{lemma}

In order to construct the auxiliary functions $X_{\ell }(t)$ defined in %
\eqref{e:xi_i_to_X0}, one must first choose a basis $\{\psi _{\ell }\}$.
Though any orthonormal basis of $L^{2}([0,1];\mathbb{R})$ will suffice, when
the observations $Y_{\ell }$ satisfy $E|Y_{1}(s)|^{2}<\infty $, $s\in
\lbrack 0,1]$, Principal Component Analysis (PCA) based approaches are among
the most popular choices for selecting $\{\psi _{\ell }\}$, and typically
lead to good finite-sample performance. Under the assumption that $%
E|Y_{1}(s)|^{2}<\infty $, $s\in \lbrack 0,1]$, define 
\begin{equation}
C(t,s)=\cov\left( Y_{1}(t),Y_{1}(s)\right)  \label{e:cov(Y)}
\end{equation}%
According to the PCA approach, the $\psi _{\ell }$ in (\ref{e:xi_i}) are
chosen as the eigenfunctions of $C(t,s)$ 
\begin{equation*}
\chi _{\ell }\psi _{\ell }(t)=\int C(t,s)\psi _{\ell }(s)ds,
\end{equation*}%
where $\chi _{1}>\chi _{2}>\ldots $, and $\{\psi _{\ell },\ell \geq 1\}$ are
orthonormal - note the requirement that eigenvalues are well-separated,
which is typical of (functional) PCA\ (see e.g. %
\citealp{horvath:kokoszka:2012}). With this choice of basis, typically the
approximation $Y_{\ell }(t)\approx \sum_{i=1}^{d}\xi _{j,\ell }\psi _{\ell
}(t)$ requires only a small number $d$ of projections for good finite-sample
performance. We estimate the covariance function $C(t,s)$ in \eqref{e:cov(Y)}
as 
\begin{equation*}
\widehat{C}_{N}(t,s)=\frac{1}{N}\sum_{j=1}^{N}(Y_{j}(t)-\widehat{\mu }%
_{Y,N}(t))(Y_{j}(s)-\widehat{\mu }_{Y,N}(s)),
\end{equation*}%
where $\widehat{\mu }_{Y,N}(t)=N^{-1}\sum_{j=1}^{N}Y_{j}(t)$ is the sample
mean. If Assumption \ref{as:Yk} holds with $\beta >2$, then by the ergodic
theorem it holds that 
\begin{equation}
\iint \big(\widehat{C}_{N}(t,s)-C(t,s)\big)^{2}dtds\rightarrow 0\quad \text{%
a.s.}  \label{e:C_N_to_C}
\end{equation}%
Thus, if $(\widehat{\chi }_{\ell },\widehat{\psi }_{\ell })$ are the
eigenvalue-eigenfunction pairs defined by 
\begin{equation*}
\widehat{\chi }_{\ell }(t)\widehat{\psi }_{\ell }=\int_{0}^{1}\widehat{C}%
_{N}(t,s)\widehat{\psi }_{\ell }(s)ds,\quad \widehat{\chi }_{1}\geq \widehat{%
\chi }_{2}\geq \ldots ,
\end{equation*}%
where $\int \left\vert \widehat{\psi }_{\ell }\left( t\right) \right\vert
^{2}dt=1$, then for each fixed $1\leq \ell \leq N-1$, the eigenfunctions are
estimated consistently modulo a sign - i.e., it holds that $E\Vert \widehat{%
\psi }_{\ell }(t)-\iota _{\ell }\psi _{\ell }(t)\Vert ^{2}=o_{P}(1)$, where $%
\iota _{\ell }$ is a random sign (see Theorem 2.8 in %
\citealp{horvath:kokoszka:2012}). Since the variables $X_{\ell }(t)$ do not
depend on the sign of $\psi _{\ell }$, one can then construct $X_{\ell }(t)$
in \eqref{e:xi_i_to_X0} based on 
\begin{equation}
\widehat{\xi }_{j,\ell }=\int_{0}^{1}Y_{\ell }(s)\widehat{\psi }%
_{j}(s)ds,\quad 1\leq \ell \leq N,\quad 1\leq j\leq d.  \label{e:xi_i_hat}
\end{equation}%
With the PCA-based choice \eqref{e:xi_i_hat}, it can be verified that $%
\widehat{X}_{\ell }\left( t\right) =\exp \left( \mathbf{i}\sum_{j=1}^{d}t_{j}%
\widehat{\xi }_{j,\ell }\right) $ still satisfies Lemma \ref{x-bernoulli}.
Hence, all the results of Section \ref{asymptotics} hold when using the
empirical energy function $V_{N}\left( k\right) $ based on $\widehat{X}%
_{\ell }\left( t\right) $ to test for $H_{0}^{\prime }$ in (\ref{h0-distr})
versus $H_{A}^{\prime }$ in (\ref{hA-distr}).

\section{Simulations\label{simulations}}

We provide some Monte Carlo evidence on the performance of our test
statistics, and some guidelines on how to implement the tests; further
details and results (including a set of experiments on binary segmentation)
are reported in Section \ref{furtherMC} in the Supplement. We use the
following Data Generating Process (DGP), inspired by \citet{happ}, based on
a truncated multivariate KL representation%
\begin{equation}
X_{i}(t)=\mu _{i}(t)+\sum_{\ell =1}^{M}\lambda _{\ell }^{1/2}\mathcal{Z}%
_{\ell ,i}\phi _{\ell }(t)+\nu _{i}\left( t\right) ,  \label{dgp}
\end{equation}%
for $1\leq i\leq N$, where: $t\in \mathcal{T=}\left[ 0,1\right] $, $X_{i}(t)$
is univariate, $\mathcal{Z}_{\ell ,i}$ are $N\left( 0,1\right) $ and {%
uncorrelated across $\ell $, $\left\{ \phi _{\ell }(t),1\leq \ell \leq
M\right\} $ form an orthonormal basis, and $\nu _{i}\left( t\right) $ is an 
\textit{i.i.d.} Gaussian measurement error with mean zero and scale $%
E\left\Vert \nu _{i}\left( t\right) \right\Vert ^{2}=\sigma _{\nu }^{2}$. As
far as $\nu _{i}\left( t\right) $ is concerned, we consider two designs: a
benchmark one with no measurement error (i.e., $\sigma _{\nu }^{2}=0$), and
one with $\sigma _{\nu }^{2}=0.25$. We allow for serial dependence in the $%
X_{i}(t)$'s through an AR($1$) structure in the $\mathcal{Z}_{\ell ,i}$
across $i$, viz. $\mathcal{Z}_{\ell ,i}=\rho \mathcal{Z}_{\ell ,i-1}+e_{\ell
,i}^{\mathcal{Z}}$ for all $1\leq \ell \leq M$, with $e_{\ell ,i}^{\mathcal{Z%
}}\sim i.i.d.N\left( 0,1\right) $ across $i$ and $\ell $. Under the null, we
set $\mu _{i}(t)=0$ for all $1\leq i\leq N$, for simplicity and with no loss
of generality. The eigenvalues $\lambda _{\ell }$ in (\ref{dgp}) are
generated as%
\begin{equation}
\lambda _{\ell }=\exp \left( -\left( \ell -1\right) /2\right) ;
\label{eig_exp}
\end{equation}%
unreported simulations show that using different schemes (e.g. a linear, or
a Wiener one) does not alter the results. The observations $X_{i}(t)$ are
sampled on an equispaced grid of $S=128$ points. As is typical in FDA, a
possible approach would be to pre-process and smooth the data, converting
the discretely observed $X_{i}(t_{j})$, $1\leq j\leq S$ into functional
objects by projecting them onto a suitably chosen basis; in our case, this
would only help with dimensionality reduction, since the coefficients of the
expansion are not required by any of our procedures. However, our test
statistics are not particularly computationally demanding, and therefore
pre-processing is not strictly required. Indeed, as %
\citet{hormann2022consistently} put it \textquotedblleft for the processing
of real data we will most often use the discretised curves
anyway\textquotedblright . In our case, for example, the integral in
equation (\ref{e:eigenvalues_def}) will be computed numerically, and the
most natural choice of nodes in the numerical computations are the
discretised sampling points $t_{j}$, $1\leq j\leq S$ - hence, we suggest as
a guideline that no data smoothing/pre-processing is carried out, at least
for \textquotedblleft reasonable\textquotedblright\ values of $S$. As far as
other specifications are concerned, we compute $\widehat{\sigma }_{N}^{2}$
and $\widehat{\mathbf{D}}_{N}(t,s)$ as described in Section \ref{cr_values}.
We have used the Parzen kernel, and we have selected the bandwidth $h$
according to the optimal rules derived in \citet{andrews1991}. We simulate $%
\Delta _{N,M}(u)$ over a grid with exactly $N$ points, which we recommend in
practical applications. All results are based on using an estimate $\widehat{%
M}$ of $M$, chosen so that the first $\widehat{M}$ eigenvalues of $\widehat{%
\mathbf{D}}_{N}(t,s)$ explain a prespecified amount of the total variability
(we set this to $0.95$, which is a bit higher than in other papers, but
still comes with a great dimensionality reduction). Critical values for
weighted functionals of $\Delta _{N,M}(u)$ are computed using $500$
replications. All simulations are carried out with $1,000$ replications; all
routines have been written using GAUSS 21.0.6. }

Empirical rejection frequencies under the null, at a nominal $5\%$ level,
are reported in Table \ref{tab:ERF}; see also Section \ref{furtherMC} in the
Supplement for further cases. In the \textit{i.i.d.} case, our tests have
excellent size control in all cases: the empirical rejection frequencies lie
in the confidence interval $\left[ 0.036,0.064\right] $ even for sample
sizes as small as $N=50$, and for all the values of $\alpha $ considered in
our simulations. In general, our tests are almost never oversized,
suggesting that spurious break detection is highly unlikely. When serial
dependence (especially) and/or measurement errors are present, the tests are
somewhat conservative for small samples and large $\alpha $, but this
improves as $N$ increases. Upon closer inspection, this is due to the fact
that the bandwidth $h$ employed in (\ref{e:def_D_N}) seems too high, and
reducing it would increase the size; in turn, this suggests that, prior to
implementing the tests, some qualitative considerations based on the
presence of measurement error, and a bandwidth selection rule based on $%
\alpha $, may yield improvements.

\begin{table*}[t]
\caption{{\protect\footnotesize {Empirical rejection frequencies under the
null of no changepoint}}}
\label{tab:ERF}\centering
{\footnotesize {\ }}
\par
\resizebox{\textwidth}{!}{

\begin{tabular}{lllllllllllllllllllll}
\hline\hline
&  &  &  &  &  &  &  &  &  &  &  &  &  &  &  &  &  &  &  &  \\ 
&  &  & \multicolumn{8}{c}{\textit{i.i.d. }case, no measurement error} & 
\multicolumn{1}{c}{} & \multicolumn{1}{c}{} & \multicolumn{8}{c}{serial
dependence with measurement error} \\ 
&  &  & \multicolumn{1}{c}{} & \multicolumn{1}{c}{} & \multicolumn{1}{c}{} & 
\multicolumn{1}{c}{} & \multicolumn{1}{c}{} & \multicolumn{1}{c}{} & 
\multicolumn{1}{c}{} & \multicolumn{1}{c}{} & \multicolumn{1}{c}{} & 
\multicolumn{1}{c}{} & \multicolumn{1}{c}{} & \multicolumn{1}{c}{} & 
\multicolumn{1}{c}{} & \multicolumn{1}{c}{} & \multicolumn{1}{c}{} & 
\multicolumn{1}{c}{} & \multicolumn{1}{c}{} & \multicolumn{1}{c}{} \\ 
\multicolumn{1}{c}{$N$} & $\alpha $ &  & \multicolumn{1}{c}{$0.00$} & 
\multicolumn{1}{c}{$0.10$} & \multicolumn{1}{c}{$0.25$} & \multicolumn{1}{c}{$0.50$} & \multicolumn{1}{c}{$0.75$} & \multicolumn{1}{c}{$0.85$} & 
\multicolumn{1}{c}{$0.95$} & \multicolumn{1}{c}{$0.99$} & \multicolumn{1}{c}{
} & \multicolumn{1}{|c}{} & \multicolumn{1}{c}{$0.00$} & \multicolumn{1}{c}{$0.10$} & \multicolumn{1}{c}{$0.25$} & \multicolumn{1}{c}{$0.50$} & 
\multicolumn{1}{c}{$0.75$} & \multicolumn{1}{c}{$0.85$} & \multicolumn{1}{c}{$0.95$} & \multicolumn{1}{c}{$0.99$} \\ 
\multicolumn{1}{c}{} &  &  & \multicolumn{1}{c}{} & \multicolumn{1}{c}{} & 
\multicolumn{1}{c}{} & \multicolumn{1}{c}{} & \multicolumn{1}{c}{} & 
\multicolumn{1}{c}{} & \multicolumn{1}{c}{} & \multicolumn{1}{c}{} & 
\multicolumn{1}{c}{} & \multicolumn{1}{|c}{} & \multicolumn{1}{c}{} & 
\multicolumn{1}{c}{} & \multicolumn{1}{c}{} & \multicolumn{1}{c}{} & 
\multicolumn{1}{c}{} & \multicolumn{1}{c}{} & \multicolumn{1}{c}{} & 
\multicolumn{1}{c}{} \\ 
\multicolumn{1}{c}{$50$} &  &  & \multicolumn{1}{c}{$0.046$} & 
\multicolumn{1}{c}{$0.048$} & \multicolumn{1}{c}{$0.047$} & 
\multicolumn{1}{c}{$0.048$} & \multicolumn{1}{c}{$0.047$} & 
\multicolumn{1}{c}{$0.048$} & \multicolumn{1}{c}{$0.047$} & 
\multicolumn{1}{c}{$0.051$} & \multicolumn{1}{c}{} & \multicolumn{1}{|c}{} & 
\multicolumn{1}{c}{$0.045$} & \multicolumn{1}{c}{$0.036$} & 
\multicolumn{1}{c}{$0.042$} & \multicolumn{1}{c}{$0.027$} & 
\multicolumn{1}{c}{$0.029$} & \multicolumn{1}{c}{$0.021$} & 
\multicolumn{1}{c}{$0.016$} & \multicolumn{1}{c}{$0.011$} \\ 
\multicolumn{1}{c}{$100$} &  &  & \multicolumn{1}{c}{$0.058$} & 
\multicolumn{1}{c}{$0.058$} & \multicolumn{1}{c}{$0.056$} & 
\multicolumn{1}{c}{$0.058$} & \multicolumn{1}{c}{$0.062$} & 
\multicolumn{1}{c}{$0.058$} & \multicolumn{1}{c}{$0.062$} & 
\multicolumn{1}{c}{$0.062$} & \multicolumn{1}{c}{} & \multicolumn{1}{|c}{} & 
\multicolumn{1}{c}{$0.054$} & \multicolumn{1}{c}{$0.050$} & 
\multicolumn{1}{c}{$0.059$} & \multicolumn{1}{c}{$0.047$} & 
\multicolumn{1}{c}{$0.053$} & \multicolumn{1}{c}{$0.036$} & 
\multicolumn{1}{c}{$0.040$} & \multicolumn{1}{c}{$0.025$} \\ 
\multicolumn{1}{c}{$150$} &  &  & \multicolumn{1}{c}{$0.058$} & 
\multicolumn{1}{c}{$0.055$} & \multicolumn{1}{c}{$0.057$} & 
\multicolumn{1}{c}{$0.056$} & \multicolumn{1}{c}{$0.057$} & 
\multicolumn{1}{c}{$0.052$} & \multicolumn{1}{c}{$0.052$} & 
\multicolumn{1}{c}{$0.051$} & \multicolumn{1}{c}{} & \multicolumn{1}{|c}{} & 
\multicolumn{1}{c}{$0.059$} & \multicolumn{1}{c}{$0.070$} & 
\multicolumn{1}{c}{$0.063$} & \multicolumn{1}{c}{$0.069$} & 
\multicolumn{1}{c}{$0.051$} & \multicolumn{1}{c}{$0.037$} & 
\multicolumn{1}{c}{$0.039$} & \multicolumn{1}{c}{$0.027$} \\ 
\multicolumn{1}{c}{$200$} &  &  & \multicolumn{1}{c}{$0.051$} & 
\multicolumn{1}{c}{$0.050$} & \multicolumn{1}{c}{$0.051$} & 
\multicolumn{1}{c}{$0.052$} & \multicolumn{1}{c}{$0.059$} & 
\multicolumn{1}{c}{$0.052$} & \multicolumn{1}{c}{$0.050$} & 
\multicolumn{1}{c}{$0.051$} & \multicolumn{1}{c}{} & \multicolumn{1}{|c}{} & 
\multicolumn{1}{c}{$0.052$} & \multicolumn{1}{c}{$0.050$} & 
\multicolumn{1}{c}{$0.047$} & \multicolumn{1}{c}{$0.050$} & 
\multicolumn{1}{c}{$0.043$} & \multicolumn{1}{c}{$0.041$} & 
\multicolumn{1}{c}{$0.039$} & \multicolumn{1}{c}{$0.035$} \\ 
\multicolumn{1}{c}{} &  &  & \multicolumn{1}{c}{} & \multicolumn{1}{c}{} & 
\multicolumn{1}{c}{} & \multicolumn{1}{c}{} & \multicolumn{1}{c}{} & 
\multicolumn{1}{c}{} & \multicolumn{1}{c}{} & \multicolumn{1}{c}{} & 
\multicolumn{1}{c}{} & \multicolumn{1}{|c}{} & \multicolumn{1}{c}{} & 
\multicolumn{1}{c}{} & \multicolumn{1}{c}{} & \multicolumn{1}{c}{} & 
\multicolumn{1}{c}{} & \multicolumn{1}{c}{} & \multicolumn{1}{c}{} & 
\multicolumn{1}{c}{} \\ \hline\hline
\end{tabular}

}
\par
{\footnotesize 
\begin{tablenotes}
      \tiny
            \item The table contains the empirical rejection frequencies under the null of no changepoint, using $M=40$ orthonormal bases in (\ref{dgp}), for tests at a $5\%$ nominal level. The specifications of (\ref{dgp}) are described in the main text.
            
\end{tablenotes}
}
\end{table*}

Turning to power, we compute empirical rejection frequencies under the
at-most-one-change alternative where, in (\ref{dgp}), for simplicity in
illustration, we shift the data by a constant after the breakpoint, namely: 
\begin{equation}
\mu _{i}(t)=\mathcal{\delta }(t)I\left( i>k^{\ast }\right) ,  \label{amoc}
\end{equation}%
where $\delta(t)\equiv C_\delta>0$, and the size of the change set to $%
\|\delta\|=C_\delta \in \left\{ 0.4,0.8,1.2,1.6,2\right\} $. We consider two
scenarios: a mid-sample break, with $k^{\ast }=\left\lfloor N/2\right\rfloor 
$ (Figure \ref{fig:FigM200}), and a late-occurring changepoint, with $%
k^{\ast }=\left\lfloor 0.9N\right\rfloor $ (Figure \ref{fig:FigE200}).
Results are obtained for $N=200$, and using $200$ replications to save
computational time; for brevity, in Figures \ref{fig:FigM200} and \ref%
{fig:FigE200} we report only results for the \textit{i.i.d.} case and the
case of serial dependence and measurement error.\footnote{%
Further results, with $N=100 $, confirm the findings reported in this
section, and are reported in Section \ref{furtherMC} in the Supplement.}
Figures \ref{fig:FigM200} and \ref{fig:FigE200} confirm that using higher $%
\alpha $ is beneficial when the breakdate $k^{\ast }$ is close to sample
endpoints, whereas, in the presence of mid-sample breaks, the test generally
has good power, which tends to be lower as $\alpha $ increases (the
discrepancy increases as $\left\Vert \delta \right\Vert $ declines).

\begin{figure}[t]
\caption{{\protect\footnotesize {Empirical rejection frequencies under a
mid-sample break, $N=200$ - \textit{i.i.d.} data (left panel) and data with
serial dependence and measurement error (right panel)}}}
\label{fig:FigM200}\centering
\begin{minipage}[c]{0.49\textwidth}
\centering
    \includegraphics[width=\linewidth]{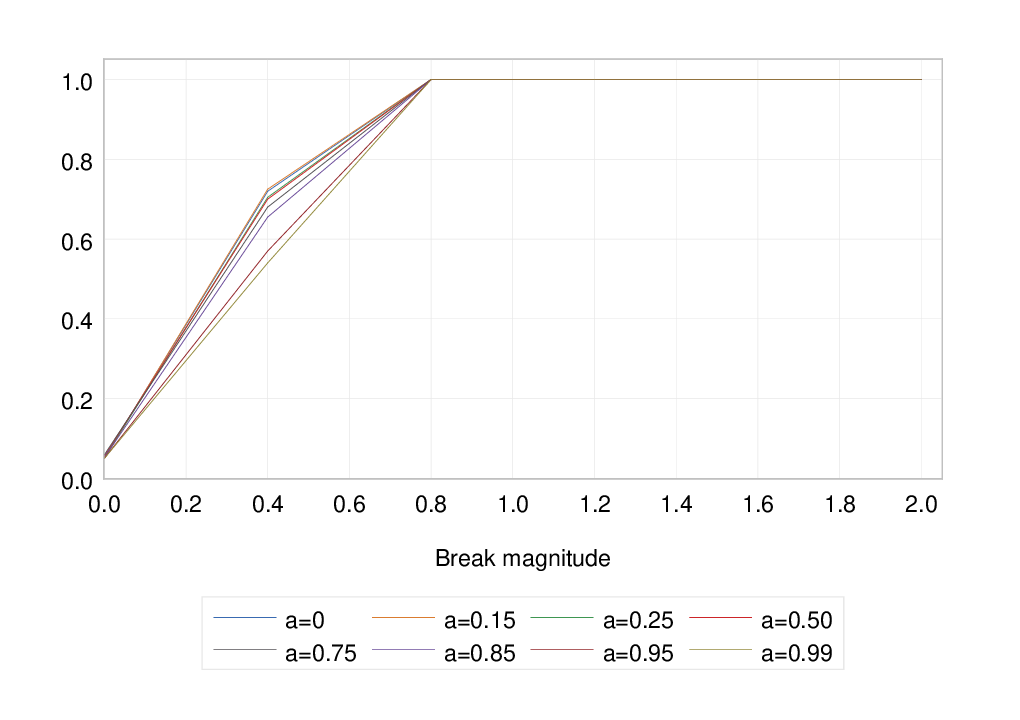}
\end{minipage}
\begin{minipage}[c]{0.49\textwidth}
\centering
    \includegraphics[width=\linewidth]{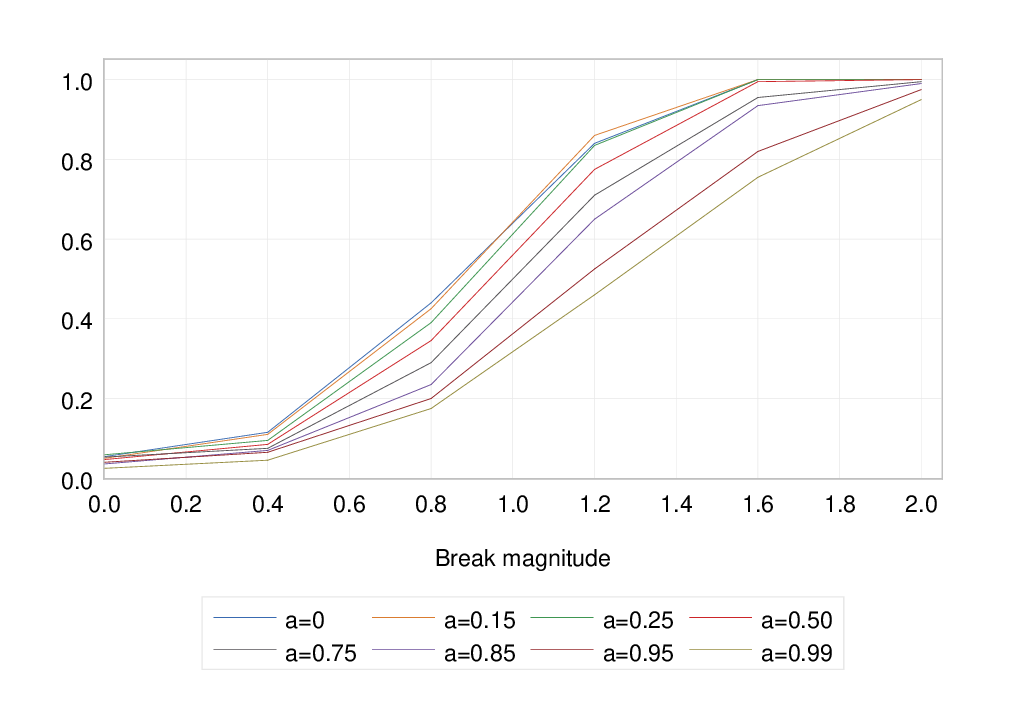}
\end{minipage}
\par
\end{figure}

\begin{figure}[b]
\caption{{\protect\footnotesize {Empirical rejection frequencies under an
end-of-sample break, $N=200$ - \textit{i.i.d.} data (left panel) and data
with serial dependence and measurement error (right panel)}}}
\label{fig:FigE200}\centering
\begin{minipage}{0.49\textwidth}
\centering
    \includegraphics[width=\linewidth]{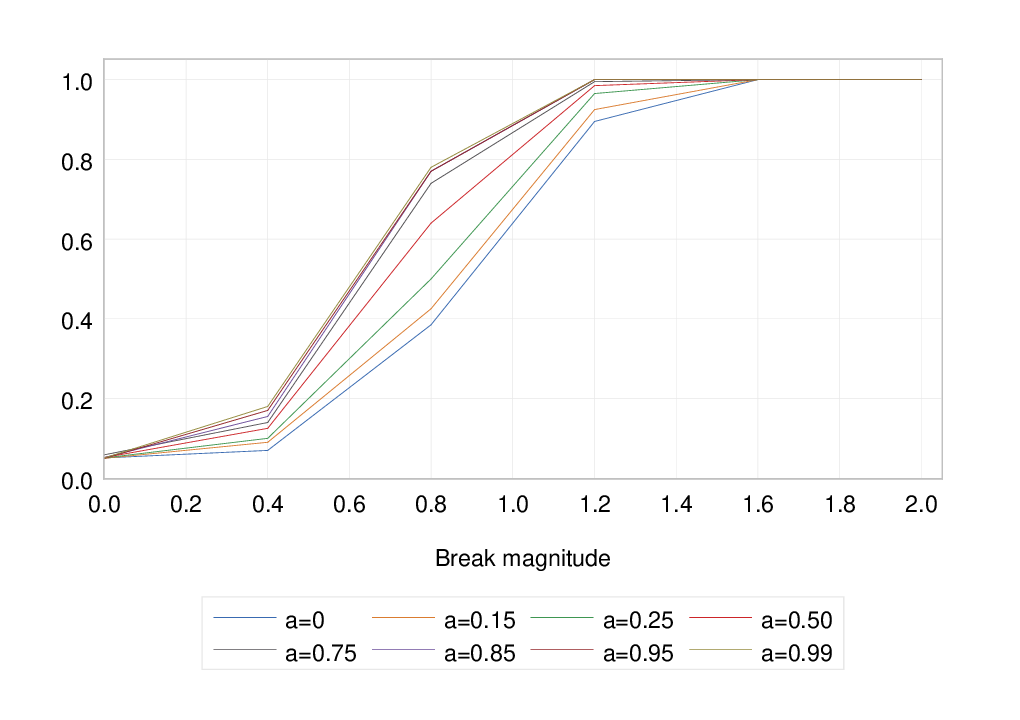}
    \label{fig:t11}
\end{minipage}
\begin{minipage}{0.49\textwidth}
\centering
    \includegraphics[width=\linewidth]{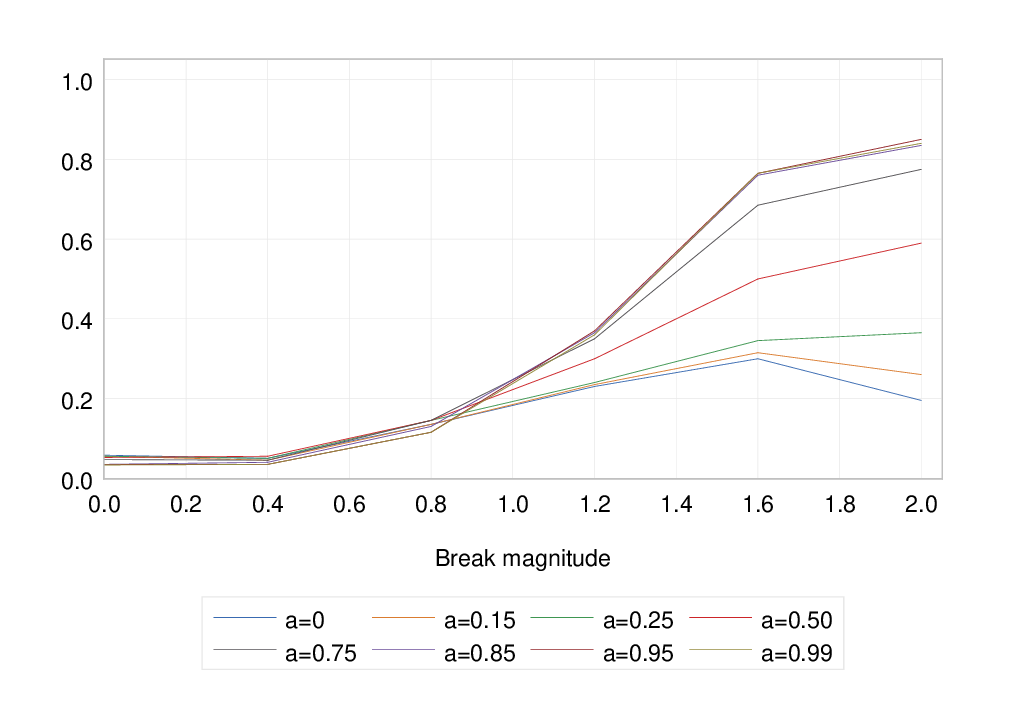}
    \label{fig:t13}
\end{minipage}
\par
\end{figure}

In a second set of experiments, we explore the performance of our
methodology to test for changes in the distribution proposed in Section \ref%
{distrib} via a small Monte Carlo exercise. We generate the one-dimensional
functional data ${Y_{\ell }}(t)$\ using (\ref{dgp}) with no measurement
error, viz. ${Y_{\ell }}(t)=\mu _{\ell }(t)+\epsilon _{\ell }(t)$ with $%
\epsilon _{\ell }(t)$%
\begin{equation}
\epsilon _{\ell }(t)=\sum_{m=1}^{M}\lambda _{m}^{1/2}\mathcal{Z}_{m,\ell
}\phi _{m}(t).  \label{dgp-distr}
\end{equation}%
We project $Y_{\ell }(t)$ onto its \textit{first} Principal Component - that
is, we use $d=1$ in (\ref{e:xi_i_to_X0}). We do this merely for
computational simplicity; when computing the eigenvalues of the long-run
variance matrix associated with $X_{\ell }(t)$, we use the algorithm in
Section 3.2 in \citet{happ}, based on the multivariate KL expansion. In
Table \ref{tab:ERFDistr}, we report the empirical rejection frequencies
under the null, showing that our methodology has excellent size control for $%
N\geq 100$; when $N=50$, tests appear to be mildly oversized.

\medskip

\begin{table*}[t]
\caption{{\protect\footnotesize {Empirical rejection frequencies under the
null of no changepoint in distribution - \textit{i.i.d.} data}}}
\label{tab:ERFDistr}\centering
\par
{\scriptsize {\ 
\begin{tabular}{lllllllllll}
\hline\hline
&  &  &  &  &  &  &  &  &  &  \\ 
\multicolumn{1}{c}{$N$} & $\alpha $ &  & \multicolumn{1}{c}{$0.00$} & 
\multicolumn{1}{c}{$0.10$} & \multicolumn{1}{c}{$0.25$} & \multicolumn{1}{c}{%
$0.50$} & \multicolumn{1}{c}{$0.75$} & \multicolumn{1}{c}{$0.85$} & 
\multicolumn{1}{c}{$0.95$} & \multicolumn{1}{c}{$0.99$} \\ 
\multicolumn{1}{c}{} &  &  & \multicolumn{1}{c}{} & \multicolumn{1}{c}{} & 
\multicolumn{1}{c}{} & \multicolumn{1}{c}{} & \multicolumn{1}{c}{} & 
\multicolumn{1}{c}{} & \multicolumn{1}{c}{} & \multicolumn{1}{c}{} \\ 
\multicolumn{1}{c}{$50$} &  &  & \multicolumn{1}{c}{$0.067$} & 
\multicolumn{1}{c}{$0.069$} & \multicolumn{1}{c}{$0.069$} & 
\multicolumn{1}{c}{$0.069$} & \multicolumn{1}{c}{$0.072$} & 
\multicolumn{1}{c}{$0.072$} & \multicolumn{1}{c}{$0.077$} & 
\multicolumn{1}{c}{$0.076$} \\ 
\multicolumn{1}{c}{$100$} &  &  & \multicolumn{1}{c}{$0.058$} & 
\multicolumn{1}{c}{$0.058$} & \multicolumn{1}{c}{$0.057$} & 
\multicolumn{1}{c}{$0.056$} & \multicolumn{1}{c}{$0.057$} & 
\multicolumn{1}{c}{$0.056$} & \multicolumn{1}{c}{$0.059$} & 
\multicolumn{1}{c}{$0.061$} \\ 
\multicolumn{1}{c}{$150$} &  &  & \multicolumn{1}{c}{$0.053$} & 
\multicolumn{1}{c}{$0.051$} & \multicolumn{1}{c}{$0.050$} & 
\multicolumn{1}{c}{$0.050$} & \multicolumn{1}{c}{$0.051$} & 
\multicolumn{1}{c}{$0.047$} & \multicolumn{1}{c}{$0.049$} & 
\multicolumn{1}{c}{$0.048$} \\ 
\multicolumn{1}{c}{$200$} &  &  & \multicolumn{1}{c}{$0.057$} & 
\multicolumn{1}{c}{$0.057$} & \multicolumn{1}{c}{$0.054$} & 
\multicolumn{1}{c}{$0.053$} & \multicolumn{1}{c}{$0.052$} & 
\multicolumn{1}{c}{$0.056$} & \multicolumn{1}{c}{$0.060$} & 
\multicolumn{1}{c}{$0.057$} \\ 
\multicolumn{1}{c}{} &  &  & \multicolumn{1}{c}{} & \multicolumn{1}{c}{} & 
\multicolumn{1}{c}{} & \multicolumn{1}{c}{} & \multicolumn{1}{c}{} & 
\multicolumn{1}{c}{} & \multicolumn{1}{c}{} & \multicolumn{1}{c}{} \\ 
\hline\hline
\end{tabular}
} }
\par
{\scriptsize {\footnotesize 
\begin{tablenotes}
      \tiny
            \item The table contains the empirical rejection frequencies under the null of no changepoint. Data are generated according to (\ref{dgp-distr}) with specifications as in the main text, for tests at a $5\%$ nominal level.
            
\end{tablenotes}
} }
\end{table*}

We separately consider the following alternative hypotheses: 
\begin{equation}
\mu _{\ell }(t)=\mathcal{\delta }(t)I\left( k^{\ast }\leq \ell \leq N\right)
,  \label{mean_distr}
\end{equation}%
with $\left\Vert \mathcal{\delta }\right\Vert =1$, to consider changes in
the mean function;%
\begin{equation}
\sigma _{\mathcal{Z}}=\sigma _{\mathcal{Z},\ell }=I\left( 1\leq \ell
<k^{\ast }\right) +2I\left( k^{\ast }\leq \ell \leq N\right) ,
\label{var_distr}
\end{equation}%
with $\mu _{\ell }(t)=0$ for all $1\leq \ell \leq N$, to consider a change
in the (unconditional) variance which is helpful to understand whether our
methodology can detect heteroskedasticity; and lastly 
\begin{equation}
\epsilon _{\ell }(t)=\left( \sum_{m=1}^{M}\lambda _{m}^{1/2}\mathcal{Z}%
_{m,\ell }\phi _{m}(t)\right) I\left( 1\leq \ell <k^{\ast }\right) +\left(
\sum_{m=1}^{M}\lambda _{m}^{1/2}t_{m,\ell }^{\left( 3\right) }\phi
_{m}(t)\right) I\left( k^{\ast }\leq \ell \leq N\right) ,  \label{tail_distr}
\end{equation}%
where $t_{m,\ell }^{\left( 3\right) }$\ are \textit{i.i.d.} random
variables, independent across $m$ and $\ell $, with a Student's t
distribution with $3$ degrees of freedom, and $\mu _{\ell }(t)=0$ for all $%
1\leq \ell \leq N$, as a more general alternative where the data, after a
period of \textquotedblleft normal\textquotedblright\ fluctuations, exhibit
heavy tails. Results in Table \ref{tab:PowerDistr1} show that our tests -
even when using $d=1$ - have excellent power under all cases in the presence
of a mid-sample break, which is also estimated correctly. Hence, the test
developed in Section \ref{distrib} can be used to detect shifts in the mean,
in the variance, or in the tails - of course, the test is an \textit{omnibus}
test, and therefore it is non-constructive in that, upon rejecting the null,
it does not indicate a specific alternative. In the case of end-of-sample
breaks, Table \ref{tab:PowerDistr2} shows that the test is sensitive, as
expected, to the choice of $\alpha $, and that as $\alpha $ approaches $1$
the power increases, as does the accuracy in estimating the changepoint. The
test performs very well, even for small sample sizes ($N=100$) in the
presence of changes in the mean and in the tails (i.e., under (\ref%
{mean_distr}) and (\ref{tail_distr}) respectively), whereas its performance
is less good in the presence of shifts in the variance (i.e., under (\ref%
{var_distr})) although it picks up as both $\alpha $ and $N$ increase.

\begin{table*}[t]
\caption{{\protect\footnotesize {Empirical rejection frequencies in the
presence of a mid-sample changepoint in distribution}}}
\label{tab:PowerDistr1}\centering
\par
{\scriptsize {\ 
\begin{tabular}{lllllllllll}
\hline\hline
&  &  &  &  &  &  &  &  &  &  \\ 
\multicolumn{11}{c}{Empirical rejection frequencies under (\ref{mean_distr})}
\\ 
&  &  &  &  &  &  &  &  &  &  \\ 
\multicolumn{1}{c}{$N$} & $\alpha $ &  & \multicolumn{1}{c}{$0.00$} & 
\multicolumn{1}{c}{$0.10$} & \multicolumn{1}{c}{$0.25$} & \multicolumn{1}{c}{%
$0.50$} & \multicolumn{1}{c}{$0.75$} & \multicolumn{1}{c}{$0.85$} & 
\multicolumn{1}{c}{$0.95$} & \multicolumn{1}{c}{$0.99$} \\ 
\multicolumn{1}{c}{} &  &  & \multicolumn{1}{c}{} & \multicolumn{1}{c}{} & 
\multicolumn{1}{c}{} & \multicolumn{1}{c}{} & \multicolumn{1}{c}{} & 
\multicolumn{1}{c}{} & \multicolumn{1}{c}{} & \multicolumn{1}{c}{} \\ 
\multicolumn{1}{c}{$100$} &  &  & \multicolumn{1}{c}{$\underset{\left(
50\right) }{0.875}$} & \multicolumn{1}{c}{$\underset{\left( 50\right) }{0.885%
}$} & \multicolumn{1}{c}{$\underset{\left( 50\right) }{0.885}$} & 
\multicolumn{1}{c}{$\underset{\left( 50\right) }{0.875}$} & 
\multicolumn{1}{c}{$\underset{\left( 50\right) }{0.870}$} & 
\multicolumn{1}{c}{$\underset{\left( 50\right) }{0.865}$} & 
\multicolumn{1}{c}{$\underset{\left( 50\right) }{0.855}$} & 
\multicolumn{1}{c}{$\underset{\left( 50\right) }{0.850}$} \\ 
\multicolumn{1}{c}{$200$} &  &  & \multicolumn{1}{c}{$\underset{\left(
100\right) }{1.000}$} & \multicolumn{1}{c}{$\underset{\left( 100\right) }{%
1.000}$} & \multicolumn{1}{c}{$\underset{\left( 100\right) }{1.000}$} & 
\multicolumn{1}{c}{$\underset{\left( 100\right) }{1.000}$} & 
\multicolumn{1}{c}{$\underset{\left( 100\right) }{1.000}$} & 
\multicolumn{1}{c}{$\underset{\left( 100\right) }{1.000}$} & 
\multicolumn{1}{c}{$\underset{\left( 100\right) }{1.000}$} & 
\multicolumn{1}{c}{$\underset{\left( 100\right) }{1.000}$} \\ 
&  &  &  &  &  &  &  &  &  &  \\ \hline
&  &  &  &  &  &  &  &  &  &  \\ 
\multicolumn{11}{c}{Empirical rejection frequencies under (\ref{var_distr})}
\\ 
&  &  &  &  &  &  &  &  &  &  \\ 
\multicolumn{1}{c}{$N$} & $\alpha $ &  & \multicolumn{1}{c}{$0.00$} & 
\multicolumn{1}{c}{$0.10$} & \multicolumn{1}{c}{$0.25$} & \multicolumn{1}{c}{%
$0.50$} & \multicolumn{1}{c}{$0.75$} & \multicolumn{1}{c}{$0.85$} & 
\multicolumn{1}{c}{$0.95$} & \multicolumn{1}{c}{$0.99$} \\ 
\multicolumn{1}{c}{} &  &  & \multicolumn{1}{c}{} & \multicolumn{1}{c}{} & 
\multicolumn{1}{c}{} & \multicolumn{1}{c}{} & \multicolumn{1}{c}{} & 
\multicolumn{1}{c}{} & \multicolumn{1}{c}{} & \multicolumn{1}{c}{} \\ 
\multicolumn{1}{c}{$100$} &  &  & \multicolumn{1}{c}{$\underset{\left(
51\right) }{0.925}$} & \multicolumn{1}{c}{$\underset{\left( 51\right) }{0.925%
}$} & \multicolumn{1}{c}{$\underset{\left( 51\right) }{0.925}$} & 
\multicolumn{1}{c}{$\underset{\left( 51\right) }{0.905}$} & 
\multicolumn{1}{c}{$\underset{\left( 51\right) }{0.850}$} & 
\multicolumn{1}{c}{$\underset{\left( 51\right) }{0.830}$} & 
\multicolumn{1}{c}{$\underset{\left( 51\right) }{0.790}$} & 
\multicolumn{1}{c}{$\underset{\left( 51\right) }{0.765}$} \\ 
\multicolumn{1}{c}{$200$} &  &  & \multicolumn{1}{c}{$\underset{\left(
101\right) }{1.000}$} & \multicolumn{1}{c}{$\underset{\left( 101\right) }{%
1.000}$} & \multicolumn{1}{c}{$\underset{\left( 101\right) }{1.000}$} & 
\multicolumn{1}{c}{$\underset{\left( 101\right) }{1.000}$} & 
\multicolumn{1}{c}{$\underset{\left( 101\right) }{1.000}$} & 
\multicolumn{1}{c}{$\underset{\left( 101\right) }{1.000}$} & 
\multicolumn{1}{c}{$\underset{\left( 101\right) }{1.000}$} & 
\multicolumn{1}{c}{$\underset{\left( 101\right) }{1.000}$} \\ 
&  &  &  &  &  &  &  &  &  &  \\ \hline
&  &  &  &  &  &  &  &  &  &  \\ 
\multicolumn{11}{c}{Empirical rejection frequencies under (\ref{tail_distr})}
\\ 
&  &  &  &  &  &  &  &  &  &  \\ 
\multicolumn{1}{c}{$N$} & $\alpha $ &  & \multicolumn{1}{c}{$0.00$} & 
\multicolumn{1}{c}{$0.10$} & \multicolumn{1}{c}{$0.25$} & \multicolumn{1}{c}{%
$0.50$} & \multicolumn{1}{c}{$0.75$} & \multicolumn{1}{c}{$0.85$} & 
\multicolumn{1}{c}{$0.95$} & \multicolumn{1}{c}{$0.99$} \\ 
\multicolumn{1}{c}{} &  &  & \multicolumn{1}{c}{} & \multicolumn{1}{c}{} & 
\multicolumn{1}{c}{} & \multicolumn{1}{c}{} & \multicolumn{1}{c}{} & 
\multicolumn{1}{c}{} & \multicolumn{1}{c}{} & \multicolumn{1}{c}{} \\ 
\multicolumn{1}{c}{$100$} &  &  & \multicolumn{1}{c}{$\underset{\left(
51\right) }{0.945}$} & \multicolumn{1}{c}{$\underset{\left( 51\right) }{0.945%
}$} & \multicolumn{1}{c}{$\underset{\left( 51\right) }{0.945}$} & 
\multicolumn{1}{c}{$\underset{\left( 51\right) }{0.950}$} & 
\multicolumn{1}{c}{$\underset{\left( 51\right) }{0.945}$} & 
\multicolumn{1}{c}{$\underset{\left( 51\right) }{0.945}$} & 
\multicolumn{1}{c}{$\underset{\left( 51\right) }{0.930}$} & 
\multicolumn{1}{c}{$\underset{\left( 51\right) }{0.915}$} \\ 
\multicolumn{1}{c}{$200$} &  &  & \multicolumn{1}{c}{$\underset{\left(
100\right) }{1.000}$} & \multicolumn{1}{c}{$\underset{\left( 100.5\right) }{%
1.000}$} & \multicolumn{1}{c}{$\underset{\left( 100.5\right) }{1.000}$} & 
\multicolumn{1}{c}{$\underset{\left( 101\right) }{1.000}$} & 
\multicolumn{1}{c}{$\underset{\left( 101\right) }{1.000}$} & 
\multicolumn{1}{c}{$\underset{\left( 101\right) }{1.000}$} & 
\multicolumn{1}{c}{$\underset{\left( 101\right) }{1.000}$} & 
\multicolumn{1}{c}{$\underset{\left( 101\right) }{1.000}$} \\ 
&  &  &  &  &  &  &  &  &  &  \\ \hline\hline
\end{tabular}
} }
\par
{\scriptsize {\footnotesize 
\begin{tablenotes}
      \tiny
            \item The table contains the empirical rejection frequencies under a changepoint occurring at $k^{\ast}=N/2$; the numbers in round brackets are the median estimated break dates. Data are generated according to (\ref{dgp-distr}) with specifications as in the main text, for tests at a $5\%$ nominal level, under the alternative hypotheses in (\ref{mean_distr})-(\ref{tail_distr}).
            
\end{tablenotes}
} }
\end{table*}

\begin{table*}[t]
\caption{{\protect\footnotesize {Empirical rejection frequencies in the
presence of an end-of-sample changepoint in distribution}}}
\label{tab:PowerDistr2}\centering
\par
{\scriptsize {\ 
\begin{tabular}{lllllllllll}
\hline\hline
&  &  &  &  &  &  &  &  &  &  \\ 
\multicolumn{11}{c}{Empirical rejection frequencies under (\ref{mean_distr})}
\\ 
&  &  &  &  &  &  &  &  &  &  \\ 
\multicolumn{1}{c}{$N$} & $\alpha $ &  & \multicolumn{1}{c}{$0.00$} & 
\multicolumn{1}{c}{$0.10$} & \multicolumn{1}{c}{$0.25$} & \multicolumn{1}{c}{%
$0.50$} & \multicolumn{1}{c}{$0.75$} & \multicolumn{1}{c}{$0.85$} & 
\multicolumn{1}{c}{$0.95$} & \multicolumn{1}{c}{$0.99$} \\ 
\multicolumn{1}{c}{} &  &  & \multicolumn{1}{c}{} & \multicolumn{1}{c}{} & 
\multicolumn{1}{c}{} & \multicolumn{1}{c}{} & \multicolumn{1}{c}{} & 
\multicolumn{1}{c}{} & \multicolumn{1}{c}{} & \multicolumn{1}{c}{} \\ 
\multicolumn{1}{c}{$100$} &  &  & \multicolumn{1}{c}{$\underset{\left(
71\right) }{0.200}$} & \multicolumn{1}{c}{$\underset{\left( 71\right) }{0.210%
}$} & \multicolumn{1}{c}{$\underset{\left( 71.5\right) }{0.220}$} & 
\multicolumn{1}{c}{$\underset{\left( 81\right) }{0.265}$} & 
\multicolumn{1}{c}{$\underset{\left( 86.5\right) }{0.340}$} & 
\multicolumn{1}{c}{$\underset{\left( 88\right) }{0.360}$} & 
\multicolumn{1}{c}{$\underset{\left( 89.5\right) }{0.380}$} & 
\multicolumn{1}{c}{$\underset{\left( 90\right) }{0.380}$} \\ 
\multicolumn{1}{c}{$200$} &  &  & \multicolumn{1}{c}{$\underset{\left(
157\right) }{0.485}$} & \multicolumn{1}{c}{$\underset{\left( 160\right) }{%
0.525}$} & \multicolumn{1}{c}{$\underset{\left( 164\right) }{0.555}$} & 
\multicolumn{1}{c}{$\underset{\left( 172.5\right) }{0.620}$} & 
\multicolumn{1}{c}{$\underset{\left( 177\right) }{0.695}$} & 
\multicolumn{1}{c}{$\underset{\left( 179\right) }{0.715}$} & 
\multicolumn{1}{c}{$\underset{\left( 179\right) }{0.725}$} & 
\multicolumn{1}{c}{$\underset{\left( 179\right) }{0.735}$} \\ 
&  &  &  &  &  &  &  &  &  &  \\ \hline
&  &  &  &  &  &  &  &  &  &  \\ 
\multicolumn{11}{c}{Empirical rejection frequencies under (\ref{var_distr})}
\\ 
&  &  &  &  &  &  &  &  &  &  \\ 
\multicolumn{1}{c}{$N$} & $\alpha $ &  & \multicolumn{1}{c}{$0.00$} & 
\multicolumn{1}{c}{$0.10$} & \multicolumn{1}{c}{$0.25$} & \multicolumn{1}{c}{%
$0.50$} & \multicolumn{1}{c}{$0.75$} & \multicolumn{1}{c}{$0.85$} & 
\multicolumn{1}{c}{$0.95$} & \multicolumn{1}{c}{$0.99$} \\ 
\multicolumn{1}{c}{} &  &  & \multicolumn{1}{c}{} & \multicolumn{1}{c}{} & 
\multicolumn{1}{c}{} & \multicolumn{1}{c}{} & \multicolumn{1}{c}{} & 
\multicolumn{1}{c}{} & \multicolumn{1}{c}{} & \multicolumn{1}{c}{} \\ 
\multicolumn{1}{c}{$100$} &  &  & \multicolumn{1}{c}{$\underset{\left(
63\right) }{0.145}$} & \multicolumn{1}{c}{$\underset{\left( 79\right) }{0.165%
}$} & \multicolumn{1}{c}{$\underset{\left( 83.5\right) }{0.200}$} & 
\multicolumn{1}{c}{$\underset{\left( 89\right) }{0.260}$} & 
\multicolumn{1}{c}{$\underset{\left( 90\right) }{0.340}$} & 
\multicolumn{1}{c}{$\underset{\left( 90\right) }{0.370}$} & 
\multicolumn{1}{c}{$\underset{\left( 91\right) }{0.400}$} & 
\multicolumn{1}{c}{$\underset{\left( 91\right) }{0.410}$} \\ 
\multicolumn{1}{c}{$200$} &  &  & \multicolumn{1}{c}{$\underset{\left(
156.5\right) }{0.300}$} & \multicolumn{1}{c}{$\underset{\left( 162.5\right) }%
{0.320}$} & \multicolumn{1}{c}{$\underset{\left( 171\right) }{0.370}$} & 
\multicolumn{1}{c}{$\underset{\left( 179\right) }{0.455}$} & 
\multicolumn{1}{c}{$\underset{\left( 180\right) }{0.545}$} & 
\multicolumn{1}{c}{$\underset{\left( 180\right) }{0.560}$} & 
\multicolumn{1}{c}{$\underset{\left( 180\right) }{0.590}$} & 
\multicolumn{1}{c}{$\underset{\left( 180\right) }{0.595}$} \\ 
&  &  &  &  &  &  &  &  &  &  \\ \hline
&  &  &  &  &  &  &  &  &  &  \\ 
\multicolumn{11}{c}{Empirical rejection frequencies under (\ref{tail_distr})}
\\ 
&  &  &  &  &  &  &  &  &  &  \\ 
\multicolumn{1}{c}{$N$} & $\alpha $ &  & \multicolumn{1}{c}{$0.00$} & 
\multicolumn{1}{c}{$0.10$} & \multicolumn{1}{c}{$0.25$} & \multicolumn{1}{c}{%
$0.50$} & \multicolumn{1}{c}{$0.75$} & \multicolumn{1}{c}{$0.85$} & 
\multicolumn{1}{c}{$0.95$} & \multicolumn{1}{c}{$0.99$} \\ 
\multicolumn{1}{c}{} &  &  & \multicolumn{1}{c}{} & \multicolumn{1}{c}{} & 
\multicolumn{1}{c}{} & \multicolumn{1}{c}{} & \multicolumn{1}{c}{} & 
\multicolumn{1}{c}{} & \multicolumn{1}{c}{} & \multicolumn{1}{c}{} \\ 
\multicolumn{1}{c}{$100$} &  &  & \multicolumn{1}{c}{$\underset{\left(
86\right) }{0.320}$} & \multicolumn{1}{c}{$\underset{\left( 88\right) }{0.345%
}$} & \multicolumn{1}{c}{$\underset{\left( 90\right) }{0.385}$} & 
\multicolumn{1}{c}{$\underset{\left( 90\right) }{0.460}$} & 
\multicolumn{1}{c}{$\underset{\left( 90\right) }{0.575}$} & 
\multicolumn{1}{c}{$\underset{\left( 91\right) }{0.610}$} & 
\multicolumn{1}{c}{$\underset{\left( 91\right) }{0.650}$} & 
\multicolumn{1}{c}{$\underset{\left( 91\right) }{0.660}$} \\ 
\multicolumn{1}{c}{$200$} &  &  & \multicolumn{1}{c}{$\underset{\left(
173\right) }{0.505}$} & \multicolumn{1}{c}{$\underset{\left( 178\right) }{%
0.540}$} & \multicolumn{1}{c}{$\underset{\left( 179\right) }{0.630}$} & 
\multicolumn{1}{c}{$\underset{\left( 180\right) }{0.750}$} & 
\multicolumn{1}{c}{$\underset{\left( 180\right) }{0.835}$} & 
\multicolumn{1}{c}{$\underset{\left( 180\right) }{0.855}$} & 
\multicolumn{1}{c}{$\underset{\left( 181\right) }{0.860}$} & 
\multicolumn{1}{c}{$\underset{\left( 181\right) }{0.860}$} \\ 
&  &  &  &  &  &  &  &  &  &  \\ \hline\hline
\end{tabular}
} }
\par
{\scriptsize {\footnotesize 
\begin{tablenotes}
      \tiny
            \item The table contains the empirical rejection frequencies under a changepoint occurring at $k^{\ast}=0.9N$; the numbers in round brackets are the median estimated break dates. Data are generated according to (\ref{dgp-distr}) with specifications as in the main text, for tests at a $5\%$ nominal level, under the alternative hypotheses in (\ref{mean_distr})-(\ref{tail_distr}).
            
\end{tablenotes}
} }
\end{table*}

\section{Changepoint detection in high-frequency financial data\label%
{empirical}}

We apply our tests for changepoint detection in the mean and in the
distribution of intraday return patterns (on a month-on-month basis) in
high-frequency trading of the S\&P 500 index.\footnote{%
In Section \ref{empiric_further} in the Supplement, we study temperature
data, which is another classical application of FDA (see e.g. %
\citealp{berkes2009detecting}).} High-frequency trading data lend themselves
to being studied through the lenses of FDA, as they typically contain a huge
amount of data for which a parsimonious representation is necessary;
furthermore prices change continuously on a daily basis, and therefore daily
prices are genuinely functional objects, whose sampling points are the
observed prices recorded over several points in time each day. Examples of
applications of FDA to high frequency financial data include e.g. %
\citet{genccay2001introduction}, and \citet{muller2011functional}; %
\citet{kokoszka2012functional} consider an alternative definition of return
(know as \textit{cumulative} intraday returns, or CIDRs) which seems to be
particularly suited for predictions using FDA.\footnote{%
In Section \ref{cidr} in the Supplement, we complement our analysis by
considering changes in the mean and in the distribution of CIDRs.}

Returns are calculated from closing prices that are recorded at equispaced
5-min intervals between $00:00$ and $23:55$ each day, corresponding to a
sampling frequency $S=276$. We have used the period spanning from January
3rd, $2022$, until September 3rd, $2023$. In order to balance the sample and
ensure that there are $S=276$ sampling points each day, we have removed $21$
trading daily curves in which some data were missing,\footnote{%
A list of the relevant days is available upon request.} for a total of $414$
functional datapoints. Denoting prices at day $i$ as $P_{i}\left( t\right) $%
, we construct month-on-month log returns as $Y_{i}\left( t\right) =\ln
P_{i}\left( t\right) -\ln P_{i-21}\left( t\right) $, having used $21$ lags
as the average amount of trading days in a month; hence, the resulting
sample size is $N=393$, effectively starting from February 3rd, $2022$. We
have implemented our test using the guidelines and specifications suggested
in Section \ref{simulations}.\footnote{%
Critical values for weighted functionals of $\Delta _{N,M}(u)$ are computed
using $500$ replications, due to the reduced computational times in the
empirical exercise; we note however that results do not differ in any
significant way upon altering this specification.} In particular, when using
the KL expansion, we employ a number of bases $\widehat{M}$, chosen so that
the first $\widehat{M}$ eigenvalues of the estimated long-run variance
explain $95\%$ of the total variability; we note that, in unreported
experiments, altering this specification did not change any of the final
results. Contrary to Section \ref{simulations}, we use an estimate of the
covariance kernel\ $\mathbf{D}\left( t,s\right) $ computed using the
estimator of the covariance functions (\ref{gammas}) with pre- and
post-break demeaning; we note that we tried to use $\widehat{\mathbf{D}}%
_{N}(t,s)$ defined in (\ref{e:def_D_N}) - that is, without demeaning before
and after the candidate breakdate - but results do not change in any way.%
\footnote{%
As far as other specifications of $\widehat{\mathbf{D}}_{N}(t,s)$ are
concerned, we have used a Parzen kernel and the bandwidth $h$ chosen
according to \citet{andrews1991}.} Upon inspection, this is due to the fact
that, across all exercises, the bandwidth is selected as at most $h=1$;
this, in turn, suggests that the data have only little serial dependence. We
tried to assess the sensitivity of our results by varying $h$, but virtually
no changes were noted. Finally, results are reported at a nominal level of $%
5\%$; however, we have also carried out - by way of sensitivity analysis -
detection at $1\%$ and $10\%$ levels, and results are discussed in the notes
of our tables.

We begin by applying our test for a changepoint in the mean. In Table \ref%
{tab:empiricsSP}, we report results for $\alpha =0.5$, but we also tried $%
\alpha =0$ and $\alpha =0.99$, obtaining exactly the same outcome: there is
only one break, located at October 20th, $2022$, whose estimate appears to
be remarkably robust (note also the discrepancy between the daily averages).
Whilst it is difficult to associate a particular event to that date, in
general the common wisdom among financial analysts is that S\&P 500 index
started recovering, after a turbulent year and after hitting its low in
October $2022$, around the second half of that month.\footnote{%
A qualitative description can be e.g. found at
https://www.usbank.com/investing/financial-perspectives/market-news/is-a-market-correction-coming.html%
} No further breaks in the mean function were found.

\medskip

\begin{table*}[h!]
\caption{{\protect\footnotesize {Changepoint in the mean in intraday,
month-on-month return curves - S\&P 500, January 3rd, 2022 - September 3rd,
2023}}}
\label{tab:empiricsSP}\centering
\par
\resizebox{\textwidth}{!}{

\begin{tabular}{ccccccccccc}
\hline\hline
&  &  &  &  &  &  &  &  &  &  \\ 
\multicolumn{11}{c}{Changepoint detection in the mean} \\ 
&  &  &  &  &  &  &  &  &  &  \\ 
& Iteration &  & Segment &  & Outcome &  & Estimated date &  & Notes &  \\ 
&  &  &  &  &  &  &  &  &  &  \\ 
& $1$ &  & Jan 3rd, 2022 -- Sep 3rd, 2023 &  & Reject &  & Oct 20th, 2022 & 
& significant also at $1\%$ &  \\ 
&  &  &  &  &  &  &  &  & break found also using $\alpha =0,$ $0.99$, at the
same date &  \\ 
&  &  &  &  &  &  &  &  & daily average = $0.00268$ &  \\ 
&  &  &  &  &  &  &  &  &  &  \\ 
& $2$ &  & Oct 20th, 2022 -- Sep 3rd, 2023 &  & Not reject &  &  &  & no
break found even at $10\%$ &  \\ 
&  &  &  &  &  &  &  &  & daily average = $-0.02493$ &  \\ 
&  &  &  &  &  &  &  &  &  &  \\ 
& $3$ &  & Jan 3rd, 2022 -- Oct 19th, 2022 &  & Not reject &  &  &  & no
break found even at $10\%$ &  \\ 
&  &  &  &  &  &  &  &  & daily average = $0.01631$ &  \\ 
&  &  &  &  &  &  &  &  &  &  \\ \hline\hline
\end{tabular}

}
\par
{\footnotesize 
\begin{tablenotes}
      \tiny
            \item Daily averages are computed as averages within each regime, and across sampling points. 
            
\end{tablenotes}
}
\end{table*}

\medskip

It is well known (see e.g. \citealp{kim2004more}) that the behaviour of
financial markets is characterised by not being adequately described by the
Gaussian distribution; hence, it is important to check if there are changes
not merely in the mean (or in the variance), but in the whole distribution.
Thus, after finding the presence of a break in the mean, we demean the data
in each of the two segments around October 20th, $2022$, and carry out the
test for distributional changes, on the demeaned data, discussed in Section %
\ref{distrib}. We use exactly the same specifications as in Section \ref%
{simulations}, using only one principal component (i.e. $d=1$) in the
construction of $X_{i}(t)$ in (\ref{e:xi_i_to_X0}). Tests are carried out at
a nominal level of $5\%$ by default (we also tried $1\%$ and $10\%$, see the
notes to Table \ref{tab:empiricsSPD}); when using binary segmentation, we
use, as threshold, $\tau _{N}=c_{a}\sqrt{\ln N}$, where $a$ is the nominal
level of the test; results are generally robust to this (we tried $\tau
_{N}=c_{a}\ln \ln N$, and $\tau _{N}=c_{a}{\ln N}$, and no changes were
noted). We report our findings using weights $\alpha =0$, $0.5$ and $0.99$;
we used binary segmentation, so the case $\alpha =0$ is not reliable \textit{%
per se}, as it may lead to overestimation of the number of regimes, but we
use it as a benchmark for the other two sets of results.

\begin{table*}[t]
\caption{{\protect\footnotesize {Changepoint in the distribution in
intraday, month-on-month return curves - S\&P 500, January 3rd, 2022 -
September 3rd, 2023}}}
\label{tab:empiricsSPD}\centering
\par
\resizebox{\textwidth}{!}{

\begin{tabular}{ccccccccccc}
\hline\hline
&  &  &  &  &  &  &  &  &  &  \\ 
\multicolumn{11}{c}{Changepoint detection using $\alpha =0.50$} \\ 
&  &  &  &  &  &  &  &  &  &  \\ 
& Iteration &  & Segment &  & Outcome &  & Estimated date &  & Notes &  \\ 
&  &  &  &  &  &  &  &  &  &  \\ 
& $1$ &  & Jan 3rd, 2022 -- Sep 3rd, 2023 &  & Reject &  & Oct 20th, 2022 & 
& significant also at $1\%$ &  \\ 
&  &  &  &  &  &  &  &  & break found also using $\alpha =0,$ $0.99$, at the
same date &  \\ 
&  &  &  &  &  &  &  &  &  &  \\ 
& $2$ &  & Oct 20th, 2022 -- Sep 3rd, 2023 &  & Not reject &  &  &  & no
break found even at $10\%$, or with $\alpha =0,$ $0.99$ &  \\ 
&  &  &  &  &  &  &  &  & $\begin{array}{c}
\widehat{\sigma }_{N}^{2}=0.002 \\ 
sk=0.950 \\ 
ku=2.008\end{array}$ &  \\ 
&  &  &  &  &  &  &  &  &  &  \\ 
& $3$ &  & Jan 3rd, 2022 -- Oct 19th, 2022 &  & Reject &  & Aug 30th, 2022 & 
& significant at $5\%$ &  \\ 
&  &  &  &  &  &  &  &  & break found also using $\alpha =0,$ at the same
date &  \\ 
&  &  &  &  &  &  &  &  & break found also using $\alpha =0.99,$ Sep 12th,
2022 &  \\ 
&  &  &  &  &  &  &  &  &  &  \\ 
& $4$ &  & Aug 30th, 2022 -- Oct 19th, 2022 &  & Not reject &  &  &  & no
break found even at $10\%$, or with $\alpha =0,$ $0.99$ &  \\ 
&  &  &  &  &  &  &  &  & $\begin{array}{c}
\widehat{\sigma }_{N}^{2}=0.006 \\ 
sk=-1.234 \\ 
ku=1.635\end{array}$ &  \\ 
&  &  &  &  &  &  &  &  &  &  \\ 
& $5$ &  & Jan 3rd, 2022 -- Aug 29th, 2022 &  & Reject &  & Jul 13th, 2022 & 
& significant at $5\%$ &  \\ 
&  &  &  &  &  &  &  &  & break found also using $\alpha =0.99$, at the same
date &  \\ 
&  &  &  &  &  &  &  &  & break found also using $\alpha =0$, at Jul 14th,
2022 &  \\ 
&  &  &  &  &  &  &  &  &  &  \\ 
& $6$ &  & Jul 13th, 2022 -- Aug 29th, 2022 &  & Not reject &  &  &  & no
break found even at $10\%$, or with $\alpha =0,$ $0.99$ &  \\ 
&  &  &  &  &  &  &  &  & $\begin{array}{c}
\widehat{\sigma }_{N}^{2}=0.005 \\ 
sk=1.233 \\ 
ku=1.656\end{array}$ &  \\ 
&  &  &  &  &  &  &  &  &  &  \\ 
& $7$ &  & Jan 3rd, 2022 -- Jul 12th, 2022 &  & Reject &  & Apr 21st, 2022 & 
& significant at $5\%$ &  \\ 
&  &  &  &  &  &  &  &  & break found also using $\alpha =0,0.99$, at the
same date &  \\ 
&  &  &  &  &  &  &  &  &  &  \\ 
& $8$ &  & Jan 3rd, 2022 -- Apr 20th, 2022 &  & Reject &  & Mar 20th, 2022 & 
& significant at $5\%$ &  \\ 
&  &  &  &  &  &  &  &  & break found also using $\alpha =0,0.99$, at the
same date &  \\ 
&  &  &  &  &  &  &  &  &  &  \\ 
& $9$ &  & Jan 3rd, 2022 -- Mar 19th, 2022 &  & Not reject &  &  &  & no
break found even at $10\%$, or with $\alpha =0,$ $0.99$ &  \\ 
&  &  &  &  &  &  &  &  & $\begin{array}{c}
\widehat{\sigma }_{N}^{2}=0.002 \\ 
sk=-1.269 \\ 
ku=1.734\end{array}$ &  \\ 
&  &  &  &  &  &  &  &  &  &  \\ 
& $10$ &  & Mar 20th, 2022 -- Apr 20th, 2022 &  & Not reject &  &  &  & no
break found even at $10\%$, or with $\alpha =0,$ $0.99$ &  \\ 
&  &  &  &  &  &  &  &  & $\begin{array}{c}
\widehat{\sigma }_{N}^{2}=0.003 \\ 
sk=1.247 \\ 
ku=1.645\end{array}$ &  \\ 
&  &  &  &  &  &  &  &  &  &  \\ 
& $11$ &  & Apr 21st, 2022 -- Jul 12th, 2022 &  & Not reject &  &  &  & no
break found even at $10\%$, or with $\alpha =0,$ $0.99$ &  \\ 
&  &  &  &  &  &  &  &  & $\begin{array}{c}
\widehat{\sigma }_{N}^{2}=0.005 \\ 
sk=-1.291 \\ 
ku=1.807\end{array}$ &  \\ 
&  &  &  &  &  &  &  &  &  &  \\ \hline\hline
\end{tabular}

}
\par
{\footnotesize 
\begin{tablenotes}
      \tiny
            \item We have used the estimator of the variance $\widehat{\sigma }_{N}^{2}$
defined in (\ref{sighat}). As far as the other descriptive statistics are concerned,  \textquotedblleft $sk$\textquotedblright\ and
\textquotedblleft $ku$\textquotedblright\ represent overall measures of
skewness and kurtosis respectively, computed within each regime, and across sampling points (see equations (\ref{sk}) and (\ref{ku}) in the Supplement).

\end{tablenotes}
}
\end{table*}

The results in Table \ref{tab:empiricsSPD} show a much richer picture that
changes in the mean alone. Interestingly, the same changepoints are found
across all values of $\alpha $, including $\alpha =0$; the only difference
is in the date of the break estimated between July and October $2022$, which
appears to be estimated $8$ trading days later when using larger values of $%
\alpha $. Otherwise, results are exactly the same; indeed, we also
experimented with other values, but results were the same even in those
cases. The same robustness was found when altering other specifications of
the procedure, e.g. the estimation the covariance kernel or of the number of
terms $\widehat{M}$ in the KL expansion. The estimated breakdates are, at
least in some cases, highly suggestive; interestingly, in all cases, a
change in regime corresponds to a change in the sign of our measure of
skewness, which confirms the stylised fact that skewness is time-varying (%
\citealp{alles1994regularities}; \citealp{bekaert1998distributional}). The
first break, recorded at March 20th, $2022$, corresponds to a peak in the
S\&P 500, after which the market entered a bear phase to stay below that
peak until July $2023$; after removing the first 21 observations, our
month-on-month return series starts effectively in February, so that the
first regime - characterised by a strongly negative measure of skewness -
reflects the uncertainty due to the war in Ukraine and its impact on the
global economy. The second regime, between March 20th and April 20th, is
characterised by a positive skewness, possibly indicating that the market -
after the stalling of the Russian offensive - was expecting an upward price
movement. This did not materialise, and in April $2022$ the market
experienced a strong correction, partly also due to inflation expectation
and underperformance of high-tech firms.\footnote{%
https://www.marketwatch.com/story/the-stock-market-swoon-just-sent-the-s-p-500-into-its-second-correction-of-2022-11651265882%
} After April 20th the market entered a bear phase, chracterised by negative
skewness, until a turning point was reached on July13th, $2022$ on account
of the FED ending (temporarily) its rate hiking. A correction occurred after
August 30th, $2022$, with a slump that lasted until approximately the second
half of October (the changepoint was recorded on October 20th, $2022$). From
thereon, the market started a rebound which lasted for the remainder of our
sample period; during this long horizon, the market was again characterised
by positive skewness.%

\section{Discussion and conclusions\label{conclusions}}

In this paper, we propose a family of weighted statistics to detect
changepoints possibly dependent, multivariate functional data. Although we
focus our exposition on the well-studied case of changes in the mean, our
tests can be applied to much more general changepoint problems, such as
detecting changes in the whole distribution. We base our test statistics on
the notion of \textit{energy distance}, a recently proposed measure of
proximity between distributions; we use a 
version of the (empirical) energy distance which is particularly suited to
determining the equality of the first moment of random variables, showing
that, under the null of no breaks, this is related to the familiar CUSUM\
process. Our statistics can be applied under very general forms of (weak)
serial dependence, thus being suitable for the analysis of several datasets,
including meteorological, financial and economic time series. By using a set
of weights which place more emphasis on observations occurring close to the
sample endpoints, we are able to detect changepoints occurring very close to
the beginning/end of the sample. Also, our approach 
is sufficiently flexible to allow for generalisations to e.g. testing for
changepoints in the (marginal) distributions of a sequence. In particular,
our approach is based on checking whether \textit{expectations} of \textit{%
functions} 
of our data remain constant over time; consequently, we can use all the
technology available in the literature, such as e.g. binary segmentation in
order to detect (and estimate the number and location of) multiple
changepoints. An important feature of our procedures is its computational
simplicity: critical values can be derived with arbitrary precision, and
this requires only the eigenvalues of the covariance operator of the data,
which can be quickly computed via any available statistical package. Our
simulations show that our statistics have excellent finite sample
performance even for small samples, thus making their use possible in
virtually all contexts involving FDA.

This work leads to several possible future directions, including extensions
to energy distances for functional data beyond the case $\eta =2$, and more
broadly further exploration of generalized energy distances. The use of the
characteristic function in Section \ref{distrib} can be viewed as a
finite-dimensional approximation of the characteristic \textit{functional};
an interesting direction would be to more deeply explore finite-dimensional
approximations of the characteristic functional and similar transformations
in the context of functional time series.


\begin{adjustwidth}{-0pt}{-0pt}

{\footnotesize {\ 
\bibliographystyle{chicago}
\bibliography{BHTbiblio}
} }

\end{adjustwidth}

\newpage 
\newpage

\clearpage
\renewcommand*{\thesection}{\Alph{section}}

\setcounter{section}{0} \setcounter{subsection}{-1} %
\setcounter{subsubsection}{-1} \setcounter{equation}{0} \setcounter{lemma}{0}
\setcounter{theorem}{0} \renewcommand{\theassumption}{A.\arabic{assumption}} 
\renewcommand{\thetheorem}{A.\arabic{theorem}} \renewcommand{\thelemma}{A.%
\arabic{lemma}} \renewcommand{\theproposition}{A.\arabic{proposition}} %
\renewcommand{\thecorollary}{A.\arabic{corollary}} \renewcommand{%
\theequation}{A.\arabic{equation}}

\section{Further Monte Carlo evidence and guidelines\label{furtherMC}}

\subsection{Empirical rejection frequencies under the null: further results}

We complement the results in Table \ref{tab:ERF} by considering the cases of 
\textit{i.i.d.} data with measurement error, and the case of serially
dependent data without measurement error.

\bigskip

\begin{table*}[h!]
\caption{{\protect\footnotesize {Empirical rejection frequencies under the
null of no changepoint}}}
\label{tab:ERFbis}\centering
{\footnotesize {\ }}
\par
\resizebox{\textwidth}{!}{

\begin{tabular}{lllllllllllllllllllll}
\hline\hline
&  &  &  &  &  &  &  &  &  &  &  &  &  &  &  &  &  &  &  &  \\ 
&  &  & \multicolumn{8}{c}{serial dependence, no measurement error} & 
\multicolumn{1}{c}{} & \multicolumn{1}{c}{} & \multicolumn{8}{c}{\textit{i.i.d. }case with measurement error} \\ 
&  &  & \multicolumn{1}{c}{} & \multicolumn{1}{c}{} & \multicolumn{1}{c}{} & 
\multicolumn{1}{c}{} & \multicolumn{1}{c}{} & \multicolumn{1}{c}{} & 
\multicolumn{1}{c}{} & \multicolumn{1}{c}{} & \multicolumn{1}{c}{} & 
\multicolumn{1}{c}{} & \multicolumn{1}{c}{} & \multicolumn{1}{c}{} & 
\multicolumn{1}{c}{} & \multicolumn{1}{c}{} & \multicolumn{1}{c}{} & 
\multicolumn{1}{c}{} & \multicolumn{1}{c}{} & \multicolumn{1}{c}{} \\ 
\multicolumn{1}{c}{$N$} & $\alpha $ &  & \multicolumn{1}{c}{$0.00$} & 
\multicolumn{1}{c}{$0.10$} & \multicolumn{1}{c}{$0.25$} & \multicolumn{1}{c}{$0.50$} & \multicolumn{1}{c}{$0.75$} & \multicolumn{1}{c}{$0.85$} & 
\multicolumn{1}{c}{$0.95$} & \multicolumn{1}{c}{$0.99$} & \multicolumn{1}{c}{
} & \multicolumn{1}{|c}{} & \multicolumn{1}{c}{$0.00$} & \multicolumn{1}{c}{$0.10$} & \multicolumn{1}{c}{$0.25$} & \multicolumn{1}{c}{$0.50$} & 
\multicolumn{1}{c}{$0.75$} & \multicolumn{1}{c}{$0.85$} & \multicolumn{1}{c}{$0.95$} & \multicolumn{1}{c}{$0.99$} \\ 
\multicolumn{1}{c}{} &  &  & \multicolumn{1}{c}{} & \multicolumn{1}{c}{} & 
\multicolumn{1}{c}{} & \multicolumn{1}{c}{} & \multicolumn{1}{c}{} & 
\multicolumn{1}{c}{} & \multicolumn{1}{c}{} & \multicolumn{1}{c}{} & 
\multicolumn{1}{c}{} & \multicolumn{1}{|c}{} & \multicolumn{1}{c}{} & 
\multicolumn{1}{c}{} & \multicolumn{1}{c}{} & \multicolumn{1}{c}{} & 
\multicolumn{1}{c}{} & \multicolumn{1}{c}{} & \multicolumn{1}{c}{} & 
\multicolumn{1}{c}{} \\ 
\multicolumn{1}{c}{$50$} &  &  & \multicolumn{1}{c}{$0.020$} & 
\multicolumn{1}{c}{$0.019$} & \multicolumn{1}{c}{$0.016$} & 
\multicolumn{1}{c}{$0.013$} & \multicolumn{1}{c}{$0.013$} & 
\multicolumn{1}{c}{$0.016$} & \multicolumn{1}{c}{$0.014$} & 
\multicolumn{1}{c}{$0.014$} & \multicolumn{1}{c}{} & \multicolumn{1}{|c}{} & 
\multicolumn{1}{c}{$0.052$} & \multicolumn{1}{c}{$0.057$} & 
\multicolumn{1}{c}{$0.049$} & \multicolumn{1}{c}{$0.053$} & 
\multicolumn{1}{c}{$0.046$} & \multicolumn{1}{c}{$0.053$} & 
\multicolumn{1}{c}{$0.052$} & \multicolumn{1}{c}{$0.059$} \\ 
\multicolumn{1}{c}{$100$} &  &  & \multicolumn{1}{c}{$0.049$} & 
\multicolumn{1}{c}{$0.045$} & \multicolumn{1}{c}{$0.043$} & 
\multicolumn{1}{c}{$0.036$} & \multicolumn{1}{c}{$0.033$} & 
\multicolumn{1}{c}{$0.027$} & \multicolumn{1}{c}{$0.022$} & 
\multicolumn{1}{c}{$0.018$} & \multicolumn{1}{c}{} & \multicolumn{1}{|c}{} & 
\multicolumn{1}{c}{$0.045$} & \multicolumn{1}{c}{$0.051$} & 
\multicolumn{1}{c}{$0.052$} & \multicolumn{1}{c}{$0.048$} & 
\multicolumn{1}{c}{$0.056$} & \multicolumn{1}{c}{$0.044$} & 
\multicolumn{1}{c}{$0.054$} & \multicolumn{1}{c}{$0.047$} \\ 
\multicolumn{1}{c}{$150$} &  &  & \multicolumn{1}{c}{$0.039$} & 
\multicolumn{1}{c}{$0.038$} & \multicolumn{1}{c}{$0.036$} & 
\multicolumn{1}{c}{$0.032$} & \multicolumn{1}{c}{$0.026$} & 
\multicolumn{1}{c}{$0.027$} & \multicolumn{1}{c}{$0.023$} & 
\multicolumn{1}{c}{$0.019$} & \multicolumn{1}{c}{} & \multicolumn{1}{|c}{} & 
\multicolumn{1}{c}{$0.044$} & \multicolumn{1}{c}{$0.052$} & 
\multicolumn{1}{c}{$0.051$} & \multicolumn{1}{c}{$0.050$} & 
\multicolumn{1}{c}{$0.054$} & \multicolumn{1}{c}{$0.045$} & 
\multicolumn{1}{c}{$0.066$} & \multicolumn{1}{c}{$0.049$} \\ 
\multicolumn{1}{c}{$200$} &  &  & \multicolumn{1}{c}{$0.058$} & 
\multicolumn{1}{c}{$0.055$} & \multicolumn{1}{c}{$0.055$} & 
\multicolumn{1}{c}{$0.052$} & \multicolumn{1}{c}{$0.047$} & 
\multicolumn{1}{c}{$0.035$} & \multicolumn{1}{c}{$0.034$} & 
\multicolumn{1}{c}{$0.033$} & \multicolumn{1}{c}{} & \multicolumn{1}{|c}{} & 
\multicolumn{1}{c}{$0.055$} & \multicolumn{1}{c}{$0.052$} & 
\multicolumn{1}{c}{$0.052$} & \multicolumn{1}{c}{$0.050$} & 
\multicolumn{1}{c}{$0.056$} & \multicolumn{1}{c}{$0.050$} & 
\multicolumn{1}{c}{$0.056$} & \multicolumn{1}{c}{$0.059$} \\ 
\multicolumn{1}{c}{} &  &  & \multicolumn{1}{c}{} & \multicolumn{1}{c}{} & 
\multicolumn{1}{c}{} & \multicolumn{1}{c}{} & \multicolumn{1}{c}{} & 
\multicolumn{1}{c}{} & \multicolumn{1}{c}{} & \multicolumn{1}{c}{} & 
\multicolumn{1}{c}{} & \multicolumn{1}{|c}{} & \multicolumn{1}{c}{} & 
\multicolumn{1}{c}{} & \multicolumn{1}{c}{} & \multicolumn{1}{c}{} & 
\multicolumn{1}{c}{} & \multicolumn{1}{c}{} & \multicolumn{1}{c}{} & 
\multicolumn{1}{c}{} \\ \hline\hline
\end{tabular}

}
\par
{\footnotesize 
\begin{tablenotes}
      \tiny
            \item The table contains the empirical rejection frequencies under the null of no changepoint, using $M=40$ orthonormal bases in (\ref{dgp}), for tests at a $5\%$ nominal level. The specifications of (\ref{dgp}) are described in the main text.
            
\end{tablenotes}
}
\end{table*}

\subsection{Empirical rejection frequencies under the alternative: further
results\label{alternative_further}}

We begin by reporting the power against one changepoint, with the same
design as in equation (\ref{amoc}) using $N=100$. As can be seen in Figures %
\ref{fig:FigM100} and \ref{fig:FigE100}, the results are similar, although
the test is less powerful compared to the results in Figures \ref%
{fig:FigM200} and \ref{fig:FigE200}, which is expected due to the smaller
value of $N$. In particular, in the presence of an end-of-sample break,
power is ensured only for large values of $\alpha $.

\begin{figure}[h!]
\caption{{\protect\footnotesize {Empirical rejection frequencies under a
mid-sample break with \textit{i.i.d.} data and data with serial dependence
and measurement error}}}
\label{fig:FigM100}\centering
\hspace{-2.5cm} 
\begin{minipage}{0.7\textwidth}
\centering
    \includegraphics[scale=0.65]{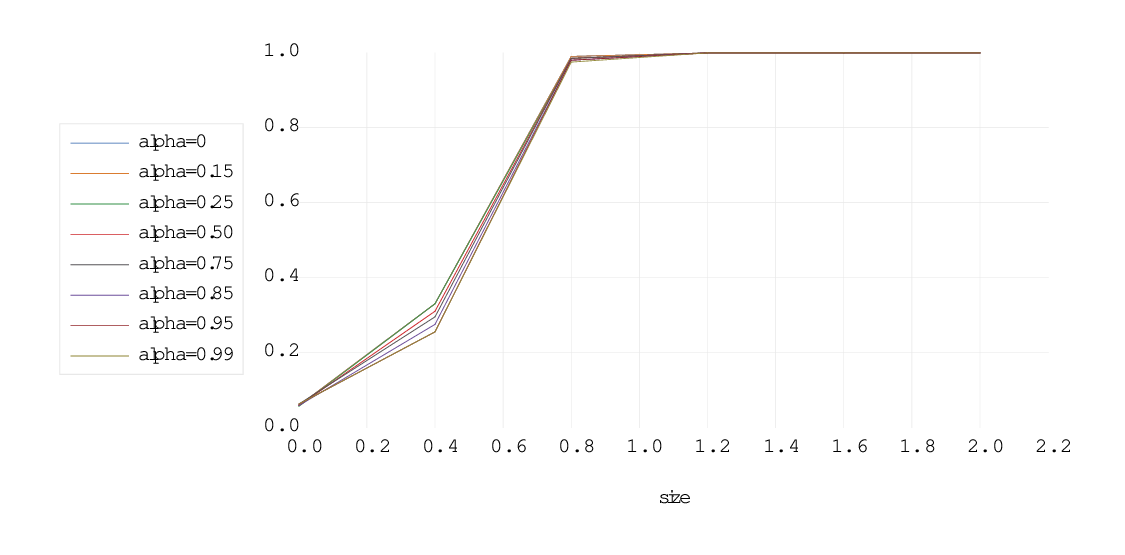}
    \label{fig:t11}
\end{minipage}
\par
\hspace{-2.5cm} 
\begin{minipage}{0.7\textwidth}
\centering
    \includegraphics[scale=0.65]{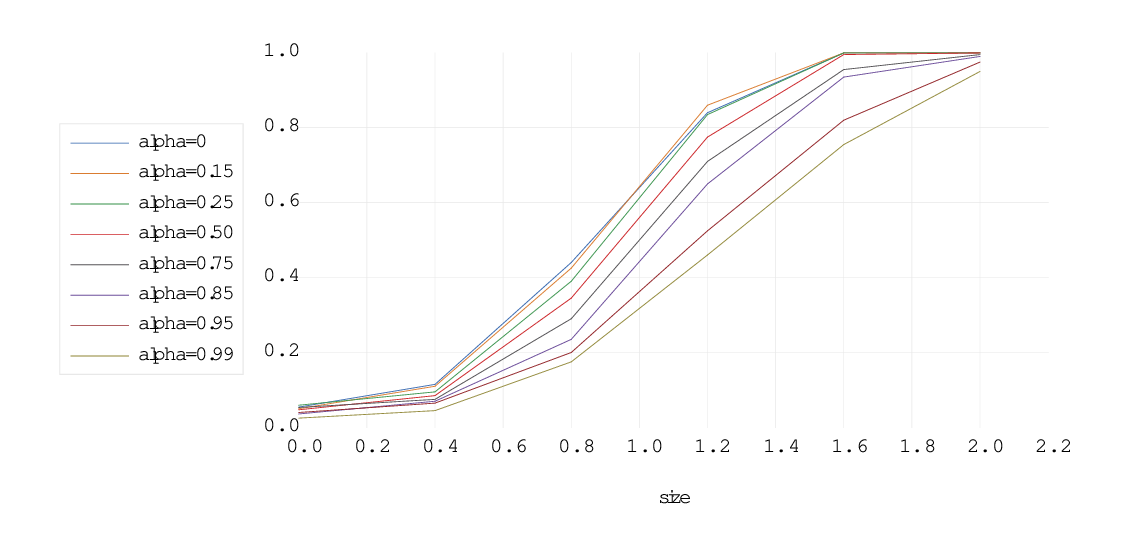}
    \label{fig:t13}
\end{minipage}
\par
\end{figure}

\medskip

\begin{figure}[h]
\caption{{\protect\footnotesize {Empirical rejection frequencies under a
mid-sample break with \textit{i.i.d.} data and data with serial dependence
and measurement error}}}
\label{fig:FigE100}\centering
\hspace{-2.5cm} 
\begin{minipage}{0.7\textwidth}
\centering
    \includegraphics[scale=0.65]{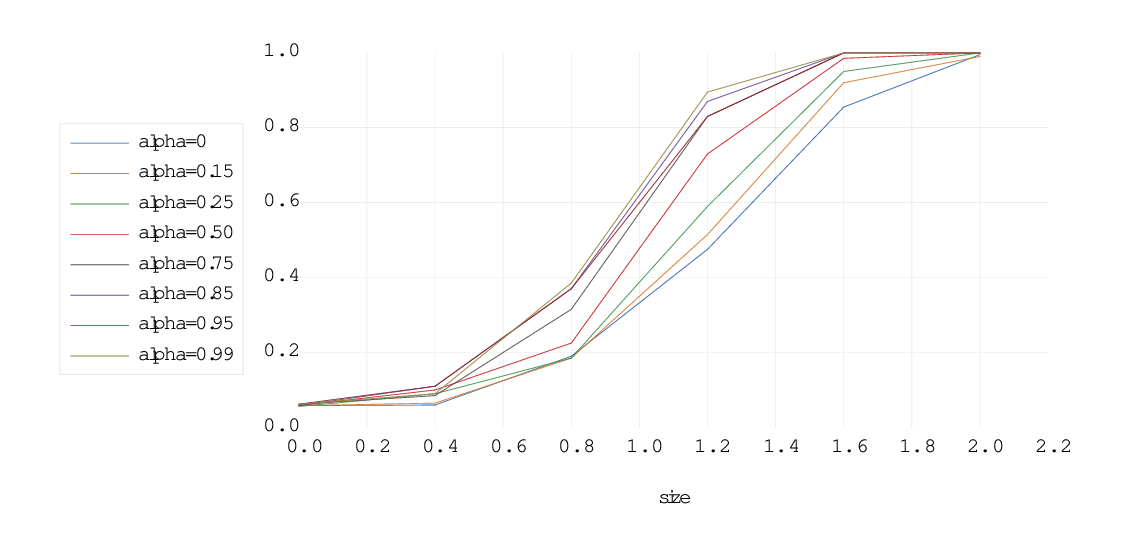}
    \label{fig:t111}
\end{minipage}
\par
\hspace{-2.5cm} 
\begin{minipage}{0.7\textwidth}
\centering
    \includegraphics[scale=0.65]{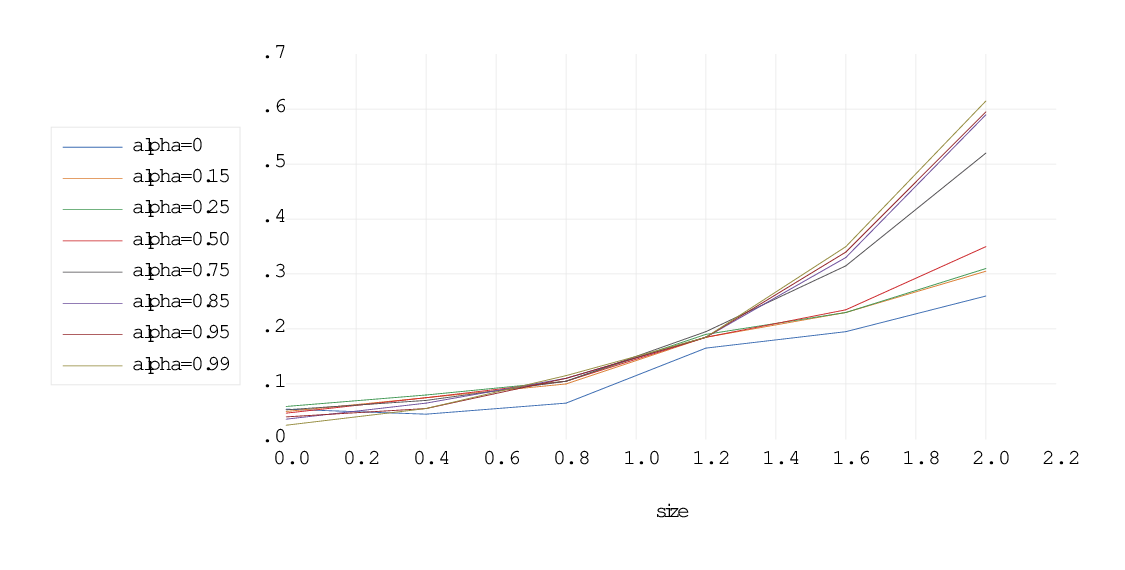}
    \label{fig:t113}
\end{minipage}
\par
\end{figure}

\medskip

We now report the median values of the estimated breakdate $k^{\ast }$, in
the case of a single changepoint (when this is detected), under the same
set-up as in Section \ref{simulations} - see equation (\ref{amoc}) in
particular. Results in Tables \ref{tab:MedianMid1}-\ref{tab:MedianLate2}
should be read in conjunction with Figures \ref{fig:FigM100}-\ref%
{fig:FigE100} and \ref{fig:FigM200}-\ref{fig:FigE200}, and broadly confirm
the theory spelled out in Theorem \ref{th-3}. In the case of mid-sample
breaks, the estimator of $k^{\ast }$ is usually very good when a changepoint
is detected (Tables \ref{tab:MedianMid1}-\ref{tab:MedianMid2}), even for
small break sizes like $\left\Vert \mathcal{\delta }\right\Vert =0.4$; this
is true across all values of $\alpha $, although, in the case of small
breaks ($\left\Vert \mathcal{\delta }\right\Vert =0.4$), the performance of $%
\widehat{k}$ when $\alpha $ gets closer to $1$ seems to worsen. Conversely,
when breaks occur close to the end of the sample ($k^{\ast }=0.9N$), results
in Table \ref{tab:MedianLate1} and \ref{tab:MedianLate2} differ dramatically
across $\alpha $: as expected, when $\alpha $ increases, $\widehat{k}$
performs better, and it performs very well when $\left\Vert \mathcal{\delta }%
\right\Vert \geq 0.8$ (and even more so when $N=200$). Interestingly, in
this case $\widehat{k}$ appears to have a downward bias, which vanishes as $%
N\left\Vert \mathcal{\delta }\right\Vert $ increases.

\medskip

\begin{table*}[t]
\caption{{\protect\footnotesize {Median estimated changepoints under a
mid-sample changepoint, \textit{i.i.d.} data, no measurement error}}}
\label{tab:MedianMid1}\centering
{\footnotesize {\ }}
\par
{\footnotesize {\ }}
\par
\resizebox{\textwidth}{!}{
\begin{tabular}{cccccccccccccccccccccc}
\hline\hline
&  &  &  &  &  &  &  &  &  &  &  &  &  &  &  &  &  &  &  &  &  \\ 
& \multicolumn{1}{c}{} & \multicolumn{1}{c}{} & \multicolumn{9}{c}{$N=100$, $k^{\ast }=50$} &  & \multicolumn{9}{c}{$N=200$, $k^{\ast }=100$} \\ 
&  &  & \multicolumn{1}{|c}{} &  &  &  &  &  &  &  &  &  & 
\multicolumn{1}{|c}{} &  &  &  &  &  &  &  &  \\ 
& $\alpha $ &  & \multicolumn{1}{|c}{} & $0.00$ & $0.10$ & $0.25$ & $0.50$ & 
$0.75$ & $0.85$ & $0.95$ & $0.99$ &  & \multicolumn{1}{|c}{} & $0.00$ & $0.10
$ & $0.25$ & $0.50$ & $0.75$ & $0.85$ & $0.95$ & $0.99$ \\ 
$\left\Vert \mathcal{\delta }\right\Vert $ &  &  & \multicolumn{1}{|c}{} & 
&  &  &  &  &  &  &  &  & \multicolumn{1}{|c}{} &  &  &  &  &  &  &  &  \\ 
&  &  & \multicolumn{1}{|c}{} &  &  &  &  &  &  &  &  &  & 
\multicolumn{1}{|c}{} &  &  &  &  &  &  &  &  \\ 
$0.40$ &  &  & \multicolumn{1}{|c}{} & $50$ & $50$ & $50$ & $50$ & $49$ & $50
$ & $49$ & $50$ &  & \multicolumn{1}{|c}{} & $100$ & $101$ & $101$ & $101$ & 
$101$ & $100$ & $100$ & $100$ \\ 
$0.80$ &  &  & \multicolumn{1}{|c}{} & $50$ & $50$ & $50$ & $50$ & $50$ & $50
$ & $50$ & $50$ &  & \multicolumn{1}{|c}{} & $100$ & $100$ & $100$ & $100$ & 
$100$ & $100$ & $100$ & $100$ \\ 
$1.20$ &  &  & \multicolumn{1}{|c}{} & $50$ & $50$ & $50$ & $50$ & $50$ & $50
$ & $50$ & $50$ &  & \multicolumn{1}{|c}{} & $100$ & $100$ & $100$ & $100$ & 
$100$ & $100$ & $100$ & $100$ \\ 
$1.60$ &  &  & \multicolumn{1}{|c}{} & $50$ & $50$ & $50$ & $50$ & $50$ & $50
$ & $50$ & $50$ &  & \multicolumn{1}{|c}{} & $100$ & $100$ & $100$ & $100$ & 
$100$ & $100$ & $100$ & $100$ \\ 
$2.00$ &  &  & \multicolumn{1}{|c}{} & $50$ & $50$ & $50$ & $50$ & $50$ & $50
$ & $50$ & $50$ &  & \multicolumn{1}{|c}{} & $100$ & $100$ & $100$ & $100$ & 
$100$ & $100$ & $100$ & $100$ \\ 
&  &  &  &  &  &  &  &  &  &  &  &  &  &  &  &  &  &  &  &  &  \\ 
\hline\hline
\end{tabular}}
\par
{\footnotesize 
\begin{tablenotes}
      \tiny
            \item The table contains the median estimated changepoint in the presence of a mid-sample break, for different values of $\alpha$, $\left\Vert 
\mathcal{\delta }\right\Vert $, and sample sizes $N$, with $\rho=0$. All figures are based on $200$ replications.
            
\end{tablenotes}
}
\end{table*}

\medskip

\begin{table*}[b]
\caption{{\protect\footnotesize {Median estimated changepoints under a
mid-sample changepoint, serially dependent data with measurement error}}}
\label{tab:MedianMid2}\centering
{\footnotesize {\ }}
\par
\resizebox{\textwidth}{!}{
\begin{tabular}{cccccccccccccccccccccc}
\hline\hline
&  &  &  &  &  &  &  &  &  &  &  &  &  &  &  &  &  &  &  &  &  \\ 
&  &  & \multicolumn{9}{c}{$N=100$, $k^{\ast }=50$} &  & \multicolumn{9}{c}{$N=200$, $k^{\ast }=100$} \\ 
&  &  & \multicolumn{1}{|c}{} &  &  &  &  &  &  &  &  &  & 
\multicolumn{1}{|c}{} &  &  &  &  &  &  &  &  \\ 
& $\alpha $ &  & \multicolumn{1}{|c}{} & $0.00$ & $0.10$ & $0.25$ & $0.50$ & 
$0.75$ & $0.85$ & $0.95$ & $0.99$ &  & \multicolumn{1}{|c}{} & $0.00$ & $0.10
$ & $0.25$ & $0.50$ & $0.75$ & $0.85$ & $0.95$ & $0.99$ \\ 
$\left\Vert \mathcal{\delta }\right\Vert $ &  &  & \multicolumn{1}{|c}{} & 
&  &  &  &  &  &  &  &  & \multicolumn{1}{|c}{} &  &  &  &  &  &  &  &  \\ 
&  &  & \multicolumn{1}{|c}{} &  &  &  &  &  &  &  &  &  & 
\multicolumn{1}{|c}{} &  &  &  &  &  &  &  &  \\ 
$0.40$ &  &  & \multicolumn{1}{|c}{} & $49$ & $49.5$ & $49.5$ & $54$ & $58.5$
& $58.5$ & $75$ & $88.5$ &  & \multicolumn{1}{|c}{} & $97$ & $104$ & $107$ & 
$107$ & $113$ & $111$ & $114.5$ & $119$ \\ 
$0.80$ &  &  & \multicolumn{1}{|c}{} & $50$ & $50$ & $51$ & $51$ & $52$ & $52
$ & $52$ & $51$ &  & \multicolumn{1}{|c}{} & $99$ & $100$ & $100$ & $101$ & $101$ & $100$ & $100$ & $100$ \\ 
$1.20$ &  &  & \multicolumn{1}{|c}{} & $50$ & $50$ & $50$ & $50$ & $50$ & $50
$ & $50$ & $50$ &  & \multicolumn{1}{|c}{} & $100$ & $100$ & $100$ & $100$ & 
$100$ & $100$ & $100$ & $100$ \\ 
$1.60$ &  &  & \multicolumn{1}{|c}{} & $50$ & $50$ & $50$ & $50$ & $50$ & $50
$ & $50$ & $50$ &  & \multicolumn{1}{|c}{} & $100$ & $100$ & $100$ & $100$ & 
$100$ & $100$ & $100$ & $100$ \\ 
$2.00$ &  &  & \multicolumn{1}{|c}{} & $50$ & $50$ & $50$ & $50$ & $50$ & $50
$ & $50$ & $50$ &  & \multicolumn{1}{|c}{} & $100$ & $100$ & $100$ & $100$ & 
$100$ & $100$ & $100$ & $100$ \\ 
&  &  &  &  &  &  &  &  &  &  &  &  &  &  &  &  &  &  &  &  &  \\ 
\hline\hline
\end{tabular}
}
\par
{\footnotesize 
\begin{tablenotes}
      \tiny
            \item The table contains the median estimated changepoint in the presence of a mid-sample break, for different values of $\alpha$, $\left\Vert 
\mathcal{\delta }\right\Vert $, and sample sizes $N$, with $\rho=0.5$. All figures are based on $200$ replications.
            
\end{tablenotes}
}
\end{table*}

\medskip

\begin{table*}[t]
\caption{{\protect\footnotesize {Median estimated changepoints under an
end-of-sample changepoint, \textit{i.i.d.} data, no measurement error}}}
\label{tab:MedianLate1}\centering
{\footnotesize {\ }}
\par
\resizebox{\textwidth}{!}{
\begin{tabular}{cccccccccccccccccccccc}
\hline\hline
&  &  &  &  &  &  &  &  &  &  &  &  &  &  &  &  &  &  &  &  &  \\ 
&  &  & \multicolumn{9}{c}{$N=100$, $k^{\ast }=90$} &  & \multicolumn{9}{c}{$N=200$, $k^{\ast }=180$} \\ 
&  &  & \multicolumn{1}{|c}{} &  &  &  &  &  &  &  &  &  & 
\multicolumn{1}{|c}{} &  &  &  &  &  &  &  &  \\ 
& $\alpha $ &  & \multicolumn{1}{|c}{} & $0.00$ & $0.10$ & $0.25$ & $0.50$ & 
$0.75$ & $0.85$ & $0.95$ & $0.99$ &  & \multicolumn{1}{|c}{} & $0.00$ & $0.10
$ & $0.25$ & $0.50$ & $0.75$ & $0.85$ & $0.95$ & $0.99$ \\ 
$\left\Vert \mathcal{\delta }\right\Vert $ &  &  & \multicolumn{1}{|c}{} & 
&  &  &  &  &  &  &  &  & \multicolumn{1}{|c}{} &  &  &  &  &  &  &  &  \\ 
&  &  & \multicolumn{1}{|c}{} &  &  &  &  &  &  &  &  &  & 
\multicolumn{1}{|c}{} &  &  &  &  &  &  &  &  \\ 
$0.40$ &  &  & \multicolumn{1}{|c}{} & $56$ & $57$ & $54.5$ & $56$ & $63$ & $81$ & $82$ & $87.5$ &  & \multicolumn{1}{|c}{} & $110.5$ & $106.5$ & $110.5$
& $145$ & $166.5$ & $175$ & $176$ & $177$ \\ 
$0.80$ &  &  & \multicolumn{1}{|c}{} & $57.5$ & $61$ & $62$ & $73$ & $85$ & $89$ & $90$ & $90$ &  & \multicolumn{1}{|c}{} & $145.5$ & $148.5$ & $161$ & $176$ & $179$ & $180$ & $180$ & $180$ \\ 
$1.20$ &  &  & \multicolumn{1}{|c}{} & $71$ & $78$ & $82$ & $89$ & $90$ & $90
$ & $90$ & $90$ &  & \multicolumn{1}{|c}{} & $161$ & $166$ & $175$ & $179$ & 
$180$ & $180$ & $180$ & $180$ \\ 
$1.60$ &  &  & \multicolumn{1}{|c}{} & $81$ & $84$ & $88$ & $90$ & $90$ & $90
$ & $90$ & $90$ &  & \multicolumn{1}{|c}{} & $177.5$ & $178$ & $179$ & $179$
& $180$ & $180$ & $180$ & $180$ \\ 
$2.00$ &  &  & \multicolumn{1}{|c}{} & $86$ & $87$ & $89$ & $90$ & $90$ & $90
$ & $90$ & $90$ &  & \multicolumn{1}{|c}{} & $178$ & $178$ & $179$ & $180$ & 
$180$ & $180$ & $180$ & $180$ \\ 
&  &  &  &  &  &  &  &  &  &  &  &  &  &  &  &  &  &  &  &  &  \\ 
\hline\hline
\end{tabular}}
\par
{\footnotesize 
\begin{tablenotes}
      \tiny
            \item The table contains the median estimated changepoint in the presence of a mid-sample break, for different values of $\alpha$, $\left\Vert 
\mathcal{\delta }\right\Vert $, and sample sizes $N$, with $\rho=0.5$ and $\sigma_{v}=0.25$. All figures are based on $200$ replications.
            
\end{tablenotes}
}
\end{table*}

\medskip

\begin{table*}[b]
\caption{{\protect\footnotesize {Median estimated changepoints under an
end-of-sample changepoint, serially dependent data with measurement error}}}
\label{tab:MedianLate2}\centering
{\footnotesize {\ }}
\par
\resizebox{\textwidth}{!}{
\begin{tabular}{cccccccccccccccccccccc}
\hline\hline
&  &  &  &  &  &  &  &  &  &  &  &  &  &  &  &  &  &  &  &  &  \\ 
&  &  & \multicolumn{9}{c}{$N=100$, $k^{\ast }=90$} &  & \multicolumn{9}{c}{$N=200$, $k^{\ast }=180$} \\ 
&  &  & \multicolumn{1}{|c}{} &  &  &  &  &  &  &  &  &  & 
\multicolumn{1}{|c}{} &  &  &  &  &  &  &  &  \\ 
& $\alpha $ &  & \multicolumn{1}{|c}{} & $0.00$ & $0.10$ & $0.25$ & $0.50$ & 
$0.75$ & $0.85$ & $0.95$ & $0.99$ &  & \multicolumn{1}{|c}{} & $0.00$ & $0.10
$ & $0.25$ & $0.50$ & $0.75$ & $0.85$ & $0.95$ & $0.99$ \\ 
$\left\Vert \mathcal{\delta }\right\Vert $ &  &  & \multicolumn{1}{|c}{} & 
&  &  &  &  &  &  &  &  & \multicolumn{1}{|c}{} &  &  &  &  &  &  &  &  \\ 
&  &  & \multicolumn{1}{|c}{} &  &  &  &  &  &  &  &  &  & 
\multicolumn{1}{|c}{} &  &  &  &  &  &  &  &  \\ 
$0.40$ &  &  & \multicolumn{1}{|c}{} & $58$ & $57.5$ & $61$ & $57.5$ & $75$
& $75.5$ & $75.5$ & $82$ &  & \multicolumn{1}{|c}{} & $93$ & $89$ & $105.5$
& $118$ & $118$ & $105.5$ & $118$ & $118$ \\ 
$0.80$ &  &  & \multicolumn{1}{|c}{} & $54.5$ & $54.5$ & $54$ & $54$ & $61$
& $61$ & $68$ & $81.5$ &  & \multicolumn{1}{|c}{} & $94$ & $90$ & $113.5$ & $116.5$ & $123$ & $142$ & $142$ & $165$ \\ 
$1.20$ &  &  & \multicolumn{1}{|c}{} & $53$ & $54.5$ & $53$ & $61$ & $76$ & $77$ & $83$ & $87.5$ &  & \multicolumn{1}{|c}{} & $103$ & $102$ & $116$ & $131.5$ & $163$ & $172$ & $177$ & $178$ \\ 
$1.60$ &  &  & \multicolumn{1}{|c}{} & $53$ & $54$ & $53$ & $73$ & $81$ & $87
$ & $89.5$ & $90$ &  & \multicolumn{1}{|c}{} & $111$ & $111$ & $125$ & $158$
& $178$ & $179$ & $180$ & $180$ \\ 
$2.00$ &  &  & \multicolumn{1}{|c}{} & $54$ & $58.5$ & $59$ & $76$ & $88$ & $90$ & $90$ & $90$ &  & \multicolumn{1}{|c}{} & $134$ & $143$ & $163$ & $179$
& $180$ & $180$ & $180$ & $180$ \\ 
&  &  &  &  &  &  &  &  &  &  &  &  &  &  &  &  &  &  &  &  &  \\ 
\hline\hline
\end{tabular}

}
\par
{\footnotesize 
\begin{tablenotes}
      \tiny
            \item The table contains the median estimated changepoint in the presence of an end-of-sample break, for different values of $\alpha$, $\left\Vert 
\mathcal{\delta }\right\Vert $, and sample sizes $N$, with $\rho=0.5$ and $\sigma_{v}=0.25$. All figures are based on $200$ replications.
            
\end{tablenotes}
}
\end{table*}

\clearpage

\subsection{Binary segmentation: pesudocode and Monte Carlo evidence\label%
{simulations_binary}}

We begin by reporting some pseudocode for the practical implementation of
the algorithm. Let, for short 
\begin{equation*}
\mathcal{Z}_{\ell ,u}^{k}=\left( \frac{k}{N}\left( 1-\frac{k}{N}\right)
\right) ^{-\alpha }V_{N}^{\left( \ell ,u\right) }\left( k\right) .
\end{equation*}%
The pseudocode is in Algorithm \ref{alg111} below.%

\bigskip

\begin{algorithm}[h!]
	\caption{Binary segmentation for functional data based on the empirical energy function: ENERGYSEG($\ell$, $u$, $\tau _{N}$)}
	\label{alg111}
	\begin{algorithmic}
		
		\Require starting index $\ell$; ending index $u$; threshold $\tau _{N}$
		\\
		\If{$u-\ell \leq 4$}  \State{STOP} 
		\ElsIf{$u-\ell > 4$} \State{Define $k_0=\sargmax_{\ell \leq k\leq u}\mathcal{Z}_{\ell,u}^{k}$, and $\mathcal{Z}=\mathcal{Z}_{\ell,u}^{k_0}$}%

		\EndIf
		\\
		\If{$\mathcal{Z}>\tau_N$} 
		\State{add $k_0$ to the
set of estimated changepoints.}
		\State{\textbf{run} ENERGYSEG($l$, $k_0$, 
$\tau _{N}$) and ENERGYSEG($k_0$, $u$, $\tau _{N}$)}

		\Else{}
STOP
		
		\EndIf

	\end{algorithmic} 	
\end{algorithm}

\bigskip

We now report a small Monte Carlo exercise to assess the performance of the
binary segmentation procedure discussed in Section \ref{segment}. In
particular, we consider the following DGP%
\begin{equation*}
X_{i}(t)=\sum_{j=1}^{R+1}\mu _{j}\left( t\right) I\left\{ k_{j-1}\leq
i<k_{j}\right\} +\sum_{\ell =1}^{M}\lambda _{\ell }^{1/2}\mathcal{Z}_{\ell
,i}\phi _{\ell }(t)+\nu _{i}\left( t\right) ,
\end{equation*}%
where the random part $\sum_{\ell =1}^{M}\lambda _{\ell }^{1/2}\mathcal{Z}%
_{\ell ,i}\phi _{\ell }(t)+\nu _{i}\left( t\right) $ is generated in the
same way as in Section \ref{simulations}, and we consider two mid-sample
changepoints (i.e., $R=2$) in the mean function $\mu _{j}(t)$, viz. 
\begin{equation}
\mu _{j}\left( t\right) =\left\{ 
\begin{array}{ll}
\mathcal{\delta }_{1}(t) & 1\leq i<k_{1}^{\ast } \\ 
\mathcal{\delta }_{2}(t) & k_{1}^{\ast }\leq i<k_{2}^{\ast } \\ 
\mathcal{\delta }_{3}(t) & k_{2}^{\ast }\leq i\geq N%
\end{array}%
\right. ,  \label{mean-seg}
\end{equation}%
with $k_{1}^{\ast }=\left\lfloor 0.35N\right\rfloor $ and $k_{2}^{\ast
}=\left\lfloor 0.7N\right\rfloor $, and $\delta_i(t)\equiv C_{\delta_i}$ are
constants, $C_{\delta_1}=0$, $C_{\delta_2}=2$ and $C_{ \mathcal{\delta }%
_{3}} =3$. When using Algorithm \ref{alg111}, we select the threshold 
\begin{equation}
\tau _{N}=c_{a}\log \log N,  \label{seg-thresh}
\end{equation}%
where $c_{a}$ is the critical value at nominal level $a=0.05$.

Results in Table \ref{tab:segment1} contain measures of location of the
estimator of the number of changepoints $\widehat{R}$, and the average
values, across simulations, of the estimated breakdates, using $500$
simulations. Results are obtained for $N=200$ and with $\mathcal{Z}_{\ell
,i} $ generated as \textit{i.i.d.} across $1\leq i\leq N$; we consider the
presence of measurement errors, setting $\sigma _{\nu }^{2}=0.25$ as in
Section \ref{simulations}, but in unreported experiments we noted that
having $\sigma _{\nu }^{2}=0$ does not change the results in any significant
way.

\bigskip

\begin{table*}[h]
\caption{{\protect\footnotesize {Measures of location for $\widehat{R}$ and
the estimated breakdates}}}
\label{tab:segment1}\centering
\par
{\footnotesize {\ }}
\par
{\footnotesize 
\begin{tabular}{ccccccccccccc}
\hline\hline
&  &  &  &  &  &  &  &  &  &  &  &  \\ 
\multicolumn{13}{c}{Measures of location for the estimated number of
changepoints $\widehat{R}$} \\ 
&  &  &  &  &  &  &  &  &  &  &  &  \\ 
&  & $\alpha $ &  & $0.00$ & $0.15$ & $0.25$ & $0.50$ & $0.75$ & $0.85$ & $%
0.95$ & $0.99$ &  \\ 
&  &  &  &  &  &  &  &  &  &  &  &  \\ 
& mean &  &  & $1.994$ & $1.998$ & $1.988$ & $1.984$ & $1.986$ & $1.990$ & $%
1.964$ & $2.000$ &  \\ 
& median &  &  & $2.000$ & $2.000$ & $2.000$ & $2.000$ & $2.000$ & $2.000$ & 
$2.000$ & $2.000$ &  \\ 
& min &  &  & $0.000$ & $0.000$ & $0.000$ & $0.000$ & $0.000$ & $0.000$ & $%
0.000$ & $0.000$ &  \\ 
& max &  &  & $3.000$ & $3.000$ & $3.000$ & $3.000$ & $3.000$ & $4.000$ & $%
4.000$ & $4.000$ &  \\ 
&  &  &  &  &  &  &  &  &  &  &  &  \\ \hline
&  &  &  &  &  &  &  &  &  &  &  &  \\ 
\multicolumn{13}{c}{Median estimated breakdates} \\ 
&  &  &  &  &  &  &  &  &  &  &  &  \\ 
&  & $\alpha $ &  & $0.00$ & $0.15$ & $0.25$ & $0.50$ & $0.75$ & $0.85$ & $%
0.95$ & $0.99$ &  \\ 
&  &  &  &  &  &  &  &  &  &  &  &  \\ 
& $\widehat{k}_{1}$ &  &  & $70$ & $70$ & $70$ & $70$ & $70$ & $70$ & $70$ & 
$70$ &  \\ 
& $\widehat{k}_{2}$ &  &  & $140$ & $140$ & $140$ & $140$ & $140$ & $140$ & $%
140$ & $140$ &  \\ 
&  &  &  &  &  &  &  &  &  &  &  &  \\ \hline\hline
\end{tabular}
}
\par
{\footnotesize \ }
\par
{\footnotesize 
\begin{tablenotes}
      \tiny
            \item The table contains various measures of location for $\widehat{R}$ and the median estimated breakdates under the alternative \eqref{mean-seg}; data have been generated as \textit{i.i.d.} with measurement error, with sample size $N=200$.
            
\end{tablenotes}
}
\end{table*}

\bigskip

In addition to investigating the performance of binary segmentation in the
presence of shifts in the mean as per (\ref{seg-thresh}), we also explore
its performance in the presence of changes in the distribution. In
particular, we consider the \textquotedblleft epidemic\textquotedblright\
alternative in a model with zero mean%
\begin{equation*}
{Y_{\ell }}(t)=\epsilon _{\ell }(t),
\end{equation*}%
and 
\begin{equation}
\epsilon _{\ell }(t)=\left\{ 
\begin{array}{ll}
\displaystyle\sum_{m=1}^{M}\lambda _{m}^{1/2}\mathcal{Z}_{m,\ell }\phi
_{m}(t) & 1\leq \ell <k_{1} \\ 
\displaystyle\sum_{m=1}^{M}\lambda _{m}^{1/2}t_{m,\ell }^{\left( 3\right)
}\phi _{m}(t) & k_{1}\leq \ell <k_{2} \\ 
\displaystyle\sum_{m=1}^{M}\lambda _{m}^{1/2}\mathcal{Z}_{m,\ell }\phi
_{m}(t) & k_{2}\leq \ell \leq N%
\end{array}%
\right. ,  \label{seg-distr}
\end{equation}%
where, as in (\ref{tail_distr}), $t_{m,\ell }^{\left( 3\right) }$\ are 
\textit{i.i.d.} random variables, independent across $m$ and $\ell $, with a
Student's t distribution with $3$ degrees of freedom; all the other
specifications are the same as above. We use $N=200$ and, as above, $%
k_{1}^{\ast }=\left\lfloor 0.35N\right\rfloor $ and $k_{2}^{\ast
}=\left\lfloor 0.7N\right\rfloor $. Alternative (\ref{seg-distr}) represents
a case, relevant in practice, where the data experience a period of
turbulence characterised by heavy tails, after which they revert to normal.
Results are in Table \ref{tab:segment2}; we found $\tau _{N}=c_{\alpha }(\ln
N)^{1/2}$ to be a better choice in this case, and we suggest this choice of
threshold when testing for changes in the distribution.

\medskip

\begin{table*}[h]
\caption{{\protect\footnotesize {Measures of location for $\widehat{R}$ and
the estimated breakdates - testing for distributional changes}}}
\label{tab:segment2}\centering
\par
{\footnotesize {\ }}
\par
{\footnotesize 
\begin{tabular}{ccccccccccccc}
\hline\hline
&  &  &  &  &  &  &  &  &  &  &  &  \\ 
\multicolumn{13}{c}{Measures of location for the estimated number of
changepoints $\widehat{R}$} \\ 
&  &  &  &  &  &  &  &  &  &  &  &  \\ 
&  & $\alpha $ &  & $0.00$ & $0.15$ & $0.25$ & $0.50$ & $0.75$ & $0.85$ & $%
0.95$ & $0.99$ &  \\ 
&  &  &  &  &  &  &  &  &  &  &  &  \\ 
& mean &  &  & $1.995$ & $2.100$ & $1.860$ & $2.050$ & $1.840$ & $1.815$ & $%
1.600$ & $1.335$ &  \\ 
& median &  &  & $2.000$ & $2.000$ & $2.000$ & $2.000$ & $2.000$ & $2.000$ & 
$0.000$ & $0.000$ &  \\ 
& min &  &  & $0.000$ & $0.000$ & $0.000$ & $0.000$ & $0.000$ & $0.000$ & $%
0.000$ & $0.000$ &  \\ 
& max &  &  & $6.000$ & $5.000$ & $5.000$ & $6.000$ & $7.000$ & $6.000$ & $%
6.000$ & $6.000$ &  \\ 
&  &  &  &  &  &  &  &  &  &  &  &  \\ \hline
&  &  &  &  &  &  &  &  &  &  &  &  \\ 
\multicolumn{13}{c}{Median estimated breakdates} \\ 
&  &  &  &  &  &  &  &  &  &  &  &  \\ 
&  & $\alpha $ &  & $0.00$ & $0.15$ & $0.25$ & $0.50$ & $0.75$ & $0.85$ & $%
0.95$ & $0.99$ &  \\ 
&  &  &  &  &  &  &  &  &  &  &  &  \\ 
& $\widehat{k}_{1}$ &  &  & $70$ & $70$ & $70$ & $70$ & $70$ & $70$ & $70$ & 
$70$ &  \\ 
& $\widehat{k}_{2}$ &  &  & $139$ & $140$ & $140$ & $140$ & $139$ & $137$ & $%
138$ & $137$ &  \\ 
&  &  &  &  &  &  &  &  &  &  &  &  \\ \hline\hline
\end{tabular}
}
\par
{\footnotesize \ }
\par
{\footnotesize 
\begin{tablenotes}
      \tiny
            \item The table contains various measures of location for $\widehat{R}$ and the median estimated breakdates under the alternative \eqref{seg-distr}; data have been generated as \textit{i.i.d.} with measurement error, with sample size $N=200$.
            
\end{tablenotes}
}
\end{table*}

\medskip

\clearpage

\newpage

\clearpage
\renewcommand*{\thesection}{\Alph{section}}

\setcounter{subsection}{-1} \setcounter{subsubsection}{-1} %
\setcounter{equation}{0} \setcounter{lemma}{0} \setcounter{theorem}{0} %
\renewcommand{\theassumption}{C.\arabic{assumption}} 
\renewcommand{\thetheorem}{C.\arabic{theorem}} \renewcommand{\thelemma}{C.%
\arabic{lemma}} \renewcommand{\theproposition}{C.\arabic{proposition}} %
\renewcommand{\thecorollary}{C.\arabic{corollary}} \renewcommand{%
\theequation}{C.\arabic{equation}}

\section{Further empirical evidence\label{further-empirical}}

\subsection{Further empirical evidence: changepoint analysis of cumulative
intraday returns\label{cidr}}

We complement our findings in Section \ref{empirical} by applying our tests
for changes in the mean and in the distribution to cumulative intraday
returns (CIDRs henceforth), whose usefulness is demonstrated in a
contribution by \citet{kokoszka2012functional}. We use the same dataset as
in Section \ref{empirical}, having removed the same curves consisting of
partial trading days. CIDRs are defined as%
\begin{equation*}
Y_{i}\left( t\right) =\log P_{i}\left( t\right) -\log P_{i}\left(
t_{0}\right) ,
\end{equation*}%
where $P_{i}\left( t\right) $ is the daily price evaluated at $t$, and $%
t_{0} $ is, for each day $1\leq i\leq N$, the start of trading for the day
(in our case, midnight). Contrary to the use of month-on-month returns, in
this case we can use the whole sample of daily curves consisting of $N=414$
functional datapoints.

Tests have been applied with the same specifications as in Section \ref%
{empirical} in the main paper. We did not find any changepoints in the mean,
irrespective of the value of $\alpha $. Conversely, applying the test for
distributional changes to the demeaned data, several changepoints are found,
summarised in Table XXX. In the table, as in the rest of the paper, we have
computed the measures of skewness and kurtosis as%
\begin{align}
&\frac{1}{T}\int_{0}^{T}\frac{\displaystyle{N^{-1}\sum_{i=1}^{N}\left(
Y_{i}\left( t\right) -N^{-1}\sum_{i=1}^{N}Y_{i}\left( t\right) \right) ^{3}}%
}{\displaystyle{\left( N^{-1}\sum_{i=1}^{N}\left( Y_{i}\left( t\right)
-N^{-1}\sum_{i=1}^{N}Y_{i}\left( t\right) \right) ^{2}\right) ^{3/2}}}dt,
\label{sk} \\
&\frac{1}{T}\int_{0}^{T}\frac{\displaystyle{N^{-1}\sum_{i=1}^{N}\left(
Y_{i}\left( t\right) -N^{-1}\sum_{i=1}^{N}Y_{i}\left( t\right) \right) ^{4}}%
}{\displaystyle{\left( N^{-1}\sum_{i=1}^{N}\left( Y_{i}\left( t\right)
-N^{-1}\sum_{i=1}^{N}Y_{i}\left( t\right) \right) ^{2}\right) ^{4/2}}}dt.
\label{ku}
\end{align}

\medskip

\begin{table*}[h]
\caption{{\protect\footnotesize {Changepoint in the distribution in
intraday, month-on-month return curves - S\&P 500, January 3rd, 2022 -
September 3rd, 2023}}}
\label{tab:empiricsSPD}\centering
\par
\resizebox{\textwidth}{!}{

\begin{tabular}{ccccccccccc}
\hline\hline
&  &  &  &  &  &  &  &  &  &  \\ 
\multicolumn{11}{c}{Changepoint detection using $\alpha =0.50$} \\ 
&  &  &  &  &  &  &  &  &  &  \\ 
& Iteration &  & Segment &  & Outcome &  & Estimated date &  & Notes &  \\ 
&  &  &  &  &  &  &  &  &  &  \\ 
& $1$ &  & Jan 3rd, 2022 -- Sep 3rd, 2023 &  & Reject &  & May 5th, 2022 & 
& significant also at $1\%$ &  \\ 
&  &  &  &  &  &  &  &  & break found also using $\alpha =0$, at the same
date, and using $\alpha =0.99$, at Apr 25th, 2022 &  \\ 
&  &  &  &  &  &  &  &  &  &  \\ 
& $2$ &  & Jan 3rd, 2022 -- May 4th, 2022 &  & Reject &  & Jan 20th, 2022 & 
& significant also at $1\%$ &  \\ 
&  &  &  &  &  &  &  &  & break found also using $\alpha =0,$ at Feb 13th,
2022, and using $\alpha =0.99$, at the same date &  \\ 
&  &  &  &  &  &  &  &  &  &  \\ 
& $3$ &  & Jan 3rd, 2022 -- Jan 19th, 2022 &  & Not reject &  &  &  & no
break found even at $10\%$, or with $\alpha =0,$ $0.99$ &  \\ 
&  &  &  &  &  &  &  &  & $\begin{array}{c}
\widehat{\sigma }_{N}^{2}=0.017 \\ 
sk=1.019 \\ 
ku=1.050\end{array}$ &  \\ 
&  &  &  &  &  &  &  &  &  &  \\ 
& $4$ &  & Jan 20th, 2022 -- May 4th, 2022 &  & Not reject &  &  &  & no
break found even at $10\%$, or with $\alpha =0,$ $0.99$ &  \\ 
&  &  &  &  &  &  &  &  & $\begin{array}{c}
\widehat{\sigma }_{N}^{2}=0.005 \\ 
sk=1.183 \\ 
ku=1.477\end{array}$ &  \\ 
&  &  &  &  &  &  &  &  &  &  \\ 
& $5$ &  & May 5th, 2022 -- Sep 3rd, 2023 &  & Reject &  & May 17th, 2023 & 
& significant also at $1\%$ &  \\ 
&  &  &  &  &  &  &  &  & break found also using $\alpha =0,$ $0.99$, at the
same date &  \\ 
&  &  &  &  &  &  &  &  &  &  \\ 
& $6$ &  & May 5th, 2022 -- May 18th, 2023 &  & Not reject &  &  &  & no
break found even at $10\%$, or with $\alpha =0,$ $0.99$ &  \\ 
&  &  &  &  &  &  &  &  & $\begin{array}{c}
\widehat{\sigma }_{N}^{2}=0.003 \\ 
sk=-1.522 \\ 
ku=2.737\end{array}$ &  \\ 
&  &  &  &  &  &  &  &  &  &  \\ 
& $7$ &  & May 18th, 2023 -- Sep 3rd, 2023 &  & Reject &  & Jun 12th, 2023 & 
& significant also at $1\%$ &  \\ 
&  &  &  &  &  &  &  &  & break found also using $\alpha =0,$ $0.99$, at the
same date &  \\ 
&  &  &  &  &  &  &  &  &  &  \\ 
& $8$ &  & May 18th, 2023 -- Jun 11th, 2023 &  & Not reject &  &  &  & no
break found even at $10\%$, or with $\alpha =0,$ $0.99$ &  \\ 
&  &  &  &  &  &  &  &  & $\begin{array}{c}
\widehat{\sigma }_{N}^{2}=0.001 \\ 
sk=1.334 \\ 
ku=2.349\end{array}$ &  \\ 
&  &  &  &  &  &  &  &  &  &  \\ 
& $9$ &  & Jun 11th, 2023 -- Sep 3rd, 2023 &  & Not reject &  &  &  & no
break found even at $10\%$, or with $\alpha =0,$ $0.99$ &  \\ 
&  &  &  &  &  &  &  &  & $\begin{array}{c}
\widehat{\sigma }_{N}^{2}=0.004 \\ 
sk=1.103 \\ 
ku=1.278\end{array}$ &  \\ 
&  &  &  &  &  &  &  &  &  &  \\ \hline\hline
\end{tabular}

}
\par
{\footnotesize 
\begin{tablenotes}
      \tiny
            \item We have used the estimator of the variance $\widehat{\sigma }_{N}^{2}$
defined in (\ref{sighat}). As far as the other descriptive statistics are concerned,  \textquotedblleft $sk$\textquotedblright\ and
\textquotedblleft $ku$\textquotedblright\ are computed according to equations (\ref{sk}) and (\ref{ku}) respectively.

\end{tablenotes}
}
\end{table*}

\subsection{Further empirical evidence: changepoint analysis of temperature
data\label{empiric_further}}

In this section, we illustrate our approach using temperature data, where
FDA is applied \textquotedblleft naturally\textquotedblright ; more broadly
speaking, recent contributions in the area of climate science show that
using time series methods can be beneficial (see %
\citealp{diebold2022probability}; \citealp{diebold2023will}; and %
\citealp{ditlevsen2023warning}).

Following \citet{berkes2009detecting}, we use a sample of $N=251$ yearly
curves, recorded on a daily basis between $1772$ and $2022$. Each curve
contains average daily temperatures (in degrees Celsius) recorded in Central
England; apart from removing the data corresponding to February 29th in leap
years in order to balance the sample, no further transformation is applied
to the data.\footnote{%
The data have been downloaded from %
\url{https://www.metoffice.gov.uk/hadobs/hadcet/}, where a brief description
of the dataset can also be found. A more complete description of the data
can be found in \citet{parker1992new}, to which we refer for details.} Our
techniques are particularly suited to this dataset for a number of reasons:
firstly, we do not need to invert any large-scale matrix, contrary to %
\citet{horvath1999testing}, and therefore we can use the daily sampling
frequency as opposed to transforming it into monthly averages; secondly,
temperature data might exhibit linear or nonlinear serial dependence (see
e.g. \citealp{bowers}), which our tests are designed to take into account,
unlike those proposed in \citet{berkes2009detecting}; and, finally, our
weighted test statistics are also suited to detect changepoints occurring
close to the end of the sample, thus allowing to shed light on the presence
and extent of changes in average temperatures in recent years. We apply our
tests for changepoints in the mean using our tests with $\alpha \in \left\{
0,0.5,0.65\right\} $, by way of comparison and robustness check; we note
that using different values of $\alpha $ does not alter the main
conclusions, although higher values of $\alpha $ seem to estimate the
changepoint date later and later. In order to take into account the possible
presence of multiple changes, we apply binary segmentation. We have
implemented our test using the same specifications as described in Section %
\ref{empirical} in the main paper, also carrying out the same robustness
checks with no noticeable changes in the results.

\bigskip

\begin{table*}[h]
\caption{{\protect\footnotesize {Changepoint detection in Central England
temperatures, $1772-2022$}}}
\label{tab:empirics}\centering
\par
{\scriptsize {\ 
\begin{tabular}{ccccccccccc}
\hline\hline
&  &  &  &  &  &  &  &  &  &  \\ 
\multicolumn{11}{c}{Changepoint detection with $\alpha =0.00$} \\ 
&  &  &  &  &  &  &  &  &  &  \\ 
& Iteration &  & Segment &  & Outcome &  & Estimated date &  & Notes &  \\ 
&  &  &  &  &  &  &  &  &  &  \\ 
& $1$ &  & $1772-2022$ &  & Reject &  & $1919$ &  & significant also at $1\%$
&  \\ 
& $2$ &  & $1772-1918$ &  & Reject &  & $1842$ &  & significant only at $5\%$
&  \\ 
& $3$ &  & $1772-1841$ &  & Not reject &  &  &  & no break found even at $%
10\%$ &  \\ 
& $4$ &  & $1842-1918$ &  & Not reject &  &  &  & no break found even at $%
10\%$ &  \\ 
& $5$ &  & $1919-2022$ &  & Reject &  & $1987$ &  & significant also at $1\%$
&  \\ 
& $6$ &  & $1919-1986$ &  & Not reject &  &  &  & no break found even at $%
10\%$ &  \\ 
& $7$ &  & $1987-2022$ &  & Not reject &  &  &  & no break found even at $%
10\%$ &  \\ 
&  &  &  &  &  &  &  &  &  &  \\ \hline
&  &  &  &  &  &  &  &  &  &  \\ 
\multicolumn{11}{c}{Changepoint detection with $\alpha =0.50$} \\ 
&  &  &  &  &  &  &  &  &  &  \\ 
& $1$ &  & $1772-2022$ &  & Reject &  & $1919$ &  & significant also at $1\%$
&  \\ 
& $2$ &  & $1772-1918$ &  & Reject &  & $1842$ &  & significant only at $5\%$
&  \\ 
& $3$ &  & $1772-1841$ &  & Not reject &  &  &  & no break found even at $%
10\%$ &  \\ 
& $4$ &  & $1842-1918$ &  & Not reject &  &  &  & no break found even at $%
10\%$ &  \\ 
& $5$ &  & $1919-2022$ &  & Reject &  & $1988$ &  & significant also at $1\%$
&  \\ 
& $6$ &  & $1919-1987$ &  & Not reject &  &  &  & no break found even at $%
10\%$ &  \\ 
& $7$ &  & $1988-2022$ &  & Not reject &  &  &  & no break found even at $%
10\%$ &  \\ 
&  &  &  &  &  &  &  &  &  &  \\ \hline
&  &  &  &  &  &  &  &  &  &  \\ 
\multicolumn{11}{c}{Changepoint detection with $\alpha =0.65$} \\ 
&  &  &  &  &  &  &  &  &  &  \\ 
& $1$ &  & $1772-2022$ &  & Reject &  & $1932$ &  & significant also at $1\%$
&  \\ 
& $2$ &  & $1772-1931$ &  & Reject &  & $1842$ &  & significant also at $1\%$
&  \\ 
& $3$ &  & $1772-1841$ &  & Not reject &  &  &  & no break found even at $%
10\%$ &  \\ 
& $4$ &  & $1842-1931$ &  & Not reject &  &  &  & no break found even at $%
10\%$ &  \\ 
& $5$ &  & $1932-2022$ &  & Reject &  & $1989$ &  & significant also at $1\%$
&  \\ 
& $6$ &  & $1932-1988$ &  & Not reject &  &  &  & no break found even at $%
10\%$ &  \\ 
& $7$ &  & $1989-2022$ &  & Not reject &  &  &  & no break found even at $%
10\%$ &  \\ 
&  &  &  &  &  &  &  &  &  &  \\ \hline\hline
\end{tabular}
} }
\par
{\scriptsize {\footnotesize \ } }
\end{table*}

\bigskip

Results are in Table \ref{tab:empirics}; in Figure \ref{fig:FigRough}, we
also report the average temperature functions between each of the estimated
changepoints for the various values of $\alpha $. With small and medium
values of $\alpha $ (i.e., $\alpha =0$ and $\alpha =0.5$), we identify three
changepoints. The first one to be identified (corresponding to the
\textquotedblleft strongest\textquotedblright\ break) is estimated to have
occurred in $1919$. This result is essentially in agreement with the
findings in \citet{berkes2009detecting} (and also in %
\citealp{horvath1999testing}), where the first changepoint is found around $%
1926$. This estimated date corresponds to the so-called Early Twentieth
Century Warming (\citealp{hegerl2018early}), a well-documented phenomenon
which partly coincides with the well-known phenomenon of warming of the
Arctic (\citealp{bengtsson2004early}), and which \textquotedblleft still
defies full explanation\textquotedblright\ (\citealp{bronnimann2009early},
p. 735). Our estimated date is ealier than that of %
\citet{berkes2009detecting}, which could be ascribed to the estimation
error, but also to the \textquotedblleft pull\textquotedblright\ effect, on
the data, of the UK heatwave of $1911$. We also estimate a breakdate at $%
1842 $, which could be ascribed to the anthropogenic effect of the
Industrial Revolution, and again it is similar to the estimate of $1850$ in %
\citet{berkes2009detecting}. On the other hand, \citet{berkes2009detecting}
also find one changepoint in $1808$. None of our statistics finds evidence
of a changepoint around this date, even at $10\%$ nominal level; on account
of the lack of serial correlation, our data could be roughly interpreted as
falling into the \textquotedblleft \textit{i.i.d.} with measurement
error\textquotedblright\ category, for which our simulations indicate no
undersizement and excellent power even in sample sizes. Finally, we find
clear evidence of a changepoint in $1987$ ($1988$ when using $\alpha =0.5$,
and $1989$ when using $\alpha =0.65$). This is significant even at $1\%$
nominal level, which corresponds to the beginning of (rapid) global warming
- the late $1980$'s date confirms the statement, in the \textquotedblleft
State of the World 1989\textquotedblright\ WorldWatch report (%
\citealp{koutaissoff1989state}), that the 90's would be the
\textquotedblleft turnaround decade\textquotedblright\ as far as climate
change is concerned. Indeed, Figure \ref{fig:FigRough} shows very clearly
the presence of a pronounced increase in average daily temperatures between
the first and the fourth subsamples. When using $\alpha =0.65$, essentially
the same results are found, but the first changepoint is estimated at $1932$%
, i.e. one decade later than the other two changepoint estimates. This can
be read in the light of the results in Table \ref{tab:MedianMid2}, which
suggest that, as $\alpha $ increases, the estimated changepoint may be
increasingly biased.

Finally, we also considered the possible presence of changes in the
distribution of temperature data, applying the test developed in Section \ref%
{distrib}. To this end, we demeaned the data in each segment, and applied
the test in Section \ref{distrib} using $d=1$. Results are in Table \ref%
{tab:empiricsD}. No changepoints were detected for $\alpha =0$ and $\alpha
=0.5$, even at $1\%$ nominal level, suggesting that, if a changepoint is
present, this is located towards the sample endpoints; indeed, when using $%
\alpha =0.65$, the presence of one changepoint emerges even at $10\%$
nominal level, with estimated date $1996$. Comparing descriptive statistics,
skewness and kurtosis in the two subperiods seem very similar (in both cases
suggesting Gaussianity); conversely, the estimated variances seem to
indicate that there is a changepoint in the \textit{variability} of
temperatures between the two subperiods.

\bigskip

\begin{table*}[h]
\caption{{\protect\footnotesize {Changepoint detection in Central England
temperatures, $1772-2022$ - changes in the distribution of demeaned data}}}
\label{tab:empiricsD}\centering
\par
{\scriptsize {\ 
\begin{tabular}{ccccccccccc}
\hline\hline
&  &  &  &  &  &  &  &  &  &  \\ 
\multicolumn{11}{c}{Changepoint detection with $\alpha =0.00$ and $\alpha
=0.50$} \\ 
&  &  &  &  &  &  &  &  &  &  \\ 
& Iteration &  & Segment &  & Outcome &  & Estimated date &  & Notes &  \\ 
&  &  &  &  &  &  &  &  &  &  \\ 
& $1$ &  & $1772-2022$ &  & Not reject &  &  &  & no break found even at $%
10\%$ &  \\ 
&  &  &  &  &  &  &  &  &  &  \\ \hline
&  &  &  &  &  &  &  &  &  &  \\ 
\multicolumn{11}{c}{Changepoint detection with $\alpha =0.65$} \\ 
&  &  &  &  &  &  &  &  &  &  \\ 
& $1$ &  & $1772-2022$ &  & Reject &  & $1996$ &  & significant also at $1\%$
&  \\ 
& $2$ &  & $1772-1995$ &  & Not reject &  &  &  & no break found even at $%
10\%$ &  \\ 
&  &  &  &  &  &  &  &  & $%
\begin{array}{c}
\widehat{\sigma }_{N}^{2}=7.42 \\ 
sk=-0.105 \\ 
ku=3.066%
\end{array}%
$ &  \\ 
& $3$ &  & $1996-2022$ &  & Not reject &  &  &  & no break found even at $%
10\%$ &  \\ 
&  &  &  &  &  &  &  &  & $%
\begin{array}{c}
\widehat{\sigma }_{N}^{2}=6.73 \\ 
sk=0.166 \\ 
ku=3.092%
\end{array}%
$ &  \\ 
&  &  &  &  &  &  &  &  &  &  \\ \hline\hline
\end{tabular}
} }
\par
{\footnotesize {\ 
\begin{tablenotes}
      \tiny
            \item We have used the estimator of the variance $\widehat{\sigma }_{N}^{2}$
defined in (\ref{sighat}). As far as the other descriptive statistics are concerned,  \textquotedblleft $sk$\textquotedblright\ and
\textquotedblleft $ku$\textquotedblright\ represent overall measures of
skewness and kurtosis respectively.
            
\end{tablenotes}
} }
\end{table*}

\medskip

\begin{figure}[h!]
\caption{{\protect\footnotesize {Average daily temperatures in Central
England}}}
\label{fig:FigRough}\centering
\hspace{0cm} 
\includegraphics[scale=0.8]{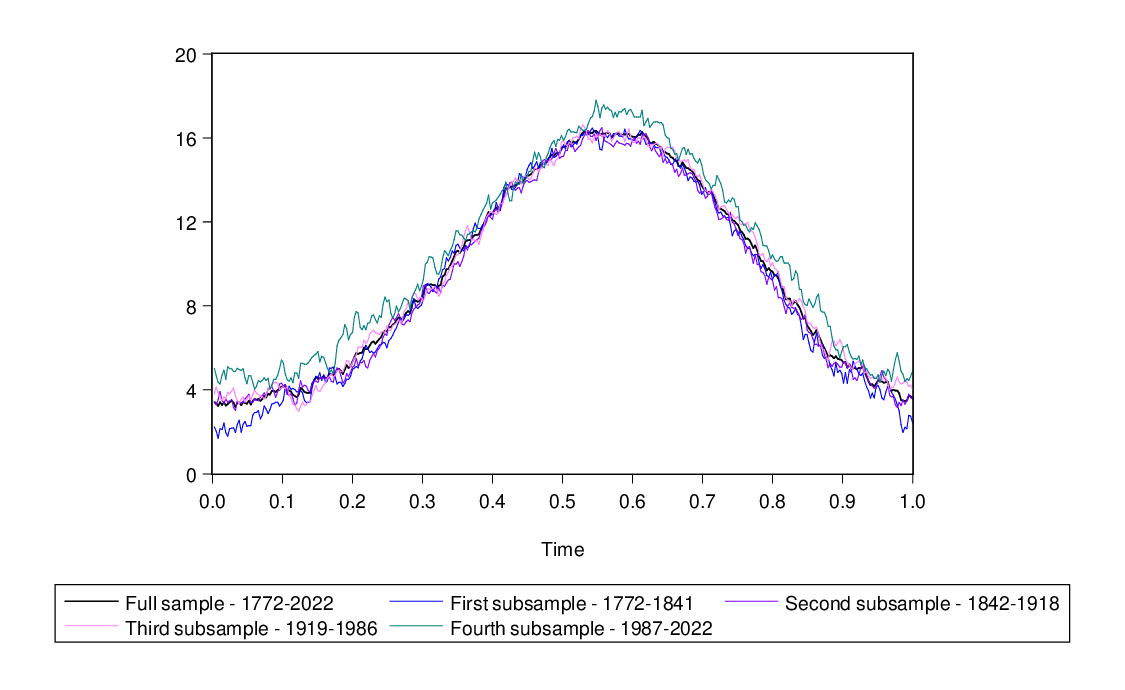} 
\end{figure}

\medskip

Descriptive statistics for the whole sample, and for the segments identified
when using $\alpha =0$, are in Table \ref{tab:empirics2}. Results using
other values of $\alpha $ (and, therefore, other estimated changepoints) are
available upon request.

\medskip

\begin{table*}[h!]
\caption{{\protect\footnotesize {Descriptive statistics for Central England
temperature data: full sample and individual segments}}}
\label{tab:empirics2}\centering
{\footnotesize {\ }}
\par
{\footnotesize {\ }}
\par
{\footnotesize {\ }}
\par
\resizebox{\textwidth}{!}{
\begin{tabular}{ccccccccccccccccc}
\hline\hline
&  &  &  &  &  &  &  &  &  &  &  &  &  &  &  &  \\ 
\multicolumn{17}{c}{Descriptive statistics - overall period $1772-2022$} \\ 
&  &  &  &  &  &  &  &  &  &  &  &  &  &  &  &  \\ 
& Average daily &  & Average low &  & Average high &  & Record low &  & 
Record high &  & First quartile &  & Median &  & Third quartile &  \\ 
&  &  &  &  &  &  &  &  &  &  &  &  &  &  &  &  \\ 
& $9.40$ $^{\circ }C$ &  & $-3.62$ $^{\circ }C$ &  & $20.99$ $^{\circ }C$ & 
& $\underset{\left( 20/01/1838\right) }{-11.90^{\circ }C}$ &  & $\underset{\left( 19/07/2022\right) }{28.10^{\circ }C}$ &  & $5.32$ $^{\circ }C$ &  & $9.34$ $^{\circ }C$ &  & $13.79$ $^{\circ }C$ &  \\ 
&  &  &  &  &  &  &  &  &  &  &  &  &  &  &  &  \\ \hline\hline
&  &  &  &  &  &  &  &  &  &  &  &  &  &  &  &  \\ 
\multicolumn{17}{c}{Descriptive statistics - subperiods} \\ 
&  &  &  &  &  &  &  &  &  &  &  &  &  &  &  &  \\ \cline{3-15}
&  &  &  &  &  &  &  &  &  &  &  &  &  &  &  &  \\ 
\multicolumn{17}{c}{Descriptive statistics - subperiod $1772-1841$} \\ 
&  &  &  &  &  &  &  &  &  &  &  &  &  &  &  &  \\ 
& Average daily &  & Average low &  & Average high &  & Record low &  & 
Record high &  & First quartile &  & Median &  & Third quartile &  \\ 
&  &  &  &  &  &  &  &  &  &  &  &  &  &  &  &  \\ 
& $9.12$ $^{\circ }C$ &  & $-4.70$ $^{\circ }C$ &  & $20.71$ $^{\circ }C$ & 
& $-11.90$ $^{\circ }C$ &  & $24.50$ $^{\circ }C$ &  & $4.87$ $^{\circ }C$ & 
& $9.14$ $^{\circ }C$ &  & $13.79$ $^{\circ }C$ &  \\ 
&  &  &  &  &  &  &  &  &  &  &  &  &  &  &  &  \\ \cline{3-15}
&  &  &  &  &  &  &  &  &  &  &  &  &  &  &  &  \\ 
\multicolumn{17}{c}{Descriptive statistics - subperiod $1842-1918$} \\ 
&  &  &  &  &  &  &  &  &  &  &  &  &  &  &  &  \\ 
& Average daily &  & Average low &  & Average high &  & Record low &  & 
Record high &  & First quartile &  & Median &  & Third quartile &  \\ 
&  &  &  &  &  &  &  &  &  &  &  &  &  &  &  &  \\ 
& $9.18$ $^{\circ }C$ &  & $-3.65$ $^{\circ }C$ &  & $20.60$ $^{\circ }C$ & 
& $-9.00$ $^{\circ }C$ &  & $24.00$ $^{\circ }C$ &  & $5.16$ $^{\circ }C$ & 
& $9.10$ $^{\circ }C$ &  & $13.51$ $^{\circ }C$ &  \\ 
&  &  &  &  &  &  &  &  &  &  &  &  &  &  &  &  \\ \cline{3-15}
&  &  &  &  &  &  &  &  &  &  &  &  &  &  &  &  \\ 
\multicolumn{17}{c}{Descriptive statistics - subperiod $1919-1986$} \\ 
&  &  &  &  &  &  &  &  &  &  &  &  &  &  &  &  \\ 
& Average daily &  & Average low &  & Average high &  & Record low &  & 
Record high &  & First quartile &  & Median &  & Third quartile &  \\ 
&  &  &  &  &  &  &  &  &  &  &  &  &  &  &  &  \\ 
& $9.48$ $^{\circ }C$ &  & $-3.28$ $^{\circ }C$ &  & $21.02$ $^{\circ }C$ & 
& $-8.70$ $^{\circ }C$ &  & $25.10$ $^{\circ }C$ &  & $5.45$ $^{\circ }C$ & 
& $9.44$ $^{\circ }C$ &  & $13.79$ $^{\circ }C$ &  \\ 
&  &  &  &  &  &  &  &  &  &  &  &  &  &  &  &  \\ \cline{3-15}
&  &  &  &  &  &  &  &  &  &  &  &  &  &  &  &  \\ 
\multicolumn{17}{c}{Descriptive statistics - subperiod $1987-2023$} \\ 
&  &  &  &  &  &  &  &  &  &  &  &  &  &  &  &  \\ 
& Average daily &  & Average low &  & Average high &  & Record low &  & 
Record high &  & First quartile &  & Median &  & Third quartile &  \\ 
&  &  &  &  &  &  &  &  &  &  &  &  &  &  &  &  \\ 
& $10.22$ $^{\circ }C$ &  & $-2.11$ $^{\circ }C$ &  & $22.28$ $^{\circ }C$ & 
& $-7.70$ $^{\circ }C$ &  & $28.10$ $^{\circ }C$ &  & $6.29$ $^{\circ }C$ & 
& $10.08$ $^{\circ }C$ &  & $14.36$ $^{\circ }C$ &  \\ 
&  &  &  &  &  &  &  &  &  &  &  &  &  &  &  &  \\ \hline\hline
\end{tabular}
}
\par
{\footnotesize \ }
\par
{\footnotesize \ }
\par
{\footnotesize 
\begin{tablenotes}
      \tiny
            \item The table contains various measures of location - for the full sample and each subsample - for the Central England temperature data. 
            
\end{tablenotes}
}
\end{table*}

\newpage

\clearpage
\renewcommand*{\thesection}{\Alph{section}}

\setcounter{subsection}{-1} \setcounter{subsubsection}{-1} %
\setcounter{equation}{0} \setcounter{lemma}{0} \setcounter{theorem}{0} %
\renewcommand{\theassumption}{C.\arabic{assumption}} 
\renewcommand{\thetheorem}{C.\arabic{theorem}} \renewcommand{\thelemma}{C.%
\arabic{lemma}} \renewcommand{\theproposition}{C.\arabic{proposition}} %
\renewcommand{\thecorollary}{C.\arabic{corollary}} \renewcommand{%
\theequation}{C.\arabic{equation}}

\section{Preliminary lemmas\label{lemmas}}

Henceforth, unless stated otherwise, we carry out our proofs for the case $%
r=d=1$, for simplicity and without loss of generality. We use the following
notation: $\left\lceil \cdot \right\rceil $ is the ceiling function, that is
the function that rounds a number to the nearest, largest integer; $C$
denotes a generic constant independent of $N$, $k$, $t$ that may change from
line to line.

\bigskip

We begin by recalling some results in \citet{berkes:horvath:rice:2013}.

\begin{lemma}
\label{l:BHR_11} We assume that Assumption \ref{as1} is satisfied. Let $%
w_{k}=w_{k}(t)=\sum_{i=1}^{k}\epsilon _{i}(t)$. Then, for each $N$, on a
suitably enlarged probability space, we may define a Gaussian process $%
\{G_{N}(u,t),u\geq 0,t\in \mathcal{T}\}$, whose distribution does not depend
on $N$, such that 
\begin{equation}
\sup_{0\leq u\leq 1}\Vert N^{-1/2}w_{\lfloor Nu\rfloor }-G_{N}(u,\cdot
)\Vert ^{2}=o_{P}(1),  \label{e:berkes_etal_th1}
\end{equation}%
where $EG_{N}(u,t)=0$, and $EG_{N}\left( u,t\right) G_{N}^{\top }\left(
u^{\prime },t^{\prime }\right) =\min \left\{ u,u^{\prime }\right\} \mathbf{D}%
\left( t,t^{\prime }\right) $ for every $N$.

\begin{proof}
For $\mathbb{R}$-valued functional observations, the desired result is
already established in Theorem 1.1 in of \citet{berkes:horvath:rice:2013}.
Minor adjustments to their proofs show they also hold for $\mathbb{R}^{r}$%
-valued functional observations.
\end{proof}
\end{lemma}

\begin{lemma}
\label{l:BHR_33} Suppose $\{Y_{\ell },-\infty <\ell <\infty \}$ is a
zero-mean Bernoulli shift sequence such that, for some $s>0$, it holds that $%
E\Vert Y_{\ell }\Vert ^{2+s}<\infty $. Then, for every $j<k$, it holds that 
\begin{equation*}
E\Bigg\|\sum_{i=j+1}^{k}Y_{i}\Bigg\|^{2+s}\leq C\left( k-j\right) ^{1+s/2}.
\end{equation*}

\begin{proof}
For $\mathbb{R}$-valued functional observations, the desired result is
already established in Theorem 3.3 of \citet{berkes:horvath:rice:2013}.
Minor adjustments to their proofs show they also hold for $\mathbb{R}^{r}$%
-valued functional observations.
\end{proof}
\end{lemma}

\bigskip

For convenience, throughout the remainder of this section we use the
notation $S_{k}(t)=\sum_{i=1}^{k}X_{i}(t),\quad t\in \mathcal{T}$, and
without loss of generality under $H_{0}$ we assume $\mu _{i}=0$ for all $i$,
so that $EX_{i}(t)=E\epsilon _{i}(t)=0$. We now establish some intermediate
weighted approximation results related to the partial sums $S_{k}$.

\begin{lemma}
\label{l:weightappx} We assume that Assumption \ref{as1} is satisfied. Then,
for all $\gamma >1/2$, under $H_{0}$ it holds that 
\begin{equation}
\max_{1\leq k\leq N}\frac{\Vert S_{k}\Vert }{k^{\gamma }}=O_{P}(1),\quad
\max_{1\leq k\leq N}\frac{\Vert S_{N}-S_{k}\Vert }{(N-k)^{\gamma }}=O_{P}(1),
\label{e:Sk_Op1}
\end{equation}%
and 
\begin{equation}
\max_{1\leq k\leq N}\frac{1}{k^{\gamma }}\sum_{i=1}^{k}\big(\Vert X_{i}\Vert
^{2}-\sigma _{0}^{2}\big)=O_{P}(1),\quad \max_{1\leq k<N}\frac{1}{%
(N-k)^{\gamma }}\sum_{i=k+1}^{N}\big(\Vert X_{i}\Vert ^{2}-\sigma _{0}^{2}%
\big)=O_{P}(1).  \label{e:X2_Op1}
\end{equation}

\begin{proof}
We begin by establishing a maximal inequality for the partial sums process $%
S_{k}$. By Lemma \ref{l:BHR_33} with $Y_{\ell }=X_{\ell }$, for any $1\leq
j<k\leq N$, 
\begin{equation*}
E\Vert S_{k}-S_{j}\Vert ^{2+{\epsilon }}=E\Vert S_{k-j}\Vert ^{2+{\epsilon }%
}\leq C(k-j)^{1+{\epsilon }/2}.
\end{equation*}%
Thus, by Theorem 3.1 in \citet{moricz:serfling:stout:1982}, 
\begin{equation}
E\Big(\max_{1\leq \ell \leq k}\Vert S_{\ell }\Vert \Big)^{2+{\epsilon }}\leq
Ck^{1+{\epsilon }/2}.  \label{e:S_k_maximal_ineq}
\end{equation}%
This gives 
\begin{align*}
P&\left\{ \max_{1\leq k\leq N}\frac{\Vert S_{k}\Vert }{k^{\gamma }}>x\right\}
\\
&\leq P\left\{ \max_{1\leq m\leq \lceil \log N\rceil }\max_{\exp \left(
m-1\right) \leq k<\exp \left( m\right) }\frac{\Vert S_{k}\Vert }{k^{\gamma }}%
>x\right\} \\
&\leq \sum_{m=1}^{\lceil \log N\rceil }P\left\{ \max_{\exp \left( m-1\right)
\leq k<\exp \left( m\right) }\frac{\Vert S_{k}\Vert }{k^{\gamma }}>x\right\}
\\
&\leq \sum_{m=1}^{\lceil \log N\rceil }P\left\{ \max_{\exp \left( m-1\right)
\leq k<\exp \left( m\right) }\Vert S_{k}\Vert >xe^{\gamma (m-1)}\right\} \\
&\leq \sum_{m=1}^{\lceil \log N\rceil }P\left\{ \max_{1\leq k\leq \exp
\left( m\right) }\Vert S_{k}\Vert ^{2+{\epsilon }}>x^{2+{\epsilon }%
}e^{\gamma (2+{\epsilon })(m-1)}\right\} \\
&\leq x^{-\left( 2+{\epsilon }\right) }\sum_{m=1}^{\lceil \log N\rceil }\exp
\left( -\gamma (2+{\epsilon })(m-1)\right) E\max_{1\leq k\leq
e^{m}}\left\Vert S_{k}\right\Vert ^{2+{\epsilon }} \\
&\leq \frac{C}{x^{2+{\epsilon }}}\sum_{m=1}^{\lceil \log N\rceil }\exp
\left( m((1+{\epsilon }/2)-\gamma (2+{\epsilon }))\right) \leq \frac{C}{x^{2+%
{\epsilon }}}
\end{align*}%
This gives \eqref{e:Sk_Op1} for $S_{k}$; the argument for $S_{N}-S_{k}$ is
analogous.

We now turn to \eqref{e:X2_Op1}, and let $X_{k}^{(m)}=\epsilon
_{k}^{(m)}+\mu _{k}$ - as above, we set $\mu _{k}=0$ for simplicity and
without loss of generality. We begin by showing that, if $X_{k}$ is a
Bernoulli shift sequence satisfying Assumption \ref{as1}, then $\Vert
X_{k}\Vert ^{2}$ also is. We begin by noting that%
\begin{align*}
&\left\Vert X_{1}\right\Vert ^{2}-\big\Vert X_{1}^{(m)}\big\Vert ^{2} \\
&\quad=\left\langle X_{1},X_{1}\right\rangle -\left\langle
X_{1}^{(m)},X_{1}^{(m)}\right\rangle =\left\langle
X_{1}-X_{1}^{(m)},X_{1}+X_{1}^{(m)}\right\rangle \leq \left\Vert
X_{1}-X_{1}^{(m)}\right\Vert \left\Vert X_{1}+X_{1}^{(m)}\right\Vert ,
\end{align*}%
having used the Cauchy-Schwartz inequality in the last passage. Hence, for
any $p\geq 1$, using Minkowski's inequality%
\begin{align*}
&\left\vert \left\Vert X_{1}\right\Vert ^{2}-\left\Vert
X_{1}^{(m)}\right\Vert ^{2}\right\vert ^{p} \\
&\quad=\left\Vert X_{1}-X_{1}^{(m)}\right\Vert ^{p}\left\Vert
X_{1}+X_{1}^{(m)}\right\Vert ^{p}\leq \left\Vert
X_{1}-X_{1}^{(m)}\right\Vert ^{p}\left( \left\Vert X_{1}\right\Vert
^{p}+\left\Vert X_{1}^{(m)}\right\Vert ^{p}\right) .
\end{align*}%
Hence, taking $p=2+\widetilde{{\epsilon }}$, where $\widetilde{{\epsilon }}={%
\epsilon }/2$ (with ${\epsilon }$ as in Assumption \ref{as1}\textit{(ii)}),
and $\widetilde{\kappa }=\kappa /2>2+\widetilde{{\epsilon }}$, it follows
that 
\begin{equation*}
\left( E\left\vert \left\Vert X_{1}\right\Vert ^{2}-\left\Vert
X_{1}^{(m)}\right\Vert ^{2}\right\vert ^{2+\widetilde{{\epsilon }}}\right)
^{1/\widetilde{\kappa }}\leq C\left( E\left\Vert
X_{1}-X_{1}^{(m)}\right\Vert ^{4+{\epsilon }}\right) ^{1/\left( 2\widetilde{%
\kappa }\right) },
\end{equation*}%
whence finally%
\begin{equation*}
\sum_{m=1}^{\infty }\left( E\left\vert \left\Vert X_{1}\right\Vert
^{2}-\left\Vert X_{1}^{(m)}\right\Vert ^{2}\right\vert ^{2+\widetilde{{%
\epsilon }}}\right) ^{1/\widetilde{\kappa }}\leq C\sum_{m=1}^{\infty }\left(
E\left\Vert X_{1}-X_{1}^{(m)}\right\Vert ^{4+{\epsilon }}\right) ^{1/\kappa
}<\infty ,
\end{equation*}%
by Assumption \ref{as1}. Then, applying Lemma \ref{l:BHR_33} with $Y_{\ell
}=\Vert X_{\ell }\Vert ^{2}$, we obtain for every $j<k$, 
\begin{equation*}
E\Bigg|\sum_{\ell =j+1}^{k}\big(\Vert X_{\ell }\Vert ^{2}-\sigma _{0}^{2}%
\big)\Bigg|^{2+\widetilde{{\epsilon }}}\leq C(k-j)^{1+\widetilde{{\epsilon }}%
/2},
\end{equation*}%
which, by Theorem 3.1 in \citet{moricz:serfling:stout:1982}, gives 
\begin{equation*}
E\max_{1\leq \ell \leq k}\Bigg|\sum_{j=1}^{\ell }\big(\Vert X_{\ell }\Vert
^{2}-\sigma _{0}^{2}\big)\Bigg|^{2+\widetilde{{\epsilon }}}\leq Ck^{1+%
\widetilde{{\epsilon }}/2}.
\end{equation*}%
Henceforth, \eqref{e:X2_Op1} follows by repeating the arguments for %
\eqref{e:Sk_Op1}.
\end{proof}
\end{lemma}

\bigskip

The next lemma provides an approximation for $V_{N}(k)$ in terms of $S_{k}$.
For each $1\leq k\leq N$, recall $u=k/N$ and define 
\begin{align}
Q_{N}(k)& =2\Big(\frac{N}{k(N-k)}\Big)^{2}\Big\|S_{k}-\frac{k}{N}S_{N}\Big\|%
^{2}-\frac{2N}{k(N-k)}\sigma _{0}^{2}  \label{e:def_Qk} \\
& =\frac{2}{N}\big(u(1-u)\big)^{-2}\bigg(\Big\|N^{-1/2}\Big(S_{k}-\frac{k}{N}%
S_{N}\Big)\Big\|^{2}-\sigma _{0}^{2}u(1-u)\bigg).  \notag
\end{align}

\begin{lemma}
\label{l:replace_V_with_Q} We assume that Assumption \ref{as1} is satisfied.
Then, under $H_{0}$, it holds that, for all $0\leq \alpha <1$%
\begin{equation}
\max_{1\leq k\leq N}\Big[\frac{k}{N}\Big(1-\frac{k}{N}\Big)\Big]^{2-\alpha }%
\big|V_{N}(k)-Q_{N}(k)\big|=o_{P}\left( \frac{1}{N}\right) .
\label{e:Vk-Qk_appx_statement}
\end{equation}

\begin{proof}
Rewrite 
\begin{align*}
V_{N}(k)& =\frac{2}{k(N-k)}\sum_{i=1}^{k}\sum_{j=k+1}^{N}\Vert
X_{i}-X_{j}\Vert ^{2} \\
& \qquad \qquad \qquad -\frac{1}{\displaystyle{{\binom{k}{2}}}}\sum_{1\leq
i<j\leq k}\Vert X_{i}-X_{j}\Vert ^{2}-\frac{1}{\displaystyle{{\binom{N-k}{2}}%
}}\sum_{k<i<j\leq N}\Vert X_{i}-X_{j}\Vert ^{2} \\
& =T_{1}-T_{2}-T_{3}.
\end{align*}%
Using the identity $\Vert x-y\Vert ^{2}=\Vert x\Vert ^{2}+\Vert y\Vert
^{2}-2\langle x,y\rangle $, we obtain 
\begin{equation*}
T_{1}=\frac{2}{k}\sum_{i=1}^{k}\Vert X_{i}\Vert ^{2}+\frac{2}{N-k}%
\sum_{j=k+1}^{N}\Vert X_{i}\Vert ^{2}-\frac{4}{k(N-k)}\langle
S_{k},S_{N}-S_{k}\rangle ,
\end{equation*}%
and 
\begin{align}
T_{2}& =\frac{1}{\displaystyle{{\binom{k}{2}}}}\sum_{1\leq i\leq j\leq
k}\Vert X_{i}-X_{j}\Vert ^{2}  \notag \\
& =\frac{2}{k-1}\sum_{i=1}^{k}\Vert X_{i}\Vert ^{2}-\frac{2}{k(k-1)}\Vert
S_{k}\Vert ^{2}  \notag \\
& =\frac{2}{k}\sum_{i=1}^{k}\Vert X_{i}\Vert ^{2}-\frac{2}{k^{2}}\Vert
S_{k}\Vert ^{2}+\frac{2\sigma _{0}^{2}}{k}  \notag \\
& \qquad +\frac{2}{k(k-1)}\sum_{i=1}^{k}\big(\Vert X_{i}\Vert ^{2}-\sigma
_{0}^{2})+\frac{2\sigma _{0}^{2}}{k(k-1)}-\frac{2}{k^{2}(k-1)}\Vert
S_{k}\Vert ^{2}.  \notag
\end{align}%
Analogously, 
\begin{align}
T_{3}& =\frac{2}{N-k}\sum_{j=k+1}^{N}\Vert X_{j}\Vert ^{2}-\frac{2}{(N-k)^{2}%
}\Vert S_{N}-S_{k}\Vert ^{2}+\frac{2\sigma _{0}^{2}}{N-k}  \notag \\
& \qquad +\frac{2}{(N-k)(N-k-1)}\sum_{j=k+1}^{N}\big(\Vert X_{j}\Vert
^{2}-\sigma _{0}^{2})  \notag \\
& \qquad +\frac{2\sigma _{0}^{2}}{(N-k)(N-k-1)}-\frac{2}{(N-k)^{2}(N-k-1)}%
\Vert S_{N}-S_{k}\Vert ^{2}.  \label{e:Vk_3rdterm}
\end{align}%
Now, turning to the leading terms in $T_{1}-T_{2}-T_{3}$, observe 
\begin{align*}
-\frac{4}{k(N-k)}\langle S_{k},& S_{N}-S_{k}\rangle +\frac{2}{k^{2}}\Vert
S_{k}\Vert ^{2}+\frac{2}{(N-k)^{2}}\Vert S_{N}-S_{k}\Vert ^{2} \\
& =2\Big\|\frac{S_{k}}{k}-\frac{S_{N}-S_{k}}{N-k}\Big\|^{2} \\
& =2\Big(\frac{N}{k(N-k)}\Big)^{2}\Big\|S_{k}-\frac{k}{N}S_{N}\Big\|^{2}.
\end{align*}%
Therefore, 
\begin{align}
V_{N}(k)-Q_{N}(k)& =\frac{2}{k(k-1)}\sum_{i=1}^{k}\big(\Vert X_{i}\Vert
^{2}-\sigma _{0}^{2})+\frac{2\sigma _{0}^{2}}{k(k-1)}+\frac{2}{k^{2}(k-1)}%
\Vert S_{k}\Vert ^{2}  \label{e:Vk-Qk_decomp} \\
& \qquad +\frac{2}{(N-k)(N-k-1)}\sum_{j=k+1}^{N}\big(\Vert X_{j}\Vert
^{2}-\sigma _{0}^{2})  \notag \\
& \qquad +\frac{2\sigma _{0}^{2}}{(N-k)(N-k-1)}-\frac{2}{(N-k)^{2}(N-k-1)}%
\Vert S_{N}-S_{k}\Vert ^{2}.  \notag
\end{align}%
For the first term in \eqref{e:Vk-Qk_decomp}, applying Lemma \ref%
{l:weightappx}, when $\alpha >0$ we may take $\delta \in \lbrack 0,1/2)$ so
that $1/2<\alpha +\delta <1$, giving 
\begin{align*}
\max_{1\leq k\leq N}& \Big[\frac{k}{N}\Big(1-\frac{k}{N}\Big)\Big]^{2-\alpha
}\frac{2}{k(k-1)}\sum_{i=1}^{k}\left( \Vert X_{i}\Vert ^{2}-\sigma
_{0}^{2}\right) \\
\leq & \max_{1\leq k\leq N}\Big(\frac{k}{N}\Big)^{2-\alpha }\frac{2}{k(k-1)}%
\sum_{i=1}^{k}\big( \Vert X_{i}\Vert ^{2}-\sigma _{0}^{2}\big) \\
\leq & CN^{\alpha -2+\delta }\max_{1\leq k\leq N}\frac{1}{k^{\alpha +\delta }%
}\sum_{i=1}^{k}\big( \Vert X_{i}\Vert ^{2}-\sigma _{0}^{2}\big) %
=O_{P}(N^{\alpha +\delta -2}),\ 
\end{align*}%
and clearly ${\alpha +\delta -2}<-1$. When $\alpha =0$, taking $1/2<\delta
<1 $ immediately yields%
\begin{align*}
\max_{1\leq k\leq N}&\left[ \frac{k}{N}\left( 1-\frac{k}{N}\right) \right]
^{2}\frac{2}{k\left( k-1\right) }\sum_{i=1}^{k}\left( \Vert X_{i}\Vert
^{2}-\sigma _{0}^{2}\right) \\
\leq &\max_{1\leq k\leq N}\left( \frac{k}{N}\right) ^{2}\frac{2}{k\left(
k-1\right) }\sum_{i=1}^{k}\left( \Vert X_{i}\Vert ^{2}-\sigma _{0}^{2}\right)
\\
\leq &CN^{-2+\delta }\max_{1\leq k\leq N}\frac{1}{k^{\delta }}%
\sum_{i=1}^{k}\left( \Vert X_{i}\Vert ^{2}-\sigma _{0}^{2}\right) \\
=&O_{P}\left( N^{-2+\delta }\right) =o_{P}\left( \frac{1}{N}\right) .
\end{align*}%
Similarly, whenever $\alpha >0$ 
\begin{align*}
\max_{1\leq k\leq N}& \Big[\frac{k}{N}\Big(1-\frac{k}{N}\Big)\Big]^{2-\alpha
}\frac{2}{k^{2}(k-1)}\Vert S_{k}\Vert ^{2} \\
& \leq N^{\alpha -2}\max_{1\leq k\leq N}\frac{2}{k^{\alpha }(k-1)}\Vert
S_{k}\Vert ^{2}\leq CN^{\alpha -2}\Big(\max_{1\leq k\leq N}k^{-(\frac{1}{2}+%
\frac{\alpha }{2})}\Vert S_{k}\Vert \Big)^{2}=O_{P}(N^{\alpha -2});
\end{align*}%
when $\alpha =0$, we may take $\delta <1$ such that%
\begin{align*}
\max_{1\leq k\leq N}&\left[ \frac{k}{N}\left( 1-\frac{k}{N}\right) \right]
^{2}\frac{2}{k^{2}\left( k-1\right) }\Vert S_{k}\Vert ^{2} \\
\leq & N^{-2}\max_{1\leq k\leq N}\frac{2k^{\delta }}{k^{\delta }\left(
k-1\right) }\Vert S_{k}\Vert ^{2}\leq CN^{\delta -2}\left( \max_{1\leq k\leq
N}\frac{\Vert S_{k}\Vert }{k^{1/2+\delta /2}}\right) ^{2} \\
=&O_{P}\left( N^{\delta -2}\right) =o_{P}\left( N^{-1}\right) .
\end{align*}%
Finally, it is readily seen that 
\begin{equation*}
\max_{1\leq k\leq N}\Big[\frac{k}{N}\Big(1-\frac{k}{N}\Big)\Big]^{2-\alpha }%
\frac{2\sigma _{0}^{2}}{k(k-1)}=O(N^{\alpha -2}).
\end{equation*}%
Analogous arguments apply for the remaining terms in \eqref{e:Vk-Qk_decomp},
ultimately giving \eqref{e:Vk-Qk_appx_statement}.
\end{proof}
\end{lemma}

\begin{lemma}
\label{l:Q_N_boundary_bound} We assume that Assumption \ref{as1} is
satisfied. Then, for every fixed $x>0$, it holds that, for all $0\leq \alpha
<1$ 
\begin{equation}
\lim_{s\rightarrow 0}\limsup_{N\rightarrow \infty }P\Big\{N\max_{1\leq k\leq
Ns}\Big[\frac{k}{N}\Big(1-\frac{k}{N}\Big)\Big]^{2-\alpha }|Q_{N}(k)|>x\Big\}%
=0,  \label{e:Q_k_bound_near0}
\end{equation}%
and 
\begin{equation}
\lim_{s\rightarrow 0}\limsup_{N\rightarrow \infty }P\Big\{N\max_{(1-s)N\leq
k\leq N}\Big[\frac{k}{N}\Big(1-\frac{k}{N}\Big)\Big]^{2-\alpha }|Q_{N}(k)|>x%
\Big\}=0.  \label{e:Q_k_bound_near1}
\end{equation}

\begin{proof}
We first turn to \eqref{e:Q_k_bound_near0}. From \eqref{e:def_Qk}, for $%
1\leq k\leq Ns$, 
\begin{align}
N\Big[\frac{k}{N}\Big(1-\frac{k}{N}\Big)\Big]^{2-\alpha }|Q_{N}(k)|& \leq 2%
\Big[\frac{k}{N}\Big(1-\frac{k}{N}\Big)\Big]^{-\alpha }\Bigg(\Big\|N^{-1/2}%
\Big(S_{k}-\frac{k}{N}S_{N}\Big)\Big\|^{2}-\sigma _{0}^{2}\Big[\frac{k}{N}%
\Big(1-\frac{k}{N}\Big)\Big]\Bigg)  \notag \\
& \leq C\Bigg(N^{\alpha -1}\frac{\Vert S_{k}\Vert ^{2}}{k^{\alpha }}+\Big(%
\frac{k}{N}\Big)^{2-\alpha }\frac{\Vert S_{N}\Vert ^{2}}{N}+\sigma _{0}^{2}%
\Big(\frac{k}{N}\Big)^{1-\alpha }\Bigg)  \label{e:Q_k_bound}
\end{align}%
For the first term on the right-hand side of \eqref{e:Q_k_bound}, using the
maximal inequality \eqref{e:S_k_maximal_ineq}, 
\begin{align*}
P\Big\{N^{\alpha -1}\max_{1\leq k\leq Ns}k^{-\alpha }\Vert S_{k}\Vert ^{2}>x%
\Big\}& =P\Big\{N^{\alpha -1}\max_{1\leq j\leq \left\lceil \log
(Ns)\right\rceil }\max_{\exp \left( j-1\right) \leq k<\exp \left( j\right)
}k^{-\alpha }\Vert S_{k}\Vert ^{2}>x\Big\} \\
& \leq \sum_{j=1}^{\lceil \log (Ns)\rceil }P\Big\{N^{\alpha -1}\max_{\exp
\left( j-1\right) \leq k<\exp \left( j\right) }k^{-\alpha }\Vert S_{k}\Vert
^{2}>x\Big\} \\
& \leq \sum_{j=1}^{\lceil \log (Ns)\rceil }P\Big\{\max_{e^{j-1}\leq k\leq
e^{j}}\Vert S_{k}\Vert ^{2+\delta }>\big(xe^{\alpha (j-1)}N^{1-\alpha }\big)%
^{1+\delta /2}\Big\} \\
& \leq C\frac{N^{(\alpha -1)(1+\delta /2)}}{x^{1+\delta /2}}%
\sum_{j=1}^{\lceil \log (Ns)\rceil }e^{-(\alpha -1)(1+\delta /2)} \\
& \leq C\frac{N^{(\alpha -1)(1+\delta /2)}}{x^{1+\delta /2}}(Ns)^{(1-\alpha
)(1+\delta /2)}\leq C\frac{s^{(1-\alpha )(1+\delta /2)}}{x^{1+\delta /2}}.
\end{align*}%
This implies, for each $x>0$, 
\begin{equation*}
\lim_{s\rightarrow 0}\limsup_{N\rightarrow \infty }P\Big\{N^{\alpha
-1}\max_{1\leq k\leq Ns}k^{-\alpha }\Vert S_{k}\Vert ^{2}>x\Big\}=0.
\end{equation*}%
For the second term in \eqref{e:Q_k_bound}, since $\Vert S_{N}\Vert
^{2}/N=O_{P}(1)$, 
\begin{equation*}
\lim_{s\rightarrow 0}\limsup_{N\rightarrow \infty }P\bigg\{\max_{1\leq k\leq
Ns}\Big(\frac{k}{N}\Big)^{2-\alpha }\frac{\Vert S_{N}\Vert ^{2}}{N}>x\bigg\}%
=\lim_{s\rightarrow 0}\limsup_{N\rightarrow \infty }P\bigg\{\frac{\Vert
S_{N}\Vert ^{2}}{N}>xs^{\alpha -2}\bigg\}=0.
\end{equation*}%
Finally, for $1\leq k\leq Ns$, the third term in \eqref{e:Q_k_bound} clearly
tends to $0$ uniformly in $N$ as $s\rightarrow 0$ since $\alpha <1$, which
gives \eqref{e:Q_k_bound_near0}. Turning to \eqref{e:Q_k_bound_near1}, for
each $k$ in the range $N(1-s)\leq k\leq N$, since 
\begin{equation*}
S_{k}-\frac{k}{N}S_{N}=(S_{N}-S_{k})-\Big(1-\frac{k}{N}\Big)S_{N},
\end{equation*}%
we have 
\begin{equation*}
N\Big[\frac{k}{N}\Big(1-\frac{k}{N}\Big)\Big]^{2-\alpha }|Q_{N}(k)|\leq C%
\bigg(N^{\alpha -1}\frac{\Vert S_{N}-S_{k}\Vert ^{2}}{k^{\alpha }}+\Big(1-%
\frac{k}{N}\Big)^{2-\alpha }\frac{\Vert S_{N}\Vert ^{2}}{N}+\Big(1-\frac{k}{N%
}\Big)^{1-\alpha }\bigg).
\end{equation*}%
In view of \eqref{e:Q_k_bound}, the same arguments for %
\eqref{e:Q_k_bound_near0} therefore give \eqref{e:Q_k_bound_near1}, \textit{%
mutatis mutandis}.
\end{proof}
\end{lemma}

\begin{lemma}
\label{l:Delta_boundary_bound} For each fixed $x>0$, it holds that 
\begin{equation}
\begin{gathered} \lim_{s\rightarrow 0}P\bigg\{\sup_{0\leq u\leq s}u^{-\alpha
}|\Delta (u)|>x\bigg\}=0,\\ \mathnormal{and}\quad \lim_{s\rightarrow
0}P\bigg\{\sup_{1-s\leq u\leq 1}(1-u)^{-\alpha }|\Delta (u)|>x\bigg\}=0,
\label{e:estimate_on_Delta} \end{gathered}
\end{equation}
for all $0\leq \alpha <1$.

\begin{proof}
First note 
\begin{equation}
|\Delta (u)|\leq \int |\Gamma (u,t)|^{2}dt+\sigma _{0}^{2}u(1-u),
\label{ssG}
\end{equation}%
and 
\begin{align*}
\{\Gamma (1-u,t),0\leq u\leq 1,t\in \mathcal{T}\}& \overset{\mathcal{D}}{=}%
\{\Gamma (u,t),0\leq u\leq 1,t\in \mathcal{T}\} \\
& \overset{\mathcal{D}}{=}\{G(u,t)-uG(1,t),0\leq u\leq 1,t\in \mathcal{T}\},
\end{align*}%
where $\{G(u,t),u\geq 0,t\in \mathcal{T}\}$ is a Gaussian process with $%
EG(u,t)=0$ and $EG(u,t)G (u^{\prime },t^{\prime} )=\min \{u,u^{\prime }\}%
\mathbf{D}(t,t^{\prime })$, which can be verified by checking the covariance
functions. Thus, it suffices to establish 
\begin{equation}
\lim_{s\rightarrow 0}P\bigg\{\sup_{0\leq u\leq s}u^{-\alpha }\int
|G(u,t)|^{2}dt>x\bigg\}=0.  \label{implied}
\end{equation}%
Note that for each $c>0$, it holds that $\left\{ G(cu,t),u\geq 0,t\in 
\mathcal{T}\right\} \overset{\mathcal{D}}{=}\left\{ c^{1/2}G(u,t),u\geq
0,t\in \mathcal{T}\right\} $. Thus, for each $s>0$, we have 
\begin{align}
\sup_{0\leq u\leq s}u^{-\alpha }\int |G(u,t)|^{2}dt&=\sup_{0\leq u^{\prime
}\leq 1}\left( su^{\prime }\right) ^{-\alpha }\int \left\vert G\left(
su^{\prime },t\right) \right\vert ^{2}dt  \notag \\
&\overset{\mathcal{D}}{=}s^{1-\alpha }\sup_{0\leq v\leq 1}v^{-\alpha }\int
\left\vert G\left( v,t\right) \right\vert ^{2}dt.  \label{implies}
\end{align}%
Following exactly the same logic as in the proof of (\ref{integral-test}) in
Theorem \ref{th-1}, it follows that $\sup_{0\leq u\leq 1}v^{-\alpha }\int
|G(u,t)|^{2}dt=O_{P}\left( 1\right) $, which by (\ref{implies}) implies (\ref%
{implied}).
\end{proof}
\end{lemma}

\begin{lemma}
\label{refinement}We assume that Assumption \ref{as1} and (\ref{k-star}) are
satisfied, and that, as $N\rightarrow \infty $, (\ref{vanishing}) holds.
Then it holds that 
\begin{equation}
\Vert \mathcal{\delta }\Vert ^{2}\left( \widehat{k}_{N}-k^{\ast }\right)
=O_{P}\left( 1\right) .  \label{refine}
\end{equation}
\end{lemma}

\begin{proof}
We begin by noting that $\widehat{\theta }_{N}-\theta =o_{P}\left( 1\right) $
implies that $\widehat{k}_{N}-k^{\ast }=o_{P}\left( N\right) $. This entails
that in our calculations below we can assume that $aN\leq k\leq bN$, for any 
$a<\theta <b$. We begin by defining%
\begin{equation*}
\gamma _{N}\left(C\right) =C\Vert \mathcal{\delta }\Vert ^{-2},
\end{equation*}%
for some positive constant $C$, and the function%
\begin{equation*}
\Pi \left( k\right) =\left\{ 
\begin{array}{ll}
\displaystyle\frac{k\left( N-k^{\ast }\right) }{N}\mathcal{\delta } & 1\leq
k\leq k^{\ast }, \\ 
\displaystyle\frac{k^{\ast }\left( N-k\right) }{N}\mathcal{\delta } & 
k^{\ast }+1\leq k\leq N.%
\end{array}%
\right.
\end{equation*}%
By standard algebra, it follows that $V_{N}\left( k\right) $ can be written
as%
\begin{align*}
V_{N}\left( k\right) = 2&\left[ \frac{N}{k\left( N-k\right) }\right]
^{2}\left\Vert S_{k}-\frac{k}{N}S_{N}\right\Vert ^{2}-\frac{2}{k\left(
k-1\right) }\sum_{i=1}^{k}\left\Vert X_{i}\right\Vert ^{2} \\
& -\frac{2}{\left( N-k\right) \left( N-k-1\right) }\sum_{i=k+1}^{N}\left%
\Vert X_{i}\right\Vert ^{2}+\frac{2}{k^{2}\left( k-1\right) }\left\Vert
S_{k}\right\Vert ^{2} \\
& +\frac{2}{\left( N-k\right) ^{2}\left( N-k-1\right) }\left\Vert
S_{N}-S_{k}\right\Vert ^{2}.
\end{align*}%
We will consider the following function%
\begin{equation*}
\widetilde{V}\left( k\right) =\frac{1}{2}\left[ \frac{N}{k\left( N-k\right) }%
\right] ^{\alpha -2}V_{N}\left( k\right) ,
\end{equation*}%
and study%
\begin{equation*}
\widetilde{V}\left( k\right) -\widetilde{V}\left( k^{\ast }\right)
=\sum_{j=1}^{10}\widetilde{V}_{k,j},
\end{equation*}%
where we have defined%
\begin{equation*}
\widetilde{V}_{k,1}=\left[ \left( \frac{N}{k\left( N-k\right) }\right)
^{\alpha }-\left( \frac{N}{k^{\ast }\left( N-k^{\ast }\right) }\right)
^{\alpha }\right] \left\Vert \sum_{i=1}^{k}\epsilon _{i}-\frac{k}{N}%
\sum_{i=1}^{N}\epsilon _{i}\right\Vert ^{2},
\end{equation*}%
\begin{align*}
\widetilde{V}_{k,2}=& \left( \frac{N}{k^{\ast }\left( N-k^{\ast }\right) }%
\right) ^{\alpha }\left\langle \sum_{i=1}^{k}\epsilon _{i}-\frac{k}{N}%
\sum_{i=1}^{N}\epsilon _{i}+\sum_{i=1}^{k^{\ast }}\epsilon _{i}-\frac{%
k^{\ast }}{N}\sum_{i=1}^{N}\epsilon _{i},\right. \\
&\qquad\qquad\qquad\qquad\qquad\qquad\qquad\qquad \left.
\sum_{i=1}^{k}\epsilon _{i}-\sum_{i=1}^{k^{\ast }}\epsilon _{i}-\frac{%
k-k^{\ast }}{N}\sum_{i=1}^{N}\epsilon _{i}\right\rangle ,
\end{align*}%
\begin{equation*}
\widetilde{V}_{k,3}=2\left\langle \left( \frac{N}{k^{\ast }\left( N-k^{\ast
}\right) }\right) ^{\alpha }\Pi \left( k\right) -\left( \frac{N}{k^{\ast
}\left( N-k^{\ast }\right) }\right) ^{\alpha }\Pi \left( k^{\ast }\right)
,\sum_{i=1}^{k}\epsilon _{i}-\frac{k}{N}\sum_{i=1}^{N}\epsilon
_{i}\right\rangle ,
\end{equation*}%
\begin{equation*}
\widetilde{V}_{k,4}=2\left\langle \left( \frac{N}{k^{\ast }\left( N-k^{\ast
}\right) }\right) ^{\alpha }\Pi \left( k^{\ast }\right) ,\frac{k-k^{\ast }}{N%
}\sum_{i=1}^{N}\epsilon _{i}\right\rangle ,
\end{equation*}%
\begin{equation*}
\widetilde{V}_{k,5}=2\left\langle \left( \frac{N}{k^{\ast }\left( N-k^{\ast
}\right) }\right) ^{\alpha }\Pi \left( k^{\ast }\right)
,\sum_{i=1}^{k}\epsilon _{i}-\sum_{i=1}^{k^{\ast }}\epsilon
_{i}\right\rangle ,
\end{equation*}%
\begin{equation*}
\widetilde{V}_{k,6}=\left( \frac{N}{k\left( N-k\right) }\right) ^{\alpha
}\left\Vert \Pi \left( k\right) \right\Vert ^{2}-\left( \frac{N}{k^{\ast
}\left( N-k^{\ast }\right) }\right) ^{\alpha }\left\Vert \Pi \left( k^{\ast
}\right) \right\Vert ^{2},
\end{equation*}%
\begin{align*}
\widetilde{V}_{k,7} =&-\left[ \frac{N}{k\left( N-k\right) }\right] ^{\alpha
-2}\frac{1}{k\left( k-1\right) }\sum_{i=1}^{k}\left\Vert X_{i}\right\Vert
^{2} \\
&+\left[ \frac{N}{k^{\ast }\left( N-k^{\ast }\right) }\right] ^{\alpha -2}%
\frac{1}{k^{\ast }\left( k^{\ast }-1\right) }\sum_{i=1}^{k^{\ast
}}\left\Vert X_{i}\right\Vert ^{2},
\end{align*}%
\begin{align*}
\widetilde{V}_{k,8}=& -\left[ \frac{N}{k\left( N-k\right) }\right] ^{\alpha
-2}\frac{1}{\left( N-k\right) \left( N-k-1\right) }\sum_{i=k+1}^{N}\left%
\Vert X_{i}\right\Vert ^{2} \\
& +\left[ \frac{N}{k^{\ast }\left( N-k^{\ast }\right) }\right] ^{\alpha -2}%
\frac{1}{\left( N-k^{\ast }\right) \left( N-k^{\ast }-1\right) }%
\sum_{i=k^{\ast }+1}^{N}\left\Vert X_{i}\right\Vert ^{2},
\end{align*}%
\begin{align*}
\widetilde{V}_{k,9} &=\left[ \frac{N}{k\left( N-k\right) }\right] ^{\alpha
-2}\frac{1}{k^{2}\left( k-1\right) }\left\Vert S_{k}\right\Vert ^{2} \\
&\qquad-\left[ \frac{N}{k^{\ast }\left( N-k^{\ast }\right) }\right] ^{\alpha
-2}\frac{1}{\left( k^{\ast }\right) ^{2}\left( k^{\ast }-1\right) }%
\left\Vert S_{k^{\ast }}\right\Vert ^{2},
\end{align*}%
\begin{align*}
\widetilde{V}_{k,10}=& \left[ \frac{N}{k\left( N-k\right) }\right] ^{\alpha
-2}\frac{1}{\left( N-k\right) ^{2}\left( N-k-1\right) }\left\Vert
S_{N}-S_{k}\right\Vert ^{2} \\
& -\left[ \frac{N}{k^{\ast }\left( N-k^{\ast }\right) }\right] ^{\alpha -2}%
\frac{1}{\left( N-k^{\ast }\right) ^{2}\left( N-k^{\ast }-1\right) }%
\left\Vert S_{N}-S_{k^{\ast }}\right\Vert ^{2}.
\end{align*}%
It immediately follows from the Mean Value Theorem that there exist two
positive constants $c_{1}\geq c_{2}$ such that%
\begin{equation}
-c_{1}N^{1-\alpha }\left\vert k^{\ast }-k\right\vert \Vert \mathcal{\delta }%
\Vert ^{2}\leq \widetilde{V}_{k,6}\leq -c_{2}N^{1-\alpha }\left\vert k^{\ast
}-k\right\vert \Vert \mathcal{\delta }\Vert ^{2}.  \label{hrice-1}
\end{equation}%
We note that, similarly to the proof of Theorem 2.2.1\textit{(i)} in %
\citet{chgreg}, (see also \cite{aue:gabrys:horvath:kokoszka:2009}) it
follows that%
\begin{equation*}
\max_{\left\vert k^{\ast }-k\right\vert \geq \gamma _{N}\left( C\right)
,aN\leq k\leq bN}\frac{\left\vert \widetilde{V}_{k,j}\right\vert }{%
N^{1-\alpha }\left\vert k^{\ast }-k\right\vert \Vert \mathcal{\delta }\Vert
^{2}}=o_{P}\left( 1\right) ,
\end{equation*}%
for all $1\leq j\leq 4$. Indeed,\footnote{%
We report passages for the case $k\leq k^{\ast }-\gamma _{N}\left( C\right) $%
; the case $k\geq k^{\ast }+\gamma _{N}\left( C\right) $ follows from the
same logic.} Lemma \ref{l:BHR_11} entails that%
\begin{equation}
\max_{1\leq k\leq N}\left\Vert \sum_{i=1}^{k}\epsilon _{i}\right\Vert
=O_{P}\left( N^{1/2}\right) .  \label{fclt1}
\end{equation}%
We now show that 
\begin{equation}
\max_{1\leq k\leq k^{\ast }-\gamma _{N}\left( C\right) }\frac{1}{k^{\ast }-k}%
\left\Vert \sum_{i=k+1}^{k^{\ast }}\epsilon _{i}\right\Vert =O_{P}\left(
\gamma _{N}^{-1/2}\left( C\right) \right) .  \label{fclt2}
\end{equation}%
The proof uses similar arguments as above, so we only report its main
passages. By stationarity%
\begin{align*}
&P\left( \max_{1\leq k\leq k^{\ast }-\gamma _{N}\left( C\right) }\frac{1}{%
k^{\ast }-k}\left\Vert \sum_{i=k+1}^{k^{\ast }}\epsilon _{i}\right\Vert \geq
x\gamma _{N}^{-1/2}\left( C\right) \right) \\
&\qquad=P\left( \max_{1\leq k\leq k^{\ast }-\gamma _{N}\left( C\right) }%
\frac{1}{k^{\ast }-k}\left\Vert \sum_{i=1}^{k^{\ast }-k}\epsilon
_{i}\right\Vert \geq x\gamma _{N}^{-1/2}\left( C\right) \right) .
\end{align*}%
Hence%
\begin{align}
P&\left( \max_{1\leq k\leq k^{\ast }-\gamma _{N}\left( C\right) }\frac{1}{%
k^{\ast }-k}\left\Vert \sum_{i=1}^{k^{\ast }-k}\epsilon _{i}\right\Vert \geq
x\gamma _{N}^{-1/2}\left( C\right) \right)  \notag \\
&=P\left( \max_{\gamma _{N}\left( C\right) \leq u\leq k^{\ast }}\frac{1}{u}%
\left\Vert \sum_{i=1}^{u}\epsilon _{i}\right\Vert \geq x\gamma
_{N}^{-1/2}\left( C\right) \right)  \notag \\
&\leq P\left( \max_{\left\lfloor \log \gamma _{N}\left( C\right)
\right\rfloor \leq \ell \leq \infty }\max_{\exp \left( \ell \right) \leq
u\leq \exp \left( \ell +1\right) }\frac{1}{u}\left\Vert
\sum_{i=1}^{u}\epsilon _{i}\right\Vert \geq x\gamma _{N}^{-1/2}\left(
C\right) \right)  \notag \\
&\leq \sum_{\ell =\left\lfloor \log \gamma _{N}\left( C\right) \right\rfloor
}^{\infty }P\left( \max_{\exp \left( \ell \right) \leq u\leq \exp \left(
\ell +1\right) }\frac{1}{u}\left\Vert \sum_{i=1}^{u}\epsilon _{i}\right\Vert
\geq x\gamma _{N}^{-1/2}\left( C\right) \right)  \notag \\
&\leq \sum_{\ell =\left\lfloor \log \gamma _{N}\left( C\right) \right\rfloor
}^{\infty }P\left( \max_{\exp \left( \ell \right) \leq u\leq \exp \left(
\ell +1\right) }\left\Vert \sum_{i=1}^{u}\epsilon _{i}\right\Vert \geq
x\gamma _{N}^{-1/2}\left( C\right) \exp \left( \ell +1\right) \right)  \notag
\\
&\leq x^{-\left( 2+s\right) }\gamma _{N}^{\left( 2+s\right) /2}\left(
C\right) \sum_{\ell =\left\lfloor \log \gamma _{N}\left( C\right)
\right\rfloor }^{\infty }\exp \left( \left( 2+s\right) \left( \ell +1\right)
\right) E\max_{1\leq u\leq \exp \left( \ell +1\right) }\left\Vert
\sum_{i=1}^{u}\epsilon _{i}\right\Vert ^{2+s},
\label{e:v3bound_intermediate}
\end{align}%
for some $0<s<2$. Using Lemma \ref{l:BHR_33}, expression %
\eqref{e:v3bound_intermediate} is bounded by%
\begin{align*}
& c_{0}x^{-\left( 2+s\right) }\gamma _{N}^{\left( 2+s\right) /2}\left(
C\right) \sum_{\ell =\left\lfloor \log \gamma _{N}\left( C\right)
\right\rfloor }^{\infty }\exp \left( -\left( 2+s\right) \ell \right) \exp
\left( \frac{1}{2}\left( 2+s\right) \ell \right) \\
&\leq c_{1}x^{-\left( 2+s\right) },
\end{align*}%
where $c_{0}$ and $c_{1}$\ are finite, positive constants. Now (\ref{fclt2})
follows. Finally, a routine application of the Mean Value Theorem yields%
\begin{equation}
\max_{aN\leq k\leq bN}\frac{1}{\left\vert k^{\ast }-k\right\vert }\left[
\left( \frac{N}{k\left( N-k\right) }\right) ^{\alpha }-\left( \frac{N}{%
k^{\ast }\left( N-k^{\ast }\right) }\right) ^{\alpha }\right] =O\left(
N^{-1-\alpha }\right) .  \label{mvt}
\end{equation}

Therefore%
\begin{align*}
\max_{{\left\vert k^{\ast }-k\right\vert \geq \gamma _{N}\left(
C\right) }{aN\leq k\leq bN}}&\frac{\left\vert \widetilde{V}_{k,1}\right\vert 
}{N^{1-\alpha }\left\vert k^{\ast }-k\right\vert \Vert \mathcal{\delta }%
\Vert ^{2}} \\
\leq & \frac{2}{N^{1-\alpha }\Vert \mathcal{\delta }\Vert ^{2}}\left(
\max_{\left\vert k^{\ast }-k\right\vert \geq \gamma _{N}\left( C\right)
,aN\leq k\leq bN}\left\Vert \sum_{i=1}^{k}\epsilon _{i}\right\Vert
^{2}+\max_{\left\vert k^{\ast }-k\right\vert \geq \gamma _{N}\left( C\right)
,aN\leq k\leq bN}\left\Vert \frac{k}{N}\sum_{i=1}^{N}\epsilon
_{i}\right\Vert ^{2}\right) \\
& \qquad\qquad\qquad\times \max_{\left\vert k^{\ast }-k\right\vert \geq
\gamma _{N}\left( C\right) ,aN\leq k\leq bN}\frac{1}{\left\vert k^{\ast
}-k\right\vert }\left[ \left( \frac{N}{k\left( N-k\right) }\right) ^{\alpha
}-\left( \frac{N}{k^{\ast }\left( N-k^{\ast }\right) }\right) ^{\alpha }%
\right] \\
=& \frac{O_{P}\left( N\right) +O_{P}\left( \gamma _{N}\left( C\right)
\right) }{N^{1-\alpha }\Vert \mathcal{\delta }\Vert ^{2}}N^{-1-\alpha }=%
\frac{O_{P}\left( N\right) }{N^{1-\alpha }\Vert \mathcal{\delta }\Vert ^{2}}%
N^{-1-\alpha }=o_{P}\left( 1\right) ,
\end{align*}%
where we have used (\ref{fclt1}) and (\ref{mvt}) in the final passage, and
the fact that $N\Vert \mathcal{\delta }\Vert ^{2}\rightarrow \infty $ and
that, by construction, $\gamma _{N}\left( C\right) =o_{P}\left( N\right) $.
Similarly%
\begin{align*}
\max_{\left\vert k^{\ast }-k\right\vert \geq \gamma _{N}\left( C\right)
,aN\leq k\leq bN}&\frac{\left\vert \widetilde{V}_{k,2}\right\vert }{%
N^{1-\alpha }\left\vert k^{\ast }-k\right\vert \Vert \mathcal{\delta }\Vert
^{2}} \\
& \hspace{-15ex}\leq\frac{1}{N^{1-\alpha }\Vert \mathcal{\delta }\Vert ^{2}}%
\max_{\left\vert k^{\ast }-k\right\vert \geq \gamma _{N}\left( C\right)
,aN\leq k\leq bN}\left( \frac{N}{k^{\ast }\left( N-k^{\ast }\right) }\right)
^{\alpha } \\
&\hspace{-12ex} \times \max_{\left\vert k^{\ast }-k\right\vert \geq \gamma
_{N}\left( C\right) ,aN\leq k\leq bN}\left( \left\Vert
\sum_{i=1}^{k}\epsilon _{i}\right\Vert +\left\Vert \sum_{i=1}^{k^{\ast
}}\epsilon _{i}\right\Vert +\left\Vert \frac{k}{N}\sum_{i=1}^{N}\epsilon
_{i}\right\Vert +\left\Vert \frac{k^{\ast }}{N}\sum_{i=1}^{N}\epsilon
_{i}\right\Vert \right) \\
& \hspace{-12ex}\times \max_{\left\vert k^{\ast }-k\right\vert \geq \gamma
_{N}\left( C\right) ,aN\leq k\leq bN}\frac{1}{\left\vert k^{\ast
}-k\right\vert }\left( \left\Vert \sum_{i=k+1}^{k^{\ast }}\epsilon
_{i}\right\Vert +\left\Vert \frac{k-k^{\ast }}{N}\sum_{i=1}^{N}\epsilon
_{i}\right\Vert \right) \\
&\hspace{-15ex}= O_{P}\left( 1\right) \frac{1}{N^{1-\alpha }\Vert \mathcal{%
\delta }\Vert ^{2}}N^{-\alpha }\left( N^{1/2}\right) \left( \gamma
_{N}^{-1/2}\left( C\right) +N^{-1/2}\right) =O_{P}\left( \frac{1}{%
N^{1/2}\Vert \mathcal{\delta }\Vert }\right) +O_{P}\left( \frac{1}{N\Vert 
\mathcal{\delta }\Vert ^{2}}\right) \\
&\hspace{-15ex}= o_{P}\left( 1\right) ,
\end{align*}%
\begin{align*}
\max_{\left\vert k^{\ast }-k\right\vert \geq \gamma _{N}\left( C\right)
,aN\leq k\leq bN}&\frac{\left\vert \widetilde{V}_{k,3}\right\vert }{%
N^{1-\alpha }\left\vert k^{\ast }-k\right\vert \Vert \mathcal{\delta }\Vert
^{2}} \\
&\hspace{-15ex}\leq \frac{2}{N^{1-\alpha }\Vert \mathcal{\delta }\Vert ^{2}}%
\max_{\left\vert k^{\ast }-k\right\vert \geq \gamma _{N}\left( C\right)
,aN\leq k\leq bN}\left\Vert \sum_{i=1}^{k}\epsilon _{i}-\frac{k}{N}%
\sum_{i=1}^{N}\epsilon _{i}\right\Vert \\
&\hspace{-10ex}\times \max_{\left\vert k^{\ast }-k\right\vert \geq \gamma
_{N}\left( C\right) ,aN\leq k\leq bN}\frac{1}{\left\vert k^{\ast
}-k\right\vert }\left( \frac{N}{k^{\ast }\left( N-k^{\ast }\right) }\right)
^{\alpha }\left\Vert \left( \Pi \left( k\right) -\Pi \left( k^{\ast }\right)
\right) \right\Vert \\
&\hspace{-15ex}= O_{P}\left( N^{1/2}\right) \frac{1}{N^{1-\alpha }\Vert 
\mathcal{\delta }\Vert ^{2}}N^{-\alpha }\Vert \mathcal{\delta }\Vert
=O_{P}\left( \frac{1}{N^{1/2}\Vert \mathcal{\delta }\Vert }\right)
=o_{P}\left( 1\right) ,
\end{align*}%
and 
\begin{align*}
\max_{\left\vert k^{\ast }-k\right\vert \geq \gamma _{N}\left( C\right)
,aN\leq k\leq bN}&\frac{\left\vert \widetilde{V}_{k,4}\right\vert }{%
N^{1-\alpha }\left\vert k^{\ast }-k\right\vert \Vert \mathcal{\delta }\Vert
^{2}} \\
&\hspace{-15ex}\leq \frac{2}{N^{1-\alpha }\Vert \mathcal{\delta }\Vert ^{2}}%
\max_{\left\vert k^{\ast }-k\right\vert \geq \gamma _{N}\left( C\right)
,aN\leq k\leq bN}\left( \frac{N}{k^{\ast }\left( N-k^{\ast }\right) }\right)
^{\alpha }\left\Vert \Pi \left( k^{\ast }\right) \right\Vert \\
&\hspace{-12ex}\times \max_{\left\vert k^{\ast }-k\right\vert \geq \gamma
_{N}\left( C\right) ,aN\leq k\leq bN}\frac{1}{\left\vert k^{\ast
}-k\right\vert }\left\Vert \frac{k-k^{\ast }}{N}\sum_{i=1}^{N}\epsilon
_{i}\right\Vert \\
& \hspace{-15ex}=O_{P}\left( N^{-1/2}\right) \frac{2}{N^{1-\alpha }\Vert 
\mathcal{\delta }\Vert ^{2}}N^{1-\alpha }\Vert \mathcal{\delta }\Vert
=O_{P}\left( \frac{1}{N^{1/2}\Vert \mathcal{\delta }\Vert }\right) \\
& \hspace{-15ex}= o_{P}\left( 1\right) ,
\end{align*}%
having used (\ref{fclt2}). We now consider $\widetilde{V}_{k,7}$, writing it
as%
\begin{align*}
\widetilde{V}_{k,7}=& \left[ \left( \frac{N}{k^{\ast }\left( N-k^{\ast
}\right) }\right) ^{\alpha -2}-\left( \frac{N}{k\left( N-k\right) }\right)
^{\alpha -2}\right] \frac{1}{k^{\ast }\left( k^{\ast }-1\right) }%
\sum_{i=1}^{k}\left\Vert X_{i}\right\Vert ^{2} \\
& \quad+\left( \frac{N}{k\left( N-k\right) }\right) ^{\alpha -2}\left[ \frac{%
1}{k^{\ast }\left( k^{\ast }-1\right) }-\frac{1}{k\left( k-1\right) }\right]
\sum_{i=1}^{k}\left\Vert X_{i}\right\Vert ^{2} \\
& \quad +\left( \frac{N}{k^{\ast }\left( N-k^{\ast }\right) }\right)
^{\alpha -2}\frac{1}{k^{\ast }\left( k^{\ast }-1\right) }\sum_{i=k+1}^{k^{%
\ast }}\left\Vert X_{i}\right\Vert ^{2} \\
=& \widetilde{V}_{k,7,1}+\widetilde{V}_{k,7,2}+\widetilde{V}_{k,7,3}.
\end{align*}%
Using the Law of Large Numbers, it holds that 
\begin{equation}
\max_{1\leq k\leq N}\sum_{i=1}^{k}\left\Vert X_{i}\right\Vert
^{2}=\sum_{i=1}^{N}\left\Vert X_{i}\right\Vert ^{2}=O_{P}\left( N\right) .
\label{lln}
\end{equation}%
Further, it holds that%
\begin{equation}
\max_{1\leq k\leq k^{\ast }-\gamma _{N}\left( C\right) }\frac{1}{k^{\ast }-k}%
\left\vert \sum_{i=k+1}^{k^{\ast }}\left( \left\Vert X_{i}\right\Vert
^{2}-E\left\Vert X_{i}\right\Vert ^{2}\right) \right\vert =O_{P}\left(
\gamma ^{-1/2}(N)\right) =o_{P}\left( 1\right) .  \label{m-ineq}
\end{equation}%
Indeed, for all $x$%
\begin{align*}
P&\left( \max_{1\leq k\leq k^{\ast }-\gamma _{N}\left( C\right) }\frac{1}{%
k^{\ast }-k}\left\vert \sum_{i=k+1}^{k^{\ast }}\left( \left\Vert
X_{i}\right\Vert ^{2}-E\left\Vert X_{i}\right\Vert ^{2}\right) \right\vert
>x\gamma _{N}^{-1/2}\left( C\right) \right) \\
&= P\left( \max_{\gamma _{N}\left( C\right) \leq j\leq k^{\ast }}\frac{1}{j}%
\left\vert \sum_{i=1}^{j}\left( \left\Vert X_{i}\right\Vert ^{2}-E\left\Vert
X_{i}\right\Vert ^{2}\right) \right\vert >x\gamma _{N}^{-1/2}\left( C\right)
\right) \\
&\leq P\left( \max_{\left\lceil \log \gamma _{N}\left( C\right) \right\rceil
\leq \ell \leq \left\lceil \log k^{\ast }\right\rceil }\max_{\exp \left(
\ell -1\right) \leq j\leq \exp \left( \ell \right) }\frac{1}{j}\left\vert
\sum_{i=1}^{j}\left( \left\Vert X_{i}\right\Vert ^{2}-E\left\Vert
X_{i}\right\Vert ^{2}\right) \right\vert >x\gamma _{N}^{-1/2}\left( C\right)
\right) \\
&\leq \sum_{\ell =\left\lceil \log \gamma _{N}\left( C\right) \right\rceil
}^{\left\lceil \log k^{\ast }\right\rceil }P\left( \max_{\exp \left( \ell
-1\right) \leq j\leq \exp \left( \ell \right) }\left\vert
\sum_{i=1}^{j}\left( \left\Vert X_{i}\right\Vert ^{2}-E\left\Vert
X_{i}\right\Vert ^{2}\right) \right\vert >x\gamma _{N}^{-1/2}\left( C\right)
\exp \left( -\ell +1\right) \right) \\
&\leq x^{-\left( 2+\epsilon /2\right) }\gamma _{N}\left( C\right) \sum_{\ell
=\left\lceil \log \gamma _{N}\left( C\right) \right\rceil }^{\left\lceil
\log k^{\ast }\right\rceil }\exp \left( -\left( 2+\epsilon /2\right) \left(
\ell -1\right) \right) E\left( \max_{1\leq j\leq \exp \left( \ell \right)
}\left\vert \sum_{i=1}^{j}\left( \left\Vert X_{i}\right\Vert
^{2}-E\left\Vert X_{i}\right\Vert ^{2}\right) \right\vert ^{2+\epsilon
/2}\right) \\
&\leq x^{-\left( 2+\epsilon /2\right) }\gamma _{N}\left( C\right) \sum_{\ell
=\left\lceil \log \gamma _{N}\left( C\right) \right\rceil }^{\left\lceil
\log k^{\ast }\right\rceil }\exp \left( -\left( 2+\epsilon /2\right) \left(
\ell -1\right) \right) \exp \left( \ell \left( 1+\epsilon /4\right) \right)
=c_{0}x^{-\left( 2+\epsilon /2\right) },
\end{align*}%
having used Lemma \ref{l:BHR_33}, since $\left\Vert X_{i}\right\Vert ^{2}$
is a Bernoulli shift sequence which satisfies Assumption \ref{as1}. Hence,
recalling that $k^{\ast }=\left\lfloor N\theta \right\rfloor $ with $%
0<\theta <1$, we have%
\begin{align*}
\max_{\left\vert k^{\ast }-k\right\vert \geq \gamma _{N}\left( C\right)
,aN\leq k\leq bN}&\frac{\left\vert \widetilde{V}_{k,7,3}\right\vert }{%
N^{1-\alpha }\left\vert k^{\ast }-k\right\vert \Vert \mathcal{\delta }\Vert
^{2}} \\
&\hspace{-15ex}\leq \frac{1}{N^{1-\alpha }\Vert \mathcal{\delta }\Vert ^{2}}%
\left( \frac{N}{k^{\ast }\left( N-k^{\ast }\right) }\right) ^{\alpha -2}%
\frac{1}{k^{\ast }\left( k^{\ast }-1\right) }\max_{\left\vert k^{\ast
}-k\right\vert \geq \gamma _{N}\left( C\right) ,aN\leq k\leq bN}\frac{1}{%
\left\vert k^{\ast }-k\right\vert }\sum_{i=k+1}^{k^{\ast }}\left\Vert
X_{i}\right\Vert ^{2} \\
&\hspace{-15ex}= O_{P}\left( 1\right) \frac{1}{N^{1-\alpha }\Vert \mathcal{%
\delta }\Vert ^{2}}N^{2-\alpha }\frac{1}{N^{2}}o_{P}\left( 1\right)
=O_{P}\left( \frac{1}{N\Vert \mathcal{\delta }\Vert ^{2}}\right)
=o_{P}\left( 1\right) .
\end{align*}

Further%
\begin{align*}
& \max_{\left\vert k^{\ast }-k\right\vert \geq \gamma _{N}\left( C\right)
,aN\leq k\leq bN}\frac{\left\vert \widetilde{V}_{k,7,1}\right\vert }{%
N^{1-\alpha }\left\vert k^{\ast }-k\right\vert \Vert \mathcal{\delta }\Vert
^{2}} \\
&\hspace{5ex} \leq \frac{1}{N^{1-\alpha }\Vert \mathcal{\delta }\Vert ^{2}}%
\frac{1}{k^{\ast }\left( k^{\ast }-1\right) }\sum_{i=1}^{N}\left\Vert
X_{i}\right\Vert ^{2} \\
&\hspace{15ex}\times \max_{\left\vert k^{\ast }-k\right\vert \geq \gamma
_{N}\left( C\right) ,aN\leq k\leq bN}\frac{1}{\left\vert k^{\ast
}-k\right\vert }\left[ \left( \frac{N}{k^{\ast }\left( N-k^{\ast }\right) }%
\right) ^{\alpha -2}-\left( \frac{N}{k\left( N-k\right) }\right) ^{\alpha -2}%
\right] \\
&\hspace{5ex}= O_{P}\left( N\right) \frac{1}{N^{1-\alpha }\Vert \mathcal{%
\delta }\Vert ^{2}}\frac{1}{N^{2}}O\left( N^{-1-\alpha +2}\right)
=O_{P}\left( \frac{1}{N\Vert \mathcal{\delta }\Vert ^{2}}\right)
=o_{P}\left( 1\right) ,
\end{align*}%
having used (\ref{lln}) and the Mean Value Theorem. Finally we have%
\begin{align*}
& \max_{\left\vert k^{\ast }-k\right\vert \geq \gamma _{N}\left( C\right)
,aN\leq k\leq bN}\frac{\left\vert \widetilde{V}_{k,7,2}\right\vert }{%
N^{1-\alpha }\left\vert k^{\ast }-k\right\vert \Vert \mathcal{\delta }\Vert
^{2}} \\
&\hspace{5ex}\leq \frac{1}{N^{1-\alpha }\Vert \mathcal{\delta }\Vert ^{2}}%
\sum_{i=1}^{N}\left\Vert X_{i}\right\Vert ^{2}\max_{\left\vert k^{\ast
}-k\right\vert \geq \gamma ,aN\leq k\leq bN}\left( \frac{N}{k\left(
N-k\right) }\right) ^{\alpha -2} \\
&\hspace{15ex} \times\max_{\left\vert k^{\ast }-k\right\vert \geq \gamma
_{N}\left( C\right) ,aN\leq k\leq bN}\frac{1}{\left\vert k^{\ast
}-k\right\vert }\left\vert \frac{1}{k^{\ast }\left( k^{\ast }-1\right) }-%
\frac{1}{k\left( k-1\right) }\right\vert \\
&\hspace{5ex}= \frac{1}{N^{1-\alpha }\Vert \mathcal{\delta }\Vert ^{2}}%
\sum_{i=1}^{N}\left\Vert X_{i}\right\Vert ^{2}\max_{\left\vert k^{\ast
}-k\right\vert \geq \gamma _{N}\left( C\right) ,aN\leq k\leq bN}\left( \frac{%
N}{k\left( N-k\right) }\right) ^{\alpha -2} \\
&\hspace{15ex}\times \max_{\left\vert k^{\ast }-k\right\vert \geq \gamma
_{N}\left( C\right) ,aN\leq k\leq bN}\frac{1}{\left\vert k^{\ast
}-k\right\vert }\frac{\left\vert k^{\ast }-k\right\vert \left( k^{\ast
}+k-1\right) }{kk^{\ast }\left( k-1\right) \left( k^{\ast }-1\right) } \\
&\hspace{5ex}= O_{P}\left( N\right) N^{-\alpha +2}\frac{1}{N^{1-\alpha
}\Vert \mathcal{\delta }\Vert ^{2}}N^{-3}=O_{P}\left( \frac{1}{N\Vert 
\mathcal{\delta }\Vert ^{2}}\right) =o_{P}\left( 1\right) ,
\end{align*}%
so that ultimately 
\begin{equation*}
\max_{\left\vert k^{\ast }-k\right\vert \geq \gamma _{N}\left( C\right)
,aN\leq k\leq bN}\frac{\left\vert \widetilde{V}_{k,7}\right\vert }{%
N^{1-\alpha }\left\vert k^{\ast }-k\right\vert \Vert \mathcal{\delta }\Vert
^{2}}=o_{P}\left( 1\right) .
\end{equation*}%
Using exactly the same logic, the same result can be shown for $\widetilde{V}%
_{k,8}$. Finally, consider%
\begin{align*}
\widetilde{V}_{k,9}=& \left[ \left( \frac{N}{k^{\ast }\left( N-k^{\ast
}\right) }\right) ^{\alpha -2}-\left( \frac{N}{k\left( N-k\right) }\right)
^{\alpha -2}\right] \frac{1}{\left( k^{\ast }\right) ^{2}\left( k^{\ast
}-1\right) }\left\Vert S_{k}\right\Vert ^{2} \\
& +\left( \frac{N}{k\left( N-k\right) }\right) ^{\alpha -2}\left[ \frac{1}{%
\left( k^{\ast }\right) ^{2}\left( k^{\ast }-1\right) }-\frac{1}{k^{2}\left(
k-1\right) }\right] \left\Vert S_{k}\right\Vert ^{2} \\
& +\left( \frac{N}{k^{\ast }\left( N-k^{\ast }\right) }\right) ^{\alpha -2}%
\frac{1}{\left( k^{\ast }\right) ^{2}\left( k^{\ast }-1\right) }\left(
\left\Vert S_{k^{\ast }}\right\Vert ^{2}-\left\Vert S_{k}\right\Vert
^{2}\right) \\
=& \widetilde{V}_{k,9,1}+\widetilde{V}_{k,9,2}+\widetilde{V}_{k,9,3}.
\end{align*}%
Similarly to the above, we can show that $\max_{1\leq k\leq N}\left\Vert
S_{k}\right\Vert =O_{P}\left( N\right) $; further, we will use fact that $%
\left\vert \left\Vert S_{k^{\ast }}\right\Vert ^{2}-\left\Vert
S_{k}\right\Vert ^{2}\right\vert \leq \left\Vert S_{k}+S_{k^{\ast
}}\right\Vert \left\Vert S_{k}-S_{k^{\ast }}\right\Vert $, and 
\begin{equation}
\max_{1\leq k\leq k^{\ast }-\gamma _{N}\left( C\right) }\frac{1}{k^{\ast }-k}%
\left\Vert S_{k}-S_{k^{\ast }}\right\Vert =O_{P}\left( \gamma ^{-1/2}\right)
=o_{P}\left( 1\right) ,  \label{m-ineq-2}
\end{equation}%
which can be shown by repeating the proof of (\ref{m-ineq}). Then we have%
\begin{align*}
& \max_{\left\vert k^{\ast }-k\right\vert \geq \gamma _{N}\left( C\right)
,aN\leq k\leq bN}\frac{\left\vert \widetilde{V}_{k,9,1}\right\vert }{%
N^{1-\alpha }\left\vert k^{\ast }-k\right\vert \Vert \mathcal{\delta }\Vert
^{2}} \\
&\hspace{4ex}\leq \frac{1}{N^{1-\alpha }\Vert \mathcal{\delta }\Vert ^{2}}%
\frac{1}{\left( k^{\ast }\right) ^{2}\left( k^{\ast }-1\right) } \\
& \hspace{6ex}\times \max_{\left\vert k^{\ast }-k\right\vert \geq \gamma
_{N}\left( C\right) ,aN\leq k\leq bN}\frac{1}{\left\vert k^{\ast
}-k\right\vert }\left[ \left( \frac{N}{k^{\ast }\left( N-k^{\ast }\right) }%
\right) ^{\alpha -2}-\left( \frac{N}{k\left( N-k\right) }\right) ^{\alpha -2}%
\right] \max_{1\leq k\leq N}\left\Vert S_{k}\right\Vert ^{2} \\
&\hspace{4ex}= O_{P}\left( N^{2}\right) \frac{1}{N^{1-\alpha }\Vert \mathcal{%
\delta }\Vert ^{2}}\frac{1}{N^{3}}O\left( N^{-1-\alpha +2}\right)
=O_{P}\left( \frac{1}{N\Vert \mathcal{\delta }\Vert ^{2}}\right)
=o_{P}\left( 1\right) ,
\end{align*}%
having recalled that $k^{\ast }=\left\lfloor N\theta \right\rfloor $, and
using the Mean Value Theorem. Also, using the fact that $k\leq bN$%
\begin{align*}
&\max_{\left\vert k^{\ast }-k\right\vert \geq \gamma _{N}\left( C\right)
,aN\leq k\leq bN}\frac{\left\vert \widetilde{V}_{k,9,2}\right\vert }{%
N^{1-\alpha }\left\vert k^{\ast }-k\right\vert \Vert \mathcal{\delta }\Vert
^{2}} \\
&\hspace{3ex} \leq \frac{1}{N^{1-\alpha }\Vert \mathcal{\delta }\Vert ^{2}}%
\max_{\left\vert k^{\ast }-k\right\vert \geq \gamma _{N}\left( C\right)
,aN\leq k\leq bN}\frac{1}{\left\vert k^{\ast }-k\right\vert }\left( \frac{N}{%
k\left( N-k\right) }\right) ^{\alpha -2}\left[ \frac{1}{\left( k^{\ast
}\right) ^{2}\left( k^{\ast }-1\right) }-\frac{1}{k^{2}\left( k-1\right) }%
\right] \\
& \hspace{6ex}\times \max_{1\leq k\leq N}\left\Vert S_{k}\right\Vert ^{2} \\
&\hspace{3ex} = O_{P}\left( N^{2}\right) \frac{1}{N^{1-\alpha }\Vert 
\mathcal{\delta }\Vert ^{2}} \\
&\hspace{6ex} \times \max_{\left\vert k^{\ast }-k\right\vert \geq \gamma
_{N}\left( C\right) ,aN\leq k\leq bN}\left( \frac{N}{k\left( N-k\right) }%
\right) ^{\alpha -2}\frac{1}{\left\vert k^{\ast }-k\right\vert }\left[ \frac{%
\left\vert k^{\ast }-k\right\vert \left( \left\vert k^{\ast }+k\right\vert
+\left( k^{\ast }\right) ^{2}+k^{2}+kk^{\ast }\right) }{\left( k^{\ast
}\right) ^{2}k^{2}\left( k^{\ast }-1\right) \left( k-1\right) }\right] \\
&\hspace{3ex} = O_{P}\left( N^{2}\right) \frac{N^{2-\alpha }}{N^{1-\alpha
}\Vert \mathcal{\delta }\Vert ^{2}}\frac{1}{N^{4}}=O_{P}\left( \frac{1}{%
N\Vert \mathcal{\delta }\Vert ^{2}}\right) =o_{P}\left( 1\right) .
\end{align*}%
Finally we have%
\begin{align*}
& \max_{\left\vert k^{\ast }-k\right\vert \geq \gamma _{N}\left( C\right)
,aN\leq k\leq bN}\frac{\left\vert \widetilde{V}_{k,9,3}\right\vert }{%
N^{1-\alpha }\left\vert k^{\ast }-k\right\vert \Vert \mathcal{\delta }\Vert
^{2}} \\
& \hspace{3ex} \leq \frac{1}{N^{1-\alpha }\Vert \mathcal{\delta }\Vert ^{2}}%
\frac{1}{\left( k^{\ast }\right) ^{2}\left( k^{\ast }-1\right) }\left( \frac{%
N}{k^{\ast }\left( N-k^{\ast }\right) }\right) ^{\alpha -2} \\
&\hspace{8ex} \times \left( \max_{\left\vert k^{\ast }-k\right\vert \geq
\gamma _{N}\left( C\right) ,aN\leq k\leq bN}\frac{1}{\left\vert k^{\ast
}-k\right\vert }\left\Vert S_{k}-S_{k^{\ast }}\right\Vert \right) \left(
\max_{1\leq k\leq N}\left( \left\Vert S_{k}\right\Vert +\left\Vert
S_{k^{\ast }}\right\Vert \right) \right) \\
&\hspace{3ex} =O_{P}\left( N\right) o_{P}\left( 1\right) \frac{1}{%
N^{1-\alpha }\Vert \mathcal{\delta }\Vert ^{2}}\frac{1}{N^{3}}N^{2-\alpha
}=O_{P}\left( \frac{1}{N\Vert \mathcal{\delta }\Vert ^{2}}\right)
=o_{P}\left( 1\right) .
\end{align*}

Putting all together, it follows that%
\begin{equation*}
\max_{\left\vert k^{\ast }-k\right\vert \geq \gamma _{N}\left( C\right)
,aN\leq k\leq bN}\frac{\left\vert \widetilde{V}_{k,9}\right\vert }{%
N^{1-\alpha }\left\vert k^{\ast }-k\right\vert \Vert \mathcal{\delta }\Vert
^{2}}=o_{P}\left( 1\right) ;
\end{equation*}%
a similar result can be shown for $\widetilde{V}_{k,10}$. On account of all
the results above and (\ref{hrice-1}), it follows that%
\begin{equation}
\max_{\left\vert k^{\ast }-k\right\vert \geq \gamma _{N}\left( C\right)
,aN\leq k\leq bN}\frac{\widetilde{V}_{k,j}}{N^{1-\alpha }\left\vert k^{\ast
}-k\right\vert \Vert \mathcal{\delta }\Vert ^{2}}=o_{P}\left( 1\right) ,
\label{cooper-1}
\end{equation}%
for all $1\leq j\leq 10$, $j\neq 5,6$; and, for all $0<c<1$%
\begin{equation*}
\max_{\left\vert k^{\ast }-k\right\vert \geq \gamma _{N}\left( C\right)
,aN\leq k\leq bN}\left( \sum_{j=1,j\neq 5,6}^{10}\widetilde{V}_{k,j}+c%
\widetilde{V}_{k,6}\right) \overset{\mathcal{P}}{\rightarrow }-\infty .
\end{equation*}%
Also, note that, as far as $\widetilde{V}_{k,5}$ is concerned, using (\ref%
{m-ineq-2})%
\begin{align}
\max_{\left\vert k^{\ast }-k\right\vert \geq \gamma _{N}\left( C\right)
,aN\leq k\leq bN}&\frac{\left\vert \widetilde{V}_{k,5}\right\vert }{%
N^{1-\alpha }\left\vert k^{\ast }-k\right\vert \Vert \mathcal{\delta }\Vert
^{2}}  \label{cooper-2} \\
& \hspace{-15ex}\leq2\frac{1}{N^{1-\alpha }\Vert \mathcal{\delta }\Vert ^{2}}%
\left( \frac{N}{k^{\ast }\left( N-k^{\ast }\right) }\right) ^{\alpha
}\max_{\left\vert k^{\ast }-k\right\vert \geq \gamma ,aN\leq k\leq
bN}\left\Vert \Pi \left( k^{\ast }\right) \right\Vert  \notag \\
& \hspace{-10ex}\times \max_{\left\vert k^{\ast }-k\right\vert \geq \gamma
_{N}\left( C\right) ,aN\leq k\leq bN}\frac{1}{\left\vert k^{\ast
}-k\right\vert }\left\Vert \sum_{i=k+1}^{k^{\ast }}\epsilon _{i}\right\Vert 
\notag \\
& \hspace{-15ex}=O_{P}\left( \gamma _{N}^{-1/2}\left( C\right) \right) \frac{%
1}{N^{1-\alpha }\Vert \mathcal{\delta }\Vert ^{2}}N^{-\alpha }N\Vert 
\mathcal{\delta }\Vert =C^{-1/2}O_{P}\left( 1\right) ,  \notag
\end{align}%
where we note that the $O_{P}\left( 1\right) $ term does not depend on $C$.
Further, seeing as $\widetilde{V}_{k,6}\leq -c_{2}N^{1-\alpha }\left\vert
k^{\ast }-k\right\vert \Vert \mathcal{\delta }\Vert ^{2}$ over the interval $%
aN\leq k\leq bN$, we have%
\begin{equation}
\max_{\left\vert k^{\ast }-k\right\vert \geq \gamma _{N}\left( C\right)
,aN\leq k\leq bN}\frac{\widetilde{V}\left( k\right) -\widetilde{V}\left(
k^{\ast }\right) }{N^{1-\alpha }\left\vert k^{\ast }-k\right\vert \Vert 
\mathcal{\delta }\Vert ^{2}}\leq -c_{2}+C^{-1/2}O_{P}\left( 1\right)
+o_{P}\left( 1\right) ,  \label{cooper-3}
\end{equation}%
whence it follows that%
\begin{equation}
\lim_{C\rightarrow \infty }\limsup_{N\rightarrow \infty }P\left(
\max_{\left\vert k^{\ast }-k\right\vert \geq \gamma _{N}\left( C\right)
,aN\leq k\leq bN}\frac{\widetilde{V}\left( k\right) -\widetilde{V}\left(
k^{\ast }\right) }{N^{1-\alpha }\left\vert k^{\ast }-k\right\vert \Vert 
\mathcal{\delta }\Vert ^{2}}\geq 0\right) =0.  \label{cooper-4}
\end{equation}
Hence we have%
\begin{align*}
P&\left( \max_{\left\vert k^{\ast }-k\right\vert \geq \gamma _{N}\left(
C\right) ,aN\leq k\leq bN}\widetilde{V}\left( k\right) -\widetilde{V}\left(
k^{\ast }\right) \geq 0\right) \\
&\leq P\left( \max_{\left\vert k^{\ast }-k\right\vert \geq \gamma _{N}\left(
C\right) ,aN\leq k\leq bN}\frac{\widetilde{V}\left( k\right) -\widetilde{V}%
\left( k^{\ast }\right) }{N^{1-\alpha }\left\vert k^{\ast }-k\right\vert
\Vert \mathcal{\delta }\Vert ^{2}}N^{1-\alpha }\left\vert k^{\ast
}-k\right\vert \Vert \mathcal{\delta }\Vert ^{2}\geq 0\right) \\
&\leq P\left( \max_{\left\vert k^{\ast }-k\right\vert \geq \gamma _{N}\left(
C\right) ,aN\leq k\leq bN}\frac{\widetilde{V}\left( k\right) -\widetilde{V}%
\left( k^{\ast }\right) }{N^{1-\alpha }\left\vert k^{\ast }-k\right\vert
\Vert \mathcal{\delta }\Vert ^{2}}N^{2-\alpha }\left( b-a\right) \Vert 
\mathcal{\delta }\Vert ^{2}\geq 0\right) \\
&=P\left( \max_{\left\vert k^{\ast }-k\right\vert \geq \gamma _{N}\left(
C\right) ,aN\leq k\leq bN}\frac{\widetilde{V}\left( k\right) -\widetilde{V}%
\left( k^{\ast }\right) }{N^{1-\alpha }\left\vert k^{\ast }-k\right\vert
\Vert \mathcal{\delta }\Vert ^{2}}\geq 0\right) ,
\end{align*}%
and therefore, by (\ref{cooper-4})%
\begin{equation*}
\lim_{C\rightarrow \infty }\limsup_{N\rightarrow \infty }P\left(
\max_{\left\vert k^{\ast }-k\right\vert \geq \gamma _{N}\left( C\right)
,aN\leq k\leq bN}\widetilde{V}\left( k\right) -\widetilde{V}\left( k^{\ast
}\right) \geq 0\right) =0.
\end{equation*}%
Now (\ref{refine}) follows from noting that%
\begin{align*}
P\left( \Vert \mathcal{\delta }\Vert ^{2}\left\vert \widehat{k}_{N}-k^{\ast
}\right\vert >C\right) &=P\left( \Vert \mathcal{\delta }\Vert ^{2}\left\vert 
\widehat{k}_{N}-k^{\ast }\right\vert >C,aN\leq k\leq bN\right) +o\left(
1\right) \\
&\leq P\left( \max_{\left\vert k^{\ast }-k\right\vert \geq \gamma _{N}\left(
C\right) ,aN\leq k\leq bN}\widetilde{V}\left( k\right) -\widetilde{V}\left(
k^{\ast }\right) \geq 0\right) +o(1).
\end{align*}
\end{proof}

\bigskip

The following two lemmas are useful for the proof of Theorem \ref{vostrikova}%
. For indices $1\leq \ell <u\leq N$, let%
\begin{equation}
\mathcal{M}_{a}\left( t\right) =\sum_{i=1}^{a}\sum_{j=1}^{R+1}\mu _{j}\left(
t\right) I\left\{ k_{j-1}\leq i<k_{j}\right\} ,  \label{m_a}
\end{equation}%
where $\mu _{j}\left( t\right) $ and $k_{j}$, $1\leq j\leq R+1$, are defined
in (\ref{multiple}), and introduce%
\begin{equation}
\Theta _{\ell ,u}^{k}=\displaystyle\left[ \frac{\left( u-\ell \right) ^{2}}{%
\left( k-\ell \right) \left( u-k\right) }\right] ^{\alpha }\displaystyle%
\frac{1}{u-\ell }\left\Vert \left( \mathcal{M}_{k}\left( t\right) -\mathcal{M%
}_{\ell }\left( t\right) \right) -\displaystyle\frac{k-\ell }{u-\ell }\left( 
\mathcal{M}_{u}\left( t\right) -\mathcal{M}_{\ell }\left( t\right) \right)
\right\Vert ^{2}.  \label{theta_lu}
\end{equation}%
If there are any changepoints between $\ell $ and $u$, we use the notation $%
i_{0}$ and $\beta $ to indicate the starting index and the number of
changepoints between $\ell $ and $u$, so that $k_{i_{0}}\leq \ell
<k_{i_{0}+1}<k_{i_{0}+2}<...<k_{i_{0}+\beta }<u\leq k_{i_{0}+\beta +1}$, and
we let $\mathcal{I=}\left\{ 1,2,...,\beta \right\} $ be the set of the
changepoints between $\ell $ and $u$.

\begin{lemma}
\label{drift}We assume that there exists at least one changepoint between $%
\ell $ and $u$. Letting $\overset{\circ }{k}=\sargmax_{l\leq k\leq u}\Theta
_{\ell ,u}^{k}$, it holds that $\overset{\circ }{k}=k_{j}$ for some $j\in
\left\{ 1,...,R\right\} $, with $\ell\leq k_{j}\leq u$.

\begin{proof}
The lemma is shown in Lemma D.5 in \citet{HT2022}.
\end{proof}
\end{lemma}

The next lemma provides a guarantee on the rate of divergence of the
maximally selected statistics used the binary segmentation algorithm
provided appropriate conditions are met on the lower and upper indices $%
\ell,u$. It is used in the proof of Theorem \ref{vostrikova} to demonstrate
each successive step of the algorithm detects one of the remaining
changepoints with probability tending to 1.

\begin{lemma}
\label{drift-2}Let $m_{N}=\zeta N\min_{i\in \left\{ 0,...,R\right\} }\left(
\theta _{i+1}-\theta _{i}\right) $ for some $\zeta \in \left( 0,1\right) $.
Assume that, for some integer $r\in \mathcal{I}$, on the sub-segments with
indices between $\ell $ and $u$, it holds that
\begin{equation}
\ell <k_{i_{0}+r}-m_{M}<k_{i_{0}+r}+m_{M}<u.  \label{8222}
\end{equation}%
Then,
\begin{equation}
\max_{\ell <k<u}\Theta _{\ell ,u}^{k}\geq c_{0}\left( N^{-1/2}\Delta
_{N}m_{N}\right) ^{2},  \label{8222_part1}
\end{equation}%
where $\Delta _{N}=\min_{1\leq j\leq R}\left\Vert \mu _{k_{j}+1}-\mu
_{k_{j}}\right\Vert $ and $c_{0}$ is a positive, finite constant. 

\begin{proof}
The lemma can be shown based on Lemma 8.2.2 in \citet{chgreg}, who prove it
for $\alpha =1$ (see also Lemma 3.3 of \cite{rice:zhang:2022} and Lemma 2.4
of \cite{venkatraman:1992}). However, for the sake of a self-contained
discussion, we report a sketch of the proof for arbitrary $0\leq \alpha <1$.

Note that if we shift each mean $\mu_i(t)$ by $\mu_i(t)+c(t)$, $i=1\ldots N$
for any $c(t)$, the corresponding value of $\Theta _{\ell ,u}^{k}$ remains
unchanged; thus by taking $c(t)=-\left( \mathcal{M}_{u}\left( t\right) -%
\mathcal{M}_{\ell }\left( t\right) \right)$, we may without loss of
generality assume that in \eqref{theta_lu} $\mathcal{M}_{u}\left( t\right) -%
\mathcal{M}_{\ell }\left( t\right) \equiv 0$. Now, by standard algebra, for
all $0\leq \alpha <1$%
\begin{equation*}
\frac{\left( u-\ell \right) }{\left( k-\ell \right) \left( u-k\right) }\geq 
\frac{4}{\left( u/N-\ell /N\right) }N^{-1};
\end{equation*}%
hence, for all $0\leq \alpha <1$ 
\begin{align}
\left[ \frac{\left( u-\ell \right) ^{2}}{\left( k-\ell \right) \left(
u-k\right) }\right] ^{\alpha }\frac{1}{u-\ell }& \geq \left( \frac{4}{\left(
u/N-\ell /N\right) }\right) ^{\alpha }N^{-\alpha }\left( u/N-\ell /N\right)
^{\alpha -1}N^{\alpha -1}  \label{algebr} \\
& \geq \frac{4^{\alpha }}{\left( u/N-\ell /N\right) }N^{-1}.  \notag
\end{align}%
%
%
Let $v=k_{i_{0}+r}$ and $v^{\prime }=k_{i_{0}+r+1}$ (with the convention
that if $v$ is the right most change point in the interval $\left( \ell
,u\right) $, then $v^{\prime }=u$). Let also $EX_{v}\left( t\right) =\mu
\left( t\right) $ and $EX_{v^{\prime }}\left( t\right) =\mu ^{\prime }\left(
t\right) $. By definition of $\Delta _{N}$, $\left\Vert \mu ^{\prime }\left(
t\right) -\mu \left( t\right) \right\Vert \geq \Delta _{N}$; hence, by the
triangular inequality, it follows that 
\begin{equation}
\max \left\{ \left\Vert \mu ^{\prime }\right\Vert ,\left\Vert \mu
\right\Vert \right\} \geq \Delta _{N}/2.  \label{rev-tr}
\end{equation}%
Note also that, by definition of $m_{N}$, there is no additional changepoint
between $\left[ v-m_{N},v\right) $ and $\left( v,v+m_{N}\right] $. Then, by
definition 
\begin{equation*}
\mathcal{M}_{v}\left( t\right) -\mathcal{M}_{v-m_{N}}\left( t\right)
=m_{N}\mu \left( t\right) ,\text{ \ \ and \ \ }\mathcal{M}_{v+m_{N}}\left(
t\right) -\mathcal{M}_{v}\left( t\right) =m_{N}\mu ^{\prime }\left( t\right)
,
\end{equation*}%
which, by (\ref{rev-tr}), implies%
\begin{equation}
\max \left\{ \left\Vert \mathcal{M}_{v}\left( t\right) -\mathcal{M}%
_{v-m_{N}}\left( t\right) \right\Vert ,\left\Vert \mathcal{M}%
_{v+m_{N}}\left( t\right) -\mathcal{M}_{v}\left( t\right) \right\Vert
\right\} \geq m_{N}\Delta _{N}/2.  \label{triangle}
\end{equation}%
In turn, this implies 
\begin{equation*}
\max\left\{\left\Vert \mathcal{M}_{v+m_{N}}\left( t\right) -\mathcal{M}%
_{\ell }\left( t\right) \right\Vert, \left\Vert \mathcal{M}_{v}\left(
t\right) -\mathcal{M}_{\ell }\left( t\right) \right\Vert, \left\Vert 
\mathcal{M}_{v-m_{N}}\left( t\right) -\mathcal{M}_{\ell }\left( t\right)
\right\Vert\right\} \geq m_{N}\Delta _{N}/4.
\end{equation*}
Thus, 
\begin{align*}
\Theta_{\ell,u}^k&=\max_{\ell <k<u}\displaystyle\left[ \frac{\left( u-\ell
\right) ^{2}}{\left( k-\ell \right) \left( u-k\right) }\right] ^{\alpha }%
\displaystyle\frac{1}{u-\ell }\left\Vert \left( \mathcal{M}_{k}\left(
t\right) -\mathcal{M}_{\ell }\left( t\right) \right) \right\Vert ^{2} \\
& \geq \max_{k\in\{v-m_N,v,v+m_N\}}\displaystyle\left[ \frac{\left( u-\ell
\right) ^{2}}{\left( k-\ell \right) \left( u-k\right) }\right] ^{\alpha }%
\displaystyle\frac{1}{u-\ell }\left\Vert \left( \mathcal{M}_{k}\left(
t\right) -\mathcal{M}_{\ell }\left( t\right) \right) \right\Vert ^{2} \\
&\geq \frac{4^{\left( \alpha -2\right) }}{\left( u/N-\ell /N\right) }\left(
N^{-1/2}\Delta _{N}m_{N}\right) ^{2}.
\end{align*}

\end{proof}
\end{lemma}

\newpage

\clearpage
\renewcommand*{\thesection}{\Alph{section}}

\setcounter{subsection}{-1} \setcounter{subsubsection}{-1} %
\setcounter{equation}{0} \setcounter{lemma}{0} \setcounter{theorem}{0} %
\renewcommand{\theassumption}{D.\arabic{assumption}} 
\renewcommand{\thetheorem}{D.\arabic{theorem}} \renewcommand{\thelemma}{D.%
\arabic{lemma}} \renewcommand{\theproposition}{D.\arabic{proposition}} %
\renewcommand{\thecorollary}{D.\arabic{corollary}} \renewcommand{%
\theequation}{D.\arabic{equation}}

\section{Main Proofs\label{proofs}}

\begin{proof}[Proof of Theorem \protect\ref{th-1}]
Let 
\begin{equation*}
Z_{N}(u)=\frac{1}{2}N(u(1-u))^{2-\alpha }Q_{N}\left( \lfloor Nu\rfloor
\right) .
\end{equation*}%
On account of Lemma \ref{l:replace_V_with_Q}, it is easy to see that%
\begin{align*}
\frac{1}{2}&N\sup_{0\leq u\leq 1}\left\vert (u(1-u))^{2-\alpha }\left(
V_{N}\left( \lfloor Nu\rfloor \right) -Q_{N}\left( \lfloor Nu\rfloor \right)
\right) \right\vert \\
& \leq CN\sup_{0\leq u\leq 1}\left\vert V_{N}\left( \lfloor Nu\rfloor
\right) -Q_{N}\left( \lfloor Nu\rfloor \right) \right\vert =o_{P}\left(
1\right) ,
\end{align*}%
and therefore we need only to establish 
\begin{equation}
Z_{N}(u)\underset{\mathcal{D}[0,1]}{\overset{w}{\rightarrow }}\frac{\Delta
(u)}{\left( u\left( 1-u\right) \right) ^{\alpha }}.
\label{e:convegence_of_ZN}
\end{equation}%
Note that, for some $0<\eta <1/2$ 
\begin{align*}
\sup_{0\leq u\leq 1}&\left\vert Z_{N}(u)-\frac{\Delta (u)}{\left( u\left(
1-u\right) \right) ^{\alpha }}\right\vert \\
\leq & \sup_{0\leq u\leq \eta }\left\vert Z_{N}(u)\right\vert +\sup_{0\leq
u\leq \eta }\left\vert \frac{\Delta (u)}{\left( u\left( 1-u\right) \right)
^{\alpha }}\right\vert +\sup_{\eta \leq u\leq 1-\eta }\left\vert Z_{N}(u)-%
\frac{\Delta (u)}{\left( u\left( 1-u\right) \right) ^{\alpha }}\right\vert \\
& \qquad \qquad +\sup_{1-\eta \leq u\leq 1}\left\vert Z_{N}(u)\right\vert
+\sup_{1-\eta \leq u\leq 1}\left\vert \frac{\Delta (u)}{\left( u\left(
1-u\right) \right) ^{\alpha }}\right\vert \\
=& I+II+III+IV+V.
\end{align*}%
It is easy to see that 
\begin{equation*}
II=\sup_{0\leq u\leq \eta }\left\vert \frac{\Delta (u)}{\left( u\left(
1-u\right) \right) ^{\alpha }}\right\vert \leq C\sup_{0\leq u\leq \eta
}u^{-\alpha }\Delta (u),
\end{equation*}%
and therefore, as $\eta \rightarrow 0$, by Lemma \ref{l:Delta_boundary_bound}
it follows that $II=o_{P}\left( 1\right) $; the same can be shown for $V$.
Similarly note that%
\begin{equation*}
I=\sup_{0\leq u\leq \eta }\left\vert Z_{N}(u)\right\vert \leq CN\sup_{0\leq
u\leq \eta }\left\vert (u(1-u))^{2-\alpha }Q_{N}\left( \lfloor Nu\rfloor
\right) \right\vert ,
\end{equation*}%
and therefore, using Lemma \ref{l:BHR_11}, it follows that%
\begin{equation*}
\lim_{\eta \rightarrow 0}\limsup_{N\rightarrow \infty }P\left\{ \sup_{0\leq
u\leq \eta }\left\vert Z_{N}(u)\right\vert >x\right\} =0,
\end{equation*}%
which yields $I=o_{P}\left( 1\right) $; the same can be shown for $IV$.
Finally, consider $III$ and let 
\begin{equation*}
Y_{N}(u,t)=N^{-1/2}\big(S_{\lfloor Nu\rfloor }(t)-\frac{\lfloor Nu\rfloor }{N%
}S_{N}(t)\big).
\end{equation*}%
For each $0<\eta <1/2$, as $N\rightarrow \infty ,$ Lemma \ref{l:BHR_11}
implies 
\begin{equation*}
\big(u(1-u)\big)^{-\alpha }\int |Y_{N}(u,t)|^{2}dt\underset{\mathcal{D}[\eta
,1-\eta ]}{\overset{w}{\rightarrow }}\big(u(1-u)\big)^{-\alpha }\int |\Gamma
(u,t)|^{2}dt.
\end{equation*}%
In turn, this gives 
\begin{align*}
Z_{N}(u)& =\frac{1}{2}N(u(1-u))^{2-\alpha }Q_{N}(\lfloor Nu\rfloor ) \\
& =\big(u(1-u)\big)^{-\alpha }\bigg(\big\|Y_{N}(u,\cdot )\big\|^{2}-\sigma
_{0}^{2}u(1-u)\bigg)\underset{\mathcal{D}[\eta ,1-\eta ]}{\overset{w}{%
\rightarrow }}\frac{\Delta (u)}{\big(u(1-u)\big)^{\alpha }}.
\end{align*}%
We now conclude the proof by showing that%
\begin{equation}
P\left\{ \sup_{0<u<1}\frac{\Delta (u)}{\left( u\left( 1-u\right) \right)
^{\alpha }}<\infty \right\} =1.  \label{integral-test}
\end{equation}%
This follows immediately if we show that%
\begin{equation*}
\lim_{x\rightarrow \infty }P\left\{ \sup_{0<u<1}\frac{\displaystyle{\int
|\Gamma (u,t)|^{2}dt}}{\left( u\left( 1-u\right) \right) ^{\alpha }}%
>x\right\} =0.
\end{equation*}%
Recalling (\ref{kl-gamma}), this is equivalent to showing%
\begin{equation*}
\lim_{x\rightarrow \infty }P\left\{ \sup_{0<u<1}\frac{\displaystyle%
\sum_{\ell =1}^{\infty }\lambda _{\ell }B_{\ell }^{2}(u)}{\left( u\left(
1-u\right) \right) ^{\alpha }}>x\right\} =0.
\end{equation*}%
It holds that%
\begin{align*}
P\left\{ \sup_{0<u<1}\frac{\displaystyle\sum_{\ell =1}^{\infty }\lambda
_{\ell }B_{\ell }^{2}(u)}{\left( u\left( 1-u\right) \right) ^{\alpha }}%
>x\right\} & \leq x^{-1}E\left( \sup_{0<u<1}\frac{\displaystyle\sum_{\ell
=1}^{\infty }\lambda _{\ell }B_{\ell }^{2}(u)}{\left( u\left( 1-u\right)
\right) ^{\alpha }}\right) \\
& \leq x^{-1}\sum_{\ell =1}^{\infty }\lambda _{\ell }E\left( \sup_{0<u<1}%
\frac{B_{\ell }^{2}(u)}{\left( u\left( 1-u\right) \right) ^{\alpha }}\right)
,
\end{align*}%
and recalling that the $B_{\ell }(u)$ are all standard Brownian bridges,
this entails that%
\begin{align*}
P\left\{ \sup_{0<u<1}\frac{\displaystyle\sum_{\ell =1}^{\infty }\lambda
_{\ell }B_{\ell }^{2}(u)}{\left( u\left( 1-u\right) \right) ^{\alpha }}%
>x\right\} & \leq x^{-1}\left( \sum_{\ell =1}^{\infty }\lambda _{\ell
}\right) E\left( \sup_{0<u<1}\frac{B_{0}^{2}(u)}{\left( u\left( 1-u\right)
\right) ^{\alpha }}\right) \\
& \leq Cx^{-1}E\left( \sup_{0<u<1}\frac{B_{0}^{2}(u)}{\left( u\left(
1-u\right) \right) ^{\alpha }}\right) ,
\end{align*}%
where $B_{0}(u)$ is a standard Brownian bridge and the last passage follows
from the fact that $\mathbf{D}\left( t,s\right) \in L^{2}\left( \mathcal{T}%
\right) $ entails the summability of the eigenvalues (see e.g. %
\citealp{horvath:kokoszka:2012}, p. 24). Note now that%
\begin{equation*}
B_{0}^{2}(u)\leq 2\left( W^{2}\left( u\right) +u^{2}W\left( 1\right) \right)
,
\end{equation*}%
where $\left\{ W\left( t\right) ,0\leq t\leq 1\right\} $ is a standard
Wiener process. By equation (2.6) in \citet{garsia1970}, it can be shown
that there exists a random variable $\xi $ such that $E\left\vert \xi
\right\vert ^{p}<\infty $ for all $p>0$ such that%
\begin{equation}
\left\vert W\left( u\right) \right\vert \leq \left\vert \xi \right\vert
\left( u\log \frac{1}{u}\right) ^{1/2}\text{ \ a.s.}  \label{garsia}
\end{equation}%
Hence we have%
\begin{align*}
E\left( \sup_{0<u\leq 1/2}\frac{B_{0}^{2}(u)}{\left( u\left( 1-u\right)
\right) ^{\alpha }}\right) &\leq 2E\left( \sup_{0<u\leq 1/2}\frac{%
W^{2}\left( u\right) }{u^{\alpha }}\right) +2E\left( \sup_{0<u\leq 1/2}\frac{%
u^{2}W^{2}\left( 1\right) }{u^{\alpha }}\right) \\
&\leq 2\left( E\left\vert \xi \right\vert ^{2}\right) \left( \sup_{0<u\leq
1/2}\frac{u\log \frac{1}{u}}{u^{\alpha }}\right) +2E\left( W^{2}\left(
1\right) \right) \left( \sup_{0<u\leq 1/2}\frac{u^{2}}{u^{\alpha }}\right) \\
&\leq C\sup_{0<u\leq 1/2}\frac{u\log \frac{1}{u}}{u^{\alpha }}%
+2\sup_{0<u\leq 1/2}\frac{u^{2}}{u^{\alpha }}\leq C,
\end{align*}%
where the last inequality follows from standard algebra; by symmetry, it
also follows that%
\begin{equation*}
E\left( \sup_{1/2\leq u<1}\frac{B_{0}^{2}(u)}{\left( u\left( 1-u\right)
\right) ^{\alpha }}\right) \leq C.
\end{equation*}%
Thus we finally have 
\begin{equation*}
P\left\{ \sup_{0<u<1}\frac{\displaystyle \sum_{\ell =1}^{\infty }\lambda
_{\ell }B_{\ell }^{2}(u)}{\left( u\left( 1-u\right) \right) ^{\alpha }}%
>x\right\} \leq Cx^{-1},
\end{equation*}%
whence (\ref{integral-test}) follows immediately. The desired result now
follows by putting everything together.
\end{proof}

\begin{proof}[Proof of Theorem \protect\ref{th-2}]
We show equation (\ref{e:T_N_asymp_norm}) in detail; the divergence
statement (\ref{power}) can be shown through a somewhat similar (and
shorter) proof. For brevity and simplicity we work under the assumption
that, as $N\rightarrow \infty $, it holds that%
\begin{equation*}
\left\Vert \Vert \mathcal{\delta }\Vert ^{-1}\mathcal{\delta }-\rho
\right\Vert \rightarrow 0,
\end{equation*}%
for some $\rho \in L^{2}(\mathcal{T})$; this condition can be dropped with
minor but tedious adjustments to the arguments that follow. Denote the time
of change by ${k^{\ast }}$. For $k<{k^{\ast }}$, using the identity $\Vert
\epsilon _{i}-\epsilon _{j}+\mathcal{\delta }\Vert ^{2}=\Vert \epsilon
_{i}-\epsilon _{j}\Vert ^{2}+\Vert \mathcal{\delta }\Vert ^{2}+2\big\langle%
\epsilon _{i}-\epsilon _{j},\mathcal{\delta }\big\rangle$, we have 
\begin{align*}
\sum_{i=1}^{k}\sum_{j=k+1}^{N}\Vert X_{i}-X_{j}\Vert ^{2}& =\sum_{i=1}^{k}%
\Bigg(\sum_{j=k+1}^{{k^{\ast }}}\Vert \epsilon _{i}-\epsilon _{j}\Vert
^{2}+\sum_{j={k^{\ast }}+1}^{N}\Vert \epsilon _{i}-\epsilon _{j}+\mathcal{%
\delta }\Vert ^{2}\Bigg) \\
& =\sum_{i=1}^{k}\sum_{j=k+1}^{N}\Vert \epsilon _{i}-\epsilon _{j}\Vert
^{2}+k(N-{k^{\ast }})\Vert \mathcal{\delta }\Vert ^{2}+2\big\langle(N-{%
k^{\ast }})w_{k}-k(w_{N}-w_{k^{\ast }}),\mathcal{\delta }\big\rangle,
\end{align*}%
where recall that $w_{k}=w_{k}(t)=\sum_{i=1}^{k}\epsilon _{i}(t)$.
Analogously, for $k\geq {k^{\ast }}$, it holds that 
\begin{equation*}
\sum_{i=1}^{k}\sum_{j=k+1}^{N}\Vert X_{i}-X_{j}\Vert
^{2}=\sum_{i=1}^{k}\sum_{j=k+1}^{N}\Vert \epsilon _{i}-\epsilon _{j}\Vert
^{2}+{k^{\ast }}(N-k)\Vert \mathcal{\delta }\Vert ^{2}+2\big\langle(N-k)w_{{%
k^{\ast }}}-{k^{\ast }}(w_{N}-w_{k}),\mathcal{\delta }\big\rangle.
\end{equation*}%
Similarly, 
\begin{align*}
\sum_{i,j=1}^{k}&\Vert X_{i}-X_{j}\Vert ^{2}= \\
&\sum_{i,j=1}^{k}\Vert \epsilon _{i}-\epsilon _{j}\Vert ^{2}+%
\begin{cases}
0,\quad1\leq k\leq {k^{\ast }}, &  \\ 
2{k^{\ast }}(k-{k^{\ast }})\Vert \mathcal{\delta }\Vert ^{2}+4\langle
(k-k^{\ast })w_{k^{\ast }}-k^{\ast }(w_{k}-w_{k^{\ast }}),\mathcal{\delta }%
\rangle , & {k^{\ast }}<k\leq N%
\end{cases}%
\end{align*}%
and 
\begin{align*}
& \sum_{i,j=k+1}^{N}\Vert X_{i}-X_{j}\Vert ^{2}=\sum_{i,j=k+1}^{N}\Vert
\epsilon _{i}-\epsilon _{j}\Vert ^{2} \\
& \hspace{2ex}+%
\begin{cases}
2({k^{\ast }}-k)(N-{k^{\ast }})\Vert \mathcal{\delta }\Vert ^{2}+4\langle
(N-k^{\ast })(w_{k^{\ast }}-w_{k})-(k^{\ast }-k)(w_{N}-w_{k^{\ast }}),%
\mathcal{\delta }\rangle , & 1\leq k\leq {k^{\ast },} \\ 
0,\quad {k^{\ast }}<k\leq N. & 
\end{cases}%
\end{align*}%
We therefore obtain 
\begin{equation}
V_{N}(k)=V_{N}^{0}(k)+g_{N}(k)+R_{N}(k),
\label{e:V_N(k)_decomposition_under_H1}
\end{equation}%
where 
\begin{equation*}
V_{N}^{0}(k)=\frac{2}{k(N-k)}\sum_{i=1}^{k}\sum_{j=k+1}^{N}\Vert \epsilon
_{i}-\epsilon _{j}\Vert ^{2}-\frac{1}{\displaystyle{{\binom{k}{2}}}}%
\sum_{1\leq i<j\leq k}\Vert \epsilon _{i}-\epsilon _{j}\Vert ^{2}-\frac{1}{%
\displaystyle{{\binom{N-k}{2}}}}\sum_{k<i<j\leq N}\Vert \epsilon
_{i}-\epsilon _{j}\Vert ^{2}
\end{equation*}%
with 
\begin{equation*}
g_{N}(k)=2\Vert \mathcal{\delta }\Vert ^{2}\times 
\begin{cases}
\displaystyle{\Big(\frac{N-{k^{\ast }}}{N-k}\Big)^{2}-\frac{({k^{\ast }}%
-k)(N-{k^{\ast }})}{(N-k)^{2}(N-k-1)}}, & 1\leq k\leq {k^{\ast }},\vspace{%
0.2cm} \\ 
\displaystyle{\Big(\frac{{k^{\ast }}}{k}\Big)^{2}-\frac{{k^{\ast }}(k-{%
k^{\ast }})}{k^{2}(k-1)}}, & {k^{\ast }}<k\leq N,%
\end{cases}%
\end{equation*}%
and the remainder $R_{N}(k)=R_{N}^{(1)}(k)+R_{N}^{(2)}(k)$, where 
\begin{equation}
R_{N}^{(1)}(k)=\frac{4}{k(N-k)}\times 
\begin{cases}
\big\langle(N-{k^{\ast }})w_{k}-k(w_{N}-w_{{k^{\ast }}}),\mathcal{\delta }%
\big\rangle, & 1\leq k<{k^{\ast }}\vspace{0.2cm} \\ 
\big\langle(N-k)w_{{k^{\ast }}}-{k^{\ast }}(w_{N}-w_{k}),\mathcal{\delta }%
\big\rangle, & {k^{\ast }}\leq k\leq N%
\end{cases}
\label{e:def_RN(u)}
\end{equation}%
and 
\begin{align}
&R_{N}^{(2)}(k)  \notag \\
&=4\times 
\begin{cases}
\displaystyle\frac{1}{(N-k)(N-k-1)}\langle (N-k^{\ast })(w_{k^{\ast
}}-w_{k})-(k^{\ast }-k)(w_{N}-w_{k^{\ast }}),\mathcal{\delta }\rangle , & 
1\leq k<{k^{\ast }}\vspace{0.2cm} \\ 
\displaystyle\frac{1}{k(k-1)}\langle (k-k^{\ast })w_{k^{\ast }}-k^{\ast
}(w_{k}-w_{k^{\ast }}),\mathcal{\delta }\rangle ,\quad {k^{\ast }}\leq k\leq
N. & 
\end{cases}%
\hspace{-6ex}  \label{e:def_RN(u)2}
\end{align}%
Writing 
\begin{equation*}
\widetilde{g}(u)=%
\begin{cases}
\displaystyle{(1-\theta )^{2}(1-u)^{-\alpha }u^{2-\alpha }} & 0\leq u\leq
\theta \vspace{0.2cm} \\ 
\displaystyle{\theta ^{2}u^{-\alpha }(1-u)^{2-\alpha }} & \theta <u\leq 1,%
\end{cases}%
\end{equation*}%
we have 
\begin{align}
\sup_{0\leq u\leq 1}\Big|& \frac{1}{2}N(u(1-u))^{2-\alpha }\Big(%
V_{N}(\lfloor Nu\rfloor )-R_{N}(\lfloor Nu\rfloor )\Big)-N\Vert \mathcal{%
\delta }\Vert ^{2}\widetilde{g}(u)\Big|  \notag \\
& =\sup_{0\leq u\leq 1}\Big|\frac{1}{2}N(u(1-u))^{2-\alpha
}V_{N}^{0}(\lfloor Nu\rfloor )+\frac{1}{2}N(u(1-u))^{2-\alpha }g_{N}(\lfloor
Nu\rfloor )-N\Vert \mathcal{\delta }\Vert ^{2}\widetilde{g}(u)\Big|  \notag
\\
& =O_{P}(1),  \label{e:V-gbar(u)_bound1}
\end{align}%
where we used that 
\begin{equation*}
\sup_{0\leq u\leq 1}|N(u(1-u))^{2-\alpha }V_{N}^{0}(\lfloor Nu\rfloor
)|=O_{P}(1),
\end{equation*}%
which is a consequence of Theorem \ref{th-1}. Thus, combining %
\eqref{e:V_N(k)_decomposition_under_H1} and \eqref{e:V-gbar(u)_bound1}, we
get 
\begin{equation}
\frac{T_{N}}{N\Vert \mathcal{\delta }\Vert ^{2}}=\widetilde{g}(u)+\frac{1}{2}%
\Vert \mathcal{\delta }\Vert ^{-2}(u(1-u))^{2-\alpha }R_{N}(\lfloor
Nu\rfloor )+\Psi _{N}(u),  \label{TNexp}
\end{equation}%
where $\sup_{0\leq u\leq 1}|\Psi _{N}(u)|=O_{P}(N^{-1}\Vert \mathcal{\delta }%
\Vert ^{-2})=o_{P}(1).$ Turning to $R_{N}$, by Lemma \ref{l:BHR_11}, for
each $N$ we may define a Gaussian process $\{G_{N}(u,t),u\geq 0,t\in 
\mathcal{T}\}$ such that 
\begin{equation*}
\sup_{0\leq u\leq 1}\Vert N^{-1/2}w_{\lfloor Nu\rfloor }-G_{N}(u,\cdot
)\Vert ^{2}=o_{P}(1),
\end{equation*}%
where $EG_{N}(u,t)=0$, and $EG_{N}\left( u,t\right) G_{N}^{\top }\left(
u^{\prime },t^{\prime }\right) =\min \{u,u^{\prime }\}\mathbf{D}(t,t^{\prime
})$ for every $N$. In particular, by \eqref{e:def_RN(u)}, this implies 
\begin{align}
\sup_{0\leq u\leq \theta }& \frac{1}{2}\Vert \mathcal{\delta }\Vert
^{-2}(u(1-u))^{2-\alpha }|R_{N}^{(1)}(\lfloor Nu\rfloor )|  \notag \\
& \leq CN^{-1/2}\Vert \mathcal{\delta }\Vert ^{-2}\sup_{0\leq u\leq \theta }%
\Big(|N^{-1/2}\langle w_{\lfloor Nu\rfloor },\mathcal{\delta }\rangle
|+|N^{-1/2}\langle w_{N}-w_{{k^{\ast }}},\mathcal{\delta }\rangle |\Big)%
=O_{P}(N^{-1/2}\Vert \mathcal{\delta }\Vert ^{-1}),  \label{e:R_N(u)_bound}
\end{align}%
and 
\begin{align}
\sup_{0\leq u\leq \theta }& \frac{1}{2}\Vert \mathcal{\delta }\Vert
^{-2}(u(1-u))^{2-\alpha }|R_{N}^{(2)}(\lfloor Nu\rfloor )|  \notag \\
& \leq C\Vert \mathcal{\delta }\Vert ^{-2}\sup_{0\leq u\leq \theta }\Big(%
\frac{1}{(N-k^{\ast })}|\langle w_{k^{\ast }}-w_{\lfloor Nu\rfloor },%
\mathcal{\delta }\rangle |+\frac{k^{\ast }}{(N-k^{\ast })^{2}}|\langle
w_{N}-w_{{k^{\ast }}},\mathcal{\delta }\rangle |\Big)  \notag \\
& =C\Vert \mathcal{\delta }\Vert ^{-2}\sup_{0\leq u\leq \theta } \Big(\frac{%
N^{1/2}}{(N-k^{\ast })}\frac{|\langle w_{k^{\ast }}-w_{\lfloor Nu\rfloor },%
\mathcal{\delta }\rangle |}{N^{1/2}}+\frac{k^{\ast }N^{1/2}}{(N-k^{\ast
})^{2}}\frac{|\langle w_{N}-w_{{k^{\ast }}},\mathcal{\delta }\rangle |}{%
N^{1/2}}\Big)  \notag \\
& \leq C\Vert \mathcal{\delta }\Vert ^{-2}\Big(N^{-1/2}O_{P}(\Vert \mathcal{%
\delta }\Vert )+N^{-1/2}O_{P}(\Vert \mathcal{\delta }\Vert )\Big)%
=O_{P}(N^{-1/2}\Vert \mathcal{\delta }\Vert ^{-1}).  \label{e:R_N(u)_bound1}
\end{align}%
Analogous arguments give 
\begin{equation}
\sup_{\theta \leq u\leq 1}\frac{1}{2}\Vert \mathcal{\delta }\Vert
^{-2}(u(1-u))^{2-\alpha }|R_{N}^{(i)}(\lfloor Nu\rfloor
)|=O_{P}(N^{-1/2}\Vert \mathcal{\delta }\Vert ^{-1}),\quad i=1,2.
\label{e:R_N(u)_bound2}
\end{equation}%
Thus, defining $\Phi _{N}(u)=\Phi _{N}^{(1)}(u)+\Phi _{N}^{(2)}(u)$, with 
\begin{equation*}
\Phi _{N}^{(i)}(u)=\frac{1}{2}\Vert \mathcal{\delta }\Vert
^{-2}(u(1-u))^{2-\alpha }R_{N}^{(i)}(\lfloor Nu\rfloor ),\quad i=1,2,
\end{equation*}%
we have $\sup_{0\leq u\leq 1}|\Phi _{N}(u)|=O_{P}(N^{-1/2}\Vert \mathcal{%
\delta }\Vert ^{-1})$ and by (\ref{TNexp}), 
\begin{equation*}
\big(N\Vert \mathcal{\delta }\Vert ^{2}\big)^{-1}T_{N}=\sup_{0\leq u\leq 1}%
\Big(\widetilde{g}(u)+\Psi _{N}(u)+\Phi _{N}(u)\Big)\overset{P}{\rightarrow }%
\widetilde{g}(\theta ).
\end{equation*}%
We now turn to establishing the limit behavior of $\Phi _{N}(u)$. To do
this, letting $\rho _{N}=\rho _{N}(t)=\mathcal{\delta }(t)/\Vert \mathcal{%
\delta }\Vert $, we first define processes $Z_{N}^{(1)}(u)$ and $%
Z_{N}^{(2)}(u)$, where%
\begin{equation*}
Z_{N}^{(1)}(u)=%
\begin{cases}
2(1-\theta )\big\langle G_{N}(u,\cdot ),\rho _{N}\big\rangle-2u\,\big\langle %
G_{N}(1,\cdot )-G_{N}(\theta ,\cdot ),\rho _{N}\big\rangle & 0\leq u\leq
\theta ,\vspace{0.2cm} \\ 
2(1-u)\big\langle G_{N}(\theta ,\cdot ),\rho _{N}\big\rangle-2\theta \,%
\big\langle G_{N}(1,\cdot )-G_{N}(u,\cdot ),\rho _{N}\big\rangle, & \theta<u
\leq 1,%
\end{cases}%
\end{equation*}%
and 
\begin{equation*}
Z_{N}^{(2)}(u)=%
\begin{cases}
2(1-\theta )\big\langle G_{N}(\theta ,\cdot )-G_{N}(u,\cdot ),\rho _{N}%
\big\rangle-2(\theta -u)\big\langle G_{N}(1,\cdot )-G_{N}(\theta ,\cdot
),\rho _{N}\rangle & 0\leq u\leq \theta, \vspace{0.2cm} \\ 
2(u-\theta )\big\langle G_{N}(\theta ,\cdot ),\rho _{N}\big\rangle-2\theta %
\big\langle G_{N}(u,\cdot )-G_{N}(\theta ,\cdot ),\rho _{N}\big\rangle,\quad
\theta< u \leq 1. & 
\end{cases}%
\end{equation*}%
Then, by \eqref{e:def_RN(u)}, we see that, for $0\leq u\leq \theta $, 
\begin{align*}
\frac{1}{2}u\big(1-u\big)& N^{1/2}\Vert \mathcal{\delta }\Vert
^{-1}R_{N}^{(1)}(\lfloor Nu\rfloor )-Z_{N}^{(1)}(u) \\
& =2(1-\theta )\big\langle N^{-1/2}w_{\lfloor Nu\rfloor }-G_{N}(u,\cdot
),\rho _{N}\big\rangle \\
& \qquad \qquad -2u\,\big \langle N^{-1/2}w_{N}-G_{N}(1,\cdot ),\rho _{N}%
\big\rangle+2u\big\langle N^{-1/2}w_{{k^{\ast }}}-G_{N}(\theta ,\cdot ),\rho
_{N}\big\rangle+O_{P}(N^{-1})
\end{align*}%
and therefore by \eqref{e:berkes_etal_th1}, 
\begin{equation}
\sup_{0\leq u\leq \theta }\Big|N^{1/2}\Vert \mathcal{\delta }\Vert \Phi
_{N}^{(1)}(u)-\big(u(1-u)\big)^{1-\alpha }Z_{N}^{(1)}(u)\Big|=o_{P}(1).
\label{e:R_N_appx1}
\end{equation}%
Analogously, 
\begin{equation}
\sup_{\theta \leq u\leq 1}\Big|N^{1/2}\Vert \mathcal{\delta }\Vert \Phi
_{N}^{(1)}(u)-\big(u(1-u)\big)^{1-\alpha }Z_{N}^{(1)}(u)\Big|=o_{P}(1).
\label{e:R_N_appx2}
\end{equation}%
Arguing similarly, 
\begin{equation}
\sup_{0\leq u\leq \theta }\Big|N^{1/2}\Vert \mathcal{\delta }\Vert \Phi
_{N}^{(2)}(u)-u^{2-\alpha }(1-u)^{-\alpha }Z_{N}^{(2)}(u)\Big|=o_{P}(1),
\label{e:R_N_appx3}
\end{equation}%
and 
\begin{equation}
\sup_{\theta \leq u\leq 1}\Big|N^{1/2}\Vert \mathcal{\delta }\Vert \Phi
_{N}^{(2)}(u)-u^{-\alpha }(1-u)^{2-\alpha }Z_{N}^{(2)}(u)\Big|=o_{P}(1).
\label{e:R_N_appx4}
\end{equation}%
Putting together \eqref{e:R_N_appx1}-\eqref{e:R_N_appx4} we have 
\begin{equation}
\sup_{0\leq u\leq 1}\Big|N^{1/2}\Vert \mathcal{\delta }\Vert \Phi _{N}(u)-%
\mathcal{Z}_{N}(u)\Big|=o_{P}(1),  \label{e:R_N_appx5}
\end{equation}%
where 
\begin{equation*}
\mathcal{Z}_{N}(u)=\big(u(1-u)\big)^{1-\alpha }Z_{N}^{(1)}(u)+%
\begin{cases}
u^{2-\alpha }(1-u)^{-\alpha }Z_{N}^{(2)}(u) & 0\leq u\leq \theta , \\ 
u^{-\alpha }(1-u)^{2-\alpha }Z_{N}^{(2)}(u) & \theta <u\leq 1.%
\end{cases}%
\end{equation*}%
We also remark that since $\langle G_{N}(u,\cdot ),\rho \rangle $ is
Gaussian, and 
\begin{align*}
E\left\langle G_{N}(u,\cdot ),\rho \right\rangle \left\langle
G_{N}(u^{\prime },\cdot ),\rho \right\rangle & =E\int \int \rho ^{\top
}\left( s\right) G_{N}(u^{\prime },s)G_{N}^{\top }(u,t)\rho (s)dtds \\
& =\min \left\{ u,u^{\prime }\right\} \int \int \rho ^{\top }\left( t\right) 
\mathbf{D}(t,s)\rho (s)dsdt,
\end{align*}%
the process $u\mapsto \langle G_{N}(u,\cdot ),\rho \rangle $ is a Brownian
motion in law and thus $Z_{N}^{(1)}$, $Z_{N}^{(2)}$ can be taken continuous.
Since $Z_{N}^{(2)}(\theta )=0,$ $\mathcal{Z}_{N}(u)$ is therefore
continuous. We now proceed to show the weak limit 
\begin{equation*}
\Vert \mathcal{\delta }\Vert N^{1/2}\Big(\frac{T_{N}}{N\Vert \mathcal{\delta 
}\Vert ^{2}}-\widetilde{g}(\theta )\Big)\overset{\mathcal{D}}{\rightarrow }%
\mathcal{Z}(\theta ).
\end{equation*}%
First observe since $\widetilde{g}(u)$ has a global maximum at $u=\theta $,
for every small $h>0$, 
\begin{align}
\sup_{u\in \lbrack 0,1]\setminus (\theta -h,\theta +h)}\Big|\frac{%
(u(1-u))^{2-\alpha }|V_{N}(\lfloor Nu\rfloor )|}{2\Vert \mathcal{\delta }%
\Vert ^{2}}\Big|& =\sup_{u\in \lbrack 0,1]\setminus (\theta -h,\theta +h)}%
\Big|\widetilde{g}(u)+\Psi _{N}(u)+\Phi _{N}(u)\Big|  \notag \\
& \overset{\mathcal{P}}{\rightarrow }\sup_{u\in \lbrack 0,1]\setminus
(\theta -h,\theta +h)}\widetilde{g}(u)<\widetilde{g}(\theta ).
\label{e:sup_near_theta0}
\end{align}%
This implies for each small $h>0$, 
\begin{equation}
\lim_{N\rightarrow \infty }P\Big\{T_{N}=T_{N,h}\Big\}=1,
\label{e:sup_near_theta}
\end{equation}%
where 
\begin{equation*}
T_{N,h}=\sup_{\theta -h\leq u\leq \theta +h}\frac{N}{2}(u(1-u))^{2-\alpha
}|V_{N}(\lfloor Nu\rfloor )|.
\end{equation*}%
Thus, for a fixed small $h>0$ define 
\begin{equation*}
A_{N}=\Big\{\omega :\sup_{0<u<1}\big|\Phi _{N}(u)+\Psi _{N}(u)\big|<\frac{1}{%
2}\inf_{\theta -h\leq u\leq \theta +h}g(u)\Big\}
\end{equation*}%
and note $P(A_{N})\rightarrow 1$ since $\sup_{0<u<1}|\Phi _{N}(u)+\Psi
_{N}(u)|=o_{P}(1)$. For each $\omega \in A_{N}$, clearly 
\begin{equation*}
\frac{T_{N,h}}{N\Vert \mathcal{\delta }\Vert ^{2}}-\widetilde{g}(\theta
)=\sup_{\theta -h\leq u\leq \theta +h}\big(\widetilde{g}(u)-g(\theta )+\Psi
_{N}(u)+\Phi _{N}(u)\big)=:\sup_{\theta -h\leq u\leq \theta +h}H(u).
\end{equation*}%
This gives, for every $\omega \in A_{N}$, 
\begin{equation*}
H(\theta )=\Psi _{N}(\theta )+\Phi _{N}(\theta )\leq \sup_{\theta -h\leq
u\leq \theta +h}H(u)=\Big(\frac{T_{N,h}}{N\Vert \mathcal{\delta }\Vert ^{2}}-%
\widetilde{g}(\theta )\Big)\leq \sup_{\theta -h\leq u\leq \theta +h}\big(%
\Psi _{N}(u)+\Phi _{N}(u)\big).
\end{equation*}%
Now, by \eqref{e:R_N_appx1} and \eqref{e:R_N_appx2}, on the set $A_{N}$, 
\begin{equation}
o_{P}(1)+\mathcal{Z}_{N}(\theta )\leq \Vert \mathcal{\delta }\Vert N^{1/2}%
\Big(\frac{T_{N,h}}{N\Vert \mathcal{\delta }\Vert ^{2}}-\widetilde{g}(\theta
)\Big)\leq o_{P}(1)+\sup_{\theta -h\leq u\leq \theta +h}\mathcal{Z}_{N}(u).
\label{e:Z_N_lowerupper}
\end{equation}%
Since $\rho _{N}\rightarrow \rho $ in $L^{2}(\mathcal{T})$, clearly $%
\sup_{0\leq u\leq 1}|Z_{N}(u)-\widetilde{Z}_{N}(u)|=o(1)$, where%
\begin{equation*}
\widetilde{\mathcal{Z}}_{N}(u)=\big(u(1-u)\big)^{1-\alpha }\widetilde{Z}%
_{N}^{(1)}(u)+%
\begin{cases}
u^{2-\alpha }(1-u)^{-\alpha }\widetilde{Z}_{N}^{(2)}(u), & 0\leq u\leq \theta
\\ 
u^{-\alpha }(1-u)^{2-\alpha }\widetilde{Z}_{N}^{(2)}(u), & \theta <u\leq 1,%
\end{cases}%
\end{equation*}%
with $\widetilde{Z}_{N}^{(1)}(u)$ and $\widetilde{Z}_{N}^{(2)}(u)$ the same
as ${Z}_{N}^{(1)}(u)$ and ${Z}_{N}^{(2)}(u)$ but with $\rho $ in place of $%
\rho _{N}$, namely:%
\begin{equation*}
\widetilde{Z}_{N}^{(1)}(u)=%
\begin{cases}
2(1-\theta )\big\langle G_{N}(u,\cdot ),\rho \big\rangle-2u\,\big\langle %
G_{N}(1,\cdot )-G_{N}(\theta ,\cdot ),\rho \big\rangle & 0\leq u\leq \theta 
\vspace{0.2cm} \\ 
2(1-u)\big\langle G_{N}(\theta ,\cdot ),\rho \big\rangle-2\theta \,%
\big\langle G_{N}(1,\cdot )-G_{N}(u,\cdot ),\rho \big\rangle, & \theta<u
\leq 1,%
\end{cases}%
\end{equation*}%
\begin{equation*}
\widetilde{Z}_{N}^{(2)}(u)=%
\begin{cases}
2(1-\theta )\big\langle G_{N}(\theta ,\cdot )-G_{N}(u,\cdot ),\rho %
\big\rangle-2(\theta -u)\big\langle G_{N}(1,\cdot )-G_{N}(\theta ,\cdot
),\rho \rangle & 0\leq u\leq \theta \vspace{0.2cm} \\ 
2(u-\theta )\big\langle G_{N}(\theta ,\cdot ),\rho _{N}\big\rangle-2\theta \,%
\big\langle G_{N}(u,\cdot )-G_{N}(\theta ,\cdot ),\rho \big\rangle,\quad
u<\theta \leq 1. & 
\end{cases}%
\end{equation*}%
Thus, on the set $A_{N}$, it follows from (\ref{e:Z_N_lowerupper}) that%
\begin{equation}
o_{P}(1)+\widetilde{\mathcal{Z}}_{N}(\theta )\leq \Vert \mathcal{\delta }%
\Vert N^{1/2}\Big(\frac{T_{N,h}}{N\Vert \mathcal{\delta }\Vert ^{2}}-%
\widetilde{g}(\theta )\Big)\leq o_{P}(1)+\sup_{\theta -h\leq u\leq \theta +h}%
\widetilde{\mathcal{Z}}_{N}(u).  \label{6.24bis1}
\end{equation}

Since $\widetilde{\mathcal{Z}}_{N}$ has the same distribution for each $N$,
letting $\widetilde{\mathcal{Z}}\overset{\mathcal{D}}{=}\widetilde{\mathcal{Z%
}}_{N}$, for every $x\in \mathbb{R}$, \eqref{6.24bis1} implies 
\begin{align}
\P \bigg\{\sup_{\theta -h\leq u\leq \theta +h}\widetilde{\mathcal{Z}}(u)\leq
x\bigg\}& \leq \liminf_{N\rightarrow \infty }\P \bigg\{\Vert \mathcal{\delta 
}\Vert N^{1/2}\Big(\frac{T_{N,h}}{N\Vert \mathcal{\delta }\Vert ^{2}}-%
\widetilde{g}(\theta )\Big)\leq x\bigg\}  \notag \\
& \leq \limsup_{N\rightarrow \infty }\P \bigg\{\Vert \mathcal{\delta }\Vert
N^{1/2}\Big(\frac{T_{N,h}}{N\Vert \mathcal{\delta }\Vert ^{2}}-\widetilde{g}%
(\theta )\Big)\leq x\bigg\}  \notag \\
& \leq \limsup_{N\rightarrow \infty }P\left\{ o_{P}(1)+\widetilde{\mathcal{Z}%
}_{N}(\theta )\leq x\right\}  \notag \\
& \leq P\bigg\{\widetilde{\mathcal{Z}}(\theta )\leq x+\epsilon \bigg\},
\label{e:Z_squeeze}
\end{align}%
for any $\epsilon >0$. Continuity of $\widetilde{\mathcal{Z}}$ implies 
\begin{equation*}
\lim_{h\rightarrow 0}\sup_{\theta -h\leq u\leq \theta +h}\widetilde{\mathcal{%
Z}}(u)\rightarrow \widetilde{\mathcal{Z}}(\theta ),\quad \text{a.s.}
\end{equation*}%
Therefore, using (\ref{e:sup_near_theta}) and \eqref{6.24bis1}, we obtain%
\begin{align*}
P\left( \widetilde{\mathcal{Z}}(\theta )\leq x\right) &\leq
\liminf_{N\rightarrow \infty }P\left\{ \Vert \mathcal{\delta }\Vert
N^{1/2}\left( \frac{T_{N}}{N\Vert \mathcal{\delta }\Vert ^{2}}-\widetilde{g}%
(\theta )\right) \leq x\right\} \\
&\leq \limsup_{N\rightarrow \infty }P\left\{ \Vert \mathcal{\delta }\Vert
N^{1/2}\left( \frac{T_{N}}{N\Vert \mathcal{\delta }\Vert ^{2}}-\widetilde{g}%
(\theta )\right) \leq x\right\} \\
&\leq P\left( \widetilde{\mathcal{Z}}(\theta )\leq x+\epsilon \right) ,
\end{align*}%
for any $\epsilon >0$, which gives%
\begin{equation*}
\Vert \mathcal{\delta }\Vert N^{1/2}\left( \frac{T_{N}}{N\Vert \mathcal{%
\delta }\Vert ^{2}}-\widetilde{g}(\theta )\right) \overset{\mathcal{D}}{%
\rightarrow }\widetilde{\mathcal{Z}}(\theta ).
\end{equation*}
Finally, since $E\widetilde{\mathcal{Z}}(\theta )=0$ and $EG_{N}(u,t){G}%
_{N}^{\top }\left( u^{\prime },t\right) =\min \{u,u^{\prime }\}\mathbf{D}%
(t,t^{\prime })$, we have 
\begin{align*}
E\widetilde{\mathcal{Z}}(\theta )^{2}=E\widetilde{Z}_{N}^{(1)}(\theta )^{2}&
=4E\big|\big\langle G_{N}(\theta ,\cdot )-\theta G_{N}(1,\cdot ),\rho %
\big\rangle\big|^{2} \\
& =4E\bigg(\int \big[G_{N}(\theta ,t)-\theta G_{N}(1,t)\big]^{\top }\rho
(t)dt\bigg)^{2} \\
& =4E\bigg(\iint \rho ^{\top }\left( s\right) \big[G_{N}(\theta ,s)-\theta
G_{N}(1,s)\big]\big[G_{N}(\theta ,t)-\theta G_{N}(1,t)\big]^{\top }\rho
(t)dtds\bigg) \\
& =4\theta (1-\theta )\iint \rho ^{\top }\left( t\right) \mathbf{D}(t,s)\rho
(s)dsdt.
\end{align*}
The result \eqref{e:T_N_asymp_norm} for $T_{N}$ then follows since $%
\widetilde{\mathcal{Z}}$ is Gaussian.
\end{proof}

\begin{proof}[Proof of Theorem \protect\ref{th-3}]
The proof follows on from the proof of the previous theorem. For each small $%
h>0$, write 
\begin{equation*}
c_{h}=\sup_{u\in \lbrack 0,1]\setminus (\theta -h,\theta +h)}\widetilde{g}%
(u).
\end{equation*}%
and note $c_{h}<\widetilde{g}(\theta )$. Expression \eqref{e:sup_near_theta0}
implies, for each small $h>0$, 
\begin{equation*}
\sup_{u\in \lbrack 0,1]\setminus (\theta -h,\theta +h)}\bigg|\frac{%
(u(1-u))^{2-\alpha }|V_{N}(\lfloor Nu\rfloor )|}{2\Vert \mathcal{\delta }%
\Vert ^{2}}\bigg|\overset{P}{\rightarrow }c_{h}.
\end{equation*}%
Thus, 
\begin{align*}
& P(|\widehat{\theta }_{N}-\theta |\geq h) \\
& \leq P\left( \frac{T_{N}}{N\Vert \mathcal{\delta }\Vert ^{2}}=\sup_{u\in
\lbrack 0,1]\setminus (\theta -h,\theta +h)}\bigg|\frac{(u(1-u))^{2-\alpha
}|V_{N}(\lfloor Nu\rfloor )|}{2\Vert \mathcal{\delta }\Vert ^{2}}\bigg|%
\right) \  \\
& \leq P\left( \bigg|\frac{T_{N}}{N\Vert \mathcal{\delta }\Vert ^{2}}-c_{h}%
\bigg|\leq \eta ,~\bigg|\sup_{u\in \lbrack 0,1]\setminus (\theta -h,\theta
+h)}\bigg|\frac{(u(1-u))^{2-\alpha }|V_{N}(\lfloor Nu\rfloor )|}{2\Vert 
\mathcal{\delta }\Vert ^{2}}\bigg|-c_{h}\bigg|\leq \eta \right) +o(1) \\
& \leq P\left( \bigg|\frac{T_{N}}{N\Vert \mathcal{\delta }\Vert ^{2}}-c_{h}%
\bigg|\leq \eta \right) +o(1)\leq P\left( \bigg|\frac{T_{N}}{N\Vert \mathcal{%
\delta }\Vert ^{2}}-\widetilde{g}(\theta )\bigg|\geq \widetilde{g}(\theta
)-c_{h}-\eta \right) +o(1).
\end{align*}%
We know from the above that $T_{N,h}-T_{N,h}=o_{P}\left( 1\right) $ and $%
T_{N,h}/\left( N\Vert \mathcal{\delta }\Vert ^{2}\right) -\widetilde{g}%
(\theta )=o_{P}\left( 1\right) $; hence, $T_{N}/(N\Vert \mathcal{\delta }%
\Vert ^{2})\overset{P}{\rightarrow }\widetilde{g}(\theta )$. Thus, the last
line above tends to zero by taking any $0<\eta <\widetilde{g}(\theta )-c_{h}$%
. This concludes the proof of the consistency of $\widehat{\theta }_{N}$.

We now turn to studying the limiting distribution. Assuming again that $%
\left\lfloor aN\right\rfloor \leq k\leq \left\lfloor bN\right\rfloor $ for
some $0<a<b<1$, and using a very similar logic to the proof of Lemma \ref%
{refinement}, it can be shown that $N^{-\left( 1-\alpha \right)
}\max_{\left\vert k^{\ast }-k\right\vert <C\sigma ^{2}/\Vert \mathcal{\delta 
}\Vert ^{2}}\left\vert \widetilde{V}_{k,j}\right\vert $ $=$ $o_{P}\left(
1\right) $, for all $1\leq j\leq 10$, with $j\neq 5,6$. Hence, the limiting
distribution of $\max_{\left\vert k^{\ast }-k\right\vert <C\sigma ^{2}/\Vert 
\mathcal{\delta }\Vert ^{2}}\widetilde{V}\left( k\right) -\widetilde{V}%
\left( k^{\ast }\right) $ is determined by $\widetilde{V}_{k,5}+\widetilde{V}%
_{k,6}$. Let $m=\sigma ^{2}\Vert \mathcal{\delta }\Vert ^{-2}$\ for short.
Observing that Lemma \ref{l:BHR_11} is shown, in %
\citet{berkes:horvath:rice:2013}, using a blocking argument, for any fixed $%
h\left( \cdot \right) \in L^{2}\left( \mathcal{T}\right)$ it follows that,
for each $N$, there are two independent, identically distributed Gaussian
processes $\{G_{1,N}(u,t),u\geq 0,t\in \mathcal{T}\}$ and $%
\{G_{2,N}(-u,t),u\geq 0,t\in \mathcal{T}\}$, whose distributions do not
depend on $N$, with $EG_{1,N}(u,t)=0$, and $EG_{1,N}\left( u,t\right)
G_{1,N}^{\top }\left( u^{\prime },t^{\prime }\right) =\min \left\{
u,u^{\prime }\right\} \mathbf{D}\left( t,t^{\prime }\right) $, and such that%
\begin{align*}
&\sup_{0\leq s\leq C}\left\vert \frac{1}{m^{1/2}}\sum_{i=k^{\ast
}+1}^{k^{\ast }+\left\lfloor ms\right\rfloor }\left\langle \epsilon
_{i}\left( \cdot \right) ,h\left( \cdot \right) \right\rangle -\left\langle
G_{1,N}(s,\cdot ),h\left( \cdot \right) \right\rangle \right\vert \\
&\qquad\quad +\sup_{-C\leq s\leq 0}\left\vert \frac{1}{m^{1/2}}%
\sum_{i=k^{\ast }+\left\lfloor ms\right\rfloor }^{k^{\ast }}\left\langle
\epsilon _{i}\left( \cdot \right) ,h\left( \cdot \right) \right\rangle
-\left\langle G_{2,N}(s,\cdot ),h\left( \cdot \right) \right\rangle
\right\vert =o_{P}\left( 1\right).
\end{align*}%
Given that 
\begin{align*}
&\hspace{-2ex}N^{-\left( 1-\alpha \right) }\widetilde{V}_{k^{\ast
}+\left\lfloor ms\right\rfloor ,5} \\
&= 2\left( \frac{N^{2}}{k^{\ast }\left( N-k^{\ast }\right) }\right) ^{\alpha
}\frac{k^{\ast }\left( N-k^{\ast }\right) }{N^{2}}\left( \sigma ^{2}\Vert 
\mathcal{\delta }\Vert ^{-2}\right) ^{1/2}\Vert \mathcal{\delta }\Vert \frac{%
1}{\left( \sigma ^{2}\Vert \mathcal{\delta }\Vert ^{-2}\right) ^{1/2}}%
\sum_{i=k^{\ast }+1}^{k^{\ast }+\left\lfloor ms\right\rfloor }\left\langle
\epsilon _{i}\left( \cdot \right) ,\rho \left( \cdot \right) \right\rangle
+o_P(1),
\end{align*}%
(where the $o_P(1)$ term holds uniformly in $u$ as a consequence of $%
\|\rho_N-\rho\|\to 0$) the above entails that, letting $%
G_{N}(u,t)=G_{1,N}(u,t)I\left( u\geq 0\right) +G_{2,N}(u,t)I\left( u\leq
0\right) $%
\begin{equation*}
\sup_{-C\leq s\leq C}\left\vert N^{-\left( 1-\alpha \right) }\widetilde{V}%
_{k^{\ast }+\left\lfloor ms\right\rfloor ,5}-2\sigma \left( \theta \left(
1-\theta \right) \right) ^{1-\alpha }\left\langle G_{N}(s, \cdot ),\rho
\left( \cdot \right) \right\rangle \right\vert =o_{P}\left( 1\right) .
\end{equation*}%
Note that $\left\langle G_{1,N}(s,\cdot ),\rho \left( \cdot \right)
\right\rangle =\int \rho ^{\top }\left( t\right) G_{1,N}(s,t)dt\overset{%
\mathcal{D}}{=}\sigma W_{1}\left( s\right) $, where $W_{1}\left( s\right) $
is a standard Wiener process. Indeed, both $\left\langle G_{1,N}(s,t),\rho
\left( t\right) \right\rangle $ and $\sigma W_{1}\left( s\right) $ are zero
mean Gaussian processes, with the same covariance kernel, as can be verified
by direct computation; similarly, $\left\langle G_{2,N}(s,\cdot ),\rho
\left( \cdot \right) \right\rangle \overset{\mathcal{D}}{=}\sigma
W_{2}\left( s\right) $, where $W_{2}\left( s\right) $ is a standard Wiener
process independent of $W_{1}\left( s\right) $. Hence, it follows that $%
\left\langle G_{N}(s,\cdot),\rho \left(\cdot\right) \right\rangle \overset{%
\mathcal{D}}{=}\sigma W\left( s\right) $, with $W\left( s\right) $ a
two-sided standard Wiener process. Finally, by marginally adapting equation
(2.2.13) in \citet{chgreg}, it can be shown - by elementary, if tedious,
arguments - that%
\begin{equation*}
\sup_{-C\leq s\leq C}\left\vert N^{-\left( 1-\alpha \right) }\widetilde{V}%
_{k^{\ast }+\left\lfloor ms\right\rfloor ,6}+2\sigma ^{2}\left( \theta
\left( 1-\theta \right) \right) ^{1-\alpha }\left\vert s\right\vert
m_{\alpha }\left( s\right) \right\vert =o\left( 1\right) .
\end{equation*}%
Hence, 
\begin{equation}  \label{e:weakconverg_sargmax_compacts}
N^{-\left( 1-\alpha \right) }\sum_{j=1}^{10}\widetilde{V}_{k^{\ast
}+\left\lfloor ms\right\rfloor ,j}\overset{w}{\underset{\mathcal{D}\left[
-C,C\right] }{\rightarrow }}2\left( \theta \left( 1-\theta \right) \right)
^{1-\alpha }\left( \sigma \left\langle G(s,\cdot ),\rho \right\rangle
-\left\vert s\right\vert \sigma ^{2}m_{\alpha }\left( s\right) \right) .
\end{equation}%
From the definition%
\begin{align*}
\widehat{k}_{N,C} &=\sargmax_{k\in \left\{ 1,...,N\right\} ,\left\vert
k^{\ast }-k\right\vert \leq C\sigma ^{2}/\left\Vert \mathcal{\delta }%
\right\Vert ^{2}}N\left( \frac{N^{2}}{k\left( N-k\right) }\right) ^{\alpha
-2}V_{N}\left( k\right) \\
&=\sargmax_{k\in \left\{ 1,...,N\right\} ,\left\vert k^{\ast }-k\right\vert
\leq C\sigma ^{2}/\left\Vert \mathcal{\delta }\right\Vert ^{2}}N^{\alpha
-1}\sum_{j=1}^{10}\widetilde{V}_{k,j},
\end{align*}%
where recall that \textquotedblleft $\sargmax$\textquotedblright\ denotes
the smallest integer that maximizes the relevant expression. Since the \text{%
sargmax} is continuous on $D[-C,C]$ at every point that is continuous and
possess a unique maximum (\cite{seijo:sen:2011}, Lemma 2.9), from %
\eqref{e:weakconverg_sargmax_compacts} we have 
\begin{equation*}
\left\Vert \mathcal{\delta }\right\Vert ^{2}\left( \widehat{k}_{N,C}-k^{\ast
}\right) /\sigma ^{2}\overset{\mathcal{D}}{\rightarrow }\argmax_{\left\vert
s\right\vert \leq C}\left( \sigma W(s) -\left\vert s\right\vert \sigma
^{2}m_{\alpha }\left( s\right) \right).
\end{equation*}

(Indeed, since $m_\alpha(t)$ is constant on either side of zero, $\sigma
W(t) -\left\vert t\right\vert \sigma ^{2}m_{\alpha }\left( t\right)$ is
clearly continuous, and the a.s. uniqueness of its maximizer follows from
the a.s. uniqueness and absolute continuity of its maximizer on either side
of zero and the independence of $\{W(t),t\geq 0\}$ and $\{W(t), t\leq 0\}$.)
Since by continuity, as $C\rightarrow \infty $%
\begin{equation*}
\argmax_{\left\vert s\right\vert \leq C}\left(\sigma W(s) -\left\vert
s\right\vert \sigma ^{2}m_{\alpha }(s) \right) \overset{a.s.}{\rightarrow }%
\argmax_{-\infty <s<\infty }\left(\sigma W(s) -\left\vert s\right\vert
\sigma ^{2}m_{\alpha }(s) \right) ,
\end{equation*}%
%
%
the desired result follows.
\end{proof}

\begin{proof}[Proof of Theorem \protect\ref{vostrikova}]
The proof follows a very similar logic to the proof of Theorem 8.2.2 in %
\citet{chgreg} (see also Theorem 2.2 in \cite{rice:zhang:2022}) and
therefore we report only the main two arguments where our proof differs: 
\textit{(i)} we begin by deriving a \textquotedblleft stopping
condition\textquotedblright\ for the algorithm (see (\ref{lemma823})); 
\textit{(ii)} we show that the first estimated breakdate is consistent (see (%
\ref{greg-1})). Proceeding as in the proof of Theorem 8.2.2 in \citet{chgreg}%
, the proof is completed through an inductive argument. 

Recall the definition of $\mathcal{M}_{a}\left( t\right) $ in (\ref{m_a});
we further define $S_{a}\left( t\right) =\sum_{i=1}^{a}X_{i}\left( t\right) $%
, and $\mathcal{W}_{a}\left( t\right) =\sum_{i=1}^{a}\epsilon _{i}\left(
t\right) $, omitting the index $t$ when possible. With this notation, we can
write%
\begin{align*}
\sum_{i=\ell }^{k}X_{i}\left( t\right) & =\mathcal{M}_{k}\left( t\right) -%
\mathcal{M}_{\ell }\left( t\right) +\mathcal{W}_{k}\left( t\right) -\mathcal{%
W}_{\ell }\left( t\right) , \\
\sum_{i=\ell }^{u}X_{i}\left( t\right) & =\mathcal{M}_{u}\left( t\right) -%
\mathcal{M}_{\ell }\left( t\right) +\mathcal{W}_{u}\left( t\right) -\mathcal{%
W}_{\ell }\left( t\right) ,
\end{align*}%
and therefore we can write%
\begin{align*}
&\hspace{-4ex} \left\Vert \sum_{i=\ell }^{k}X_{i}\left( t\right) -%
\displaystyle\frac{k-\ell }{u-\ell }\sum_{i=\ell }^{u}X_{i}\left( t\right)
\right\Vert ^{2} \\
=& \left\Vert \left( \mathcal{M}_{k}\left( t\right) -\mathcal{M}_{\ell
}\left( t\right) \right) -\displaystyle\frac{k-\ell }{u-\ell }\left( 
\mathcal{M}_{u}\left( t\right) -\mathcal{M}_{\ell }\left( t\right) \right)
\right\Vert ^{2} \\
& +\left\Vert \left( \mathcal{W}_{k}\left( t\right) -\mathcal{W}_{\ell
}\left( t\right) \right) -\displaystyle\frac{k-\ell }{u-\ell }\left( 
\mathcal{W}_{u}\left( t\right) -\mathcal{W}_{\ell }\left( t\right) \right)
\right\Vert ^{2} \\
& +2\left\langle \left( \mathcal{M}_{k}\left( t\right) -\mathcal{M}_{\ell
}\left( t\right) \right) -\displaystyle\frac{k-\ell }{u-\ell }\left( 
\mathcal{M}_{u}\left( t\right) -\mathcal{M}_{\ell }\left( t\right) \right)
\right. , \\
&\qquad\qquad \left. \left( \mathcal{W}_{k}\left( t\right) -\mathcal{W}%
_{\ell }\left( t\right) \right) -\displaystyle\frac{k-\ell }{u-\ell }\left( 
\mathcal{W}_{u}\left( t\right) -\mathcal{W}_{\ell }\left( t\right) \right)
\right\rangle .
\end{align*}%
Hence, under the alternative, after some algebra it holds that%
\begin{align*}
\frac{1}{2}&\left( u-\ell \right) \left( \frac{\left( k-\ell \right) \left(
u-k\right) }{\left( u-\ell \right) ^{2}}\right) ^{2-\alpha }V_{N}^{\left(
\ell ,u\right) }\left( k\right) \\
&=\left[ \frac{\left( u-\ell \right) ^{2}}{\left( k-\ell \right) \left(
u-k\right) }\right] ^{\alpha }\frac{1}{u-\ell }\left\Vert \left( \mathcal{M}%
_{k}\left( t\right) -\mathcal{M}_{\ell }\left( t\right) \right) -\frac{%
k-\ell }{u-\ell }\left( \mathcal{M}_{u}\left( t\right) -\mathcal{M}_{\ell
}\left( t\right) \right) \right\Vert ^{2} \\
&\quad +\left[ \frac{\left( u-\ell \right) ^{2}}{\left( k-\ell \right)
\left( u-k\right) }\right] ^{\alpha }\frac{1}{u-\ell }\left\Vert \left( 
\mathcal{W}_{k}\left( t\right) -\mathcal{W}_{\ell }\left( t\right) \right) -%
\frac{k-\ell }{u-\ell }\left( \mathcal{W}_{u}\left( t\right) -\mathcal{W}%
_{\ell }\left( t\right) \right) \right\Vert ^{2} \\
& \quad+2\left[ \frac{\left( u-\ell \right) ^{2}}{\left( k-\ell \right)
\left( u-k\right) }\right] ^{\alpha }\frac{1}{u-\ell }\left\langle \left( 
\mathcal{M}_{k}\left( t\right) -\mathcal{M}_{\ell }\left( t\right) \right) -%
\frac{k-\ell }{u-\ell }\left( \mathcal{M}_{u}\left( t\right) -\mathcal{M}%
_{\ell }\left( t\right) \right) ,\right. \\
&\qquad \qquad\qquad \qquad \qquad\qquad \qquad \qquad\qquad \left. \left( 
\mathcal{W}_{k}\left( t\right) -\mathcal{W}_{\ell }\left( t\right) \right) -%
\frac{k-\ell }{u-\ell }\left( \mathcal{W}_{u}\left( t\right) -\mathcal{W}%
_{\ell }\left( t\right) \right) \right\rangle \\
&\quad -\left[ \frac{\left( u-\ell \right) ^{2}}{\left( k-\ell \right)
\left( u-k\right) }\right] ^{\alpha -2}\frac{u-\ell }{\left( k-\ell \right)
\left( k-\ell -1\right) }\sum_{i=\ell }^{k}\left\Vert X_{i}\right\Vert ^{2}
\\
&\quad -\left[ \frac{\left( u-\ell \right) ^{2}}{\left( k-\ell \right)
\left( u-k\right) }\right] ^{\alpha -2}\frac{u-\ell }{\left( u-k\right)
\left( u-k-1\right) }\sum_{i=k+1}^{u}\left\Vert X_{i}\right\Vert ^{2} \\
&\quad +\left[ \frac{\left( u-\ell \right) ^{2}}{\left( k-\ell \right)
\left( u-k\right) }\right] ^{\alpha -2}\frac{u-\ell }{\left( k-\ell \right)
^{2}\left( k-\ell -1\right) }\left\Vert S_{k}-S_{\ell }\right\Vert ^{2} \\
&\quad +\left[ \frac{\left( u-\ell \right) ^{2}}{\left( k-\ell \right)
\left( u-k\right) }\right] ^{\alpha -2}\frac{u-\ell }{\left( u-k\right)
^{2}\left( u-k-1\right) }\left\Vert S_{u}-S_{k}\right\Vert ^{2} \\
&= \Theta _{\ell ,u}^{k}+\sum_{h=1}^{6}A_{\ell ,u}^{k,\left( h\right)} ,
\end{align*}%
where recall that $\Theta _{\ell ,u}^{k}$ is defined in (\ref{theta_lu}).

Recall that, if there are any changepoints between $\ell $ and $u$, we use
the notation $i_{0}$ and $\beta $ to indicate the starting index and the
number of changepoints between $\ell $ and $u$, so that%
\begin{equation*}
k_{i_{0}}\leq \ell <k_{i_{0}+1}<k_{i_{0}+2}<...<k_{i_{0}+\beta }<u\leq
k_{i_{0}+\beta +1},
\end{equation*}%
and we let $\mathcal{I=}\left\{ 1,2,...,\beta \right\} $ be the set of the
changepoints between $\ell $ and $u$. We begin by showing the following
intermediate result. Let $a_N$ be any given positive sequence and let $%
\mathcal{A}_{N}$ be any event on which 
\begin{equation}
\max_{\ell <k<u}\left\vert \sum_{h=1}^{6}A_{\ell ,u}^{k,\left( h\right)}
\right\vert \leq a_{N}.  \label{a-n}
\end{equation}%
Then if: \textit{(i)} $\beta =0$ and $k_{i_{0}}<\ell <u<k_{i_{0}+1}$; or 
\textit{(ii) }$\beta =1$ and $\min \left\{ k_{i_{0}+1}-\ell ,\right. $ $%
\left. u-k_{i_{0}+1}\right\} \leq f_{N}$; or \textit{(iii)} $\beta =2$ and $%
\max \left\{ k_{i_{0}+1}-\ell ,\right. $ $\left. u-k_{i_{0}+2}\right\} \leq
f_{N}$, for some sequence $f_{N}$; it holds that%
\begin{equation}
\max_{\ell <k<u}\left\vert \frac{1}{2}\left( u-\ell \right) \left( \frac{%
\left( k-\ell \right) \left( u-k\right) }{\left( u-\ell \right) ^{2}}\right)
^{2-\alpha }V_{N}^{\left( \ell ,u\right) }\left( k\right) \right\vert \leq
c_{0}\max \left\{ a_{N},f_{N}\right\} .  \label{lemma823}
\end{equation}%
This result can be shown similarly to Lemma 8.2.3 in \citet{chgreg}, who
prove it for $\alpha =1$. Indeed, under condition \textit{(i)}, there is no
break in the interval $(\ell,u)$ and, by (\ref{a-n}), it follows readily
that 
\begin{equation*}
\max_{\ell <k<u}\left\vert \Theta _{\ell ,u}^{k}+\sum_{h=1}^{6}A_{\ell
,u}^{k,\left( h\right)} \right\vert =\max_{\ell <k<u}\left\vert
\sum_{h=1}^{6}A_{\ell ,u}^{k,\left( h\right)}\right\vert \leq a_{N}.
\end{equation*}%
Under condition \textit{(ii)}, let the mean functions before and after $%
k_{i_{0}}$ be defined as $\mu \left( t\right) $ and $\mu ^{\prime }\left(
t\right) $; we know from Lemma \ref{drift} that 
\begin{align*}
\max_{\ell <k<u}\Theta _{\ell ,u}^{k}&=\Theta _{\ell ,u}^{k_{i_{0}+1}} \\
&= \left[ \frac{\left( u-\ell \right) ^{2}}{\left( k_{i_{0}+1}-\ell \right)
\left( u-k_{i_{0}+1}\right) }\right] ^{\alpha }\frac{1}{u-\ell } \\
&\qquad \times \left\Vert \left( k_{i_{0}+1}-\ell \right) \mu -\frac{%
k_{i_{0}+1}-\ell }{u-\ell }\left( \left( u-k_{i_{0}+1}\right) \mu ^{\prime
}+\left( k_{i_{0}+1}-\ell \right) \mu \right) \right\Vert ^{2} \\
&= \frac{\left( k_{i_{0}+1}-\ell \right) ^{2-\alpha }\left(
u-k_{i_{0}+1}\right) ^{2-\alpha }}{\left( u-\ell \right) ^{3-2\alpha }}%
\left\Vert \mu -\mu ^{\prime }\right\Vert ^{2} \\
&\leq c_{0}\min \left\{ k_{i_{0}+1}-\ell ,u-k_{i_{0}+1}\right\} \leq f_{N},
\end{align*}%
and since%
\begin{equation*}
\max_{\ell <k<u}\left\vert \Theta _{\ell ,u}^{k}+\sum_{h=1}^{6}A_{\ell
,u}^{k,\left( h\right)} \right\vert \leq \max_{\ell <k<u}\Theta _{\ell
,u}^{k}+\max_{\ell <k<u}\left\vert \sum_{h=1}^{6}A_{\ell ,u}^{k,\left(
h\right)}\right\vert \leq f_{N}+a_{N},
\end{equation*}%
the desired result follows. Finally, under condition \textit{(iii)}, by
Lemma \ref{drift} it follows that $\max_{\ell <k<u}\Theta _{\ell
,u}^{k}=\max \left\{ \Theta _{\ell ,u}^{k_{i_{0}+1}},\Theta _{\ell
,u}^{k_{i_{0}+2}}\right\} $, and after some elementary if tedious algebra it
can be shown that%
\begin{align*}
&\max \left\{ \Theta _{\ell ,u}^{k_{i_{0}+1}},\Theta _{\ell
,u}^{k_{i_{0}+2}}\right\} \\
&\quad\leq \max \left\{ \min \left\{ k_{i_{0}+1}-\ell ,u-k_{i_{0}+1}\right\}
,\min \left\{ k_{i_{0}+2}-\ell ,u-k_{i_{0}+2}\right\} \right\} \\
&\quad\leq \max \left\{ k_{i_{0}+1}-\ell ,u-k_{i_{0}+2}\right\} \leq f_{N},
\end{align*}%
whence \eqref{lemma823}. (N.b.: \eqref{lemma823} is used toward the very end
of this proof.)

We are now ready to start the proof. We begin by noting that the
segmentation procedure starts with indices $\ell =0$ and $u=N$. Also, by
Lemma \ref{drift-2}, it follows that $\max_{1\leq k\leq N}\Theta
_{1,N}^{k}\geq c_{0}N$, for some $c_{0}>0$ (see e.g. %
\citealp{rice:zhang:2022}). Moreover, consider $\max_{1\leq k\leq
N}\left\vert A_{1,N}^{k,\left( h\right) }\right\vert $, for $1\leq h\leq 6$,
and note that, by Lemma \ref{l:BHR_33}, by arguing similarly as in the proof
of Lemma \ref{l:weightappx}, it follows that%
\begin{align}
\max_{1\leq k\leq N}k^{-1/2}\left\Vert \sum_{i=1}^{k}\epsilon
_{i}\right\Vert &=O_{P}\left( \left( \ln N\right) ^{1/\nu }\right) ,
\label{mineq1} \\
\max_{1\leq k\leq N}\left( N-k\right) ^{-1/2}\left\Vert
\sum_{i=k+1}^{N}\epsilon _{i}\right\Vert &=O_{P}\left( \left( \ln N\right)
^{1/\nu }\right) .  \label{mineq2}
\end{align}%
Then we have%
\begin{align*}
\max_{1\leq k\leq N}\left\vert A_{1,N}^{k,\left( 1\right) }\right\vert &
=\max_{1\leq k\leq N}\left[ \frac{N^{2}}{k\left( N-k\right) }\right]
^{\alpha }\frac{1}{N}\left\Vert \frac{N-k}{N}\sum_{i=1}^{k}\epsilon _{i}-%
\frac{k}{N}\sum_{i=k+1}^{N}\epsilon _{i}\right\Vert ^{2} \\
& \leq 2\max_{1\leq k\leq N}\left[ \frac{N^{2}}{k\left( N-k\right) }\right]
^{\alpha }\frac{1}{N}\left( \frac{N-k}{N}\right) ^{2}\frac{k}{k}\left\Vert
\sum_{i=1}^{k}\epsilon _{i}\right\Vert ^{2} \\
& +2\max_{1\leq k\leq N}\left[ \frac{N^{2}}{k\left( N-k\right) }\right]
^{\alpha }\frac{1}{N}\left( \frac{k}{N}\right) ^{2}\frac{N-k}{N-k}\left\Vert
\sum_{i=k+1}^{N}\epsilon _{i}\right\Vert ^{2}=O_{P}\left( \left( \ln
N\right) ^{2/\nu }\right) ;
\end{align*}%
further, after some algebra%
\begin{align}
\frac{1}{2}\max_{1\leq k\leq N}\left\vert A_{1,N}^{k,\left( 2\right)
}\right\vert & =\max_{1\leq k\leq N}\left[ \frac{N^{2}}{k\left( N-k\right) }%
\right] ^{\alpha }\frac{1}{N}\left\langle \mathcal{M}_{k}\left( t\right) -%
\frac{k}{N}\mathcal{M}_{N}\left( t\right) ,\mathcal{W}_{k}\left( t\right) -%
\frac{k}{N}\mathcal{W}_{N}\left( t\right) \right\rangle  \notag \\
& \leq \max_{1\leq k\leq N}\left[ \frac{N^{2}}{k\left( N-k\right) }\right]
^{\alpha }\frac{1}{N}\left\Vert \mathcal{M}_{k}\left( t\right) -\frac{k}{N}%
\mathcal{M}_{N}\left( t\right) \right\Vert \left\Vert \mathcal{W}_{k}\left(
t\right) -\frac{k}{N}\mathcal{W}_{N}\left( t\right) \right\Vert  \notag \\
& \leq \max_{1\leq k\leq N}\left[ \frac{N^{2}}{k\left( N-k\right) }\right]
^{\alpha }\frac{1}{N}\left\Vert \mathcal{M}_{k}\left( t\right) -\frac{k}{N}%
\mathcal{M}_{N}\left( t\right) \right\Vert \frac{N-k}{N}\frac{k^{1/2}}{%
k^{1/2}}\left\Vert \sum_{i=1}^{k}\epsilon _{i}\right\Vert  \notag \\
&\quad +\max_{1\leq k\leq N}\left[ \frac{N^{2}}{k\left( N-k\right) }\right]
^{\alpha }\frac{1}{N}\left\Vert \mathcal{M}_{k}\left( t\right) -\frac{k}{N}%
\mathcal{M}_{N}\left( t\right) \right\Vert \frac{k}{N}\frac{\left(
N-k\right) ^{1/2}}{\left( N-k\right) ^{1/2}}\left\Vert
\sum_{i=k+1}^{N}\epsilon _{i}\right\Vert  \notag \\
& =O_{P}\left( N^{1/2}\left( \ln N\right) ^{1/\nu }\right)  \label{e:Ak2_1N}
\end{align}%
By similar passages as in Lemma \ref{refinement}, it can also be shown that $%
\max_{1\leq k\leq N}\left\vert A_{1,N}^{k,\left( h\right) }\right\vert
=O_{P}\left( 1\right) $, for $3\leq h\leq 6$. Putting all together, 
\begin{equation*}
\max_{1\leq k\leq N}\frac{1}{2}N\left( \frac{k}{N}\left( 1-\frac{k}{N}%
\right) \right) ^{2-\alpha }V_{N}^{\left( 1,N\right) }\left( k\right) \geq
\max_{1\leq k\leq N}\Theta _{1,N}^{k} -\max_{1\leq k\leq N}\ \left\vert
\sum_{h=1}^{2}A_{1 ,N}^{N,\left( h\right) }\right\vert - O_P(1)
\end{equation*}
and since $\tau _{N}^{-1}\Big(\max_{1\leq k\leq N}\Theta _{1,N}^{k} +
\left\vert \sum_{h=1}^{2}A_{1 ,N}^{N,\left( h\right) }\right\vert \Big)\geq
(N/\tau_N)\Big(c_0 - o_P(1)\Big)$, recalling that $\tau_N \to \infty$, $\tau
_{N}/N\rightarrow 0$, it follows that a changepoint is detected with
probability tending to 1, i.e.,%
\begin{align*}
\lim_{N\rightarrow \infty }P\left( \max_{1\leq k\leq N}\frac{1}{2}N\left( 
\frac{k}{N}\left( 1-\frac{k}{N}\right) \right) ^{2-\alpha }V_{N}^{\left(
1,N\right) }\left( k\right) >\tau _{N}\right) =1.
\end{align*}
Let now $\mathcal{H}=\{k_1,\ldots,k_R\}$ denote the set of all changepoints,
and let $\widehat{k}_{1}$ be defined as%
\begin{equation*}
\widehat{k}_{1}=\sargmax_{1\leq k\leq N}\frac{1}{2}N\left( \frac{k}{N}\left(
1-\frac{k}{N}\right) \right) ^{2-\alpha }V_{N}^{\left( 1,N\right) }\left(
k\right) .
\end{equation*}%
We now turn to showing that $\widehat k_1$ is consistent for some
changepoint, i.e., 
\begin{equation}
\dist\left( \widehat{k}_{1},\mathcal{H}\right) =O_{P}\left( 1\right) ,
\label{greg-1}
\end{equation}%
where $\dist\left( \widehat{k}_{1},\mathcal{H}\right)=\min_{k_i \in \mathcal{%
H}}|\widehat{k}_1-k_i|$ is the distance between $\widehat{k}_{1}$ and the
set $\mathcal{H}$ (note \eqref{greg-1} will serve as part of the base step
in the eventual induction argument). Let $\mathcal{H}_{\max }=\big\{ %
k_{i}:\Theta _{1,N}^{k_{i}}=\max_{1\leq k\leq N}\Theta _{1,N}^{k}\big\} $,
and let $0<a_{i}<b_{i}<1$ be two constants such that $\Theta _{1,N}^{k}$ is
strictly increasing over $\left\{ \left\lfloor Na_{i}\right\rfloor
,...,k_{i}\right\} $ and strictly decreasing over $\left\{
k_{i},...,\left\lfloor Nb_{i}\right\rfloor \right\} $ - these constants can
always be defined this way on account of Lemma \ref{drift}. Define the set $%
L_{N}=\cup _{i:k_{i}\in \mathcal{H}_{\max }}\left\{ \left\lfloor
Na_{i}\right\rfloor ,...,\left\lfloor Nb_{i}\right\rfloor \right\} $.
Following the arguments in the proof of Theorem 8.2.2 in \citet{chgreg}, it
follows $\lim_{N\rightarrow \infty }P\left( \widehat{k}_{1}=\widetilde{k}%
_{1}\right) =1$, where%
\begin{equation*}
\widetilde{k}_{1}=\sargmax_{k\in L_{N}}\frac{1}{2}N\left( u\left( 1-u\right)
\right) ^{2-\alpha }V_{N}^{\left( 1,N\right) }\left( k\right) .
\end{equation*}

This means that (\ref{greg-1}) can be shown if we show $\dist\left( 
\widetilde{k}_{1},\mathcal{H}\right) =O_{P}\left( 1\right) $. Let $%
I_{N,i}\left( M\right) =\left\{ k_{i-1}+M\right. $ $,...,$ $k_{i}-M,$ $%
k_{i}+M,$ $...,$ $\left. k_{i+1}-M\right\} \cap L_{N}$. Consider the cases $%
k\in \left\{ \left\lfloor Na_{i}\right\rfloor ,...,k_{i}\right\} $ and $k\in
\left\{ k_{i},...,\left\lfloor Nb_{i}\right\rfloor \right\} $. A routine
application of the Mean Value Theorem yields that there are positive
constants $c_{1}$, $c_{2}$, $c_{3}$ and $c_{4}$ such that%
\begin{align}
-c_{1}\left( k_{i}-k\right) &\leq \Theta _{1,N}^{k}-\Theta
_{1,N}^{k_{i}}\leq -c_{2}\left( k_{i}-k\right) ,\text{ \ \ for }k\in \left\{
\left\lfloor Na_{i}\right\rfloor ,...,k_{i}\right\} ,  \label{8.2.55} \\
c_{3}\left( k_{i}-k\right) &\leq \Theta _{1,N}^{k}-\Theta _{1,N}^{k_{i}}\leq
c_{4}\left( k_{i}-k\right) ,\text{ \ \ for }k\in \left\{
k_{i},...,\left\lfloor Nb_{i}\right\rfloor \right\} .  \label{8.2.56}
\end{align}%
Consider now the difference%
\begin{align}
\frac{1}{2}&N\left( u\left( 1-u\right) \right) ^{2-\alpha }V_{N}^{\left(
1,N\right) }\left( k\right) -\frac{1}{2}N\left( u\left( 1-u\right) \right)
^{2-\alpha }V_{N}^{\left( 1,N\right) }\left( k_{i}\right) \hspace{-15ex}
\label{domnate} \\
=& \Theta _{1,N}^{k}-\Theta _{1,N}^{k_{i}}+\left( \left[ \frac{N^{2}}{%
k\left( N-k\right) }\right] ^{\alpha }-\left[ \frac{N^{2}}{k_{i}\left(
N-k_{i}\right) }\right] ^{\alpha }\right) \left( \frac{1}{N}\left\Vert 
\mathcal{W}_{k}-\frac{k}{N}\mathcal{W}_{N}\right\Vert ^{2}\right) \hspace{%
-15ex}  \notag \\
& -\left[ \frac{N^{2}}{k_{i}\left( N-k_{i}\right) }\right] ^{\alpha }\left( 
\frac{1}{N}\left\Vert \mathcal{W}_{k_{i}}-\frac{k_{i}}{N}\mathcal{W}%
_{N}\right\Vert ^{2}-\frac{1}{N}\left\Vert \mathcal{W}_{k}-\frac{k}{N}%
\mathcal{W}_{N}\right\Vert ^{2}\right)  \notag \\
& +\frac{2}{N}\left( \left[ \frac{N^{2}}{k\left( N-k\right) }\right]
^{\alpha }-\left[ \frac{N^{2}}{k_{i}\left( N-k_{i}\right) }\right] ^{\alpha
}\right) \left\langle \mathcal{M}_{k}-\frac{k}{N}\mathcal{M}_{N},\mathcal{W}%
_{k}-\frac{k}{N}\mathcal{W}_{N}\right\rangle  \notag \\
& -\frac{2}{N}\left[ \frac{N^{2}}{k_{i}\left( N-k_{i}\right) }\right]
^{\alpha }\left( \left\langle \mathcal{M}_{k}-\frac{k}{N}\mathcal{M}_{N},%
\mathcal{W}_{k}-\frac{k}{N}\mathcal{W}_{N}\right\rangle -\left\langle 
\mathcal{M}_{k_{i}}-\frac{k_{i}}{N}\mathcal{M}_{N},\mathcal{W}_{k_{i}}-\frac{%
k_{i}}{N}\mathcal{W}_{N}\right\rangle \right) \hspace{-10ex}  \notag \\
& -\left\{ \left[ \frac{N^{2}}{k\left( N-k\right) }\right] ^{\alpha }\frac{N%
}{k\left( k-1\right) }\sum_{i=1}^{k}\left\Vert X_{i}\right\Vert ^{2}-\left[ 
\frac{N^{2}}{k_{i}\left( N-k_{i}\right) }\right] ^{\alpha }\frac{N}{%
k_{i}\left( k_{i}-1\right) }\sum_{i=1}^{k_{i}}\left\Vert X_{i}\right\Vert
^{2}\right\}  \notag \\
& -\left\{ \left[ \frac{N^{2}}{k\left( N-k\right) }\right] ^{\alpha }\frac{N%
}{\left( N-k\right) \left( N-k-1\right) }\sum_{i=k+1}^{N}\left\Vert
X_{i}\right\Vert ^{2}\right.  \notag \\
& -\left. \left[ \frac{N^{2}}{k_{i}\left( N-k_{i}\right) }\right] ^{\alpha }%
\frac{N}{\left( N-k_{i}\right) \left( N-k_{i}-1\right) }\sum_{i=k_{i}+1}^{N}%
\left\Vert X_{i}\right\Vert ^{2}\right\}  \notag \\
& +\left\{ \left[ \frac{N^{2}}{k\left( N-k\right) }\right] ^{\alpha }\frac{N%
}{k^{2}\left( k-1\right) }\left\Vert S_{k}\right\Vert ^{2}-\left[ \frac{N^{2}%
}{k_{i}\left( N-k_{i}\right) }\right] ^{\alpha }\frac{N}{k_{i}^{2}\left(
k_{i}-1\right) }\left\Vert S_{k_{i}}\right\Vert ^{2}\right\}  \notag \\
& +\left\{ \left[ \frac{N^{2}}{k\left( N-k\right) }\right] ^{\alpha }\frac{N%
}{\left( N-k\right) ^{2}\left( N-k-1\right) }\left\Vert
S_{N}-S_{k}\right\Vert ^{2}\right.  \notag \\
& -\left. \left[ \frac{N^{2}}{k_{i}\left( N-k_{i}\right) }\right] ^{\alpha }%
\frac{N}{\left( N-k_{i}\right) ^{2}\left( N-k_{i}-1\right) }\left\Vert
S_{N}-S_{k_{i}}\right\Vert ^{2}\right\}  \notag \\
=& \Theta _{1,N}^{k}-\Theta _{1,N}^{k_{i}}+\sum_{h=1}^{8}B_{1,N}^{k,\left(
h\right) },  \notag
\end{align}%
We will show that $\Theta _{1,N}^{k}-\Theta _{1,N}^{k_{i}}$ is the
dominating term, i.e. that%
\begin{align}
\max_{k\in I_{N,i}\left( M\right) ,k<k_{i}}\left\vert \frac{%
B_{1,N}^{k,\left( h\right) }}{\Theta _{1,N}^{k}-\Theta _{1,N}^{k_{i}}}%
\right\vert &=o_{P}\left( 1\right) ,  \label{v1} \\
\max_{k\in I_{N,i}\left( M\right) ,k\geq k_{i}}\left\vert \frac{%
B_{1,N}^{k,\left( h\right) }}{\Theta _{1,N}^{k}-\Theta _{1,N}^{k_{i}}}%
\right\vert &=o_{P}\left( 1\right) ,  \label{v2}
\end{align}%
for $1\leq h\leq 8$, using (\ref{8.2.55}) and (\ref{8.2.56}); indeed, we
will show (\ref{v1}), and (\ref{v2}) can be then derived by symmetry. In all
cases, we will use the fact that, by the Mean Value Theorem 
\begin{equation}
\left\vert \left[ \frac{N^{2}}{k\left( N-k\right) }\right] ^{\alpha }\frac{1%
}{N}-\left[ \frac{N^{2}}{k_{i}\left( N-k_{i}\right) }\right] ^{\alpha }\frac{%
1}{N}\right\vert \leq c_{0}N^{-2}\left\vert k_{i}-k\right\vert .
\label{mvt-1}
\end{equation}%
It holds that%
\begin{align*}
\max_{k\in I_{N,i}\left( M\right) ,k<k_{i}}&\left\vert \frac{%
B_{1,N}^{k,\left( 1\right) }}{\Theta _{1,N}^{k}-\Theta _{1,N}^{k_{i}}}%
\right\vert \\
& \hspace{-5ex}\leq C\max_{k\in I_{N,i}\left( M\right) ,k<k_{i}}\frac{1}{%
k_{i}-k}\left( \left[ \frac{N^{2}}{k\left( N-k\right) }\right] ^{\alpha }-%
\left[ \frac{N^{2}}{k_{i}\left( N-k_{i}\right) }\right] ^{\alpha }\right) 
\frac{1}{N}\left( \left\Vert \mathcal{W}_{k}-\frac{k}{N}\mathcal{W}%
_{N}\right\Vert ^{2}\right) \\
& \hspace{-5ex}\leq C\max_{k\in I_{N,i}\left( M\right) ,k<k_{i}}\frac{1}{%
k_{i}-k}N^{-2}\left\vert k_{i}-k\right\vert \left\Vert \mathcal{W}_{k}-\frac{%
k}{N}\mathcal{W}_{N}\right\Vert ^{2} \\
& \hspace{-5ex}\leq C\left( N^{-2}\max_{k\in I_{N,i}\left( M\right)
,k<k_{i}}\left\Vert \mathcal{W}_{k}\right\Vert ^{2}+N^{-2}\max_{k\in
I_{N,i}\left( M\right) ,k<k_{i}}\left\Vert \frac{k}{N}\mathcal{W}%
_{N}\right\Vert ^{2}\right) =O_{P}\left( \frac{\left( \ln N\right) ^{2/\nu }%
}{N}\right) ,
\end{align*}%
having used (\ref{mineq1}) and (\ref{mineq2}). Also, noting that $\mathcal{W}%
_{N}=O_{P}\left( N^{1/2}\right) $ and 
\begin{equation*}
\max_{k\in I_{N,i}\left( M\right) ,k<k_{i}}\frac{1}{\left( k_{i}-k\right)
^{1/2}}\left\Vert \sum_{i=k+1}^{k_{i}}\epsilon _{i}\right\Vert =O_{P}\left(
\left( \ln N\right) ^{1/\nu }\right) ,
\end{equation*}%
and noting that $\left\vert k_{i}-k\right\vert \geq M$, we have%
\begin{align*}
& \max_{k\in I_{N,i}\left( M\right) ,k<k_{i}}\left\vert \frac{%
B_{1,N}^{k,\left( 2\right) }}{\Theta _{1,N}^{k}-\Theta _{1,N}^{k_{i}}}%
\right\vert \\
\leq & C\max_{k\in I_{N,i}\left( M\right) ,k<k_{i}}\frac{1}{N\left(
k_{i}-k\right) }\left( \left\Vert \mathcal{W}_{k_{i}}-\mathcal{W}%
_{k}\right\Vert +\left\Vert \frac{k_{i}-k}{N}\mathcal{W}_{N}\right\Vert
\right) \left( \left\Vert \mathcal{W}_{k}\right\Vert +\left\Vert \mathcal{W}%
_{k_{i}}\right\Vert +\left\Vert \frac{k_{i}+k}{N}\mathcal{W}_{N}\right\Vert
\right) \\
\leq & C\max_{k\in I_{N,i}\left( M\right) ,k<k_{i}}\frac{1}{N\left(
k_{i}-k\right) }\left\Vert \mathcal{W}_{k_{i}}-\mathcal{W}_{k}\right\Vert
\left( \left\Vert \mathcal{W}_{k}\right\Vert +\left\Vert \mathcal{W}%
_{k_{i}}\right\Vert \right) \\
& +C\max_{k\in I_{N,i}\left( M\right) ,k<k_{i}}\frac{1}{N\left(
k_{i}-k\right) }\left\Vert \mathcal{W}_{k_{i}}-\mathcal{W}_{k}\right\Vert
\left\Vert \frac{k_{i}+k}{N}\mathcal{W}_{N}\right\Vert \\
& +C\max_{k\in I_{N,i}\left( M\right) ,k<k_{i}}\frac{1}{N\left(
k_{i}-k\right) }\left\Vert \frac{k_{i}-k}{N}\mathcal{W}_{N}\right\Vert
\left( \left\Vert \mathcal{W}_{k}\right\Vert +\left\Vert \mathcal{W}%
_{k_{i}}\right\Vert \right) \\
& +C\max_{k\in I_{N,i}\left( M\right) ,k<k_{i}}\frac{1}{N\left(
k_{i}-k\right) }\left\Vert \frac{k_{i}-k}{N}\mathcal{W}_{N}\right\Vert
\left\Vert \frac{k_{i}+k}{N}\mathcal{W}_{N}\right\Vert \\
=& M^{-1/2}O_{P}\left( \frac{\left( \ln N\right) ^{1/\nu }}{N^{1/2}}\right)
+O_{P}\left( \frac{1}{N}\right) =o_{P}\left( 1\right) .
\end{align*}%
Similarly%
\begin{align*}
\frac{1}{2}&\max_{k\in I_{N,i}\left( M\right) ,k<k_{i}}\left\vert \frac{%
B_{1,N}^{k,\left( 3\right) }}{\Theta _{1,N}^{k}-\Theta _{1,N}^{k_{i}}}%
\right\vert \\
& \leq c_{0}\max_{k\in I_{N,i}\left( M\right) ,k<k_{i}}\frac{1}{k_{i}-k}%
N^{-2}\left\vert k_{i}-k\right\vert \left\Vert \mathcal{M}_{k}-\frac{k}{N}%
\mathcal{M}_{N}\right\Vert \left\Vert \mathcal{W}_{k}-\frac{k}{N}\mathcal{W}%
_{N}\right\Vert \\
& \leq c_{0}\max_{k\in I_{N,i}\left( M\right) ,k<k_{i}}N^{-1}\left(
\left\Vert \mathcal{W}_{k}\right\Vert +\frac{k}{N}\left\Vert \mathcal{W}%
_{N}\right\Vert \right) \\
& =O_{P}\left( \frac{\left( \ln N\right) ^{1/\nu }}{N^{1/2}}\right)
+O_{P}\left( N^{-1/2}\right) =o_{P}\left( 1\right) .
\end{align*}%
We now write%
\begin{align*}
B_{1,N}^{k,\left( 4\right) }=& \frac{2}{N}\left[ \frac{N^{2}}{k_{i}\left(
N-k_{i}\right) }\right] ^{\alpha }\left\langle \mathcal{M}_{k_{i}}-\frac{%
k_{i}}{N}\mathcal{M}_{N},\left( \mathcal{W}_{k}-\frac{k}{N}\mathcal{W}%
_{N}\right) -\left( \mathcal{W}_{k_{i}}-\frac{k_{i}}{N}\mathcal{W}%
_{N}\right) \right\rangle \\
& +\frac{2}{N}\left[ \frac{N^{2}}{k_{i}\left( N-k_{i}\right) }\right]
^{\alpha }\left\langle \left( \mathcal{M}_{k}-\frac{k}{N}\mathcal{M}%
_{N}\right) -\left( \mathcal{M}_{k_{i}}-\frac{k_{i}}{N}\mathcal{M}%
_{N}\right) ,\mathcal{W}_{k}-\frac{k}{N}\mathcal{W}_{N}\right\rangle \\
=& B_{1,N,1}^{k,\left( 4\right) }+B_{1,N,2}^{k,\left( 4\right) },
\end{align*}%
and study%
\begin{align*}
& \frac{1}{2}\max_{k\in I_{N,i}\left( M\right) ,k<k_{i}}\left\vert \frac{%
B_{1,N,1}^{k,\left( 4\right) }}{\Theta _{1,N}^{k}-\Theta _{1,N}^{k_{i}}}%
\right\vert \\
& \leq C\max_{k\in I_{N,i}\left( M\right) ,k<k_{i}}\Bigg[\frac{1}{N\left(
k_{i}-k\right) }\left[ \frac{N^{2}}{k_{i}\left( N-k_{i}\right) }\right]
^{\alpha } \\
& \qquad\qquad\qquad\qquad\left\langle \mathcal{M}_{k_{i}}-\frac{k_{i}}{N}%
\mathcal{M}_{N},\left( \mathcal{W}_{k}-\frac{k}{N}\mathcal{W}_{N}\right)
-\left( \mathcal{W}_{k_{i}}-\frac{k_{i}}{N}\mathcal{W}_{N}\right)
\right\rangle \Bigg] \\
& \leq C\max_{k\in I_{N,i}\left( M\right) ,k<k_{i}}\frac{1}{N\left(
k_{i}-k\right) }\left\Vert \mathcal{M}_{k_{i}}-\frac{k_{i}}{N}\mathcal{M}%
_{N}\right\Vert \left\Vert \mathcal{W}_{k}-\mathcal{W}_{k_{i}}\right\Vert \\
& \qquad+ C\max_{k\in I_{N,i}\left( M\right) ,k<k_{i}}\frac{1}{N\left(
k_{i}-k\right) }\frac{k_{i}-k}{N}\left\Vert \mathcal{M}_{k_{i}}-\frac{k_{i}}{%
N}\mathcal{M}_{N}\right\Vert \left\Vert \mathcal{W}_{N}\right\Vert \\
& =C\max_{k\in I_{N,i}\left( M\right) ,k<k_{i}}\left( k_{i}-k\right) ^{\zeta
-1}+O_{P}\left( N^{-1/2}\right) =M^{\zeta -1}O_{P}\left( 1\right)
+O_{P}\left( N^{-1/2}\right) ,
\end{align*}%
for some $1/2<\zeta <1$, having noted that 
\begin{equation*}
\left\Vert \mathcal{M}_{k_{i}}-\frac{k_{i}}{N}\mathcal{M}_{N}\right\Vert
=c_{0}k_{i}\frac{N-k_{i}}{N},
\end{equation*}%
for some positive $c_{0}$, and having used the fact that, by Lemma \ref%
{l:BHR_33}, for all $\zeta >1/2$ we have%
\begin{equation*}
\max_{j\leq k\leq l}\frac{1}{\left( k-j\right) ^{\zeta }}\sum_{i=j+1}^{k}%
\epsilon _{i}=O_{P}\left( 1\right) .
\end{equation*}%
(\ref{mineq1}) and (\ref{mineq2}) and the definition of $M$. Hence it holds
that, for all $x>0$%
\begin{equation}
\lim_{M\rightarrow \infty }\limsup_{N\rightarrow \infty }P\left( \frac{1}{2}%
\max_{k\in I_{N,i}\left( M\right) ,k<k_{i}}\left\vert \frac{%
B_{1,N,1}^{k,\left( 4\right) }}{\Theta _{1,N}^{k}-\Theta _{1,N}^{k_{i}}}%
\right\vert >x\right) =0.  \label{b141}
\end{equation}%
Similarly%
\begin{align*}
\frac{1}{2}&\max_{k\in I_{N,i}\left( M\right) ,k<k_{i}}\left\vert \frac{%
B_{1,N,2}^{k,\left( 4\right) }}{\Theta _{1,N}^{k}-\Theta _{1,N}^{k_{i}}}%
\right\vert \\
&\leq c_{0}\max_{k\in I_{N,i}\left( M\right) ,k<k_{i}}\frac{1}{N\left(
k_{i}-k\right) }\left\Vert \mathcal{M}_{k}-\mathcal{M}_{k_{i}}\right\Vert
\left\Vert \mathcal{W}_{k}-\frac{k}{N}\mathcal{W}_{N}\right\Vert \\
&+c_{0}\max_{k\in I_{N,i}\left( M\right) ,k<k_{i}}\frac{1}{N\left(
k_{i}-k\right) }\frac{k_{i}-k}{N}\left\Vert \mathcal{M}_{N}\right\Vert
\left\Vert \mathcal{W}_{k}-\frac{k}{N}\mathcal{W}_{N}\right\Vert
=O_{P}\left( \frac{\left( \ln N\right) ^{1/\nu }}{N^{1/2}}\right) ,
\end{align*}%
having used the fact that, for some positive $c_{0}$, $\left\Vert \mathcal{M}%
_{k}-\mathcal{M}_{k_{i}}\right\Vert =c_{0}\left\vert k_{i}-k\right\vert $,
whence%
\begin{equation}
\frac{1}{2}\max_{k\in I_{N,i}\left( M\right) ,k<k_{i}}\left\vert \frac{%
B_{1,N,1}^{k,\left( 4\right) }}{\Theta _{1,N}^{k}-\Theta _{1,N}^{k_{i}}}%
\right\vert =O_{P}\left( \frac{1}{N^{1/2}}\right) .  \label{b142}
\end{equation}%
Combining (\ref{b141}) and (\ref{b142}), it follows that%
\begin{equation*}
\frac{1}{2}\max_{k\in I_{N,i}\left( M\right) ,k<k_{i}}\left\vert \frac{%
B_{1,N}^{k,\left( 4\right) }}{\Theta _{1,N}^{k}-\Theta _{1,N}^{k_{i}}}%
\right\vert =o_{P}\left( 1\right) .
\end{equation*}%
Continuing our proof, we have%
\begin{align*}
B_{1,N}^{k,\left( 5\right) }=& -\left\{ \left[ \frac{N^{2}}{k\left(
N-k\right) }\right] ^{\alpha }\frac{N}{k\left( k-1\right) }-\left[ \frac{%
N^{2}}{k_{i}\left( N-k_{i}\right) }\right] ^{\alpha }\frac{N}{k_{i}\left(
k_{i}-1\right) }\right\} \sum_{i=1}^{k}\left\Vert X_{i}\right\Vert ^{2} \\
& +\left[ \frac{N^{2}}{k_{i}\left( N-k_{i}\right) }\right] ^{\alpha }\frac{N%
}{k_{i}\left( k_{i}-1\right) }\sum_{i=k+1}^{k_{i}}\left\Vert
X_{i}\right\Vert ^{2}=B_{1,N,1}^{k,\left( 5\right) }+B_{1,N,2}^{k,\left(
5\right) }.
\end{align*}%
Using the Mean Value Theorem, 
\begin{equation*}
\left\vert \left[ \frac{N^{2}}{k\left( N-k\right) }\right] ^{\alpha }\frac{N%
}{k\left( k-1\right) }-\left[ \frac{N^{2}}{k_{i}\left( N-k_{i}\right) }%
\right] ^{\alpha }\frac{N}{k_{i}\left( k_{i}-1\right) }\right\vert \leq c_{0}%
\frac{k_{i}-k}{N^{2}},
\end{equation*}%
and therefore 
\begin{equation*}
\max_{k\in I_{N,i}\left( M\right) ,k<k_{i}}\left\vert \frac{%
B_{1,N,1}^{k,\left( 5\right) }}{\Theta _{1,N}^{k}-\Theta _{1,N}^{k_{i}}}%
\right\vert \leq \max_{k\in I_{N,i}\left( M\right) ,k<k_{i}}\frac{1}{k_{i}-k}%
\frac{k_{i}-k}{N^{2}}\sum_{i=1}^{k}\left\Vert X_{i}\right\Vert
^{2}=O_{P}\left( \frac{1}{N}\right) ,
\end{equation*}%
noting that $\sum_{i=1}^{k}\left\Vert X_{i}\right\Vert ^{2}\leq
\sum_{i=1}^{N}\left\Vert X_{i}\right\Vert ^{2}=O_{P}\left( N\right) $ by the
ergodic theorem. Also%
\begin{equation*}
\max_{k\in I_{N,i}\left( M\right) ,k<k_{i}}\left\vert \frac{%
B_{1,N,2}^{k,\left( 5\right) }}{\Theta _{1,N}^{k}-\Theta _{1,N}^{k_{i}}}%
\right\vert \leq c_{0}\max_{k\in I_{N,i}\left( M\right) ,k<k_{i}}\frac{1}{%
N\left( k_{i}-k\right) }\sum_{i=k+1}^{k_{i}}\left\Vert X_{i}\right\Vert
^{2}=O_{P}\left( \frac{1}{N}\right) ,
\end{equation*}%
using similar arguments as in (\ref{e:X2_Op1}) to show that $%
\sum_{i=k+1}^{k_{i}}\left\Vert X_{i}\right\Vert ^{2}=O_{P}\left(
k_{i}-k\right) $. Hence%
\begin{equation*}
\max_{k\in I_{N,i}\left( M\right) ,k<k_{i}}\left\vert \frac{%
B_{1,N}^{k,\left( 5\right) }}{\Theta _{1,N}^{k}-\Theta _{1,N}^{k_{i}}}%
\right\vert =O_{P}\left( \frac{1}{N}\right) =o_{P}\left( 1\right) ,
\end{equation*}%
and the same can be shown for $B_{1,N}^{k,\left( 6\right) }$. Finally we
study%
\begin{align*}
B_{1,N}^{k,\left( 7\right) } &=\left\{ \left[ \frac{N^{2}}{k\left(
N-k\right) }\right] ^{\alpha }\frac{N}{k^{2}\left( k-1\right) }-\left[ \frac{%
N^{2}}{k_{i}\left( N-k_{i}\right) }\right] ^{\alpha }\frac{N}{%
k_{i}^{2}\left( k_{i}-1\right) }\right\} \left\Vert S_{k}\right\Vert ^{2} \\
&\qquad-\left[ \frac{N^{2}}{k_{i}\left( N-k_{i}\right) }\right] ^{\alpha }%
\frac{N}{k_{i}^{2}\left( k_{i}-1\right) }\left( \left\Vert
S_{k_{i}}\right\Vert ^{2}-\left\Vert S_{k}\right\Vert ^{2}\right)
=B_{1,N,1}^{k,\left( 7\right) }+B_{1,N,2}^{k,\left( 7\right) }.
\end{align*}%
The Mean Value Theorem yields%
\begin{equation*}
\left\vert \left[ \frac{N^{2}}{k\left( N-k\right) }\right] ^{\alpha }\frac{N%
}{k^{2}\left( k-1\right) }-\left[ \frac{N^{2}}{k_{i}\left( N-k_{i}\right) }%
\right] ^{\alpha }\frac{N}{k_{i}^{2}\left( k_{i}-1\right) }\right\vert \leq
c_{0}\frac{k_{i}-k}{N^{3}},
\end{equation*}%
and therefore%
\begin{equation*}
\max_{k\in I_{N,i}\left( M\right) ,k<k_{i}}\left\vert \frac{%
B_{1,N,1}^{k,\left( 7\right) }}{\Theta _{1,N}^{k}-\Theta _{1,N}^{k_{i}}}%
\right\vert \leq c_{0}\max_{k\in I_{N,i}\left( M\right) ,k<k_{i}}\frac{1}{%
k_{i}-k}\frac{k_{i}-k}{N^{3}}\left\Vert S_{k}\right\Vert ^{2}=O_{P}\left( 
\frac{1}{N}\right) ,
\end{equation*}%
using the bound $\max_{1\leq k\leq N}\left\Vert S_{k}\right\Vert
=O_{P}\left( N\right) $. Further%
\begin{equation*}
\max_{k\in I_{N,i}\left( M\right) ,k<k_{i}}\left\vert \frac{%
B_{1,N,2}^{k,\left( 7\right) }}{\Theta _{1,N}^{k}-\Theta _{1,N}^{k_{i}}}%
\right\vert \leq c_{0}\max_{k\in I_{N,i}\left( M\right) ,k<k_{i}}\frac{1}{%
N^{2}\left( k_{i}-k\right) }\left\Vert S_{k}+S_{k_{i}}\right\Vert \left\Vert
S_{k_{i}}-S_{k}\right\Vert =O_{P}\left( \frac{1}{N}\right) ,
\end{equation*}%
recalling that $\max_{1\leq k\leq N}\left\Vert S_{k}\right\Vert =O_{P}\left(
N\right) $ and using (\ref{m-ineq-2}). Thus%
\begin{equation*}
\max_{k\in I_{N,i}\left( M\right) ,k<k_{i}}\left\vert \frac{%
B_{1,N}^{k,\left( 7\right) }}{\Theta _{1,N}^{k}-\Theta _{1,N}^{k_{i}}}%
\right\vert =o_{P}\left( 1\right) ,
\end{equation*}%
and the sample applies to $B_{1,N}^{k,\left( 8\right) }$. Putting all
together, we have shown (\ref{v1}). The inequalities (\ref{8.2.55}) and (\ref%
{8.2.56}) imply that, for $i$ such that $k_{i}\in \mathcal{H}_{\max }$,
there are constants $c_{i}>0$ such that $\max_{k\in I_{N,i}\left( M\right)
,k<k_{i}}\left( \Theta _{1,N}^{k}-\Theta _{1,N}^{k_{i}}\right) \leq -c_{i}M$%
, and therefore we have%
\begin{equation}
\lim_{M\rightarrow \infty }\limsup_{N\rightarrow \infty }\max_{k\in
I_{N,i}\left( M\right) ,k<k_{i}}\left( \Theta _{1,N}^{k}-\Theta
_{1,N}^{k_{i}}\right) =-\infty .  \label{dominate2}
\end{equation}%
On the other hand, 
\begin{equation}
\max_{k\in L_{N}}\left( \frac{1}{2}N\left( u\left( 1-u\right) \right)
^{2-\alpha }V_{N}^{\left( 1,N\right) }\left( k\right) -\frac{1}{2}N\left(
u\left( 1-u\right) \right) ^{2-\alpha }V_{N}^{\left( 1,N\right) }\left(
k_{i}\right) \right) \geq 0,  \label{dominate3}
\end{equation}%
for all $k_{i}\in \mathcal{H}_{\max }$, and therefore %
\begin{align*}
&P\left( \dist\left( \widetilde{k}_{1},\mathcal{H}\right)>M \right) \\
&\leq\hspace{-3mm} \sum_{i:k_{i}\in \mathcal{H}_{\max }}P\Bigg( \max_{k\in
I_{N,i}\left( M\right) }\bigg( \frac{1}{2}N\left( u\left( 1-u\right) \right)
^{2-\alpha }V_{N}^{\left( 1,N\right) }\left( k\right) - \frac{1}{2}N\left(
u\left( 1-u\right) \right) ^{2-\alpha }V_{N}^{\left( 1,N\right) }\left(
k_{i}\right)\geq 0 \bigg)\Bigg)
\end{align*}%
%
%
Combining (\ref{dominate2}), (\ref{dominate3}), and (\ref{v1}), we obtain%
\begin{align*}
& \lim_{M\rightarrow \infty }\limsup_{N\rightarrow \infty }P\left(
\max_{k\in I_{N,i}\left( M\right) }\left( \frac{1}{2}N\left( u\left(
1-u\right) \right) ^{2-\alpha }V_{N}^{\left( 1,N\right) }\left( k\right)
\right. \right. \\
& \qquad\qquad\qquad\qquad\qquad\qquad-\left. \left. \frac{1}{2}N\left(
u\left( 1-u\right) \right) ^{2-\alpha }V_{N}^{\left( 1,N\right) }\left(
k_{i}\right)\right) \geq 0 \right) =0,
\end{align*}%
whence finally $\dist\left( \widetilde{k}_{1},\mathcal{H}\right)
=O_{P}\left( 1\right) $, thus implying (\ref{greg-1}). Further note that, as
a consquence of \eqref{greg-1} and the fact that $k_{j}=\left\lfloor N\theta
_{j}\right\rfloor$, 
\begin{equation}  \label{e:khat_wellsep}
\lim_{N\rightarrow \infty }P\left( \min \{\widehat{k}_{1}, N-\widehat{k}_1\}
>a^{\prime }N\right) =1.
\end{equation}
for some $a^{\prime }\in(0,1)$. The proof now proceeds by induction, making
use of \eqref{lemma823}, and essentially by repeating the same arguments as
in the proof of Theorem 8.2.2 in \citet{chgreg}, which we summarize
hereafter for the sake of a complete discussion. Now, let $\widehat{k}_i(r)$
denote the $i$--th changepoint in increasing order appearing at the $r$--th
iteration of the binary segmentation algorithm. By way of induction, assume
that, at the $r$--th step, $1\leq r\leq R$ changepoints have been estimated, 
$1=\widehat{k}_{0}(r)<\widehat{k}_{1}(r)<\ldots <\widehat{k}_{r}(r)<\widehat{%
k}_{r+1}(r)=N$, satisfying 
\begin{equation}
\max_{1\leq i\leq r}\dist\left( \widehat{k}_{i}(r),\mathcal{H}\right)
=O_{P}\left( 1\right) ,  \label{ind1}
\end{equation}%
and as in \eqref{e:khat_wellsep}, the $\widehat{k}_i(r)$ are well-separated,
i.e., 
\begin{equation}
\lim_{N\rightarrow \infty }P\left( \min_{0\leq i\leq r}\widehat{k}_{i+1}(r)-%
\widehat{k}_{i}(r)>a^{\prime }N\right) =1,  \label{ind2}
\end{equation}%
for some $a^{\prime }\in \left( 0,1\right) $ (note \eqref{ind2} implies the $%
\widehat{k}_i(r)$, $1\leq i\leq r$ are concentrated around $r$ distinct
changepoints and remain separated from the boundary). Under (\ref{ind1}) and
(\ref{ind2}), following arguments analogous to those for $A_{1,N}^{k,\left(
h\right)}$, $h=1,\ldots,6$, it holds that%
\begin{equation}
\max_{0\leq i\leq r}\max_{\widehat{k}_{i}(r)\leq k\leq \widehat{k}%
_{i+1}(r)}\left\Vert \sum_{h=1}^{6}A_{\widehat{k}_{i}(r),\widehat{k}%
_{i+1}(r)}^{k,\left( h\right)}\right\Vert =O_{P}\left( N^{1/2}\left( \ln
N\right) ^{1/\nu }\right)  \label{ind3}
\end{equation}%
(c.f.~\eqref{e:Ak2_1N}). Consider now a sequence $i_{N}$ such that $%
N^{1/2}\left( \ln N\right) ^{1/\nu }/i_{N}+i_{N}/\tau _{N}\rightarrow 0$ as $%
N\rightarrow \infty $, and define the events 
\begin{equation*}
\mathcal{B}_{N,r}\left( N^{\prime }\right) =\left\{ \max_{1\leq i\leq r}\dist%
\left( \widehat{k}_{i}(r),\mathcal{H}\right) \leq N^{\prime }\right\} ,
\end{equation*}
and 
\begin{equation*}
\mathcal{A}_{N,r}=\left\{ \max_{0\leq i\leq r}\max_{\widehat{k}_{i}\leq
k\leq \widehat{k}_{i+1}}\left\Vert \sum_{h=1}^{6} A_{\widehat{k}_{i}(r),%
\widehat{k}_{i+1}(r)}^{k,\left( h\right) }\right\Vert \leq i_{N}\right\} .
\end{equation*}%
Recall the result in (\ref{lemma823}), using $a_{N}=i_{N}$ and $%
f_{N}=N^{\prime }$. From the induction hypothesis \eqref{ind1}, 
\begin{equation}  \label{e:B(N')_to_1}
\lim_{N^{\prime }\to\infty} \liminf_{N\to\infty}P\left( \mathcal{B}%
_{N,r}\left( N^{\prime }\right) \right) =1.
\end{equation}
Now, if $r=R$, it holds that as $N^{\prime },N\rightarrow \infty $ 
\begin{equation*}
P\left( \mathcal{A}_{N,r}\cap \mathcal{B}_{N,r}\left( N^{\prime }\right)
\right) \to 1
\end{equation*}%
Thus by \eqref{lemma823}, 
\begin{equation*}
\max_{\ell <k<u}\left\vert \frac{1}{2}\left( u-\ell \right) \left( \frac{%
\left( k-\ell \right) \left( u-k\right) }{\left( u-\ell \right) ^{2}}\right)
^{2-\alpha }V_{N}^{\left( \ell ,u\right) }\left( k\right) \right\vert \leq
c_{0}\max \left\{ i_N,N^{\prime }\right\}
\end{equation*}
with probability tending to 1 as $N,N^{\prime }\to\infty$. This implies that
the procedure terminates on the set $\mathcal{A}_{N,r}\cap \mathcal{B}%
_{N,r}\left( N^{\prime }\right) $, because $i_{N}=o\left( \tau _{N}\right) $
and $N^{\prime }$ can be chosen to diverge simultaneously with $N$
arbitrarily slowly, giving $c_{0}\max \left\{ i_N,N^{\prime }\right\}<\tau_N$
for all large $N$.

If instead we have $r<R$, then, on one of the subsegments determined by $%
\ell =\widehat{k}_{i}(r)$ and $u=\widehat{k}_{i+1}(r)$, equation \eqref{8222}
must hold with some $\zeta<1/2$ for all large $N$ (indeed, if not, then
since $r<R$, there would be at least one $i$, $1\leq i \leq R$, with two
changepoints in the interval $(\widehat{k}_i(r)-m_N,\widehat{k}_i(r) +m_N)$,
but this interval has length $2m_N <2\zeta \min\{k_i-k_{i-1} \} +O(1/N)<
\min\{k_i-k_{i-1} \}$ for all large $N$). Denote by $(\ell_*,u_*)=(\widehat{k%
}_{i_*}(r),\widehat{k}_{i_*+1}(r))$ where $i_*$ is the smallest index (say)
on which \eqref{8222} holds. From Lemma \ref{drift-2}, we have $%
\max_{\ell_*\leq k\leq u_*}\Theta _{\ell_*,u_*}^{k}\geq c_{0}N$ almost
surely, and arguing similarly to the case of $u=1,\ell=N$, another
changepoint is detected in the interval $(\ell_*,u_*)$ with probability
tending to 1. It now remains to show the next estimated changepoint remains
a bounded distance from $\mathcal{H}$ and is well-separated from the
previous estimates $\widehat{k}_i(r)$. To this end, let 
\begin{equation*}
\widehat{k}^{\ast }=\sargmax_{\ell_* <k<u_*}\frac{1}{2}\left( u_*-\ell_*
\right) \left( \frac{\left( k-\ell_* \right) \left( u_*-k\right) }{\left(
u_*-\ell_* \right) ^{2}}\right) ^{2-\alpha }V_{N}^{\left( \ell_* ,u_*\right)
}\left( k\right) ,
\end{equation*}%
and define the event%
\begin{equation*}
\mathcal{B}_{N}^{\ast }\left( N^{\ast }\right) =\left\{ \dist\left( \widehat{%
k}^{\ast },\mathcal{H}\right) >N^{\ast },\min_{0\leq i\leq r+1}\left\vert 
\widehat{k}^{\ast }-\widehat{k}_{i}(r)\right\vert >a^{\prime }N\right\}.
\end{equation*}%
with some $N^*>0.$ Note that on the event $\mathcal{B}_{N,r}(N^{\prime })$,
both $\ell_*$ and $u_*$ are within $N^{\prime }$ distance of $\mathcal{H}$,
i.e., $\min_{1\leq i<i^{\prime }\leq R}\max\{|\ell_*-k_i|,|u_*-k_{i^{\prime
}}|\}<N^{\prime }$. Using \eqref{e:B(N')_to_1}, it holds that, for each $%
\epsilon >0$ we may fix a large $N^{\prime }$ such that for all large $N$%
\begin{align}
P\left( \mathcal{B}_{N}^{\ast }\left( N^{\ast }\right) \right) &\leq P\left( 
\mathcal{B}_{N}^{\ast }\left( N^{\ast }\right) \cap \mathcal{B}_{N,r}\left(
N^{\prime }\right) \right) +\epsilon  \notag \\
&\leq P\left( \mathcal{B}_{N}^{\ast }\left( N^{\ast }\right) \cap \Big\{%
\min_{1\leq i<i^{\prime }\leq R}\max\{|\ell_*-k_i|,|u_*-k_{i^{\prime
}}|\}<N^{\prime }\Big\}\right) +\epsilon  \notag \\
&\leq \sum_{1\leq i<i^{\prime }\leq R }\sum_{\left\{ (\ell,u): \left\vert
\ell -k_{i}\right\vert <N^{\prime }\left\vert u-k_{i^{\prime }}\right\vert
<N^{\prime }\right\} }P\left( \mathcal{B}_{N}^{\ast }\left( N^{\ast }\right)
,\ell_* =\ell ,u_*=u\right) +\epsilon .  \label{eP(BNstar)_bound}
\end{align}
Repeating the proof of (\ref{greg-1}), it can be shown that, for all $%
(\ell,u)$ in the set $\{ (\ell,u):\left\vert \ell -k_{i}\right\vert
<N^{\prime },\left\vert u-k_{i^{\prime }}\right\vert <N^{\prime }\} $, it
holds that $\lim_{N^{\ast }\rightarrow \infty }\limsup_{N\rightarrow \infty
}P\left( \mathcal{B}_{N}^{\ast }\left( N^{\ast }\right) ,\ell_* =\ell ,u_*=u
\right) =0$. Since $\big|\!\left\{ (\ell,u): \left\vert \ell
-k_{i}\right\vert <N^{\prime }\left\vert u-k_{i^{\prime }}\right\vert
<N^{\prime }\right\}\!\big|< 4N^{\prime }$, and $N^{\prime }$ is fixed, from %
\eqref{eP(BNstar)_bound} we obtain the limit $\lim_{N^{\ast }\rightarrow
\infty }\limsup_{N\rightarrow \infty }P\left( \mathcal{B}_{N}^{\ast }\left(
N^{\ast }\right) \right) <\epsilon$ for every $\epsilon>0$; in turn, this
implies
\begin{equation*}
\dist\left( \widehat{k}^{\ast },\mathcal{H}\right) =O_{P}\left( 1\right)
\quad\text{and}\quad \lim_{N\rightarrow \infty }P\left( \min_{0\leq i\leq r}%
\widehat{k}^{\ast }-\widehat{k}_{i}(r)>a^{\prime }N\right) =1,
\end{equation*}%
implying $\widehat{k}^*$ is concentrated around a new changepoint, yielding
the final result.
\end{proof}

\begin{proof}[Proof of Lemma \protect\ref{vcv}]
The lemma can be shown by following, with minor modifications, the arguments
in \citet{horvath:kokoszka:reeder:2013} and \citet{berkes:horvath:rice:2016}.
\end{proof}

\begin{proof}[Proof of Lemma \protect\ref{x-bernoulli}]
We begin by noting that, on account of Assumption \ref{as:Yk}\textit{(iii)} 
\begin{equation}
E\left\vert \xi _{1,\ell }-\xi _{1,\ell }^{(m)}\right\vert ^{\beta }\leq
\Vert Y_{j}-Y_{j}^{(m)}\Vert ^{\beta }\leq Cm^{-\alpha _{0}}.
\label{e:xi-xi^(m)_bound}
\end{equation}%
Using the elementary inequality $|\exp (\mathbf{i}x)-\exp (\mathbf{i}y)|\leq
\min \{2,|x-y|\}$, valid for all $x,y\in \mathbb{R}$, it holds that 
\begin{equation*}
\Vert X_{1}-X_{1}^{(m)}\Vert \leq C\sum_{\ell =1}^{d}\min \left\{ 2,|\xi
_{1,\ell }-\xi _{1,\ell }^{(m)}|\right\} .
\end{equation*}%
Now, for $0<u<2$, 
\begin{align*}
& E\min \left\{ 2,|\xi _{1,\ell }-\xi _{1,\ell }^{(m)}|\right\} ^{\gamma } \\
& =E\min \left\{ 2,|\xi _{1,\ell }-\xi _{1,\ell }^{(m)}|\right\} ^{\gamma
}\left( \mathbf{1}_{\{|\xi _{1,\ell }-\xi _{1,\ell }^{(m)}|>u\}}+\mathbf{1}%
_{\{|\xi _{1,\ell }-\xi _{1,\ell }^{(m)}|\leq u\}}\right) \\
& \leq 2^{\gamma }P\left( |\xi _{1,\ell }-\xi _{1,\ell }^{(m)}|>u\right)
+u^{\gamma }\leq 2^{\gamma }u^{-\beta }E|\xi _{1,\ell }-\xi _{1,\ell
}^{(m)}|^{\beta }+u^{\gamma }\leq Cu^{-\beta }m^{-\alpha _{0}}+u^{\gamma },
\end{align*}%
where we used \eqref{e:xi-xi^(m)_bound} on the fourth line above. By picking 
$u=m^{-\alpha _{0}/(\beta +\gamma )}$, we obtain $E\min \left\{ 2,|\xi
_{1,\ell }-\xi _{1,\ell }^{(m)}|\right\} ^{\gamma }\leq Cm^{-\gamma \alpha
_{0}/(\beta +\gamma )}$. Since $\gamma \alpha _{0}/(\beta +\gamma )>2$ if
and only if $\gamma >2\beta /(\alpha _{0}-2)$, the desired result follows.
\end{proof}

\end{document}